\documentclass{report}
\usepackage{graphicx} % Required for inserting images
\usepackage{amsmath}
\usepackage{amsfonts}
\usepackage{amssymb}
\usepackage{amsthm}
\usepackage{tikz-cd}
\usetikzlibrary{matrix,fit,backgrounds}
\usepackage[dvipsnames]{xcolor}
\usepackage{mathrsfs}
\usepackage{mathdots}
\usepackage{mathtools}
\usepackage{comment}
\usepackage[table]{xcolor}
\usepackage{array}
\usepackage{algorithm}
\usepackage{algpseudocode}  % breakable algorithmic environment
\usepackage{caption}
\usepackage[most]{tcolorbox}

\newtcolorbox[use counter=algorithm, number within=section]{BreakableAlgorithm}[2][]{%
  enhanced,
  breakable,
  colback=white,
  colframe=black,
  boxrule=0.8pt,
  arc=2pt,
  left=8pt,right=8pt,top=8pt,bottom=8pt,
  attach title to upper,
  title={Algorithm~\thetcbcounter: #2},
  title after break={Algorithm~\thetcbcounter\ (continued): #2},
  fonttitle=\bfseries,
  colbacktitle=white,
  coltitle=black,
  titlerule=0.8pt,
  borderline north={1pt}{0pt}{black},
  borderline south={1pt}{0pt}{black},
  #1
}

\algrenewcommand\alglinenumber[1]{\scriptsize\color{gray}{#1}}
\algrenewcommand{\algorithmiccomment}[1]{\hfill\(\triangleright\)\ #1}

\newcommand{\algsep}{\Statex \rule{\linewidth}{0.6pt}}

\usepackage[margin=1.2in]{geometry}
\usepackage{lmodern}
\usepackage{microtype}
\usepackage[colorlinks=true, linkcolor=blue, urlcolor=blue, citecolor=blue]{hyperref}
\usepackage{cleveref}
\usepackage{listings}
\usepackage{enumitem}
\usepackage{caption}
\usepackage[backend=biber,style=numeric]{biblatex}
\definecolor{colorx1}{rgb}{1.00, 0.00, 0.00} % Red
\definecolor{colorx2}{rgb}{1.00, 0.75, 0.00} % Orange-yellow
\definecolor{colorx3}{rgb}{0.50, 1.00, 0.00} % Yellow-green
\definecolor{colorx4}{rgb}{0.00, 1.00, 0.50} % Green-cyan
\definecolor{colorx5}{rgb}{0.00, 1.00, 1.00} % Cyan
\definecolor{colorx6}{rgb}{0.00, 0.50, 1.00} % Blue
\definecolor{colorx7}{rgb}{0.50, 0.00, 1.00} % Indigo
\definecolor{colorx8}{rgb}{1.00, 0.00, 1.00} % Magenta

\definecolor{codegray}{rgb}{0.5,0.5,0.5}
\definecolor{codebg}{rgb}{0.95,0.95,0.95}
\definecolor{keywordcolor}{rgb}{0.0, 0.0, 0.6}     % dark blue
\definecolor{identifiercolor}{rgb}{0.5, 0.0, 0.5}  % purple
\definecolor{stringcolor}{rgb}{0.58,0,0.82}        % magenta
\definecolor{bgcolor}{rgb}{0.95,0.95,0.95}         % light gray

\usepackage{xcolor}
\definecolor{codebg}{rgb}{0.98, 0.98, 0.95}
\definecolor{codebase}{rgb}{0.25, 0.25, 0.28}
\definecolor{codecomment}{rgb}{0.45, 0.50, 0.53}
\definecolor{codekeyword}{rgb}{0.82, 0.37, 0.37}
\definecolor{codestring}{rgb}{0.37, 0.55, 0.29}
\definecolor{codefunc}{rgb}{0.48, 0.42, 0.93}
\definecolor{codetype}{rgb}{0.76, 0.47, 0.96}
\definecolor{codenumber}{rgb}{0.92, 0.47, 0.39}
\definecolor{codegray}{rgb}{0.60, 0.60, 0.60}
\definecolor{darkpastelpurple}{rgb}{0.59, 0.44, 0.84}
\definecolor{darkerpastelpurple}{rgb}{0.47, 0.34, 0.66}

\usepackage{listings}
\lstdefinelanguage{cutedsl}{
    language=Python,
    morekeywords={},                     % 'cute' as a keyword
    morekeywords=[2]{},  % method names
    keywordstyle=\color{keywordcolor}\bfseries,
    morekeywords=[2]{cute, tract}
    keywordstyle=[2]\color{red},
    morestring=[b]',                        % single-quoted strings
    morestring=[b]",                        % double-quoted strings
    stringstyle=\color{stringcolor},
    showstringspaces=false,
    basicstyle=\ttfamily\footnotesize,
    backgroundcolor=\color{bgcolor},
    frame=single,
    breaklines=true,
    columns=flexible
}

\newtheorem{theorem}{Theorem}[subsection]
\newtheorem{lemma}[theorem]{Lemma}
\newtheorem{corollary}[theorem]{Corollary}
\newtheorem{proposition}[theorem]{Proposition}
\newtheorem{mainthm}{Theorem} % Define a theorem environment named 'mainthm'
\theoremstyle{definition}
\newtheorem{definition}[theorem]{Definition}
\newtheorem{example}[theorem]{Example}
\newtheorem{construction}[theorem]{Construction}
\newtheorem{observation}[theorem]{Observation}
\newtheorem{notation}[theorem]{Notation}

\theoremstyle{remark}
\newtheorem{remark}[theorem]{Remark}
\newtheorem{aside}[theorem]{Aside}

\newcommand{\m}{\langle m \rangle}
\newcommand{\n}{\langle n \rangle}

\newcommand{\mathterm}[1]{\mathsf{#1}} % sans serif math terms

\newcommand{\sort}{\mathterm{sort}}
\newcommand{\coalesce}{\mathterm{coal}}
\newcommand{\coal}{\coalesce}
\renewcommand{\complement}{\mathterm{complement}}
\newcommand{\comp}{\mathterm{comp}}
\newcommand{\flatten}{\mathterm{flat}}
\newcommand{\cosize}{\mathterm{cosize}}
\newcommand{\size}{\mathterm{size}}
\newcommand{\len}{\mathterm{len}}
\newcommand{\rank}{\mathterm{rank}}
\newcommand{\depth}{\mathterm{depth}}
\newcommand{\mode}{\mathterm{mode}}
\newcommand{\entry}{\mathterm{entry}}

\newcommand{\sub}{\mathterm{sub}}
\newcommand{\shape}{\mathterm{shape}}
\newcommand{\stride}{\mathterm{stride}}
\newcommand{\colex}{\mathterm{colex}}
\newcommand{\squeeze}{\mathterm{squeeze}}
\newcommand{\Image}{\mathterm{Image}}
\newcommand{\codomain}{\mathterm{codomain}}
\newcommand{\domain}{\mathterm{domain}}
\newcommand{\complexity}{\mathterm{complexity}}
\newcommand{\id}{\mathterm{id}}
\newcommand{\incl}{\mathterm{incl}}
\newcommand{\Hom}{\mathterm{Hom}}
\renewcommand{\max}{\mathterm{max}}
\newcommand{\op}{\mathterm{op}}

\newcommand{\ob}{\mathterm{ob}}
\newcommand{\filter}{\mathterm{filter}}

\newcommand{\profile}{\mathterm{prof}}
\newcommand{\mor}{\mathterm{mor}}
\newcommand{\row}{\mathterm{row}}
\newcommand{\col}{\mathterm{col}}
\newcommand{\tiled}{\mathterm{tiled}}

\newcommand{\catstyle}[1]{{\boldsymbol{\mathsf{#1}}}}

\newcommand{\Fin}{\catstyle{Fin}}
\newcommand{\FinSet}{\catstyle{FinSet}}
\newcommand{\Tuple}{\catstyle{Tuple}}

\definecolor{amethyst}{rgb}{0.6, 0.4, 0.8}

\title{ {\Huge Categorical Foundations for CuTe Layouts} }
\author{Jack Carlisle \and Jay Shah \and Reuben Stern \and Paul VanKoughnett }

\date{%
  \begin{tabular}{c}
    Colfax Research\\[0.5em]
    \texttt{research@colfax-intl.com}\\[0.5em]
    January 2026
  \end{tabular}%
}

\begin{document}

\maketitle

\begin{abstract}
    NVIDIA's CUTLASS library provides a robust and expressive set of methods for describing and manipulating multi-dimensional tensor data on the GPU. These methods are conceptually grounded in the abstract notion of a CuTe layout and a rich algebra of such layouts, including operations such as composition, logical product, and logical division. In this paper, we present a categorical framework for understanding this layout algebra by focusing on a naturally occurring class of {\it tractable layouts}. To this end, we define two categories $\catstyle{Tuple}$ and $\catstyle{Nest}$ whose morphisms give rise to layouts. We define a suite of operations on morphisms in these categories and prove their compatibility with the corresponding layout operations. Moreover, we give a complete characterization of the layouts which arise from our construction. Finally, we provide a Python implementation of our categorical constructions, along with tests that demonstrate alignment with CUTLASS behavior. This implementation can be found at our git repository \url{https://github.com/ColfaxResearch/layout-categories}.
\end{abstract}

\newpage

\setcounter{tocdepth}{2}
\tableofcontents

\newpage 

\chapter{Introduction} 
In modern computing, particularly in GPU programming, performance depends critically on how multi-dimensional data is stored and accessed in memory. While most data that we care about—such as images, videos, and tensors in machine learning—are inherently multi-dimensional, a computer's memory is fundamentally one-dimensional. This means that when we want to load, store, or otherwise manipulate data, we need to map its multi-dimensional logical coordinates to one-dimensional physical coordinates. This mapping, known as a {\bf layout}, is essential for reading from and writing to memory correctly and efficiently. Moreover, with respect to the GPU's SIMT execution model, layouts are used to describe and manipulate partitionings of threads over data. This is important to ensure optimized memory access patterns and correct invocation of specialized hardware instructions such as those used to target tensor cores. 

As a motivating example, suppose we want to store the $4 \times 8$ matrix
\[
A = 
\begin{bmatrix}
12.47 & 87.21 & 34.08 & 56.93 & 45.65 &  9.17 & 73.02 & 21.39 \\
64.88 & 30.41 &  1.72 & 88.04 & 92.55 & 17.06 & 50.91 & 68.77 \\
 3.33 & 77.19 & 61.58 & 29.46 & 15.82 & 80.75 & 44.62 & 39.28 \\
91.40 & 26.12 &  6.97 & 53.03 & 58.66 & 33.79 & 11.20 & 70.55
\end{bmatrix}
\]
in memory. In order to do so, we need to specify a memory address for each entry of $A$. We do this by choosing some address for the $(0,0)$th entry of $A$, and specifying an {\bf offset} for each other entry of $A$. One common choice is the {\bf row-major} layout\\ 

\begin{centering}

\begin{tikzpicture}[x={(0cm,-1cm)},y={(1cm,0cm)},every node/.style={minimum size=1cm, outer sep=0pt}]

\node[fill=gray!20] at (0,0) {0};
\node[fill=gray!20] at (0,1) {1};
\node[fill=gray!20] at (0,2) {2};
\node[fill=gray!20] at (0,3) {3};
\node[fill=gray!20] at (0,4) {4};
\node[fill=gray!20] at (0,5) {5};
\node[fill=gray!20] at (0,6) {6};
\node[fill=gray!20] at (0,7) {7};
\node[fill=gray!20] at (1,0) {8};
\node[fill=gray!20] at (1,1) {9};
\node[fill=gray!20] at (1,2) {10};
\node[fill=gray!20] at (1,3) {11};
\node[fill=gray!20] at (1,4) {12};
\node[fill=gray!20] at (1,5) {13};
\node[fill=gray!20] at (1,6) {14};
\node[fill=gray!20] at (1,7) {15};
\node[fill=gray!20] at (2,0) {16};
\node[fill=gray!20] at (2,1) {17};
\node[fill=gray!20] at (2,2) {18};
\node[fill=gray!20] at (2,3) {19};
\node[fill=gray!20] at (2,4) {20};
\node[fill=gray!20] at (2,5) {21};
\node[fill=gray!20] at (2,6) {22};
\node[fill=gray!20] at (2,7) {23};
\node[fill=gray!20] at (3,0) {24};
\node[fill=gray!20] at (3,1) {25};
\node[fill=gray!20] at (3,2) {26};
\node[fill=gray!20] at (3,3) {27};
\node[fill=gray!20] at (3,4) {28};
\node[fill=gray!20] at (3,5) {29};
\node[fill=gray!20] at (3,6) {30};
\node[fill=gray!20] at (3,7) {31};
\draw[color=black,thick,shift={(-0.5,-0.5)}] (0,0) grid (4,8);

\node[anchor=east] at (1.5,-1) {$L^\row = (4,8):(8,1) =$};
\end{tikzpicture}

\end{centering}

\bigskip

\noindent The notation $L^\row = (4,8):(8,1)$ indicates that the offset of the $(i,j)$th entry of our matrix is 
\[
(i,j) \cdot (8,1) = 8i + j. 
\]

\noindent Another common choice is the {\bf column-major} layout \\

\begin{centering}

\begin{tikzpicture}[x={(0cm,-1cm)},y={(1cm,0cm)},every node/.style={minimum size=1cm, outer sep=0pt}]

\node[fill=gray!20] at (0,0) {0};
\node[fill=gray!20] at (0,1) {4};
\node[fill=gray!20] at (0,2) {8};
\node[fill=gray!20] at (0,3) {12};
\node[fill=gray!20] at (0,4) {16};
\node[fill=gray!20] at (0,5) {20};
\node[fill=gray!20] at (0,6) {24};
\node[fill=gray!20] at (0,7) {28};
\node[fill=gray!20] at (1,0) {1};
\node[fill=gray!20] at (1,1) {5};
\node[fill=gray!20] at (1,2) {9};
\node[fill=gray!20] at (1,3) {13};
\node[fill=gray!20] at (1,4) {17};
\node[fill=gray!20] at (1,5) {21};
\node[fill=gray!20] at (1,6) {25};
\node[fill=gray!20] at (1,7) {29};
\node[fill=gray!20] at (2,0) {2};
\node[fill=gray!20] at (2,1) {6};
\node[fill=gray!20] at (2,2) {10};
\node[fill=gray!20] at (2,3) {14};
\node[fill=gray!20] at (2,4) {18};
\node[fill=gray!20] at (2,5) {22};
\node[fill=gray!20] at (2,6) {26};
\node[fill=gray!20] at (2,7) {30};
\node[fill=gray!20] at (3,0) {3};
\node[fill=gray!20] at (3,1) {7};
\node[fill=gray!20] at (3,2) {11};
\node[fill=gray!20] at (3,3) {15};
\node[fill=gray!20] at (3,4) {19};
\node[fill=gray!20] at (3,5) {23};
\node[fill=gray!20] at (3,6) {27};
\node[fill=gray!20] at (3,7) {31};
\draw[color=black,thick,shift={(-0.5,-0.5)}] (0,0) grid (4,8);

\node[anchor=east] at (1.5,-1) {$L^\col = (4,8):(1,4) = $};
\end{tikzpicture}

\end{centering} 
\bigskip

\noindent Again, the notation $L^\col = (4,8):(1,4)$ indicates that the offset of the $(i,j)$th entry of our matrix is 
\[
(i,j) \cdot (1,4) = i + 4j.
\]

\noindent These layouts are extremely useful, but do not suffice for all purposes. For example, in high-performance computing, one often computes matrix products $AB$ by
\begin{enumerate}
    \item dividing the operand matrices $A$ and $B$ into tiles,
    \item computing matrix products of the various tiles, and 
    \item combining these partial results to obtain the full result $AB$.
\end{enumerate}
For instance, we could divide our $4 \times 8$ matrix $A$ into $2 \times 2$ tiles, as depicted below.
\[
A = \left[\begin{array}{cccc}
\begin{bmatrix}
12.47 & 87.21 \\
64.88 & 30.41 
\end{bmatrix}
&
\begin{bmatrix} 
34.08 & 56.93 \\
1.72 & 88.04 
\end{bmatrix}
&
\begin{bmatrix}
45.65 &  9.17 \\
92.55 & 17.06
\end{bmatrix}
&
\begin{bmatrix} 
73.02 & 21.39 \\
50.91 & 68.77
\end{bmatrix} 
\\[2ex]
\begin{bmatrix}
 3.33 & 77.19 \\
 91.40 & 26.12
\end{bmatrix} 
&
\begin{bmatrix} 
 61.58 & 29.46 \\
6.97 & 53.03 
\end{bmatrix}
&
\begin{bmatrix} 
 15.82 & 80.75 \\
58.66 & 33.79
\end{bmatrix} 
&
\begin{bmatrix} 
 44.62 & 39.28 \\
11.20 & 70.55
\end{bmatrix} 
\end{array}\right]
\]
% \noindent We can consider $A^\tiled$ as a tensor of shape $((2,2),(2,4))$, where 
% \[
% A^\tiled_{(i,j),(k,\ell)} = \begin{matrix}
%     \text{the }(i,j)\text{th entry of} \\
%     \; \;\text{the } (k,\ell)\text{th tile of }A.
% \end{matrix}
% \]

\noindent Suppose now that we wanted to slice out individual tiles of $A$, which we assume is laid out in column-major format in memory. To do this, one could manually compute offsets as follows: for the $(i,j)$th tile, the offset to index into the top-left entry of the tile is given by $2i + 8j$. On the other hand, to better organize this computation, we could use the {\bf interleaved} layout of tiles
\\

\begin{centering}

\begin{tikzpicture}[x={(0cm,-1cm)},y={(1cm,0cm)},every node/.style={minimum size=1cm, outer sep=0pt}]

\node[fill=gray!20]       at (0,0) {0};
\node[fill=gray!20]       at (0,1) {2};
\node[fill=gray!20]      at (0,2) {8};
\node[fill=gray!20]      at (0,3) {10};
\node[fill=gray!20]  at (0,4) {16};
\node[fill=gray!20]  at (0,5) {18};
\node[fill=gray!20]        at (0,6) {24};
\node[fill=gray!20]        at (0,7) {26};

\node[fill=gray!20]       at (1,0) {1};
\node[fill=gray!20]       at (1,1) {3};
\node[fill=gray!20]      at (1,2) {9};
\node[fill=gray!20]      at (1,3) {11};
\node[fill=gray!20]  at (1,4) {17};
\node[fill=gray!20]  at (1,5) {19};
\node[fill=gray!20]        at (1,6) {25};
\node[fill=gray!20]        at (1,7) {27};

\node[fill=gray!20]       at (2,0) {4};
\node[fill=gray!20]       at (2,1) {6};
\node[fill=gray!20]      at (2,2) {12};
\node[fill=gray!20]      at (2,3) {14};
\node[fill=gray!20]  at (2,4) {20};
\node[fill=gray!20]  at (2,5) {22};
\node[fill=gray!20]        at (2,6) {28};
\node[fill=gray!20]        at (2,7) {30};

\node[fill=gray!20]       at (3,0) {5};
\node[fill=gray!20]       at (3,1) {7};
\node[fill=gray!20]      at (3,2) {13};
\node[fill=gray!20]      at (3,3) {15};
\node[fill=gray!20]  at (3,4) {21};
\node[fill=gray!20]  at (3,5) {23};
\node[fill=gray!20]        at (3,6) {29};
\node[fill=gray!20]        at (3,7) {31};
\draw[color=black,thick,shift={(-0.5,-0.5)}] (0,0) grid (4,8);

\node[anchor=east] at (1.5,-1) {$L^\tiled =$};

% \node at (0,-1) {\Large{\texttt{0}}};
% \node at (1,-1) {\Large{\texttt{1}}};
% \node at (2,-1) {\Large{\texttt{2}}};
% \node at (3,-1) {\Large{\texttt{3}}};
% \node at (-1,0) {\Large{\texttt{0}}};
% \node at (-1,1) {\Large{\texttt{1}}};
% \node at (-1,2) {\Large{\texttt{2}}};
% \node at (-1,3) {\Large{\texttt{3}}};
% \node at (-1,4) {\Large{\texttt{4}}};
% \node at (-1,5) {\Large{\texttt{5}}};
% \node at (-1,6) {\Large{\texttt{6}}};
% \node at (-1,7) {\Large{\texttt{7}}};
\end{tikzpicture}

\end{centering}

\bigskip

\noindent where the columns are given by tiles of $A$ and the rows are given by coordinates within the tile shape. Here, we use {\bf colexicographic ordering} to linearly enumerate tiles and coordinates within tiles, hence the top-level shape $(4, 8)$ of the layout $L^\tiled$.

However, note that the interleaving pattern shown for $L^\tiled$ means that it can't be expressed as a layout $(4,8):(a,b)$ for any strides $a,b$. Instead, we can factor the modes of the shape $(4,8)$ and define 
\[ L^\tiled = ((2,2),(2,4)):((1,4),(2,8)). \]

\noindent The prior offset calculation $2i + 8j$ then appears through evaluating $L^\tiled$ on the coordinate $(0, (i,j))$, and the tile layout itself is given by the first mode. Thus, after endowing $A$ with the layout $L^\tiled$ to form $A^\tiled$, we can obtain the $(i,j)$th tile of $A$ as the slice

\[ A_{i,j} = A^\tiled(\; \rule{.7em}{0.4pt} \; , (i,j)). \]

A key idea developed in CUTLASS is that useful but more complex auxiliary layouts such as $L^\tiled$ may be systematically deduced from simpler layouts via certain fundamental operations. In the case of $L^\tiled$, the operation in question is called {\bf logical division}. If we write
\\

\begin{centering}

\begin{tikzpicture}[x={(0cm,-1cm)},y={(1cm,0cm)},every node/.style={minimum size=1cm, outer sep=0pt}]
\node[fill=gray!20]       at (0,0) {0};
\node[fill=gray!20]       at (1,0) {1};
\node[fill=gray!20]      at (0,1) {4};
\node[fill=gray!20]      at (1,1) {5};
\draw[color=black,thick,shift={(-0.5,-0.5)}] (0,0) grid (2,2);

\node[anchor=east] at (0.5,-1) {$T = (2,2):(1,4) =$};
\end{tikzpicture}

\end{centering}

\bigskip

\noindent for the tile layout, then $L^\tiled$ is the logical division

\[
L^\tiled = L^\col \oslash T
\]
as depicted below.\\

\begin{centering}
\begin{tikzpicture}[x={(0cm,-1cm)},y={(1cm,0cm)},every node/.style={minimum size=1cm, outer sep=0pt}]

\node[fill=colorx1!60] at (0,0) {0};
\node[fill=colorx1!60] at (0,1) {4};
\node[fill=colorx3!60] at (0,2) {8};
\node[fill=colorx3!60] at (0,3) {12};
\node[fill=colorx5!60] at (0,4) {16};
\node[fill=colorx5!60] at (0,5) {20};
\node[fill=colorx7!60] at (0,6) {24};
\node[fill=colorx7!60] at (0,7) {28};
\node[fill=colorx1!60] at (1,0) {1};
\node[fill=colorx1!60] at (1,1) {5};
\node[fill=colorx3!60] at (1,2) {9};
\node[fill=colorx3!60] at (1,3) {13};
\node[fill=colorx5!60] at (1,4) {17};
\node[fill=colorx5!60] at (1,5) {21};
\node[fill=colorx7!60] at (1,6) {25};
\node[fill=colorx7!60] at (1,7) {29};
\node[fill=colorx2!60] at (2,0) {2};
\node[fill=colorx2!60] at (2,1) {6};
\node[fill=colorx4!60] at (2,2) {10};
\node[fill=colorx4!60] at (2,3) {14};
\node[fill=colorx6!60] at (2,4) {18};
\node[fill=colorx6!60] at (2,5) {22};
\node[fill=colorx8!60] at (2,6) {26};
\node[fill=colorx8!60] at (2,7) {30};
\node[fill=colorx2!60] at (3,0) {3};
\node[fill=colorx2!60] at (3,1) {7};
\node[fill=colorx4!60] at (3,2) {11};
\node[fill=colorx4!60] at (3,3) {15};
\node[fill=colorx6!60] at (3,4) {19};
\node[fill=colorx6!60] at (3,5) {23};
\node[fill=colorx8!60] at (3,6) {27};
\node[fill=colorx8!60] at (3,7) {31};
\draw[color=black,thick,shift={(-0.5,-0.5)}] (0,0) grid (4,8);

\node[anchor = east] at (1.5,-1) {$L^\col = $ };
\node[anchor = east] at (1.5,9.5) {$T = $};

\node[fill=gray!5] at (1,10) {0};
\node[fill=gray!20] at (2,10) {1};
\node[fill=gray!40] at (1,11) {4};
\node[fill=gray!60] at (2,11) {5};
\draw[color=black,thick,shift={(-0.5,-0.5)}] (1,10) grid (3,12);
\end{tikzpicture}

\end{centering}

\vspace{0.2in}

\begin{centering}
\begin{tikzpicture}[x={(0cm,-1cm)},y={(1cm,0cm)},every node/.style={minimum size=1cm, outer sep=0pt}]

\node[fill=colorx1!10]       at (0,0) {0};
\node[fill=colorx2!10]       at (0,1) {2};
\node[fill=colorx3!10]      at (0,2) {8};
\node[fill=colorx4!10]      at (0,3) {10};
\node[fill=colorx5!10]  at (0,4) {16};
\node[fill=colorx6!10]  at (0,5) {18};
\node[fill=colorx7!10]        at (0,6) {24};
\node[fill=colorx8!10]        at (0,7) {26};

\node[fill=colorx1!25]       at (1,0) {1};
\node[fill=colorx2!25]       at (1,1) {3};
\node[fill=colorx3!25]      at (1,2) {9};
\node[fill=colorx4!25]      at (1,3) {11};
\node[fill=colorx5!25]  at (1,4) {17};
\node[fill=colorx6!25]  at (1,5) {19};
\node[fill=colorx7!25]        at (1,6) {25};
\node[fill=colorx8!25]        at (1,7) {27};

\node[fill=colorx1!40]       at (2,0) {4};
\node[fill=colorx2!40]       at (2,1) {6};
\node[fill=colorx3!40]      at (2,2) {12};
\node[fill=colorx4!40]      at (2,3) {14};
\node[fill=colorx5!40]  at (2,4) {20};
\node[fill=colorx6!40]  at (2,5) {22};
\node[fill=colorx7!40]        at (2,6) {28};
\node[fill=colorx8!40]        at (2,7) {30};

\node[fill=colorx1!60]       at (3,0) {5};
\node[fill=colorx2!60]       at (3,1) {7};
\node[fill=colorx3!60]      at (3,2) {13};
\node[fill=colorx4!60]      at (3,3) {15};
\node[fill=colorx5!60]  at (3,4) {21};
\node[fill=colorx6!60]  at (3,5) {23};
\node[fill=colorx7!60]       at (3,6) {29};
\node[fill=colorx8!60]       at (3,7) {31};
\draw[color=black,thick,shift={(-0.5,-0.5)}] (0,0) grid (4,8);

\node[anchor=east] at (1.5,-1) {$ L^\col \oslash T = $};

% \node at (0,-1) {\Large{\texttt{0}}};
% \node at (1,-1) {\Large{\texttt{1}}};
% \node at (2,-1) {\Large{\texttt{2}}};
% \node at (3,-1) {\Large{\texttt{3}}};
% \node at (-1,0) {\Large{\texttt{0}}};
% \node at (-1,1) {\Large{\texttt{1}}};
% \node at (-1,2) {\Large{\texttt{2}}};
% \node at (-1,3) {\Large{\texttt{3}}};
% \node at (-1,4) {\Large{\texttt{4}}};
% \node at (-1,5) {\Large{\texttt{5}}};
% \node at (-1,6) {\Large{\texttt{6}}};
% \node at (-1,7) {\Large{\texttt{7}}};
\end{tikzpicture}

\end{centering}

\bigskip

In addition to logical division, other fundamental layout operations include {\bf logical products}, {\bf complements}, and most importantly, {\bf composition}. These layout operations are the backbone of CUTLASS, and a deep understanding of their behavior is helpful for writing correct and highly performant code. However, the definitions and constructions of these operations are fairly subtle. For example, the composition $B \circ A$ of layouts $A$ and $B$ is well-defined only if $A$ and $B$ satisfy certain divisibility constraints, which CUTLASS checks under the hood. In particular, it is not always obvious when two layouts are composable, or how to interpret their composition.

% To support these tasks, NVIDIA's CUTLASS library provides a powerful toolkit for working with layouts. The core abstraction of the CUTLASS library is the {\bf CuTe layout}, or simply {\bf layout}, which specifies a logical-to-physical coordinate transformation (\cite{cutedsldocumentation}, \cite{cutedocumentation}). Throughout this work, we include code snippets written in NVIDIA’s CuTe DSL, a Python-based implementation of CuTe, which we refer to simply as $\texttt{cute}$. CuTe layouts and the operations which they support abstract away the complexity of layout management, allowing programmers to work with high-dimensional structures efficiently and intuitively. The power of layouts comes from the vast collection of operations they support.

% Given the subtlety of these operations, there is a clear need for a more accessible and mathematically grounded treatment of layouts.

% TODO: Get Cris Cecka consent or not to mention his forthcoming work.
% Mathematically, the theory of {\bf categories} and {\bf operads} provides a convenient language and set of standard mathematical tools for describing and reasoning about algebraic structures with a notion of composition. In this work, we will adopt a categorical perspective in order to develop one possible mathematical framework for CuTe layouts, although knowledge of category theory or operads isn't a strict prerequisite to read this work. We mention here that another framework has been developed in forthcoming work by Cris Cecka, the original author of CuTe, and the reader may wish to compare his perspective with ours.

\section{Summary of main results}

The main idea of this work is that we can develop an intuitive and powerful mathematical framework for working with layouts by restricting our attention to {\bf tractable layouts}, whose entries satisfy a simple divisibility condition (see Definition \ref{definitionoftractablelayouts}). Tractable layouts include almost all layouts one encounters in practice, such as
\begin{itemize}
\item {\it row-major} and {\it column-major} layouts, which are ubiquitous,
\item {\it compact} layouts, which store data in consecutive memory addresses,
\item {\it projections}, which broadcast multiple copies of data, and
\item {\it dilations}, which enable padded loads and stores.
\end{itemize}
If $L$ is a tractable layout, then we can represent $L$ with a {\bf diagram}. For example, the layouts $L^\row$, $L^\col$, and $L^\tiled$ are represented by the following diagrams.

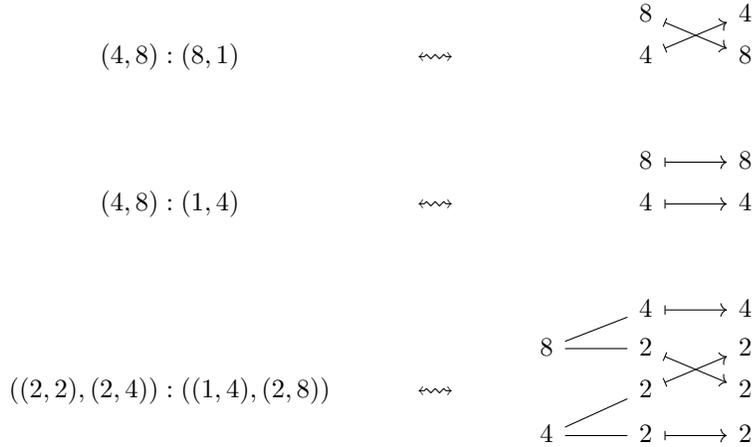
\[\begin{tikzcd}[row sep = 1, column sep = 12]
  & & & &  & & 8 \ar[drr,mapsto] & & 4 \\
 (4,8):(8,1) & & \leftrightsquigarrow & &   & & 4 \ar[urr,mapsto] & & 8 \\
  & & & &  & &   & &  \\
 & & & &   & &   & &   \\[2em]
 & & & &   & &   & &  \\
 & & & &   & & 8 \ar[rr,mapsto] & & 8  \\
(4,8):(1,4)& & \leftrightsquigarrow& &   & & 4 \ar[rr,mapsto] & & 4  \\
  & & & &  & &   & &    \\
 & & & &   & &   & &   \\[2em]
 & & & &   & &   & &    \\
 & & & &  & & 4 \ar[rr,mapsto] & & 4  \\
& & & & 8 \ar[rr,-] \ar[urr,-] & & 2 \ar[drr,mapsto] & & 2 & & & & \\
 ((2,2),(2,4)):((1,4),(2,8)) & & \leftrightsquigarrow & &   & & 2 \ar[urr,mapsto] & & 2 \\
& & & & 4 \ar[urr,-] \ar[rr,-]  & & 2 \ar[rr,mapsto] & & 2  \\
& & & &   & &   & &   \\
\end{tikzcd}\]
% \[\begin{tikzcd}[row sep = 1, column sep = 8]
% & & & & & & & &  & & \\
% & & & & & & & & & & 4\\
%     (4,4) \ar[rrrr,"f"] \ar[rrrr,swap,"(1{,}3)"] & & & & (4,2,4)  & & \rightsquigarrow & & 4 \ar[urr,mapsto] & & 2\\
%     & & & & & & & & 4 \ar[rr,mapsto] & & 4\\
%     & & & & & & & &  & f &\\
% & & & & & & & & & & \\[2em]
% & & & & & & & & 2 \ar[drr,mapsto] & & 2\\
%     (2,2,2) \ar[rrrr,"g"] \ar[rrrr,swap,"(1{,}3{,}2)"] & & & & (2,2,2)  & & \rightsquigarrow & & 2 \ar[urr,mapsto] & & 2\\
%     & & & & & & & & 2 \ar[rr,mapsto] & & 2\\
%     & & & & & & & &  & g &
% & & & & & & & & & & \\[2em]
% & & & & & & & & 4 & & \\
% & & & & & & & & 4 \ar[drr,mapsto] & & \\
% & & & & & & & & 4 & & 4\\
%     (16,(4,4),(4,4)) \ar[rrrr,"h"] \ar[rrrr,swap,"(1{,}2{,}*{,}3{,}*)"] & & & & (16,4,4)  & & \rightsquigarrow & & 4 \ar[rr,mapsto] & & 4\\
%     & & & & & & & & 16 \ar[rr,mapsto] & & 16\\
%     & & & & & & & &  & h &
% \end{tikzcd}\]

\noindent These diagrams may be interpreted as morphisms in a {\bf category}. This allows us to leverage the power of {\bf category theory} to describe layouts and their operations.\footnote{We provide a primer on category theory in Appendix \ref{categorytheoryappendix} for those unfamiliar with the subject.}

More precisely, we define a category $\catstyle{Nest}$ whose objects are nested tuples of positive integers, and whose morphisms $f:S \to T$ correspond to diagrams such as those above (see Definition \ref{definitionofTuple} and Definition \ref{definitionofNest} for details). If $L$ is a {\bf non-degenerate} tractable layout (see Definition \ref{definitionofnondegeneratelayout}), then there is an essentially unique $\catstyle{Nest}$-morphism $f$ which encodes $L$, as illustrated by the following correspondence theorem. 
% An object of $\catstyle{Nest}$ is a nested tuples of positive integers, and a morphism $f:S \to T$ in $\catstyle{Nest}$ associates entries of $S$ to entries of $T$ (see Definition \ref{definitionofTuple} and Definition \ref{definitionofNest}). If $f:S \to T$ is a nested tuple morphism, then $f$ encodes a tractable layout $L_f$. The shape of $L_f$ is $S$, and the stride of $L_f$ is a nested tuple whose entries are certain prefix products of $T$ (See Construction \ref{layoutfromnestedtuplemorphism}). For example, the layouts encoded by the morphisms $f$, $g$, and $h$ above are
% \begin{align*}
%     L_f & = (4,4):(1,8)\\
%     L_g & = (2,2,2):(1,4,2)\\
%     L_h & = (16,(4,4),(4,4)): (1,(16,0),(64,0)).
% \end{align*}

 % Conversely, if $L$ is a tractable layout, then there exists a nested tuple morphism $f_L$ which encodes $L$. The nested tuple morphism which encodes $L$ is not unique. However, we introduce a notion of {\it standard form} for nested tuple morphisms (Definition \ref{definitionofnestedstandardform}), and prove the following correspondence theorem.
\begin{mainthm}(see \ref{nestedonetoonecorrespondence})
There is a one-to-one correspondence
    \[ \begin{tikzcd} 
    \begin{Bmatrix}
    \text{Non-degenerate}\\
        \text{tractable layouts}
    \end{Bmatrix} 
    \ar[r] & \ar[l] 
        \begin{Bmatrix}
        \text{Non-degenerate}\\
        \catstyle{Nest}\text{-morphisms}\\
        \text{ of standard form}
    \end{Bmatrix}
    \end{tikzcd} 
    \]
\end{mainthm}

Layout operations such as composition, logical division, and logical products may be interpreted naturally in the category $\catstyle{Nest}$. If 
\[
\begin{tikzcd}
    S \ar[r,"f"]  & T \ar[r,"g"]  & U
\end{tikzcd}
\]
are $\catstyle{Nest}$-morphisms, then we may form the composite 
\[\begin{tikzcd}
S \ar[r,"g \circ f"] &  U
\end{tikzcd} \]
by pasting the associated diagrams together. For example,
% \[\begin{tikzcd}
% (2,2) \ar[rr,"f"] \ar[rr,swap,"(2{,}1)"] & & (2,2) \ar[rr,"g"] \ar[rr,swap,"(2{,}3)"] & & (5,(2,2))
% \end{tikzcd}\]
% is 
% \[\begin{tikzcd}
% (2,2) \ar[rr,"g \circ f"] \ar[rr,swap,"(3{,}2)"] & &  (5,(2,2))
% \end{tikzcd}\]
% as depicted below.
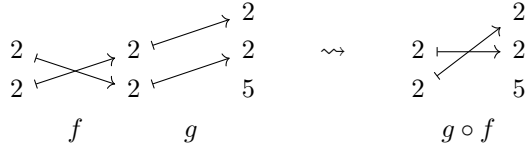
\[ \begin{tikzcd}[row sep = 1, column sep = 8]
 & &   & & 2 & & & & & & 2\\
2 \ar[drr,mapsto]  & & 2 \ar[urr,mapsto] & & 2 & & \rightsquigarrow & & 2 \ar[rr,mapsto]  & & 2\\
2 \ar[urr,mapsto] & & 2 \ar[urr,mapsto] & & 5 & & & & 2 \ar[uurr,mapsto] & & 5\\
 & f & & g & & & & & & \mathclap{g \circ f} & 
\end{tikzcd}\]
We prove that composition in $\catstyle{Nest}$ is compatible with layout composition.
\begin{mainthm}(see \ref{compatibilityofcompositioninD}) If $f$ and $g$ are non-degenerate composable $\catstyle{Nest}$-morphisms, then
\[
L_{g \circ f} = L_g \circ L_f.
\]
\end{mainthm}
% For example, if $f$ and $g$ are the morphisms above, then 
% \begin{align*}
%     L_f & =(2,2):(2,1) \\
%     L_g & = (2,2):(5,10)\\
%     L_{g \circ f} & = (2,2):(10,5)\\
%     L_g \circ L_f& = (2,2):(10,5).
% \end{align*}

We can coalesce a $\catstyle{Nest}$-morphism $f$ by collapsing adjacent arrows. For example,
% on nested tuple morphisms.
% For example if $f$ is the morphism 
% \[
% \begin{tikzcd} 
% ((2,2),(10,10)) \ar[rr,"f"] \ar[rr,swap,"(1{,}2{,}4{,}5)"] & & (2,2,2,10,10)
% \end{tikzcd}
% \]
% then $\coalesce(f)$ is the morphism 
% \[
% \begin{tikzcd} 
% (4,100) \ar[rr,"\coalesce(f)"] \ar[rr,swap,"(1{,}3)"] & & (4,2,100)
% \end{tikzcd}
% \]
% as depicted below.
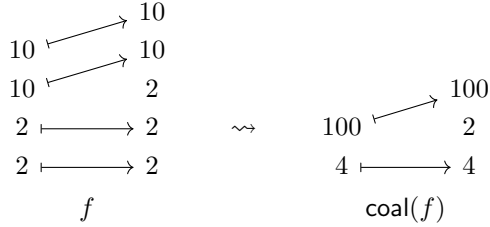
\[\begin{tikzcd}[row sep = 1, column sep = 8]
 & & 10 & & & & & & \\
 10 \ar[urr,mapsto] & & 10 & & & & & & \\
 10 \ar[urr,mapsto] & & 2 & & & & & & 100\\
 2 \ar[rr,mapsto] & & 2 & & \rightsquigarrow & & 100 \ar[urr,mapsto] & & 2\\
 2 \ar[rr,mapsto] & & 2 & & & & 4 \ar[rr,mapsto] & & 4\\
    & f & & & & & & \mathclap{\coalesce(f)} & 
\end{tikzcd}\]
We prove that this operation is compatible with layout coalesce.
\begin{mainthm}(see \ref{coalesceagreement}) If $f$ is a $\catstyle{Nest}$-morphism, then
\[
L_{\coalesce(f)} = \coalesce(L_f).
\]
\end{mainthm}
% For example, if $f$ is the morphism above, then 
% \begin{align*}
%     L_f & = ((2,2),(10,10)):((1,2),(8,80))\\
%     L_{\coalesce(f)} & = (4,100):(1,8)\\
%     \coalesce(L_f) & = (4,100):(1,8).
% \end{align*}

The complement of a  $\catstyle{Nest}$-morphism $f$ is the inclusion of the entries not hit by $f$. For example, 
% if $f$ is the morphism 
% \[
% \begin{tikzcd} 
% (2,2) \ar[rr,"f"] \ar[rr,swap,"(1{,}3)"] & & (2,5,2,5)
% \end{tikzcd}
% \]
% then $f^c$ is the morphism 
% \[
% \begin{tikzcd} 
% (5,5) \ar[rr,"f^c"] \ar[rr,swap,"(2{,}4)"] & & (2,5,2,5)
% \end{tikzcd}
% \]
% as depicted below.
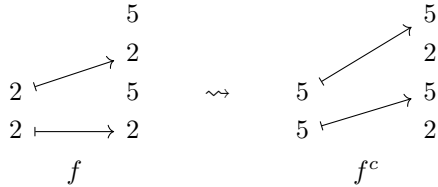
\[\begin{tikzcd}[row sep = 1, column sep = 8]
 & & 5 & & & & & & 5\\
 & & 2 & & & & & & 2\\
 2 \ar[urr,mapsto] & & 5 & & \rightsquigarrow &  & 5 \ar[uurr,mapsto] & & 5\\
 2 \ar[rr,mapsto] & & 2 & & & &  5 \ar[urr,mapsto]  & & 2\\
    & f & & & & & &  f^c & 
\end{tikzcd}\]
We prove that complements in $\catstyle{Nest}$ are compatible with layout complements. 
\begin{mainthm}(see \ref{coalescedlayoutoftuplemorphismcomplement}) If $f:S \to T$ is an injective $\catstyle{Nest}$-morphism and $N = \size(T)$, then
\[
\coalesce(L_{f^c}) = \comp(L_f,N).
\]
\end{mainthm} 
% For example, if $f$ is the morphism above, then 
% \begin{align*}
% L_f & = (2,2):(1,10)\\
% L_{f^c} & = (5,5):(2,20)\\
% \comp(L_f,100)&  = (5,5):(2,20).
% \end{align*}

We define divisibility of $\catstyle{Nest}$-morphisms, and a logical division operation 
\[
f,g \mapsto f \oslash g
\]
when $g$ divides $f$. For example, 
% if $f$ and $g$ are the morphisms 
% \[\begin{tikzcd}
%     ((4,8),(4,8)) \ar[rr,"f"] \ar[rr,swap,"(1{,}2{,}3{,}4)"] & & ((4,8),(4,8))\\
%     (4,4) \ar[rr,"g"] \ar[rr,swap,"(1{,}3)"] & & ((4,8),(4,8))
% \end{tikzcd}\]
% then $g$ divides $f$, and $f \oslash g$ is the morphism 
% \[\begin{tikzcd}
%     ((4,4),(8,8)) \ar[rr,"f\oslash g"] \ar[rr,swap,"(1{,}3{,}2{,}4)"] & & ((4,8),(4,8))
% \end{tikzcd}\]
% as depicted below.
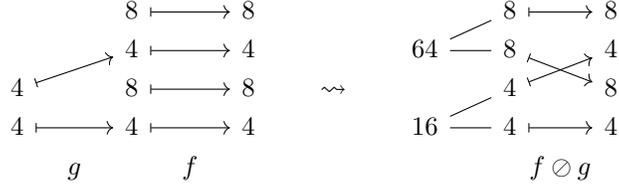
\[\begin{tikzcd}[row sep = 1, column sep = 8]
  & & 8 \ar[rr,mapsto] & & 8 & & & & & &  8 \ar[rr,mapsto]  & & 8\\
  & & 4 \ar[rr,mapsto] & & 4 & & & & 64 \ar[rr,-] \ar[urr,-]& &  8 \ar[drr,mapsto] & & 4\\
4 \ar[urr,mapsto] & & 8 \ar[rr,mapsto] & & 8 & & \rightsquigarrow & &  & &  4 \ar[urr,mapsto] & & 8\\
4 \ar[rr,mapsto] & & 4 \ar[rr,mapsto] & & 4 & & & & 16 \ar[rr,-] \ar[urr,-]& & 4 \ar[rr,mapsto] & & 4\\
  & g &  & f & & & & & & & & \mathclap{f \oslash g} & 
\end{tikzcd}\]

We prove that logical division in $\catstyle{Nest}$ is compatible with logical division of layouts.
\begin{mainthm}(see \ref{coalesceoflogicaldivision}) If $f$ and $g$ are non-degenerate $\catstyle{Nest}$-morphisms and $g$ divides $f$, then 
\[
\coalesce(L_{f \oslash g}) = \coalesce(L_f \oslash L_g).
\]
\end{mainthm}
% For example, if $f$ and $g$ are the morphisms above, then 
% \begin{align*}
%     L_f & = (4,8,4,8):(1,4,32,128)\\
%     L_g & = (4,4):(1,32)\\
%     L_{f \oslash g} & = ((4,4),(8,8)):((1,32),(4,128))\\
%     L_f \oslash L_g & = ((4,4),(8,8)):((1,32),(4,128)).
% \end{align*}

We define product admissibility of $\catstyle{Nest}$-morphisms, and a logical product operation 
\[
f,g \mapsto f \otimes g
\]
when $f$ and $g$ are product admissible. For example, 
% if $f$ and $g$ are the nested tuple morphisms 
% \[\begin{tikzcd}
%     (2,2) \ar[rr,"f"] \ar[rr,swap,"(1{,}2)"] & & (2,2,5,5)\\
%     (5,5) \ar[rr,"g"] \ar[rr,swap,"(2{,}1)"] & & (5,5)
% \end{tikzcd}\]
% then $f$ and $g$ are product admissible, and $f \otimes g$ is the morphism 
% \[\begin{tikzcd}
%     ((2,2),(5,5)) \ar[rr,"f\oslash g"] \ar[rr,swap,"(1{,}2{,}4{,}3)"] & & (2,2,5,5)
% \end{tikzcd}\]
% as depicted below.
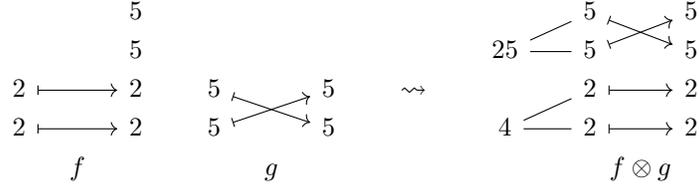
\[\begin{tikzcd}[row sep = 1, column sep = 8]
  & & 5 & &   & &  & &  & & & & 5 \ar[drr,mapsto] & & 5\\
  & & 5 & &   & & & & & & 25 \ar[rr,-] \ar[urr,-] & & 5 \ar[urr,mapsto] & & 5\\
2  \ar[rr,mapsto] & & 2 & & 5  \ar[drr,mapsto] & & 5 & & \rightsquigarrow & & &  & 2 \ar[rr,mapsto] & & 2\\
2 \ar[rr,mapsto] & & 2 & & 5 \ar[urr,mapsto]  & & 5 & & & & 4 \ar[rr,-] \ar[urr,-] & & 2 \ar[rr,mapsto] & & 2\\
  & f & & &  & g & & & & & & & & \mathclap{f \otimes g} & 
\end{tikzcd}\]

We prove that the logical products in $\catstyle{Nest}$ are compatible with logical products of layouts.
\begin{mainthm}(see \ref{logicalproductcompatibility}) If $f$ and $g$ are non-degenerate $\catstyle{Nest}$-morphisms and $f$ and $g$ are product admissible, then
\[
L_{f \otimes g} = L_f \otimes L_g.
\]
\end{mainthm}
% For example, if $f$ and $g$ are the morphisms above, then 
% \begin{align*}
%     L_f & = (2,2):(1,2)\\
%     L_g & = (5,5):(5,1) \\
%     L_{f \otimes g} & = ((2,2),(5,5)):((1,2),(20,4))\\
%     L_f \otimes L_g & = ((2,2),(5,5)):((1,2),(20,4))
% \end{align*}

In Chapter \ref{computationschapter}, we illustrate how our new framework may be used to compute important layout operations such as composition, logical division, and logical products. In particular, we present an algorithm (Algorithm \ref{tractablelayoutcompositionalgorithm}) for computing the composition $B \circ A$ of tractable layouts $A$ and $B$. Eliding details, the basic idea of our algorithm is that if we want to compute the composition $B \circ A$, we can represent $A$ and $B$ by suitably chosen $\catstyle{Nest}$-morphisms $f$ and $g$, compose these morphisms to form $g \circ f$, then take the encoded layout to obtain 
\[
B \circ A = L_{g \circ f}.
\]
We illustrate this algorithm with many examples.

\section{Organization}

The current work is organized as follows. 

In section \ref{implementationsection}, we provide details regarding the $\texttt{cute}$ implementation of layouts. We provide a Python implementation of the category $\catstyle{Nest}$ in the form of a module $\texttt{tract}$, and illustrate the compatibility of $\texttt{tract}$ with $\texttt{cute}$. Our Python implementation may be found at our git repository \url{https://github.com/ColfaxResearch/layout-categories}.

Chapter \ref{layoutschapter} serves as a comprehensive reference for layouts and their algebra. It provides rigorous definitions of layouts and the operations they support, and establishes the fundamental properties of these operations. This chapter is replete with examples, and may be of use to the working programmer. 

In Chapter \ref{categorieschapter}, we present a new mathematical framework for working with tractable layouts. In particular, we connect layouts and their algebra to the theory of {\bf categories} and {\bf operads}. The content of this chapter is of independent mathematical interest. It is also of practical value, as it provides a new framework for visualizing layouts and computing their various operations.  

In Chapter \ref{computationschapter}, we provide an algorithm for computing the composite of tractable layouts $A$ and $B$ using the framework developed in Chapter \ref{categorieschapter}. We illustrate the composition algorithm with many examples.

\section{Related work}

While the current work is theoretical in nature, it is motivated by practical applications in GPU programming, most notably CUTLASS. We emphasize that the theory developed here is implementation-agnostic: it is independent of the particular programming language or runtime system used to realize layouts in practice. Nevertheless, certain practical considerations arise when working with concrete implementations. For instance, CUTLASS distinguishes between compile-time constants (static variables) and runtime values (dynamic variables). This information enables compiler optimizations during code generation. Such implementation-specific details, while important for performance, lie outside the scope of our mathematical framework. Further discussion of this can be found in the CuTe documentation \cite{cutedocumentation}.

The mathematical framework we develop for layouts draws connections to several areas of computer science and mathematics. We briefly review relevant work on GPU programming and adjacent areas to provide a greater context for our contributions.

\begin{itemize}
    \item {\bf Applications of CUTLASS.} State-of-the-art applications of CUTLASS include FlashAttention \cite{dao2023flashattention2, shah2024flashattention3}, EVT \cite{chen2024evt}, and SonicMoE \cite{sonic}. For readers seeking a deeper understanding of CUTLASS and CuTe in practice, we recommend the comprehensive tutorial series from NVIDIA~\cite{cecka2024cutlass, sun2025cutlass, thakkar2025cutlass} and Colfax Research~\cite{shah2024cutlass_streamk, shah2025cutlass_blackwell_tensor_memory, shah2025cutlass_blackwell_clusters, shah2025cutlass_blackwell_subbyte} on GPU programming with these libraries.
    \item {\bf Data layout optimization} 
    Data layout optimization techniques seek to improve cache locality and memory access patterns by carefully considering how tensors are stored in memory \cite{xu2023alt, bacon1994compiler, raman2007structure, ju1992reduction}, \cite{sharma2015data}. Choosing efficient memory storage and access patterns is crucial for GPU performance, where memory bandwidth is often a bottleneck.
    \item {\bf Modern layout systems} Layout systems such as CuTe \cite{cutedocumentation, cutedsldocumentation, shah2024layout} and Triton Linear Layouts \cite{tritonlinearlayouts, zhou2026linear} have become industry standards for managing memory storage and access in tensor computations. Triton linear layouts are based on $\mathbb{F}_2$-linear algebra, and inheret compositional structure from the composition of $\mathbb{F}_2$-linear operators. These are also naturally compatible with layout swizzles, which can generally not be represented as a CuTe layout. On the other hand, these layouts are not as expressive as CuTe layouts since they are required to have size and cosize equal to a power of $2$, and can not express transformations such as scaling by a non power-of-two integer. Recently, it was shown that both of these layout systems may be expressed in terms of integer set relations \cite{bhaskaracharya2025}. This provides a common ground for working with CuTe and Triton linear layouts, as well as more general layouts, such as those with non-rectangular shapes.
    
    \item {\bf Polyhedral compilation} 
    The polyhedral model \cite{verdoolaege2010isl}, \cite{verdoolaege2021presburger}, \cite{thangamani2024survey} provides a mathematical framework for analyzing and transforming loop nests with affine bounds and array accesses. The primary abstraction of this model is the representation of an iteration space as the collection of integer points in some polyhedron. This formalism allows for complex loop transformations that preserve program semantics while optimizing for locality and parallelism. Tools such as Pluto \cite{bondhugula2008pluto}, Polly \cite{grosser2012polly}, and Tensor Comprehensions \cite{vasilache2018tensor} leverage polyhedral techniques to automatically generate optimized code.
    \item {\bf Tensor contraction/decomposition} Tensor contractions \cite{ShiNiranjanAnandkumarCecka2016, zhao2022polyhedral, KoldaBader2009} generalize matrix multiplication to higher-rank tensors, and are ubiquitous in machine learning and scientific computing. The efficient implementation of tensor contractions relies on optimal choices of contraction order and intermediate tensor layouts.  
\end{itemize}

\clearpage

\section{Implementation} \label{implementationsection}
In this section, we illustrate how to work with layouts in NVIDIA's CuTe DSL, which we denote as $\texttt{cute}$. We provide an implementation of our categorical framework in the form of a $\texttt{Python}$ module $\texttt{tract}$ in our git repository \url{https://github.com/ColfaxResearch/layout-categories}. Here, we show the compatibility of $\texttt{cute}$ and $\texttt{tract}$.

\begin{enumerate}[left=0pt]
\item {\bf Constructing tuples and nested tuples:} We construct tuples and nested tuples in \texttt{Python} as follows. 
\begin{lstlisting}[language=cutedsl]
S = (2,2,2)
T = ((2,2),(5,5))
U = ((2,2),4,(9,(3,3)))
\end{lstlisting}
Note that if we want to construct a tuple of length $1$, we must include a comma following the tuple's entry. For example,
\begin{lstlisting}[language=cutedsl]
S = (10,)
T = (10)
\end{lstlisting}
returns 
\begin{lstlisting}[language=cutedsl]
S = (10,)
T = 10
\end{lstlisting}
\item {\bf Constucting layouts and morphisms:} We construct a layout 
\[
L = S:D
\]
in $\texttt{cute}$ as follows.
\begin{lstlisting}[language=cutedsl]
L = cute.make_layout(shape=S, stride=D)
\end{lstlisting}
For example,
\begin{lstlisting}[language=cutedsl]
A = cute.make_layout(shape=((4,4),4), stride=((16,1),4))
B = cute.make_layout(shape=(8,64), stride=(64,1))
C = cute.make_layout(shape=100, stride=2)
\end{lstlisting}
returns
\begin{lstlisting}[language=cutedsl]
A = ((4,4),4):((16,1),4)
B = (8,64):(64,1)
C = 100:2
\end{lstlisting}
We construct a nested tuple morphism 
\[\begin{tikzcd} 
S \ar[r,"f"] \ar[r,swap,"\alpha"] & T
\end{tikzcd} \]
in $\texttt{tract}$ as follows.
\begin{lstlisting}[language=cutedsl]
f = tract.make_morphism(domain=S, codomain=T, map_=alpha)
\end{lstlisting}
For example,
\begin{lstlisting}[language=cutedsl]
f = tract.make_morphism(domain=(4,4), codomain=(4,2,4), map_=(1,3))
g = tract.make_morphism(domain=(2,2,2,2), codomain=(2,2,2,2), map_=(1,0,4,2))
h = tract.make_morphism(domain=(16,(4,4),(4,4)), codomain=(16,4,4), map_=(1,2,0,3,0))
\end{lstlisting}
returns 
\begin{lstlisting}[language=cutedsl]
f = (4,4)--(1,3)-->(4,2,4)
g = (2,2,2,2)--(1,0,4,2)-->(2,2,2,2)
h = (16,(4,4),(4,4))--(1,2,0,3,0)-->(16,4,4)
\end{lstlisting}
Note that we use the symbol \texttt{0} rather than $\ast$ when specifying maps in \texttt{tract}.

\item {\bf Translating between tractable layouts and morphisms:}
If $L$ is a layout, we can check if $L$ is tractable with
\begin{lstlisting}[language=cutedsl]
tract.is_tractable(L)
\end{lstlisting}
For example, 
\begin{lstlisting}[language=cutedsl]
A = cute.make_layout(shape=(2,2,2), stride=(1,2,4))
B = cute.make_layout(shape=(2,2,2), stride=(1,7,4))
A_is_tractable = tract.is_tractable(A)
B_is_tractable = tract.is_tractable(B)
\end{lstlisting}
returns 
\begin{lstlisting}[language=cutedsl]
A = (2,2,2):(1,2,4)
B = (2,2,2):(1,7,4)
A_is_tractable = True
B_is_tractable = False
\end{lstlisting}
If $L$ is a tractable layout, then we can construct the standard representation $f_L$ with
\begin{lstlisting}[language=cutedsl]
tract.compute_morphism(L)
\end{lstlisting}
For example,
\begin{lstlisting}[language=cutedsl]
L = cute.make_layout(shape=(2,2,2), stride=(1,2,4))
f_L = tract.compute_morphism(L)
\end{lstlisting}
returns
\begin{lstlisting}[language=cutedsl]
L = (2,2,2):(1,2,4)
f_L = (2,2,2)--(1,2,3)-->(2,2,2)
\end{lstlisting}
If $f$ is a nested tuple morphism, we can construct the layout $L_f$ encoded by $f$ with
\begin{lstlisting}[language=cutedsl]
tract.compute_layout(f)
\end{lstlisting}
For example,
\begin{lstlisting}[language=cutedsl]
f = tract.make_morphism(domain=((5,5),8), codomain=(5,8,5), map_=(1,3,2))
L_f = tract.compute_layout(f)
\end{lstlisting}
returns 
\begin{lstlisting}[language=cutedsl]
f = ((5,5),8)--(1,3,2)-->(5,8,5)
L_f = ((5,5),8):((1,40),5)
\end{lstlisting}

\item {\bf Composition}: When defined, this operation produces a layout $B \circ A$ from a pair of layouts $A$ and $B$. See Definition \ref{definitionoflayoutcomposition} for a precise definition.  We can compute the composition $B \circ A$ in $\texttt{cute}$ with
\begin{lstlisting}[language=cutedsl]
cute.composition(B,A)
\end{lstlisting}
For example, running 
\begin{lstlisting}[language=cutedsl]
A = cute.make_layout(shape=((4,4),4), stride=((16,1),4))
B = cute.make_layout(shape=(8,64), stride=(64,1))
B_o_A = cute.composition(B,A)
\end{lstlisting}
returns 
\begin{lstlisting}[language=cutedsl]
A = ((4,4),4):((16,1),4)
B = (8,64):(64,1)
B_o_A = ((4,4),(2,2)):((2,64),(256,1))
\end{lstlisting}
If $f$ and $g$ are composable nested tuple morphisms, we can compute the composition $g \circ f$ in $\texttt{tract}$ with
\begin{lstlisting}[language=cutedsl]
tract.compose(f,g)
\end{lstlisting}
For example,
\begin{lstlisting}[language=cutedsl]
f = tract.make_morphism(domain=((2,2),(2,2)), codomain=((2,2,2),(2,2,2)), map_=(3,2,6,5))
g = tract.make_morphism(domain=((2,2,2),(2,2,2)), codomain=(2,2,2,2), map_=(1,0,2,0,3,4))
g_o_f = tract.compose(f,g)
\end{lstlisting}
returns 
\begin{lstlisting}[language=cutedsl]
f = ((2,2),(2,2))--(3,2,6,5)-->((2,2,2),(2,2,2))
g = ((2,2,2),(2,2,2))--(1,0,2,0,3,4)-->(2,2,2,2)
g_o_f = ((2,2),(2,2))--(2,0,4,3)-->(2,2,2,2)
\end{lstlisting}

\item {\bf Coalesce}: This operation produces a layout $\coalesce(A)$ from a layout $A$. See Definition \ref{definitionofnestedlayoutcoalesce} for details. We can compute $\coalesce(A)$ in $\texttt{cute}$ with 
\begin{lstlisting}[language=cutedsl]
cute.coalesce(A)
\end{lstlisting}
For example, 
\begin{lstlisting}[language=cutedsl]
A = cute.make_layout(shape = ((2,2),(2,2),(5,5)), stride = ((1,2),(16,32),(64,640)))
coal_A = cute.coalesce(A)
\end{lstlisting}
returns 
\begin{lstlisting}[language=cutedsl]
A = ((2,2),(2,2),(5,5)):((1,2),(16,32),(64,640))
coal_A = (4,20,5):(1,16,640)
\end{lstlisting}
There is also a {\it relative coalesce} operation $A \mapsto \coalesce(A,S)$, which receives as input an additional nested tuple $S$ which is {\it refined} by the shape of $A$. See Definition \ref{definitionofrelativecoalesce} for details. We can compute $\coalesce(A,S)$ in $\texttt{cute}$ with 
\begin{lstlisting}[language=cutedsl]
A = cute.make_layout(shape = ((2,2),(3,3),(5,5)), stride = ((1,2),(4,12),(36,180)))
S = ((2,2),9,25)
coal_A_over_S = cute.coalesce(A,target_profile=S)
\end{lstlisting}
returns 
\begin{lstlisting}[language=cutedsl]
A = ((2,2),(3,3),(5,5)):((1,2),(4,12),(36,180))
S = ((2,2),9,25)
coal_A_over_S = ((2,2),9,25):((1,2),4,36)
\end{lstlisting}
If $f$ is a nested tuple morphism, we may form $\coalesce(f)$. See Definition \ref{definitionofcoalesceofnestedtuplemorphism} for details. We compute $\coalesce(f)$ in $\texttt{tract}$ with 
\begin{lstlisting}[language=cutedsl]
tract.coalesce(f)
\end{lstlisting}
For example, 
\begin{lstlisting}[language=cutedsl]
f = tract.make_morphism(domain=(2,2,10,10), codomain = (2,2,2,10,10), map_=(1,2,4,5))
coal_f = tract.coalesce(f)
\end{lstlisting}
returns
\begin{lstlisting}[language=cutedsl]
f = (2,2,10,10)--(1,2,4,5)-->(2,2,2,10,10)
coal_f = (4,100)--(1,3)-->(4,2,100)
\end{lstlisting}

\item {\bf Complement}: When defined, this operation produces a layout $\comp(A,N)$ from a layout $A$ and positive integer $N$. See Definition \ref{definitionoflayoutcomplements} for details. We can compute $\comp(A,N)$ in $\texttt{cute}$ with 
\begin{lstlisting}[language=cutedsl]
cute.complement(A,N)
\end{lstlisting}
For example,
\begin{lstlisting}[language=cutedsl]
A = cute.make_layout(shape = ((2,2),(2,2)), stride = ((8,2),(64,256)))
comp_A = cute.complement(A,4096)
\end{lstlisting}
returns 
\begin{lstlisting}[language=cutedsl]
A = ((2,2),(2,2)):((8,2),(64,256))
comp_A = (2,2,4,2,8):(1,4,16,128,512)
\end{lstlisting}
If $f$ is a nested tuple morphism, then we may form the complement $f^c$ of $f$. See Definition \ref{complementofmorphisminD} for details. We compute $f^c$ in $\texttt{tract}$ with 
\begin{lstlisting}[language=cutedsl]
tract.complement(f)
\end{lstlisting}
For example,
\begin{lstlisting}[language=cutedsl]
f = tract.make_morphism(domain=(2,2), codomain=(2,5,2,5), map_=(1,3))
comp_f = tract.complement(f)
\end{lstlisting}
returns 
\begin{lstlisting}[language=cutedsl]
f = (2,2)--(1,3)-->(2,5,2,5)
comp_A = (5,5)--(2,4)-->(2,5,2,5)
\end{lstlisting}
\item {\bf Logical Division}: When defined, this operation produces a layout $A \oslash B$ from a pair of layouts $A$ and $B$. See Definition \ref{definitionoflogicaldivide} for details. We compute $A \oslash B$ in $\texttt{cute}$ with 
\begin{lstlisting}[language=cutedsl]
cute.logical_divide(A,B)
\end{lstlisting}
For example,
\begin{lstlisting}[language=cutedsl]
A = cute.make_layout((64,32), stride = (32,1))
B = cute.make_layout((4,4), stride = (1,64))
quotient = cute.logical_divide(A,B)
\end{lstlisting}
returns 
\begin{lstlisting}[language=cutedsl]
A = (64,32):(32,1)
B = (4,4):(1,64)
quotient = ((4,4),(16,8)):((32,1),(128,4))
\end{lstlisting}
If $f$ and $g$ are nested tuple morphisms and $g$ divides $f$, then we may form the logical division $f \oslash g$. See Definition \ref{definitionoflogicaldivisionofnestedtuplemorphisms} for details. We compute $f \oslash g$ in $\texttt{tract}$ with 
\begin{lstlisting}[language=cutedsl]
tract.logical_divide(f,g)
\end{lstlisting}
For example,
\begin{lstlisting}[language=cutedsl]
f = tract.make_morphism(domain=(4,8,4,8), codomain=(4,8,4,8), map_=(1,2,3,4))
g = tract.make_morphism(domain=(4,4), codomain=(4,8,4,8), map_=(1,3))
quotient = tract.logical_divide(f,g)
\end{lstlisting}
returns 
\begin{lstlisting}[language=cutedsl]
f = (4,8,4,8)--(1,2,3,4)-->(4,8,4,8)
g = (4,4)--(1,3)-->(4,8,4,8)
quotient = ((4,4),(8,8))--(1,3,2,4)-->(4,8,4,8)
\end{lstlisting}
\item {\bf Logical Product}: When defined, this operation produces a layout $A \otimes B$ from a pair of layouts $A$ and $B$. See Definition \ref{definitionoflogicalproduct} for details. We compute $A \otimes B$ in $\texttt{cute}$ with 
\begin{lstlisting}[language=cutedsl]
cute.logical_product(A,B)
\end{lstlisting}
For example, running 
\begin{lstlisting}[language=cutedsl]
A = cute.make_layout((3,10,10), stride = (200,1,20))
B = cute.make_layout((2,2), stride = (1,2))
product = cute.logical_product(A,B)
\end{lstlisting}
returns 
\begin{lstlisting}[language=cutedsl]
A = (3,10,10):(200,1,20)
B = (2,2):(1,2)
product = ((3,10,10),(2,2)):((200,1,20),(10,600))
\end{lstlisting}

\end{enumerate}

If $f$ and $g$ are nested tuple morphisms and $f$ and $g$ are product admissible, then we may form the logical product $f \otimes g$. See Definition \ref{definitionoflogicalproductofnestedtuplemorphisms} for details. We compute $f \otimes g$ in $\texttt{tract}$ with 
\begin{lstlisting}[language=cutedsl]
tract.logical_product(f,g)
\end{lstlisting}
For example,
\begin{lstlisting}[language=cutedsl]
f = tract.make_morphism(domain=(2,2), codomain=(2,2,5,5), map_=(1,2))
g = tract.make_morphism(domain=(5,5), codomain=(5,5), map_=(2,1))
product = tract.logical_product(f,g)
\end{lstlisting}
returns 
\begin{lstlisting}[language=cutedsl]
f = (2,2)--(1,2)-->(2,2,5,5)
g = (5,5)--(2,1)-->(5,5)
product = ((2,2),(5,5))--(1,2,4,3)-->(2,2,5,5)
\end{lstlisting}

\clearpage

\section{Notation}

\begin{align*}
\mathbb{Z} & = \{\dots,-1,0,1,2,\dots\}\\
\mathbb{N} & = \{ 0 , 1, 2 ,\dots\}\\
\mathbb{Z}_{>0} & = \{1,2,\dots \}\\
\mathbb{F}_2 & = \{0,1\}\text{, the finite field of order } 2.\\
[0,n) & = \{0,\dots,n-1\}\text{, and }[0,0) = \varnothing. \\
\langle n \rangle & = \{1,2,\dots,n\}\text{, and }\langle 0 \rangle = \varnothing. \\
\langle n \rangle_* & = \{*,1,2,\dots,n\}\\
\delta_i^m & = (0,\dots,1,\dots,0)\text{, the tuple of length } m \text{ with }\\
& \hspace{0.2in} i\text{th entry }1\text{ and all other entries } 0. \\
\Sigma_n & = \text{ the symmetric group on } \langle n \rangle.\\
X^\sigma & = (x_{\sigma(1)},\dots,x_{\sigma(m)}) \text{ for a tuple }X = (x_1,\dots,x_m)\\
& \hspace{0.2in} \text{ and a permutation } \sigma \in \Sigma_m.\\
X \star Y & = \text{ the flat concatenation of }X \text{ and }Y.\\
X^\flat & = \text{ the flattening of a nested tuple }X.\\
\profile(X) & = \text{ the profile of a nested tuple }X.\\
(X_1,\dots,X_k) & = \text{ the (nested) concatenation of }X_1,\dots,X_k. \\
(X_1,\dots,X_k)_Q & = \text{ the }Q\text{-substitution of }X_1,\dots,X_k\text{ for a profile }Q.\\
\mathterm{Tuple}(V) & = \text{ the set of tuples with entries in a set }V.\\
\mathterm{Nest}(V) & = \text{ the set of nested tuples with entries in a set }V.\\
\mathterm{Profile} & = \text{ the set of profiles.}\\
\mathterm{FlatLayout} & = \text{ the set of flat layouts.}\\
\mathterm{Layout} & = \text{ the set of layouts.}\\
B \circ A & = \text{ the composition of }A \text{ and }B. \\
A \oslash B & = \text{ the logical division of }A \text{ by }B. \\
A \otimes B & = \text{ the logical product of }A \text{ and }B. \\
\catstyle{Set} & = \text{ the category of sets.} \\
\catstyle{FinSet} & = \text{ the category of finite sets.}\\
\catstyle{Fin} & = \text{ the full subcategory of }\catstyle{FinSet} \text{ spanned by }\n \text{ for } n \geq 0.\\
\catstyle{FinSet}_* & = \text{ the category of pointed finite sets.}\\
\catstyle{Fin}_*&  = \text{ the full subcategory of }\catstyle{FinSet}_*\text{ spanned by } \n_* \text{ for } n \geq 0. \\
\catstyle{Tuple}& = \text{ the category of tuples and tuple morphisms.}\\
\catstyle{Nest} & = \text{ the category of nested tuples and nested tuple morphisms.}\\
\catstyle{Ref} & = \text{ the category of nested tuples and refinements.}\\
\catstyle{Cat} & = \text{ the category of (small) categories and functors.}
\end{align*}

\newpage

\newpage

\chapter{Layouts and their algebra}\label{layoutschapter}

The goal of this chapter is to provide a comprehensive and mathematically grounded theory of layouts. We begin by developing a theory of flat layouts in section \ref{flatlayoutssection}. We introduce the necessary background on nested tuples in section \ref{nestedtuplessection}, so that we may cover layouts in full generality in section \ref{layoutssection}.

\section{Flat Layouts}\label{flatlayoutssection}

In this section, we examine {\bf flat layouts}, an important subclass of layouts in which both shape and stride are tuples, rather than more general nested tuples. To formalize our discussion, we begin by fixing notation related to tuples.

\subsection{Tuples}\label{tuplessection}

\begin{definition}
    If $V$ is a set, then a {\it tuple} with entries in $V$ is a finite ordered list
    \[
    X = (x_1,\dots,x_m)
    \]
    of elements $x_i \in V$ for each $1 \le i \le m$. The {\it length} of such a tuple $X=(x_1,\dots,x_m)$ is 
    \[\len(X) = m.\]
    We write $\mathterm{Tuple}(V)$ for the collection of all tuples with entries in $V$. We are especially interested in the case $V = \mathbb{Z}$, in which case we refer to $X \in \mathterm{Tuple}(\mathbb{Z})$ as a {\it tuple of integers. } If $X$ is a tuple of integers, then the {\it size} of $X$ is the product 
    \[
        \size(X) = x_1 \cdots x_m.
        \]
\end{definition}

\begin{example}
Here are some examples of tuples, together with their length and size:
\[
\begin{aligned}
X &= (3, 128, 128), & \quad \len(X) &= 3, \quad \size(X) = 49152 \\
X &= (512),         & \quad \len(X) &= 1, \quad \size(X) = 512 \\
X &= (),            & \quad \len(X) &= 0, \quad \size(X) = 1
\end{aligned}
\]
\end{example}

\begin{definition}\label{definitionofconcatenationoftuples}
If $ X = (x_1,\dots,x_m)$ and $Y = (y_1,\dots,y_n)$ are tuples, then we write 
\[
X \star Y = (x_1,\dots,x_m,y_1,\dots,y_n)
\]
for the {\it concatenation} of $X$ and $Y$.
\end{definition}

\begin{example}
    If $X = (64,32)$ and $Y = (8,8,8)$, then 
    \[
    X \star Y = (64,32,8,8,8).
    \]
\end{example}

\begin{remark}
    If $V$ is a set, then the collection
    \[\mathterm{Tuple}(V) = \coprod_{m \geq 0} V^{\times m}\]
    of all tuples with entries in $V$ is the free associative monoid on $V$. The monoidal product is concatenation, and the unit is the empty tuple $()$. 
\end{remark}
 
\begin{definition}\label{definitionofdivisibilityoftuples}
If $X$ and $X'$ are tuples, we say $X'$ {\it divides} $X$ if there exists a tuple $X''$ with 
\[
X' \star X'' = X.
\]
\end{definition}

\begin{example}
    If $X' = (81,9)$ and $X = (81,9,64,8)$, then $X'$ divides $X$, since the tuple $X'' = (64,8)$ satisfies 
    \[
    X' \star X'' = X.
    \]
\end{example}
\begin{definition} If $X=(x_1,\dots,x_m)$ is a tuple and $\sigma \in \Sigma_m$ is a permutation, then we write
\[
X^\sigma = (x_{\sigma(1)},\dots,x_{\sigma(m)})
\]
for the {\it permutation of} $X$ {\it by} $\sigma$. This specifies a right action of $\Sigma_m$ on $\mathbb{Z}^{\times m}$. 
\end{definition}

\begin{example}
    If $X = (8,16,32,64)$ and $\sigma = (1 \; 2)(3 \; 4)$, then
    \[
    X^\sigma = (16,8,64,32).
    \]
\end{example}

\begin{notation}
If $n$ is a positive integer, we write 
\[
[0,n) = \{ 0,1,\dots,n-1\},
\]
and if $S = (s_1,\dots,s_m)$ is a tuple of positive integers, we write 
\begin{align*}
[0,S) &  = [0,s_1) \times \cdots \times [0,s_m)
\end{align*}
for the collection of tuples $(x_1,\dots,x_m)$ with $0 \leq x_i < s_i$. 
\end{notation}

\begin{example}
    If $S = (3,2)$, then 
    \[
    [0,S) = \{(0,0),(1,0),(2,0),(0,1),(1,1),(2,1) \}
    \]
\end{example}
% \begin{example}
%     If $S = (2,2,2)$, then 
%     \begin{align*}
%     [0,S) & = \{(0,0,0),(1,0,0),(0,1,0),(1,1,0),(0,0,1),(1,0,1),(0,1,1),(1,1,1)\}.
%     \end{align*}
% \end{example}

\subsection{Basic definitions}

Having fixed notation, we are ready to define flat layouts.

\begin{definition}\label{definitionofflatlayout}
    A {\it flat layout} is a pair 
    \begin{align*}
        L & = S:D \\
        & = (s_1,\dots,s_m):(d_1,\dots,d_m)
    \end{align*}
    consisting of a tuple of positive integers 
    \begin{align*} 
    \shape(L) & = S\\
    & =(s_1,\dots,s_m)
    \end{align*}
    called the {\it shape} of $L$, and a tuple of non-negative integers 
    \begin{align*} 
    \stride(L) & = D\\
    & = (d_1,\dots,d_m)
    \end{align*}
    called the {\it stride} of $L$.
\end{definition}

\begin{remark}
    If $L$ is a flat layout, then by definition, $\shape(L)$ and $\stride(L)$ have the same length.
\end{remark}

\begin{remark}
    A flat layout is an example of the more general {\it layout} of Definition \ref{definitionofnestedlayout}, so we sometimes refer to a flat layout $L$ as a {\it layout}.
\end{remark}

\begin{example}
    Here are some examples of flat layouts:
    \begin{align*} 
    L_1 & = (2,2,2):(1,2,4),\\
    L_2 & = (128):(5),\\
    L_3 & = (16,12,512,512):(0,0,1,512),\\
    L_4 & = (6,1,12,2,2):(2,0,12,144,1),\\
    L_5 & = () : ().
    \end{align*}
\end{example}

\begin{example} We can depict the layout $L = (8):(5)$ as

\bigskip 

\begin{centering}

\begin{tikzpicture}[x={(0cm,-1cm)},y={(1cm,0cm)},every node/.style={minimum size=1cm, outer sep=0pt}]

\node[fill=gray!20] at (0,0) {0};
\node[fill=gray!20] at (0,1) {5};
\node[fill=gray!20] at (0,2) {10};
\node[fill=gray!20] at (0,3) {15};
\node[fill=gray!20] at (0,4) {20};
\node[fill=gray!20] at (0,5) {25};
\node[fill=gray!20] at (0,6) {30};
\node[fill=gray!20] at (0,7) {35};
\draw[color=black,thick,shift={(-0.5,-0.5)}] (0,0) grid (1,8);

\node[anchor = east] at (0,-1) {$L = $};
\end{tikzpicture}

\end{centering}

\bigskip

\noindent and we can depict the layout $L = (3,5):(2,10)$ as

\bigskip 

\begin{centering}

\begin{tikzpicture}[x={(0cm,-1cm)},y={(1cm,0cm)},every node/.style={minimum size=1cm, outer sep=0pt}]

\node[fill=gray!20] at (0,0) {0};
\node[fill=gray!20] at (0,1) {10};
\node[fill=gray!20] at (0,2) {20};
\node[fill=gray!20] at (0,3) {30};
\node[fill=gray!20] at (0,4) {40};
\node[fill=gray!20] at (1,0) {2};
\node[fill=gray!20] at (1,1) {12};
\node[fill=gray!20] at (1,2) {22};
\node[fill=gray!20] at (1,3) {32};
\node[fill=gray!20] at (1,4) {42};
\node[fill=gray!20] at (2,0) {4};
\node[fill=gray!20] at (2,1) {14};
\node[fill=gray!20] at (2,2) {24};
\node[fill=gray!20] at (2,3) {34};
\node[fill=gray!20] at (2,4) {44};
\draw[color=black,thick,shift={(-0.5,-0.5)}] (0,0) grid (3,5);

\node[anchor = east] at (1,-1) {$L = $};
\end{tikzpicture}

\end{centering}

\bigskip 

\noindent We make precise the sense in which these pictures represent the associated layout in Remark \ref{picturesoflayouts}.

\end{example}

\noindent Perhaps the most important examples of flat layouts are the column-major and row-major layouts, which we define below.

\begin{definition}\label{definitionofcolumnmajor} Suppose 
\[
L = (s_1,\dots,s_m):(d_1,\dots,d_m)
\]
is a flat layout. We say $L$ is {\it column-major} if 
    \[
    d_i = s_1 \cdots s_{i-1}
    \]
for each $1 \leq i \leq m$. We say $L$ is {\it row-major} if  
\[
d_i = s_{i+1} \cdots s_m.
\]
for each $1 \leq i \leq m$.
\end{definition}

\begin{example} The layout

\bigskip

\begin{centering} 

\begin{tikzpicture}[x={(0cm,-1cm)},y={(1cm,0cm)},every node/.style={minimum size=1cm, outer sep=0pt}]

\node[fill=gray!20] at (0,0) {0};
\node[fill=gray!20] at (0,1) {3};
\node[fill=gray!20] at (0,2) {6};
\node[fill=gray!20] at (0,3) {9};
\node[fill=gray!20] at (1,0) {1};
\node[fill=gray!20] at (1,1) {4};
\node[fill=gray!20] at (1,2) {7};
\node[fill=gray!20] at (1,3) {10};
\node[fill=gray!20] at (2,0) {2};
\node[fill=gray!20] at (2,1) {5};
\node[fill=gray!20] at (2,2) {8};
\node[fill=gray!20] at (2,3) {11};
\draw[color=black,thick,shift={(-0.5,-0.5)}] (0,0) grid (3,4);

\node[anchor=east] at (1,-1) {$L = (3,4):(1,3) = $};

\end{tikzpicture}

\end{centering}

\bigskip

\noindent is column-major, while the layout 

\bigskip

\begin{centering} 

\begin{tikzpicture}[x={(0cm,-1cm)},y={(1cm,0cm)},every node/.style={minimum size=1cm, outer sep=0pt}]

\node[fill=gray!20] at (0,0) {0};
\node[fill=gray!20] at (0,1) {1};
\node[fill=gray!20] at (0,2) {2};
\node[fill=gray!20] at (0,3) {3};
\node[fill=gray!20] at (1,0) {4};
\node[fill=gray!20] at (1,1) {5};
\node[fill=gray!20] at (1,2) {6};
\node[fill=gray!20] at (1,3) {7};
\node[fill=gray!20] at (2,0) {8};
\node[fill=gray!20] at (2,1) {9};
\node[fill=gray!20] at (2,2) {10};
\node[fill=gray!20] at (2,3) {11};
\draw[color=black,thick,shift={(-0.5,-0.5)}] (0,0) grid (3,4);

\node[anchor=east] at (1,-1) {$L = (3,4):(4,1) = $};

\end{tikzpicture}

\end{centering}

\bigskip

\noindent is row-major. These pictures make clear the reason for the terminology: If $L$ is a column-major layout of rank $2$, then the columns of $L$ are contiguous, and if $L$ is a row-major layout of  rank $2$, then the rows of $L$ are contiguous.

\end{example} 

\begin{example}
The layouts
    \begin{align*}
    L_1 & = (2,2,2,2,2):(1,2,4,8,16)\\
    L_2 & = (3,128,128):(1,3,384)\\
    L_3 & = (64):(1)
    \end{align*}
    are column-major, while the layouts
    \begin{align*}
    L_4 & = (2,2,2,2,2):(16,8,4,2,1)\\
    L_5 & = (3,128,128):(16384,128,1)\\
    L_6 & = (64):(1)
    \end{align*}
    are row-major. 
\end{example}

Now that we've seen a few examples, lets define some important attributes of flat layouts.

\begin{definition}\label{definitionofcosize} Suppose $L = (s_1,\dots,s_m):(d_1,\dots,d_m)$ is a flat layout.
\begin{itemize} 
\item The {\it rank} of $L$ is
\[\rank(L) = m.\]
\item The {\it size} of $L$ is 
    \[
    \size(L) = \prod_{i=1}^m s_i.
    \]
\item The {\it cosize} of $L$ is 
    \[
    \cosize(L) = 1 + \sum_{i=1}^m (s_i-1) \cdot d_i.
    \]
    \item For any $1 \leq i \leq \rank(L)$, the $i${\it th mode} of $L$ is the pair 
    \[
    \mode_i(L) = s_i:d_i.
    \]
\end{itemize}
\end{definition}

\begin{example}
    The layout 
    \[L = (64,32):(1,128)\]
    has $\rank(L) = 2$, $\size(L) = 2048$, and $\cosize(L) = 4032$. The modes of $L$ are 
    \begin{align*} \mode_1(L) & = 64:1\\
    \mode_2(L) & =  32:128.
    \end{align*}
\end{example}

\begin{example}
The layout
\[
L = (3,8,8,8) : (1,3,24,192).
\]
has $\rank(L) = 4$, $\size(L) = 1536$, and $\cosize(L) = 1536$. The layout $L$ has four modes, for example $\mode_3(L) = 8:24$. 
\end{example}

\begin{example}
The layout 
\[
L = (2,2,2,2,2) : (160,80,40,20,10).
\]
has $\rank(L) = 5$, $\size(L) = 32$, and $\cosize(L) = 311$. The layout $L$ has $5$ modes, for example $\mode_5(L) = 2:10$. 
\end{example}

If $L$ is a flat layout, then $L$ encodes a {\it coordinate function} $\varphi_L$. The coordinate function of $L$ is a multi-dimensional to one-dimensional transformation given by taking a dot product with $\stride(L)$. Recall that if $S = (s_1,\dots,s_m)$ is a tuple of positive integers, then 
\[
[0,S) = [0,s_1) \times \cdots \times [0,s_m)
\]
is the set of all tuples $(x_1,\dots,x_m)$ such that $0 \leq x_i < s_i$. In particular, if $S = ()$ is the empty tuple, then $[0,S) = \{()\}$. 

\begin{construction}[Coordinate functions]\label{constructionofcoordinatefunctions}
If
\[
L = (s_1,\dots,s_m):(d_1,\dots,d_m)
\]
is a flat layout, then the {\it coordinate function} of $L$ is the function 
\[\begin{tikzcd}
{[}0,\shape(L){)} \ar[r,"\varphi_L"] &  \mathbb{Z}
\end{tikzcd} \]
given by 
\begin{align*}
\varphi_L(x_1,\dots,x_m) & = (x_1,\dots,x_m) \cdot (d_1,\dots,d_m)\\
& = x_1 d_1 + \cdots + x_m d_m.
\end{align*}
    The coordinate function $\varphi_L$ factors through the inclusion $[0,\cosize(L)) \subset \mathbb{Z}$, and we write 
    \[\begin{tikzcd} {[}0,\shape(L){)} \ar[rr,"\varphi_L^{\cosize(L)}"] & & {[}0,\cosize(L){)}\subset \mathbb{Z}
    \end{tikzcd} \]
    for the factored map. More generally, for any $N \geq \cosize(L)$, we write $\varphi_L^N$ for the factorization of $\varphi_L$ through $[0,N) \subset \mathbb{Z}$, and by a mild abuse of terminology, we refer to such a map $\varphi_L^N$ as the {\it coordinate function} of $L$.
\end{construction}

\begin{example}\label{layoutfunctionexample1}
    If $L = (2,3):(1,5)$, then the coordinate function
    \[
    \varphi_L:[0,2) \times [0,3) \to \mathbb{Z}
    \]
    is given by 
    \begin{align*}
    \varphi_L(0,0) & = (0,0) \cdot (1,5) = 0,\\
    \varphi_L(1,0) & = (1,0) \cdot (1,5) = 1,\\
    \varphi_L(0,1) & = (0,1) \cdot (1,5) = 5,\\
    \varphi_L(1,1) & = (1,1) \cdot (1,5) = 6,\\
    \varphi_L(0,2) & = (0,2) \cdot (1,5) = 10,\\
    \varphi_L(1,2) & = (1,2) \cdot (1,5) = 11.
    \end{align*}
\end{example}

\begin{example}\label{layoutfunctionexample1}
    If $L = (2,2):(64,2)$, then the coordinate function
    \[
    \varphi_L:[0,2) \times [0,2) \to \mathbb{Z}
    \]
    is given by 
    \begin{align*}
    \varphi_L(0,0) & = (0,0) \cdot (64,2) = 0,\\
    \varphi_L(1,0) & = (1,0) \cdot (64,2) = 64,\\
    \varphi_L(0,1) & = (0,1) \cdot (64,2) = 2,\\
    \varphi_L(1,1) & = (1,1) \cdot (64,2) = 66.\\
    \end{align*}
\end{example}

\begin{example}
    If $E = ():()$ is the empty layout, then the coordinate function of $E$ is the map 
    \[
    \varphi_E: \{()\} \to \mathbb{Z}
    \]
    given by 
    \[
    \varphi(()) = 0.
    \]
\end{example}

\begin{remark}\label{picturesoflayouts} We can now, for example, give a precise description of the sense in which the image 
\bigskip

\begin{centering} 

\begin{tikzpicture}[x={(0cm,-1cm)},y={(1cm,0cm)},every node/.style={minimum size=1cm, outer sep=0pt}]

\node[fill=gray!20] at (0,0) {0};
\node[fill=gray!20] at (0,1) {10};
\node[fill=gray!20] at (0,2) {20};
\node[fill=gray!20] at (0,3) {30};
\node[fill=gray!20] at (0,4) {40};
\node[fill=gray!20] at (1,0) {2};
\node[fill=gray!20] at (1,1) {12};
\node[fill=gray!20] at (1,2) {22};
\node[fill=gray!20] at (1,3) {32};
\node[fill=gray!20] at (1,4) {42};
\node[fill=gray!20] at (2,0) {4};
\node[fill=gray!20] at (2,1) {14};
\node[fill=gray!20] at (2,2) {24};
\node[fill=gray!20] at (2,3) {34};
\node[fill=gray!20] at (2,4) {44};
\draw[color=black,thick,shift={(-0.5,-0.5)}] (0,0) grid (3,5);

% \node[anchor = east] at (1,-1) {$L = $};
\end{tikzpicture}

\end{centering}

\bigskip

\noindent depicts the layout $L = (3,5):(2,10)$: The $(i,j)$th cell of the grid is labeled by the value 
\[\varphi_L(i,j) = (i,j) \cdot (2,10) = 2i + 10j\]
of the coordinate function of $L$. 
\end{remark}

In practice, the most important invariant of a flat layout $L$ is its {\it layout function} $\Phi_L$, which is obtained by precomposing the coordinate function 
\[
\varphi_L:[0,S) \to \mathbb{Z}
\]
with the inverse of the {\it colexicographic isomorphism}
\[
\colex_S:[0,S) \to [0,\size(S)).
\]

\begin{definition}\label{definitionofcolex}
    Suppose $S = (s_1,\dots,s_m)$ is a tuple of positive integers and recall that 
    \[
    [0,S) = [0,s_1) \times \cdots \times [0,s_m).
    \]
    The {\it colexicographic isorphism} is the map 
    \[ \begin{tikzcd}
    {[}0,S{)} \ar[rr,"\colex_{S}"] & & \text{[}0,\size(S)\text{)}\\
    (x_1,\dots,x_m) \ar[rr,mapsto] & & \sum_{i=1}^m s_1 \cdots s_{i-1} x_i.
    \end{tikzcd} \]
    We sometimes write $\colex = \colex_{S}$ when the tuple $S$ is clear from context. The inverse of the colexicographic isomorphism is the map 
    \[ \begin{tikzcd}
    \text{[}0,\size(S)\text{)} \ar[rr,"\colex_{S}^{-1}"] & & {[}0,S{)}
    \end{tikzcd} \]
    given by 
    \[
    \colex_S^{-1}(x) = (x_1,\dots,x_m)
    \]
    where 
    \[
    x_i = \left\lfloor \dfrac{x}{s_1 \cdots s_{i-1}} \right\rfloor \mod s_i.
    \]
    Note that if $S = ()$ is the empty tuple, then
    \[
    \colex_{()}:\{()\} \to \{0\}
    \]
    and 
    \[
    \colex_{()}^{-1}:\{0\} \to \{()\}
    \]
    are the canonical isomorphisms.
\end{definition}

\begin{construction}[Layout functions]\label{constructionoflayoutfunctions}
    If
    \[L = (s_1,\dots,s_m):(d_1,\dots,d_m),\]
    is a flat layout, then the {\it layout function of }$L$ is the composite 
    \[ \begin{tikzcd} 
    {[}0,\size(L){)} \ar[rr,"\Phi_L"] \ar[dr,swap,"\colex^{-1}_{\shape(L)}"] & & \mathbb{Z}. \\
     & {[}0,\shape(L){)} \ar[ur,swap,"\varphi_L"] & 
    \end{tikzcd} \]
    Explicitly, $\Phi_L$ is given by 
    \[
    \Phi_L(x) = x_1 d_1 + \cdots + x_m d_m
    \]
    where 
    \[
    x_i = \left \lfloor \dfrac{x}{s_1 \cdots s_{i-1}} \right \rfloor \mod s_i.
    \]
    The layout function $\Phi_L$ factors through the inclusion $[0,\cosize(L)) \subset \mathbb{Z}$, and we write 
    \[\begin{tikzcd} {[}0,\size(L){)} \ar[rr,"\Phi_L^{\cosize(L)}"] & & {[}0,\cosize(L){)}\subset \mathbb{Z}
    \end{tikzcd} \]
    for the factored map. More generally, for any $N \geq \cosize(L)$, we write $\Phi_L^N$ for the factorization of $\Phi_L$ through $[0,N) \subset \mathbb{Z}$, and by a mild abuse of terminology, we refer to such a map $\varphi_L^N$ as the {\it layout function} of $L$.
\end{construction}

\begin{example}\label{layoutfunctionexample1}
    If $L = (2,3):(1,5)$, then the layout function
    \[
    \Phi_L:[0,6) \to \mathbb{Z}
    \]
    is given by 
    \begin{align*}
    \Phi_L(0) & = (0,0) \cdot (1,5) = 0,\\
    \Phi_L(1) & = (1,0) \cdot (1,5) = 1,\\
    \Phi_L(2) & = (0,1) \cdot (1,5) = 5,\\
    \Phi_L(3) & = (1,1) \cdot (1,5) = 6,\\
    \Phi_L(4) & = (0,2) \cdot (1,5) = 10,\\
    \Phi_L(5) & = (1,2) \cdot (1,5) = 11.
    \end{align*}
\end{example}

\begin{example}\label{layoutfunctionexample1}
    If $L = (2,2):(64,2)$, then the layout function
    \[
    \Phi_L:[0,4) \to \mathbb{Z}
    \]
    is given by 
    \begin{align*}
    \Phi_L(0) & = (0,0) \cdot (64,2) = 0,\\
    \Phi_L(1) & = (1,0) \cdot (64,2) = 64,\\
    \Phi_L(2) & = (0,1) \cdot (64,2) = 2,\\
    \Phi_L(3) & = (1,1) \cdot (64,2) = 66.\\
    \end{align*}
\end{example}

\begin{example}\label{layoutfunctionexample2}
    If $L = (4,2,2):(3,3,100)$, then for example, the layout function of $L$ satisfies
    \begin{align*}
    \Phi_L(7) & = (3,1,0) \cdot (3,3,100) = 12\text{,}\\
    \Phi_L(9) & = (1,0,1) \cdot (3,3,100) = 103.
    \end{align*}
\end{example}
\begin{example}
    If $E = ():()$ is the empty layout, then 
    \[
    \Phi_E:\{0\} \to \mathbb{Z}
    \]
    is given by 
    \[
    \Phi_E(0) = 0.
    \]
\end{example}

\begin{example}
    If $L$ is any flat layout, then the layout function $\Phi_L$ of $L$ satisfies
    \[
    \Phi_L(0) = 0.
    \]
\end{example}

\begin{remark}
    If $S = (s_1,\dots,s_m)$ is a tuple of positive integers, then the colexicographic isomorphism
    \[\begin{tikzcd}
    {[}0,S{)} \ar[rr,"\colex_S"] & & {[}0,\size(S){)}
    \end{tikzcd} \]
    is equal to the coordinate function $\varphi_{L}^{\cosize(L)}$ of the column major layout 
    \[
    L = (s_1,s_2,\dots,s_m):(1,s_1,\dots,s_1 \cdots s_{m-1}). 
    \]
    This implies that if a flat layout $L$ is column-major, then 
    \begin{align*}
        \Phi_L^{\cosize(L)} & = \varphi_L^{\cosize(L)} \circ \colex_{\shape(L)}^{-1} \\
        & =  \varphi_L^{\cosize(L)} \circ \left(\varphi_L^{\cosize(L)}\right)^{-1}\\
        & = \id_{{[}0,\size(L){)}}
    \end{align*}
    is the identity map on $[0,\size(L))$.
\end{remark}

% \begin{remark}
%     The notion of coordinate functions and layout functions motivates the definition of cosize: If $L$ is a flat layout, then the cosize of $L$ is the smallest integer $N$ such that $\Phi_L$ factors through $[0,N)$:
%     \begin{align*}
%     \cosize(L) & = 1 + \max(\varphi_L)\\
%     & = 1 + \max (\Phi_L).
%     \end{align*}
% \end{remark}

\begin{remark}
    There exist distinct layouts $A \neq B$ with $\Phi_A = \Phi_{B}$. For example, the layouts 
    \begin{align*}
        A & = (7,7) : (1,7)\\
        B & = (49) : (1)
    \end{align*}
    are not equal, yet $\Phi_A = \Phi_{B}$. Later, we will characterize precisely when two flat layouts $A$ and $B$ have the same layout function (see Proposition \ref{coalesceproposition}).
\end{remark}

Before moving on to our discussion of layout operations, we need to define the notion of non-degeneracy.

\begin{definition}\label{definitionofnondegenerateflatlayout}
    Suppose 
    \[
    L = (s_1,\dots,s_m):(d_1,\dots,d_m)
    \]
    is a flat layout. We say $L$ is {\it non-degenerate} if for any $1 \leq i \leq m$, we have 
    \[
    s_i = 1 \quad \Rightarrow \quad d_i = 0.
    \]
\end{definition}

\begin{example}
    The layouts 
    \begin{align*}
        L_1 & = (4,1):(1,0)\\
        L_2 & = (8,1,8,1):(2,0,16,0)
    \end{align*}
    are non-degenerate, while the layouts 
    \begin{align*}
    L_3 & = (4,1):(1,4)\\
    L_4 & = (8,1,8,1):(2,16,16,256)
    \end{align*}
    are degenerate.
\end{example}

\begin{observation}
There is no real loss of generality in assuming that a layout $L$ is non-degenerate. More precisely, if 
\begin{align*}
    L & = (s_1,\dots,s_m):(d_1,\dots,d_m)\\
    L' & = (s_1,\dots,s_m):(d_1',\dots,d_m')
\end{align*}
are flat layouts with the same shape, and $d_i = d_i'$ whenever $s_i >1$, then $\varphi_L = \varphi_{L'}$, and $\Phi_{L} = \Phi_{L'}$. In particular, we are free to set $d_i = 0$ whenever $s_i = 1$ without altering the coordinate function or layout function of $L$. 
\end{observation}

\subsection{Basic operations}

Having established the basic vocabulary for flat layouts, we turn to the operations they support. In this section, we define basic operations that will be needed to construct more sophisticated operations such as {\it coalesce}, {\it complement}, and {\it composition}.

\subsubsection{Restriction}

If $L$ is a flat layout, it is often useful to restrict to a subset of the modes of $L$. Recall that for a non-negative integer $m$, we write 
\[
\langle m \rangle = \{1,\dots,m\}.
\]

\begin{definition}
    Suppose 
    \[L = (s_1,\dots,s_m) : (d_1,\dots,d_m)\]
    is a flat layout, and suppose 
    \[
    I = \{i_1<\dots<i_k\} \subset \langle m \rangle
    \]
    is a subset. We define the {\it restriction of} $L$ {\it to} $I$ to be the flat layout 
    \[
    L \mid_I = (s_{i_1},\dots,s_{i_k}):(d_{i_1},\dots,d_{i_k}).
    \]
\end{definition}

\begin{example}
    If 

\[
\begin{tikzpicture}[x={(0cm,-1cm)},y={(1cm,0cm)},every node/.style={minimum size=1cm, outer sep=0pt}]

\node[fill=gray!20] at (0,0) {0};
\node[fill=gray!20] at (0,1) {5};
\node[fill=gray!20] at (0,2) {10};
\node[fill=gray!20] at (0,3) {15};
\node[fill=gray!20] at (0,4) {20};
\node[fill=gray!20] at (0,5) {25};
\node[fill=gray!20] at (1,0) {10};
\node[fill=gray!20] at (1,1) {15};
\node[fill=gray!20] at (1,2) {20};
\node[fill=gray!20] at (1,3) {25};
\node[fill=gray!20] at (1,4) {30};
\node[fill=gray!20] at (1,5) {35};
\node[fill=gray!20] at (2,0) {20};
\node[fill=gray!20] at (2,1) {25};
\node[fill=gray!20] at (2,2) {30};
\node[fill=gray!20] at (2,3) {35};
\node[fill=gray!20] at (2,4) {40};
\node[fill=gray!20] at (2,5) {45};
\draw[color=black,thick,shift={(-0.5,-0.5)}] (0,0) grid (3,6);

\node[anchor = east] at (1,-1) {$L = (3,6):(10,5) = $};
\end{tikzpicture}
\]
and $I = \{2\}$, then 

\[
\begin{tikzpicture}[x={(0cm,-1cm)},y={(1cm,0cm)},every node/.style={minimum size=1cm, outer sep=0pt}]
\node[fill=gray!20] at (0,0) {0};
\node[fill=gray!20] at (0,1) {5};
\node[fill=gray!20] at (0,2) {10};
\node[fill=gray!20] at (0,3) {15};
\node[fill=gray!20] at (0,4) {20};
\node[fill=gray!20] at (0,5) {25};
\draw[color=black,thick,shift={(-0.5,-0.5)}] (0,0) grid (1,6);

\node[anchor = east] at (0,-1) {$L\mid_I  = (6):(5) = $};
\end{tikzpicture}
\]
    
\end{example}

% \begin{example}
%     If 
%     \[L = (2,2,2,2,2):(16,8,4,2,1)\]
%     and $I = \{2,5\}$, then 
%     \[
%     L \mid_I = (2,2):(8,1).
%     \]
% \end{example}

\begin{example}
    If 
    \[L = (3,8,8,8):(1,3,24,192)\]
    and $I = \{1,2,3\}$, then 
    \[
    L\mid_I = (3,8,8):(1,3,24).
    \]
\end{example}

\begin{example}
    If 
    \[L = (s_1,\dots,s_m):(d_1,\dots,d_m)\]
    is a flat layout and $I = \langle m \rangle$, then
    \[
    L \mid_I = L.
    \]
\end{example}

\begin{example} 
If 
\[L = (s_1,\dots,s_m):(d_1,\dots,d_m)\]
is a flat layout and $I = \varnothing$ is the empty set, then 
\[
L \mid_I = ():()
\]
is the empty layout. 
\end{example}

\subsubsection{Squeeze}

If $L$ is a flat layout, then the operation $L \mapsto \squeeze(L)$ removes all modes $s_i:d_i$ of $L$ where $s_i = 1$. 

\begin{construction}\label{Constructionofsqueeze}
    Suppose 
    \[L = (s_1,\dots,s_m):(d_1,\dots,d_m)\]
    is a flat layout, and let
    \[
    I = \{ i \in \langle m \rangle \mid s_i > 1\}
    \]
    be the collection of indices whose corresponding shape entry is not equal to $1$. We define
    \[
    \squeeze(L) = L \mid_I.
    \]
\end{construction}

\begin{example}
    If 
    \[L = (64,64,1):(1,64,0),\]
    then 
    \[
    \squeeze(L) = (64,64):(1,64).
    \]
\end{example}

\begin{example}
    If 
    \[
    L = (64,64,1,32,1):(2048,32,0,1,0)
    \]
    then 
    \[
    \squeeze(L) = (64,64,32):(2048,32,1).
    \]
\end{example}

\begin{example}
    If $L$ is a flat layout, then
    \[
    \squeeze(L) = L
    \]
    if and only if $\shape(L)$ contains no entries equal to $1$.
\end{example}

\begin{example}
    If $L$ is a flat layout, then
    \[
    \squeeze(L) = ():()
    \]
    is the empty layout if and only if all entries of $\shape(L)$ are equal to $1$.
\end{example}

\noindent An essential property of this construction is that $L \mapsto \squeeze(L)$ leaves the layout function of $L$ unchanged. 

\begin{lemma}\label{squeezelemma}
    If $L$ is a flat layout, then 
    \begin{enumerate}
    \item $\size(\squeeze(L)) = \size(L)$,
    \item $\cosize(\squeeze(L)) = \cosize(L)$, and 
    \item $\Phi_{\squeeze(L)} = \Phi_L$.
    \end{enumerate}
\end{lemma}

\begin{proof} Let 
\[
I = \{i_1<\dots<i_{k}\} \subset \langle m \rangle
\] denote the collection of indices with $s_{i_j} > 1$, so that 
\[
    \squeeze(L) = (s_{i_1},\dots,s_{i_{k}}):(d_{i_1},\dots,d_{i_{k}}).
\]
For the first assertion, we compute
\begin{align*} \size(\squeeze(L))  = \prod_{j=1}^{k} s_{i_j} = \left( \prod_{j=1}^{k} s_{i_j} \right)\cdot\left( \prod_{\langle m \rangle \setminus I} 1 \right) = \prod_{i=1}^m s_i = \size(L).
\end{align*}
For the second assertion, we compute
\begin{align*} \cosize(\squeeze(L))  = 1 + \sum_{j=1}^{k} (s_{i_j} - 1)\cdot d_{i_j}
& = 1 + \sum_{j=1}^{k} (s_{i_j}-1)\cdot d_{i_j} + \left( \sum_{\langle m \rangle \setminus I} 0 \right)\\
& = 1 + \sum_{i=1}^m (s_i-1)\cdot d_i\\
& = \cosize(L).
\end{align*}
For the third assertion, it suffices to show that removing a mode of the form $1:d_i$ from a flat layout leaves the layout function unchanged. Suppose $L=(s_1,\dots,s_m):(d_1,\dots,d_m)$, and suppose that some $s_i=1$. Let 
    \begin{align*} L' & = (s_1',\dots,s_{m-1}'):(d_1',\dots,d_{m-1}')
    \end{align*}
    denote the flat layout obtained from $L$ by removing its $i$th mode, so that 
    \[
    s_j' = \begin{cases} s_j & j < i\\
    s_{j+1} & i \leq j < m,
    \end{cases}
    \hspace{0.2in}\text{ and } \hspace{0.2in} d_j'  = \begin{cases} d_j & j < i\\
    d_{j+1} & i \leq j < m.
    \end{cases}
    \]
    
    \noindent The layout function for $L$ is given by 
    \[
    \Phi_L(x) = x_1 d_1 + \cdots + x_m d_m
    \]
    where $x_j = \left\lfloor \dfrac{x}{s_1\cdots s_{j-1}} \right \rfloor \mod s_j$, and the layout function for $L'$ is given by 
    \[
    \Phi_{L'}(x) = x_1' d_1' + \cdots + x_{m-1}'d_{m-1}'
    \]
    where  $x_j' = \left\lfloor \dfrac{x}{s_1'\cdots s_{j-1}'} \right \rfloor \mod s_j'$. We observe that 
    \[
    x_j' = \begin{cases} x_j & j < i\\
    x_{j+1} & i \leq j < m,
    \end{cases}
    \]
    and since $x_i \in {[}0,1{)}$ is necessarily $0$, we have 
    \begin{align*}
        \Phi_{L}(x) & = x_1 d_1 + \cdots + x_m d_m \\
        & = x_1 d_1 + \cdots + x_{i-1} d_{i-1} + x_{i+1} d_{i+1} + \cdots + x_m d_m \\
        & = x_1' d_1' + \cdots + x_{m-1}'d_{m-1}'\\
        & = \Phi_{L'}(x).
    \end{align*}
\end{proof}

\subsubsection{Filter zeros}
If $L$ is a flat layout, then the operation $L \mapsto \filter(L)$ removes all modes $s_i:d_i$ with $d_i = 0$.
\begin{definition}\label{filterzerosdefinition}
    Suppose 
    \[L = (s_1,\dots,s_m):(d_1,\dots,d_m)\]
    is a flat layout, and let 
    \[I = \{i \in \langle m \rangle \mid d_i >0 \}\]
    be the collection of indices whose corresponding stride entry is not equal to $0$. We define 
    \[
    \filter(L) = L \mid_I. 
    \]
\end{definition}

\begin{example}
    If 
    \[L = (64,8,8,128):(8,1,0,512)\]
    then 
    \[
    \filter(L) = (64,8,128):(8,1,512).
    \]
\end{example}

\begin{example}
    If 
\[
\begin{tikzpicture}[x={(0cm,-1cm)},y={(1cm,0cm)},every node/.style={minimum size=1cm, outer sep=0pt}]

\node[fill=gray!20] at (0,0) {0};
\node[fill=gray!20] at (0,1) {0};
\node[fill=gray!20] at (1,0) {12};
\node[fill=gray!20] at (1,1) {12};
\node[fill=gray!20] at (2,0) {24};
\node[fill=gray!20] at (2,1) {24};
\draw[color=black,thick,shift={(-0.5,-0.5)}] (0,0) grid (3,2);

\node[anchor=east] at (1,-1) {$L = (3,2):(12,0) =$ };
\end{tikzpicture}
\]
then 
\[
\begin{tikzpicture}[x={(0cm,-1cm)},y={(1cm,0cm)},every node/.style={minimum size=1cm, outer sep=0pt}]

\node[fill=gray!20] at (0,0) {0};
\node[fill=gray!20] at (1,0) {12};
\node[fill=gray!20] at (2,0) {24};

\draw[color=black,thick,shift={(-0.5,-0.5)}] (0,0) grid (3,1);

\node[anchor=east] at (1,-1) {$\filter(L) = (3):(12) =$ };
\end{tikzpicture}
\]
\end{example}

\begin{example}
    If
    \[
    L = (3,8,8,8):(16,0,0,0)
    \]
    then 
    \[
    \filter(L)  = (3):(16). 
    \]
\end{example}

\begin{example}
    If $L$ is a flat layout, then 
    \[
    \filter(L) = L
    \]
    if and only if all entries of $\stride(L)$ are nonzero. 
\end{example}

\begin{example}
    If $L$ is a flat layout, then 
    \[
    \filter(L) = ():()
    \]
    is the empty layout if and only if all entries of $\stride(L)$ are equal to $0$. 
\end{example}

\subsubsection{Permute}

Recall that if $X = (x_1,\dots,x_m)$ is a tuple of length $m$ and $\sigma \in \Sigma_m$ is a permutation, then we write
\[
X^\sigma = (x_{\sigma(1)},\dots,x_{\sigma(m)}).
\]
for the permutation of $X$ by $\sigma$.

\begin{definition}\label{definitionofpermutationofflatlayout} If $L=(s_1,\dots,s_m):(d_1,\dots,d_m)$ is a flat layout of rank $m$ and $\sigma \in \Sigma_m$ is a permutation, we define 
\begin{align*}
L^\sigma & =  \shape(L)^\sigma : \stride(L)^\sigma \\
& = (s_{\sigma(1)},\dots,s_{\sigma(m)}) : (d_{\sigma(1)},\dots,d_{\sigma(m)}).
\end{align*}
\end{definition}

\begin{example}
    If 
\[
\begin{tikzpicture}[x={(0cm,-1cm)},y={(1cm,0cm)},every node/.style={minimum size=1cm, outer sep=0pt}]

\node[fill=gray!20] at (0,0) {0};
\node[fill=gray!20] at (0,1) {2};
\node[fill=gray!20] at (1,0) {12};
\node[fill=gray!20] at (1,1) {14};
\node[fill=gray!20] at (2,0) {24};
\node[fill=gray!20] at (2,1) {26};
\node[fill=gray!20] at (3,0) {36};
\node[fill=gray!20] at (3,1) {38};
\draw[color=black,thick,shift={(-0.5,-0.5)}] (0,0) grid (4,2);

\node[anchor=east] at (1.5,-1) {$L = (4,2):(12,2) =$ };
\end{tikzpicture}
\]
    and $\sigma = (1\; 2) \in \Sigma_2$ is the transposition, then 

\[
\begin{tikzpicture}[x={(0cm,-1cm)},y={(1cm,0cm)},every node/.style={minimum size=1cm, outer sep=0pt}]

\node[fill=gray!20] at (0,0) {0};
\node[fill=gray!20] at (0,1) {12};
\node[fill=gray!20] at (0,2) {24};
\node[fill=gray!20] at (0,3) {36};
\node[fill=gray!20] at (1,0) {2};
\node[fill=gray!20] at (1,1) {14};
\node[fill=gray!20] at (1,2) {26};
\node[fill=gray!20] at (1,3) {38};
\draw[color=black,thick,shift={(-0.5,-0.5)}] (0,0) grid (2,4);

\node[anchor=east] at (0.5,-1) {$L^\sigma = (2,4):(2,12) =$ };
\end{tikzpicture}
\]
is the transposed layout.
\end{example}

\begin{example}
    If 
    \[L = (15,12,10):(240,1,24)\]
    and $\sigma = (1 \; 2 ) \in \Sigma_3$, then 
    \[
    L^\sigma = (12,15,10):(1,240,24).
    \]
\end{example}

\begin{example}
    If 
    \[L = (2,2,2,2,2):(1,2,4,8,16)\]
    and $\sigma = (1 \; 5)(3 \; 2 \; 4) \in \Sigma_5$, then 
    \[
    L^\sigma = (2,2,2,2,2):(16,8,2,4,1).
    \]
\end{example}

\begin{example}
    If 
    \[
    L = (s,\dots,s):(d,\dots,d)
    \]
    is a flat layout all of whose modes are equal, then for any $\sigma \in \Sigma_m$, we have 
    \[
    L^\sigma = L.
    \]
\end{example}

\subsubsection{Sort}

If $L$ is a flat layout, it is often useful to permute $L$ so that its modes are increasing, in the following sense.

\begin{definition}
We define a linear ordering on pairs $s:d$ of integers by 
\[
s:d \preceq s':d' \hspace{0.2in} \Leftrightarrow \hspace{0.2in} \begin{matrix} d < d' \text{, or}\\ 
d = d'\text{ and }s \leq s'.
\end{matrix}
\]

\end{definition}

\begin{example}
    We have 
    \[
    5:8 \preceq 4:12 \preceq 5:12.
    \]
\end{example}

\begin{definition}\label{definitionofsorted}
    Suppose $L$ is a flat layout. We say $L$ is {\it sorted} if for any $1 \leq i < \rank(L)$, we have 
    \[
    \mode_i(L) \preceq \mode_{i+1}(L).
    \]
\end{definition}

\begin{example}\label{sortedlayoutexample}
The layouts
\begin{align*}
L_1 & = (128,64,2,2):(1,128,8192,16384)\\
L_2 & = (2,2,2):(1,1,1)
\end{align*}
are sorted, while the layouts
\begin{align*}
L_3 & = (2,4,8,16):(64,1,2,4)\\
L_4 & = (5,32,16):(1,5,5)
\end{align*}
are not sorted.
\end{example}

\begin{example}
    The empty layout $E = ():()$ is sorted.
\end{example}

\begin{example}
    If 
    \[L = (s_1,\dots,s_m):(0,\dots,0)\]
    is a flat layout with all entries of $\stride(L)$ equal to $0$, then $L$ is sorted if and only if 
    \[
    s_1 \leq s_2 \leq \cdots \leq s_m.
    \]
\end{example}

Whether or not a flat layout $L$ is sorted is intimately related to the behavior of the layout function $\Phi_L$ of $L$, as described in the following lemma.
\begin{lemma}
    Suppose $L$ is a flat layout. If $\Phi_L$ is non-decreasing, then $L$ is sorted.
\end{lemma}
\begin{proof}
We prove the contrapositive. Suppose that $L$ is not sorted. We will show that there exists some $x \leq y$ in the domain of $\Phi_L$ with $\Phi_L(x)>\Phi_L(y)$. If there exists some $1 \leq i < m$ such that $d_i > d_{i+1}$, then we can let 
    \begin{align*}
    x & = \prod_{j<i}s_j, \hspace{0.2in}\text{and}\hspace{0.2in} y = \prod_{j<i+1}s_j,
    \end{align*}
    in which case $x < y$, but 
    \begin{align*}
    \Phi_L(x) & = (0,\dots,1,0,\dots,0) \cdot (d_1,\dots,d_i,d_{i+1},\dots,d_m) \\
    & = d_i\\
    & > d_{i+1}\\
    & = (0,\dots,0,1,\dots,0) \cdot (d_1,\dots,d_m)\\
    & = \Phi_L(y).
    \end{align*}
    On the other hand, if there exists some $1 \leq i < m$ such that $d_i = d_{i+1}$ and $s_i > s_{i+1}$, we can set 
    \begin{align*}
    x & = (s_i-1)\left(\prod_{j<i} s_j\right),\hspace{0.2in} \text{ and } \hspace{0.2in}y = (s_{i+1} - 1)\left( \prod_{j<i+1}s_j\right),
    \end{align*}
    in which case $x < y$, but 
    \begin{align*}
    \Phi_L(x) & = (0,\dots,s_i-1,0,\dots,0) \cdot (d_1,\dots,d_i,d_{i+1},\dots,d_m) \\
    & = (s_i-1)d_i\\
    & > (s_{i+1}-1)d_{i}\\
    & = (s_{i+1}-1)d_{i+1}\\
    & = (0,\dots,0,s_{i+1}-1,\dots,0) \cdot (d_1,\dots,d_m)\\
    & = \Phi_L(y).
    \end{align*}
    We conclude that $\Phi_L$ is not non-decreasing. 
\end{proof}

\begin{remark}
    The converse of the previous lemma is false. For example, the flat layout 
    \[L = (3,5,7):(1,1,1)\]
    is sorted, but 
    \[
    \Phi_L(7) = (0,2,0) \cdot (1,1,1) = 2
    \]
    is strictly greater than 
    \[\Phi_L(16) = (0,0,1)\cdot(1,1,1) = 1.
    \]
\end{remark}

If $L$ is a flat layout, then we can permute the modes of $L$ to obtain a sorted layout $\sort(L)$. 

\begin{construction}
    Suppose 
    \[L = (s_1,\dots,s_m):(d_1,\dots,d_m)\]
    is a flat layout. Define a linear ordering $\preceq$ on $\langle m \rangle$ by $i \preceq j$ if 
    \begin{enumerate}
        \item $\mode_i(L) \preceq \mode_j(L)$, and 
        \item if $\mode_i(L) = \mode_j(L)$ then $i \leq j$. 
    \end{enumerate}
    Let $\sigma \in \Sigma_m$ be the permutation associated to the linear ordering $\preceq$ of $\langle m \rangle$. We define $\sort(L)$ to be permutation of $L$ by $\sigma$:
    \[
    \sort(L) = L^\sigma.
    \]
\end{construction}

\begin{example}
    If 
    \[L = (2,4,8,16):(64,1,2,4)\]
    then 
    \[\sort(L) = (4,8,16,2):(1,2,4,64).\]
\end{example}
\begin{example}
If 
\[L = (5,32,16):(1,5,5)\]
then 
\[
    \sort(L) = (5,16,32):(1,5,5).
\]
\end{example}
\begin{example}
    If $L$ is sorted, then $\sort(L) = L$. In particular, this implies that $\sort(-)$ is an idempotent operation:
    \[
    \sort(\sort(L)) = \sort(L).
    \]
\end{example}

\begin{observation}
    If $L$ is a flat layout, then typically $\Phi_{\sort(L)} \neq \Phi_{L}$. However, the layout functions $\Phi_L$ and $\Phi_{\sort(L)}$ always have the same image. To see this, let's write 
    \begin{align*} L & = (s_1,\dots,s_m):(d_1,\dots,d_m)\text{, and}\\
    \sort(L) & = (s_{\sigma(1)},\dots,s_{\sigma(m)}):(d_{\sigma(1)},\dots,d_{\sigma(m)})
    \end{align*}
    for some permutation $\sigma \in \Sigma_m$. If an integer $n$ is in the image of $\Phi_L$, then there exists a tuple $(x_1,\dots,x_m) \in \prod_{i=1}^m {[}0,s_i{)}$ such that 
    \[
    x_1d_1 + \cdots + x_m d_m = n
    \]
    in which case the tuple $(x_{\sigma(1)},\dots,x_{\sigma(m)}) \in \prod_{i=1}^m {[}0,s_{\sigma(i)}{)}$ satisfies 
    \[
    x_{\sigma(1)}d_{\sigma(1)} + \cdots + x_{\sigma(m)}d_{\sigma(m)} = n.
    \]
    This proves that $\Image(\Phi_\sort(L)) \subseteq \Image(\Phi_L)$, and the reverse inclusion is proved similarly. 
\end{observation}

\subsubsection{Concatenate}

\noindent Recall that if $X= (x_1,\dots,x_m)$ and $Y  = (y_1,\dots,y_n)$ are tuples, then the concatenation of $X$ and $Y$ is the tuple
    \[
    X \star Y = (x_1,\dots,x_m,y_1,\dots,y_n).
    \]
This definition extends naturally to the concatenation of flat layouts.

\begin{definition}\label{definitionofflatconcatenate}
    Suppose 
    \begin{align*}
        L_1 & = S_1:D_1\\
        L_2 & = S_2:D_2
    \end{align*}
    are flat layouts. Then the {\it concatenation} of $L_1$ and $L_2$ is the flat layout 
    \[
    L_1 \star L_2 = S_1 \star S_2 : D_1 \star D_2.
    \]
    Concatenation of flat layouts is associative, so more generally, if $L_1,\dots,L_k$ are flat layouts, we may form the concatenation
    \[
    L_1 \star \cdots \star L_k. 
    \]
\end{definition}

\begin{example}\label{flatconcatenateexample1}
    If $L_1 = (7,2):(2,1)$ and $L_2 = (3,3,3):(0,10,30)$, then 
    \[
    L_1 \star L_2 = (7,2,3,3,3):(2,1,0,10,30).
    \]
\end{example}

\begin{example}\label{flatconcatenateexample2}
    If $E = ():()$ is the empty layout, then for any flat layout $L$ we have 
    \[
    L \star E = L = E \star L.
    \]
\end{example}

\begin{observation}
    Suppose 
    \[
    L = (s_1,\dots,s_m):(d_1,\dots,d_m)
    \]
    is a flat layout. If we write 
    \[L_i = (s_i):(d_i),\]
    then we can write $L$ as the concatenation
    \[
    L = L_1 \star \cdots \star L_m.
    \]
\end{observation}

If $L_1,\dots,L_k$ are flat layouts, then the layout function of the concatenation $L_1 \star \cdots \star L_k$ is determined by the layout functions of $L_1,\dots,L_k$ as follows. 

\begin{proposition}\label{layoutfunctionofconcatenatedlayout}
    Suppose $L_1,\dots,L_k$ are flat layouts of shape $S_1,\dots,S_k$, and size $N_1,\dots,N_k$, respectively. Then the coordinate function
        \[ \begin{tikzcd} 
        {[}0,S_1 \star \cdots \star S_k{)} \ar[rr,"\varphi_{L_1 \star \cdots \star  L_k} "] & & \mathbb{Z}
        \end{tikzcd}
        \]
        of $L_1 \star \cdots \star L_k$ is equal to the composite
        \[\begin{tikzcd} [column sep = 30]
     {[}0,S_1 \star \cdots \star S_k{)} \ar[r,"\cong"]  & {[}0,S_1{)} \times \cdots \times {[}0,S_k{)} \ar[rr,"\varphi_{L_1} +\cdots + \varphi_{L_k}"] & & \mathbb{Z} , \\
     X_1 \star \cdots \star X_k \ar[r]  & \ar[l] (X_1,\dots,X_k) & & & 
    \end{tikzcd} \]
        and the layout function 
    \[ \begin{tikzcd} 
     {[}0,N_1 \cdots N_k{)} \ar[rr,"\Phi_{L_1 \star \cdots \star  L_k} "] & & \mathbb{Z}
    \end{tikzcd}
    \]
    of $L_1 \star \cdots \star L_k$ is equal to the composite 
    \[\begin{tikzcd} [column sep = 30]
    {[}0,N_1\cdots N_k{)} \ar[rr,"\colex_{(N_1,\dots,N_k)}^{-1}"] & & {[}0,N_1{)} \times \cdots \times {[}0,N_k{)} \ar[rr,"\Phi_{L_1} + \cdots +\Phi_{L_k}"] & &  \mathbb{Z}. 
    \end{tikzcd} \]
\end{proposition}

\begin{proof}
    Let's write $L_i = S_i:D_i$ for each $1 \leq i \leq k$. The first claim holds because if 
    \[X \in [0,S_1 \star \cdots \star S_k)\]
    corresponds to 
    \[
    X_1 \star \cdots \star X_k \in [0,S_1) \times \cdots \times [0,S_k)
    \]
    under the canonical isomorphism $[0,S_1 \star \cdots \star S_k) \cong [0,S_1) \times \cdots \times [0,S_k)$, then 
    \begin{align*}
        \varphi_{L_1 \star \cdots \star L_k}(X) & = X \cdot (D_1 \star \cdots \star D_k)\\
        & = (X_1 \star \cdots \star X_k) \cdot (D_1 \star \cdots \star D_k)\\
        & = (X_1 \cdot D_1) + \cdots + (X_k \cdot D_k)\\
        & = \varphi_{L_1}(X_1) + \cdots + \varphi_{L_k}(X_k).
    \end{align*}

    For the second claim, we argue that the diagram 
    \[\begin{tikzcd} [row sep = 20]
    {[}0,N_1 {)} \times \cdots \times {[}0,N_1 {)} \ar[rrrr,"\colex_{S_1}^{-1} \times \cdots \times \colex_{S_k}^{-1} "] & & & & {[}0,S_1 {)} \times \cdots \times {[}0,S_1 {)} \ar[drr,"\varphi_{L_1} + \cdots + \varphi_{L_k}"] & & \\
    {[}0,N_1 \cdots N_k {)} \ar[u,"\colex_{(N_1,\dots,N_k)}^{-1}"] \ar[rrrr,swap,"\colex_{S_1\star \cdots \star S_k}^{-1}"] & & & & {[}0,S_1 \star \cdots \star S_k{)} \ar[rr,swap,"\varphi_{L_1 \star \cdots \star L_k}"] \ar[u,"\cong"] & & \mathbb{Z} 
    \end{tikzcd} \]
    commutes. The left-hand square commutes since colexicographic isomorphisms are associative, and the right-hand triangle commutes by the first claim.
\end{proof}

We can describe the important attributes of a concatenated layout as follows.

\begin{proposition}\label{attributesofconcatenatedflatlayouts}
    Suppose $L_1,\dots,L_k$ are flat layouts. Then
    \begin{enumerate}
        \item the rank of $L_1 \star \cdots \star L_k$ is \[\rank(L_1 \star \cdots \star L_k)  = \sum_{i=1}^k\rank(L_i),\]
        \item the size of $L_1 \star \cdots \star L_k$ is 
        \[\size(L_1 \star \cdots \star L_k) = \prod_{i=1}^k\size(L_i),\]
        \item the cosize of $L_1 \star \cdots \star L_k$ is 
        \[\cosize(L_1 \star \cdots \star L_k) = 1 - k + \sum_{i=1}^k\cosize(L_i).\]
    \end{enumerate}
\end{proposition}

\begin{proof}
    Let's write $L_i = S_i:D_i$ for each $1 \leq i \leq k$. For 1, we compute
    \begin{align*}
        \rank(L_1 \star \cdots \star L_k) & = \len(S_1 \star \cdots \star S_k) = \sum_{i=1}^k \len(S_i) = \sum_{i=1}^k \rank(L_i).\\
    \end{align*}
    For 2, we compute 
    \begin{align*}
        \size(L_1 \star \cdots \star L_k) & = \size(S_1 \star \cdots \star S_k) = \prod_{i=1}^k \size(S_i) = \prod_{i=1}^k \size(L_i).
    \end{align*}
    For 3, we compute 
    \begin{align*}
    \cosize(L_1 \star \cdots \star L_k)  & = 1 + \max(\Phi_{L_1 \star \cdots \star L_k})\\
    & = 1 + \sum_{i=1}^k\max(\Phi_{L_i})\\
    & = 1 - k + \left( 1 + \max(\Phi_{L_1}) \right) + \cdots + \left( 1 + \max(\Phi_{L_1}) \right)\\
    & = 1 - k + \cosize(L_1) + \cdots + \cosize(L_k).
    \end{align*}
    where we have used our identification of $\Phi_{L_1 \star \cdots \star L_k}$ from Proposition \ref{layoutfunctionofconcatenatedlayout}.
\end{proof}

\subsection{Flat coalesce}
We have seen that the layout function $\Phi_L$ of a flat layout $L$ is an important invariant. In many cases, we are only interested in the layout function $\Phi_L$, and are free to work with any layout whose layout function is $\Phi_L$. The flat coalesce operation 
\[L \mapsto \coalesce^\flat(L)\]
provides us with the simplest flat layout whose layout function is $\Phi_L$ (see Proposition \ref{characterizationofflatcoalesce}). 

We begin by defining the notion of a coalesced flat layout.
\begin{definition} \label{definitionofflatlayoutcoalesce}
Suppose $L = (s_1,\dots,s_m):(d_1,\dots,d_m)$ is a flat layout. We say $L$ is {\it coalesced} if 
\begin{enumerate}
    \item for any $1 \leq i \leq m$, we have $s_i \neq 1$, and 
    \item for any  $1 \leq i <m$, we have $s_id_i \neq d_{i+1}$.
\end{enumerate}
\end{definition}

\begin{example}
    The flat layout 
    \begin{align*} L =(3,5,2):(7,21,4)
    \end{align*}
    is not coalesced because $3\cdot 7 = 21$.
\end{example}
\begin{example}The flat layout 
    \[
    L = (2,7,6):(1,3,10)
    \]
    is coalesced.
\end{example}
\begin{example}
    The empty layout $E = ():()$ is coalesced.
\end{example}
\begin{example}
    If $L=(s):(d)$ and $s \neq 1$, then $L$ is coalesced.
\end{example}

\begin{example}
    If $L = (s_1,\dots,s_m):(d_1,\dots,d_m)$ is a column-major layout with $\rank(L)>1$, then $L$ is not coalesced, since for any $1 \leq i < m$, we have 
    \[
    s_id_i = s_i (s_1\cdots s_{i-1}) = s_1 \cdots s_i = d_{i+1}.
    \]
\end{example}

\begin{example}
    If $L = (s_1,\dots,s_m):(d_1,\dots,d_m)$ is a row-major layout with $s_i>1$ for all $1 \leq i \leq m$, then $L$ is coalesced: If $ 1\leq i < m$, then 
    \[
    s_id_i = s_is_{i+1}\cdots s_m > s_{i+2}\cdots s_m = d_{i+1}.
    \]
\end{example}

\begin{example}
    A flat layout of the form
    \[L = (s_1,\dots,s_m):(0,\dots,0)\]
    is coalesced if and only if $m \leq 1$.
\end{example}

If $L$ is a flat layout, then we may obtain a coalesced layout $\coalesce^\flat(L)$ with the same layout function as $L$ by removing modes with $s_i = 1$, and combining modes with $s_id_i = d_{i+1}$. More precisely, we make the following construction.

\begin{construction}\label{constructionofflatcoalesce}
Suppose $L$ is a flat layout, and write
\[\squeeze(L) = (s_1,\dots,s_m):(d_1,\dots,d_m).\]
Let $\sim$ be the equivalence relation on $\langle m \rangle$ generated by $i \sim i+1$ if 
\[
s_id_i = d_{i+1}.
\]
The quotient $\langle m \rangle / \sim$ is ordered by $[i] \leq [i']$ if $i \leq i'$, so we may identify $\langle m \rangle / \sim$ with $\langle \bar{m} \rangle$, where $\bar{m}$ is the size of $\langle m \rangle / \sim$. If $i \in \langle \bar{m} \rangle$ corresponds to the equivalence class 
\[I = \{i', i' + 1, \dots, i'+k\} \in \langle m \rangle/\sim,\]
then we define integers $\bar{s}_i$ and $\bar{d}_i$ as 
\begin{align*}
\bar{s}_i & = s_{i'} s_{i'+1}\cdots s_{i'+k}
\end{align*}
and
\begin{align*}
\bar{d}_i & = d_{i'},
\end{align*}
and define
\[
\coalesce^\flat(L) = (\bar{s}_1,\dots,\bar{s}_{\bar{m}}):(\bar{d}_1,\dots,\bar{d}_{\bar{m}}).
\] 
\end{construction}

\begin{observation}\label{flatcoalesceobservation}
    Examining the definition, we could equivalently define $\coalesce^\flat(L)$ to be the flat layout obtained from 
    \[
    \squeeze(L) = (s_1,\dots,s_m):(d_1,\dots,d_m)
    \]
    by iteratively performing the operation 
    \[
    s_i,s_{i+1} : d_i,s_id_i \hspace{0.2in} \rightsquigarrow \hspace{0.2in} s_is_{i+1} : d_i
    \]
    until the result is coalesced. 
\end{observation}

\begin{example}
    If $L = (2,2,2,2,2):(8,16,1024,2048,4096)$, then
    \[
    \coalesce^\flat(L) = (4,8):(8,1024).
    \]
\end{example}

\begin{example}
    If $L = (3,4,1,5):(1,8,3,32)$, then
    \[
    \coalesce^\flat(L) = (3,20):(1,8).
    \]
\end{example}

\begin{example}
    If $L = (s_1,\dots,s_m):(d_1,\dots,d_m)$ is column-major, and not all $s_i$ are equal to $1$, then 
    \[
    \coalesce^\flat(L) = (s_1\cdots s_m):(1).
    \]
\end{example}

\begin{example}
    If $L$ is row-major, then 
    \[
    \coalesce^\flat(L) = \squeeze(L). 
    \]
\end{example}

Let's justify that the operation $L \mapsto \coalesce^\flat(L)$ results in a coalesced layout.
\begin{lemma}
    If $L$ is a flat layout, then $\coalesce^\flat(L)$ is coalesced.
\end{lemma}

\begin{proof}
Borrowing the notation of Construction \ref{constructionofflatcoalesce}, let 
    \[
    \squeeze(L) = (s_1,\dots,s_m):(d_1,\dots,d_m)
    \]
    and let 
    \[
    \coalesce^\flat(L) = (\bar{s}_1,\dots,\bar{s}_{\bar{m}}) : ( \bar{d}_1,\dots,\bar{d}_{\bar{m}}).
    \]
    We want to show that $\coalesce^\flat(L)$ is coalesced. Suppose $1 \leq i \leq \bar{m}$. Then $i$ corresponds to a (non-empty) equivalence class $I \in \langle m \rangle / \sim$, and 
    \[
    \bar{s}_i = \prod_{i' \in I} s_{i'}
    \]
    is a product of integers $s_{i'} > 1$, so $\bar{s}_i >1$.

    Suppose $1 \leq i < \bar{m}$. We claim that $\bar{s}_i \bar{d}_i \neq \bar{d}_{i+1}$. Suppose $i$ corresponds to the equivalence class 
    \[\{i',i' + 1 , \dots , i' + k\} \in \langle m \rangle / \sim,\] and suppose $i+1$ corresponds to the equivalence class
    \[ \{i'+k+1,i'+k+2,\dots,i'+k+\ell\} \in \langle m \rangle / \sim.\]
    Then by using the equalities $s_{i'+t}d_{i'+t} = d_{i'+t+1}$ for $0 \leq t < k$, we may write 
    \begin{align*}
        \bar{s}_i \bar{d}_i = \bar{d}_i \bar{s}_i & = d_{i'} s_{i'} s_{i' + 1 } \cdots s_{i'+k}\\
        & = d_{i'+1}s_{i'+1} \cdots s_{i'+k}\\
        & \hspace{0.08in} \vdots \\
        & = d_{i'+k} s_{i'+k}\\
        & = s_{i'+k}d_{i'+k}
    \end{align*}
    and since $i'+k$ and $i'+k+1$ do not lie in the same equivalence class, we have 
    \begin{align*}
        \bar{s}_i\bar{d}_i = s_{i'+k}d_{i'+k} \neq d_{i'+k+1} = \bar{d}_{i+1}.
    \end{align*}
\end{proof}

\begin{example}
    If $L$ is coalesced, then $\coalesce^\flat(L) = L$. In particular, this implies that $\coalesce^\flat(-)$ is an idempotent operation:
    \[
    \coalesce^\flat(\coalesce^\flat(L)) = \coalesce^\flat(L).
    \]
\end{example}

Next, we argue that coalescing a flat layout leaves the layout function unchanged.

\begin{lemma} \label{coalescelemma}
    If $L$ is a flat layout, then $\Phi_{\coalesce^\flat(L)} = \Phi_L$.
\end{lemma}

\begin{proof}
By Observation \ref{flatcoalesceobservation}, it suffices to show that replacing an instance of $s_i,s_{i+1}:d_i,s_id_i$ with $s_is_{i+1}:d_i$ leaves the layout function unchanged. Suppose 
\[
L = (s_1,\dots,s_m):(d_1,\dots,d_m)
\]
is a flat layout, and there exists some $1 \leq i < m$ such that $d_{i+1} = s_id_i$. Let 
    \begin{align*} L' & = (s_1',\dots,s_{m-1}'):(d_1',\dots,d_{m-1}')
    \end{align*}
    denote the flat layout obtained from $L$ by combining the $i$th and $(i+1)$th modes of $L$. More precisely, we have 
    
    \[
    s_j' = \begin{cases} s_j & j < i\\
    s_is_{i+1} & j = i \\
    s_{j+1} & i<j<m,
    \end{cases}
    \hspace{0.2in}\text{ and } \hspace{0.2in} d_j' = \begin{cases} d_j & j \leq i\\
    d_{j+1} & i<j<m.
    \end{cases}
    \]
    The layout function for $L$ is given by 
    \[
    \Phi_L(x) = x_1 d_1 + \cdots + x_m d_m
    \]
    where $x_j = \left\lfloor \dfrac{x}{s_1\cdots s_{j-1}} \right \rfloor \mod s_j$, and the layout function for $L'$ is given by 
    \[
    \Phi_{L'}(x) = x_1' d_1' + \cdots + x_{m-1}'d_{m-1}'
    \]
    where  $x_j' = \left\lfloor \dfrac{x}{s_1'\cdots s_{j-1}'} \right \rfloor \mod s_j'$. We observe that 
    
    \[
    x_j' =
    \begin{cases}
        x_j & j < i \\
        x_i + x_{i+1}s_i & j = i\\
        x_{j+1} & i < j < m,
    \end{cases}
    \]
    and so
    \begin{align*}
        \Phi_L(x) & = x_1 d_1 + \cdots + x_m d_m \\
        & = x_1 d_1 + \cdots + x_i d_i + x_{i+1}s_id_i + \cdots + x_m d_m\\
        & = x_1d_1 + \cdots + (x_i + x_{i+1} s_i) d_i + \cdots + x_m d_m \\
        & = x_1'd_1' + \cdots + x_{m-1}'d_{m-1}'\\
        & = \Phi_{L'}(x).
    \end{align*}

\end{proof}

\noindent We can use the coalesce operation to characterize when two flat layouts have the same layout function. 

\begin{proposition} \label{coalesceproposition}
    Suppose $A$ and $B$ are flat layouts. Then 
    \[
    \Phi_A = \Phi_B \quad \Leftrightarrow \quad \coalesce^\flat(A) = \coalesce^\flat(B).
    \]
\end{proposition}
\begin{proof}
    If $\coalesce^\flat(A) = \coalesce^\flat(B)$, then by Lemma \ref{coalescelemma}, we have 
    \[ \Phi_A = \Phi_{\coalesce^\flat(A)} = \Phi_{\coalesce^\flat(B)} = \Phi_B.\]

    \noindent Inversely, suppose that $\coalesce^\flat(A) \neq \coalesce^\flat(B)$. We will argue that $\Phi_A \neq \Phi_B$. Let's write 
    \begin{align*}
    \coalesce^\flat(A) & = (s_1,\dots,s_m) : (d_1,\dots,d_m)\text{,} \\
    \coalesce^\flat(B) & = (t_1,\dots,t_n) : (e_1,\dots,e_m).
    \end{align*}
    If one of $m,n$ is nonzero and the other is $0$, then clearly $\Phi_A \neq \Phi_B$, so we may assume $m,n\geq 1$. Let $i$ denote the least integer such that $(s_i,d_i) \neq (t_i,e_i)$. Then, in particular, we have $s_1 \cdots s_j = t_1 \cdots t_j$ for any $j < i$. There are two cases to consider:\\

    \begin{itemize}
    \item (Case 1): Suppose $d_i \neq e_i$. Let $N = s_1 \cdots s_{i-1} = t_1 \cdots t_{i-1}$. Then 
        \[
        \Phi_{\coalesce^\flat(A)}(N) = d_i \neq e_i = \Phi_{\coalesce^\flat(B)}(N)
        \]
        so $\Phi_{\coalesce^\flat(A)} \neq \Phi_{\coalesce^\flat(B)}$, and hence $\Phi_A \neq \Phi_B$.\\
        
     \item (Case 2): Suppose $d_i = e_i$, so that $s_i \neq t_i$. Without loss of generality we may assume $s_i < t_i$. Let $N = s_1 \cdots s_i = (t_1 \cdots t_{i-1}) s_i$. Then 
        \[
        \Phi_{\coalesce^\flat(A)}(N)= d_{i+1}
        \]
        while 
        \begin{align*}
        \Phi_{\coalesce^\flat(B)}(N) & = s_ie_i\\
        & = s_id_i,
        \end{align*}
        and since $\coalesce^\flat(A)$ is coalesced, we have $d_{i+1} \neq s_id_i$. We deduce that $\Phi_{\coalesce^\flat(A)} \neq \Phi_{\coalesce^\flat(B)}$, and hence $\Phi_A \neq \Phi_B$.
    \end{itemize}
\end{proof}

The previous proposition affords us the following abstract characterization of $\coalesce^\flat(L)$.

\begin{proposition}\label{characterizationofflatcoalesce}
    If $L$ is a flat layout, then $\coalesce^\flat(L)$ is the unique flat layout of minimal rank whose layout function is $\Phi_L$. 
\end{proposition}

\begin{proof}
    Suppose $L'$ is a flat layout with $\Phi_{L'} = \Phi_L$. Then by Proposition \ref{coalesceproposition}, we have 
    \[
    \coalesce^\flat(L) = \coalesce^\flat(L'),
    \]
    and it follows that 
    \[
    \rank(\coalesce^\flat(L)) = \rank(\coalesce^\flat(L')) \leq \rank(L'),
    \]
    where equality holds if and only if
    \[
    L ' = \coalesce^\flat (L') = \coalesce^\flat(L).
    \]
\end{proof}

\subsection{Compact flat layouts}

Before treating layout complements, we must define an important family of layouts called {\it compact} flat layouts. These are the flat layouts whose layout functions are bijective. In terms of the standard grid diagrams depicting layouts, a flat layout $L$ is compact if each integer $0 \leq i< \size(L)$ appears exactly once. For instance, the layout 
\[
\begin{tikzpicture}[x={(0cm,-1cm)},y={(1cm,0cm)},every node/.style={minimum size=1cm, outer sep=0pt}]

\node[fill=gray!20] at (0,0) {0};
\node[fill=gray!20] at (0,1) {3};
\node[fill=gray!20] at (0,2) {6};
\node[fill=gray!20] at (0,3) {9};
\node[fill=gray!20] at (0,4) {12};
\node[fill=gray!20] at (0,5) {15};
\node[fill=gray!20] at (1,0) {1};
\node[fill=gray!20] at (1,1) {4};
\node[fill=gray!20] at (1,2) {7};
\node[fill=gray!20] at (1,3) {10};
\node[fill=gray!20] at (1,4) {13};
\node[fill=gray!20] at (1,5) {16};
\node[fill=gray!20] at (2,0) {2};
\node[fill=gray!20] at (2,1) {5};
\node[fill=gray!20] at (2,2) {8};
\node[fill=gray!20] at (2,3) {11};
\node[fill=gray!20] at (2,4) {14};
\node[fill=gray!20] at (2,5) {17};
\draw[color=black,thick,shift={(-0.5,-0.5)}] (0,0) grid (3,6);

\node[anchor = east] at (1,-1) {$A = (3,6):(1,3) = $};
\end{tikzpicture}
\]
is compact, while the layouts
\[
\begin{tikzpicture}[x={(0cm,-1cm)},y={(1cm,0cm)},every node/.style={minimum size=1cm, outer sep=0pt}]

\node[fill=gray!20] at (0,0) {0};
\node[fill=gray!20] at (0,1) {6};
\node[fill=gray!20] at (0,2) {12};
\node[fill=gray!20] at (0,3) {18};
\node[fill=gray!20] at (0,4) {24};
\node[fill=gray!20] at (0,5) {30};
\node[fill=gray!20] at (1,0) {2};
\node[fill=gray!20] at (1,1) {8};
\node[fill=gray!20] at (1,2) {14};
\node[fill=gray!20] at (1,3) {20};
\node[fill=gray!20] at (1,4) {26};
\node[fill=gray!20] at (1,5) {32};
\node[fill=gray!20] at (2,0) {4};
\node[fill=gray!20] at (2,1) {10};
\node[fill=gray!20] at (2,2) {16};
\node[fill=gray!20] at (2,3) {22};
\node[fill=gray!20] at (2,4) {28};
\node[fill=gray!20] at (2,5) {34};
\draw[color=black,thick,shift={(-0.5,-0.5)}] (0,0) grid (3,6);

\node[anchor = east] at (1,-1) {$B = (3,6):(2,6) = $};
\end{tikzpicture}
\]
and
\[
\begin{tikzpicture}[x={(0cm,-1cm)},y={(1cm,0cm)},every node/.style={minimum size=1cm, outer sep=0pt}]

\node[fill=gray!20] at (0,0) {0};
\node[fill=gray!20] at (0,1) {2};
\node[fill=gray!20] at (0,2) {4};
\node[fill=gray!20] at (0,3) {6};
\node[fill=gray!20] at (0,4) {8};
\node[fill=gray!20] at (0,5) {10};
\node[fill=gray!20] at (1,0) {1};
\node[fill=gray!20] at (1,1) {3};
\node[fill=gray!20] at (1,2) {5};
\node[fill=gray!20] at (1,3) {7};
\node[fill=gray!20] at (1,4) {9};
\node[fill=gray!20] at (1,5) {11};
\node[fill=gray!20] at (2,0) {2};
\node[fill=gray!20] at (2,1) {4};
\node[fill=gray!20] at (2,2) {6};
\node[fill=gray!20] at (2,3) {8};
\node[fill=gray!20] at (2,4) {10};
\node[fill=gray!20] at (2,5) {12};
\draw[color=black,thick,shift={(-0.5,-0.5)}] (0,0) grid (3,6);

\node[anchor = east] at (1,-1) {$C = (3,6):(1,2) = $};
\end{tikzpicture}
\]
are not compact. More precisely, we have the following definition.

\begin{definition}\label{definitionofcompact}
    Suppose $L$ is a flat layout. We say $L$ is {\it compact} if 
    \[\begin{tikzcd}
    {[}0,\size(L){)} \ar[rr,"\Phi_L^{\cosize(L)}"] & & {[}0,\cosize(L){)}
    \end{tikzcd} \]
    is an isomorphism.
\end{definition}

\begin{example}
    The flat layout 
    \[
    L = (2,2,2,2):(1,2,4,8)
    \]
    is compact. More generally, if $L$ is column-major, then $L$ is compact.
\end{example}

\begin{example}
    The flat layout 
    \[
    L = (3,64,32):(2048,32,1)
    \]
    is compact. More generally, if $L$ is row-major, then $L$ is compact.
\end{example}

\begin{example}
    The empty layout 
    \[
    E = ():()
    \]
    is compact.
\end{example}

\begin{example}
    Suppose 
    \[
    L = (s_1,\dots,s_m):(d_1,\dots,d_m)
    \]
    is a flat layout. If there is some mode of $L$ with $s_i > 1$ and $d_i = 0$, then $L$ is not compact. 
\end{example}

We can give an explicit characterization of compact layouts as follows.

\begin{proposition}\label{compactisomorphismproposition}
    Suppose $L$ is a flat layout, and write 
    \[
    \squeeze(L) = (s_1,\dots,s_m):(d_1,\dots,d_m).
    \]
    then $L$ is compact if and only if there exists a permutation $\sigma \in \Sigma_m$ such that 
    \[
    d_{\sigma(i)} = s_{\sigma(1)}\cdots s_{\sigma(i-1)}
    \]
    for all $1 \leq i \leq m$. In other words, $L$ is compact if and only if there exists a permutation $\sigma \in \Sigma_m$ such that $\squeeze(L)^\sigma$ is column-major. 
\end{proposition}

\begin{proof}
    Suppose $L$ is a flat layout, and write
    \[
    \squeeze(L) = (s_1,\dots,s_m):(d_1,\dots,d_m).
    \]
    Suppose first that $L$ is compact, so there exists a permutation $\sigma \in \Sigma_m$ such that $d_{\sigma(i)} = s_{\sigma(1)} \cdots s_{\sigma(i-1)}$ for each $1 \leq i \leq m$. If we write $S^\sigma = (s_{\sigma(1)},\dots,s_{\sigma(m)})$, then we can write $\Phi_L^{\cosize(L)}$ as the composite 
    \[ \begin{tikzcd} [column sep = 12]
    {[}0,\size(L){)} \ar[rr,"\colex_S^{-1}"] & & {[}0,S{)} \ar[rr,"\cong"] & & {[}0,S^\sigma{)}\ar[rr,"\colex_{S^\sigma}"] & & {[}0,\cosize(L){)}\\
    & & (x_1,\dots,x_m) \ar[rr,mapsto] & & (x_{\sigma(1)},\dots,x_{\sigma(m)})
    \end{tikzcd} \] 
    and since each of these maps is an isomorphism, so is the composite $\Phi_L^{\cosize(L)}$.

    Conversely, suppose that $\Phi_L^{\cosize(L)}$ is an isomorphism. First, we note that the strides $d_1,\dots,d_m$ must be pairwise distinct: Suppose $d_i = d_j$, and let  $\delta_i^m$ and $\delta_j^m$ denote the tuples whose $i$th (resp. $j$th) entry is $1$, and all of whose other entries are $0$. These tuples satisfy 
    \[
    \delta_i^m \cdot (d_1,\dots,d_m) = d_i = d_j = \delta_j^m \cdot (d_1,\dots,d_m),
    \]
    and since $\Phi_L^{\cosize(L)}$ is injective, we must have $i=j$. Given that the strides $d_1,\dots,d_m$ are pairwise distinct, let $\sigma \in \Sigma_m$ be the permutation such that 
    \[d_{\sigma(1)} < d_{\sigma(2)} < \cdots < d_{\sigma(m)}.\]
    We will argue by induction on $i \geq 1$ that $d_{\sigma(i)} = s_{\sigma(1)} \cdots s_{\sigma(i-1)}$. For the base case $i=1$, we note that $1$ is in the image of $\Phi_L^{\cosize(L)}$, and the smallest non-zero value of $\Phi_L^{\cosize(L)}$ is $d_{\sigma(1)}$, so it follows that $d_{\sigma(1)} = 1$. Suppose $i>1$, and that we have proved the claim for all $j < i$. Consider the stride $d_{\sigma(i)}$. We know that there is no tuple of the form $(x_1,\dots,x_{i-1},0,\dots,0)^\sigma$ such that 
    \[(x_1,\dots,x_{i-1},0,\dots,0)^\sigma \cdot (d_1,\dots,d_m) = s_{\sigma(1)}\cdots s_{\sigma(i-1)},\] since the largest possible value of such an expression is 
    \[
    \sum_{j = 1}^{i-1} (s_{\sigma(j)}-1)(s_{\sigma(1)}\cdots s_{\sigma(j-1)}) = s_{\sigma(1)} \cdots s_{\sigma(i-1)} - 1. 
    \]
    Since $\Phi_L^\cosize(L)$ is surjective, and $d_{\sigma(i)} <  d_{\sigma(i+1)} <  \cdots <  d_{\sigma(m)}$, it follows that the next largest value of $\Phi_L^{\cosize(L)}$ is $d_{\sigma(i)}$, so we must have $d_{\sigma(i)} = s_{\sigma(1)}\cdots s_{\sigma(i-1)}$, as claimed.
\end{proof}

We conclude this section by giving a family of equivalent conditions for a flat layout $L$ to be compact. 

\begin{proposition}\label{equivalentconditionsforcompact}
    Suppose $L$ is a flat layout. Then the following are equivalent. 
    \begin{enumerate}
        \item $L$ is compact.
        \item $\coalesce^\flat(L)$ is compact.
        \item $\squeeze(L)$ is compact. 
        \item $\sort(L)$ is compact.
    \end{enumerate}
\end{proposition}

\begin{proof}
    The equivalence of 1, 2, and 3, follows from the fact that 
    \[
    \Phi_L = \Phi_{\coalesce^\flat(L)} = \Phi_{\squeeze(L)}.
    \]
    It remains to prove that $L$ is compact if and only if $\sort(L)$ is compact. Using the fact that 
    \[
    \squeeze(\sort(L)) = \sort(\squeeze(L)),
    \]
    we have 
    \begin{align*}
    \sort(L) \text{ is compact.} \hspace{0.2in} & \Leftrightarrow \hspace{0.2in} \squeeze(\sort(L)) \text{ is compact.}\\
    & \Leftrightarrow \hspace{0.2in} \sort(\squeeze(L)) \text{ is compact.}
    \end{align*}
    Now $\sort(\squeeze(L)) = \squeeze(L)^\tau$ for some permutation $\tau \in \Sigma_m$, so there exists a permutation $\sigma$ such that $\squeeze(L)^\sigma$ is column-major if and only if there exists a permutation $\sigma' \in \Sigma_m$ such that $\sort(\squeeze(L))$ is column-major, namely $\sigma' = \tau^{-1} \sigma$. It follows that
        \begin{align*}
    \sort(\squeeze(L)) \text{ is compact.}
    & \Leftrightarrow \hspace{0.2in} \squeeze(L) \text{ is compact.}\\
    & \Leftrightarrow \hspace{0.2in} L \text{ is compact.}
    \end{align*}
\end{proof}

\subsection{Complements}

In this section, we define the notion of complementary flat layouts. Recall from Definition \ref{definitionofcompact} that a flat layout $L$ is {\it compact} if the layout function 
\[\Phi_L^{\cosize(L)}: [0,\size(L)) \to [0,\cosize(L))\]
is an isomorphism.

\begin{definition}\label{definitionofcomplement}
Suppose $A$ and $B$ are flat layouts. We say $B$ is a {\it complement} of $A$, and write $A \perp B$, if the concatenated layout $A \star B$ is compact.
\end{definition}

\begin{example}\label{flatcomplementexample1}
If $A=(3):(5)$ and $B= (5):(1)$, then $A \perp B$ since 
\[
A \star B = (3,5):(5,1)
\]
is compact.

\[
\begin{tikzpicture}[x={(0cm,-1cm)},y={(1cm,0cm)},every node/.style={minimum size=1cm, outer sep=0pt}]

\node[fill=gray!20] at (-2,0) {0};
\node[fill=gray!20] at (-2,1) {1};
\node[fill=gray!20] at (-2,2) {2};
\node[fill=gray!20] at (-2,3) {3};
\node[fill=gray!20] at (-2,4) {4};
\draw[color=black,thick,shift=
{(-0.5,-0.5)}] (-2,0) grid (-1,5);
\node[anchor=east] at (-2,-1) {$B = $};

\node[fill=gray!20] at (0,-2) {0};
\node[fill=gray!20] at (1,-2) {5};
\node[fill=gray!20] at (2,-2) {10};
\draw[color=black,thick,shift=
{(-0.5,-0.5)}] (0,-2) grid (3,-1);
\node[anchor = east] at (1,-3) {$A = $};

\node[fill=gray!20] at (0,0) {0};
\node[fill=gray!20] at (0,1) {1};
\node[fill=gray!20] at (0,2) {2};
\node[fill=gray!20] at (0,3) {3};
\node[fill=gray!20] at (0,4) {4};
\node[fill=gray!20] at (1,0) {5};
\node[fill=gray!20] at (1,1) {6};
\node[fill=gray!20] at (1,2) {7};
\node[fill=gray!20] at (1,3) {8};
\node[fill=gray!20] at (1,4) {9};
\node[fill=gray!20] at (2,0) {10};
\node[fill=gray!20] at (2,1) {11};
\node[fill=gray!20] at (2,2) {12};
\node[fill=gray!20] at (2,3) {13};
\node[fill=gray!20] at (2,4) {14};
\draw[color=black,thick,shift={(-0.5,-0.5)}] (0,0) grid (3,5);

\node[anchor = west] at (1,5) {$= A \star B$};
\end{tikzpicture}
\]
\end{example}

\begin{example}\label{flatcomplementexample2}
    If $A = (4,2,10):(1400,2,20)$ and $B = (2,5,7,2):(1,4,200,5600)$, then $A \perp B$ since 
    \[
    A \star B = (4,2,10,2,5,7,2) : (1400,2,20,1,4,200,5600) 
    \]
    is compact.
\end{example}

\begin{example}
    If $A$ is a flat layout and $E = ():()$ is the empty layout, then $A \perp A$ if and only if $A$ is compact, since 
    \[
    A \star E = A.
    \]
\end{example}

\begin{example}
    If $A$ and $B$ are flat layouts, then 
    \[A \perp B \quad \Leftrightarrow \quad B \perp A.\] 
\end{example}

\begin{example}
    If $A$ is a flat layout, then $A\perp A$ if and only if $\size(A) = 1$. 
\end{example}

\begin{observation}
    In order for $A$ to admit a complement, it is necessary that $\Phi_A$ is injective. There do, however, exist flat layouts $A$ such that $\Phi_A$ is injective, and $A$ does not admit a complement. For example, consider the layout 
    \[
    A = (2,2):(1,3). 
    \]
    The layout function of $A$ is injective since 
    \begin{align*}
    \Phi_A(0) = 0\text{, }\Phi_A(1) = 1\text{, } \Phi_A(2) = 3\text{, and }\Phi_A(3) = 4,
    \end{align*}
    but $A$ does not admit a complement: Suppose 
    \[
    B = (s_1,\dots,s_m):(d_1,\dots,d_m)
    \]
    is any other flat layout. If there does not exist a tuple 
    \[
    (x_1,x_2,y_1,\dots,y_m) \in [0,2) \times [0,2) \times [0,s_1) \times \cdots \times  [0,s_m)
    \]
    such that $\varphi_{A \star B}(x_1,x_2,y_1,\dots,y_m) = 2$,
    then $A \star B$ is not compact. Suppose otherwise that there is such a tuple $(x_1,x_2,y_1,\dots,y_m)$. Then $\varphi_{B}(y_1,\dots,y_m) \in \{0,1,2\}$.
    \begin{itemize}
        \item (Case 1): If $\varphi_{B}(y_1,\dots,y_m) = 0$, then 
        \begin{align*}
        \varphi_{A \star B}(0,0,0,\dots,0) = 0 = \varphi_{A \star B}(0,0,y_1,\dots,y_m).
        \end{align*}
        \item (Case 2): If $\varphi_{B}(y_1,\dots,y_m) = 1$, then 
        \begin{align*}
        \varphi_{A \star B}(1,0,0,\dots,0) = 1 = \varphi_{A \star B}(0,0,y_1,\dots,y_m).
        \end{align*}
        \item (Case 3): If $\varphi_{B}(y_1,\dots,y_m) = 2$, then
        \[
        \varphi_{A \star B}(0,1,0,\dots,0) = 3 = \varphi_{A \star B}(1,0,y_1,\dots,y_m).
        \]
    \end{itemize}
    In any case, we deduce that $\varphi_{A\star B}$ is not injective, hence neither is $\Phi_{A \star B}$. This implies that $A \star B$ is not compact, so $B$ is not a complement of $A$.
\end{observation}

\begin{observation}
Complements are not unique. For example, if 
\[A = (8,8):(2,32),\]
then each of the layouts 
\begin{align*}
B_1 & = (2,2):(1,16)\\
B_2 & = (2,2):(16,1)\\
B_3 & = (5,2,2,1):(256,1,16,0)
\end{align*} 
is a complement of $A$. Instead, there is a (possibly empty) set 
    \[
    \mathterm{complements}^\flat(A) = \{\text{flat layouts }B \mid B \text{ is a complement of }A\}.
    \]
of layouts which are complementary to $A$.
\end{observation}

It will be useful to provide a family of equivalent conditions for $B$ to be a complement of $A$ (see Proposition \ref{equivalentconditionsforcomplement}). In order to do so, we need the following technical lemma, which describes the interplay between concatenation, and the operations $\squeeze(-)$, $\sort(-)$, and $\coalesce^\flat(-)$. 

\begin{lemma}\label{concatenationdistributions}
    Suppose $A$ and $B$ are flat layouts. Then 
    \begin{enumerate}
        \item $\squeeze(A\star B) = \squeeze(A)\star\squeeze(B)$,
        \item $\sort(A \star B) = \sort(L \star \sort(B))$, and 
        \item $\coalesce^\flat(A\star B) = \coalesce^\flat(A \star \coalesce^\flat(B))$.
    \end{enumerate}
\end{lemma}

\begin{proof}
    Write 
    \begin{align*}
        A & = (s_1,\dots,s_m):(d_1,\dots,d_m) \\
        B & = (t_1,\dots,t_n):(e_1,\dots,e_n).
    \end{align*}
    If we let $\{ i_1< \dots < i_{m'} \} \subset \langle m \rangle$ denote the indices with $s_{i_k} \neq 1$, and $\{j_1,\dots,j_{n'}\} \subset \langle n \rangle$ denote the indices with $t_{j_\ell} \neq 1$, then 
    \begin{align*}
    \squeeze(A \star B) & = (s_{i_1},\dots,s_{i_{m'}} , t_{j_1},\dots,t_{j_{n'}}) : (d_{i_1},\dots,d_{i_{m'}} , e_{j_1},\dots,e_{j_{n'}})\\
    & = \squeeze(A) \star \squeeze(B).
    \end{align*}
    This proves 1. For 2, we note that for any flat layout $L$, and any permutation $\sigma \in \Sigma_{\len(L)}$, we have $\sort(L) = \sort(L^\sigma)$. The result follows from the observation that 
    \[
    A \star \sort(B) = (A \star B)^\sigma
    \]
    where $\sigma$ is a block permutation of the form $\sigma = \id \times \sigma' \in \Sigma_m \times \Sigma_n \subset \Sigma_{m+n}$. For 3., it suffices to show that $A \star B$ and $A \star \coalesce^\flat(B)$ have the same layout function. This follows from Proposition \ref{layoutfunctionofconcatenatedlayout}.
\end{proof}

\begin{proposition}\label{equivalentconditionsforcomplement}
    Suppose $A$ and $B$ are flat layouts. Then the following are equivalent. 
    \begin{enumerate}
        \item $A \perp B$.
        \item $B \perp A$.
        \item $A \perp \squeeze(B)$.
        \item $A \perp \coalesce^\flat (B)$.
        \item $A \perp \sort(B)$.
    \end{enumerate}
\end{proposition}

\begin{proof}
    We use Proposition \ref{equivalentconditionsforcompact} and Lemma \ref{concatenationdistributions} to prove the equivalence of these conditions. First, we note that $\sort(A \star B) = \sort(B \star A)$, which implies the equivalence of 1 and 2. Next, we note that, by Lemma \ref{concatenationdistributions}, if $\op(-)$ is any of the operations $\squeeze(-)$, $\sort(-)$, or $\coalesce^\flat(-)$, then 
    \[
    \op(A \star B) = \op(A \star \op(B)),
    \]
    and so 
    \begin{align*}
    A \perp B \hspace{0.2in} & \Leftrightarrow \hspace{0.2in} A \star B \text{ is compact}. \\
    & \Leftrightarrow \hspace{0.2in} \op(A \star B) \text{ is compact}. \\
    & \Leftrightarrow \hspace{0.2in} \op(A \star \op(B)) \text{ is compact}. \\
    & \Leftrightarrow \hspace{0.2in} \op(B) \text{ is a complement of }A.
    \end{align*} 
\end{proof}

\noindent We would like to characterize when a flat layout admits a complement. To this end, we make the following definition.

\begin{definition}\label{definitionofcomplementable}
    Suppose $A$ is a flat layout, and write
    \[\sort(\squeeze(A)) = (s_1,\dots,s_m):(d_1,\dots,d_m).\]
    We say $A$ is {\it complementable} if for any $1 \leq i < m$, the integer $s_id_i$ divides $d_{i+1}$.
\end{definition}

\begin{example}
    The flat layout
    \[A_1 = (4,1,1,4,4):(64,0,0,1,8)\] is complementable, while the flat layout
    \[ A_2 = (4,4,4):(64,1,1)\] is not complementable.
\end{example}

\begin{example}
    The flat layout
    \[A_1 = (10,2):(4,80)\] complementable, while the flat layout
    \[ A_2 = (10,2):(80,4)\] is not complementable.
\end{example}

\begin{example}
    If $A$ is compact, then by Proposition \ref{compactisomorphismproposition}, $A$ is complementable.
\end{example}

\begin{example}
    Suppose
    \[A = (s_1,\dots,s_m):(d_1,\dots,d_m)\]
    is a flat layout. If there is any $1 \leq i \leq m$ such that $s_i \neq  1$ and $d_i = 0$, then $A$ is not complementable.
\end{example}

If $A$ is complementable, then we can construct a complement of $A$ as follows.

\begin{construction}\label{complementconstruction}
    Suppose $A$ is a flat layout, and write
    \[\sort(\squeeze(A)) = (s_1,\dots,s_m):(d_1,\dots,d_m).\]
    If $A$ is complementable, then we define a flat layout $\comp^\flat(A)$ as 
    \[
    \comp^\flat(A) = \coalesce^\flat(C)
    \]
    where
    \[C = \Bigl(d_1 , \dfrac{d_2}{s_1d_1} , \dfrac{d_3}{s_2d_2} , \dots ,\dfrac{d_m}{s_{m-1}d_{m-1}} \Bigr) : \Bigl(1 , s_1d_1 , s_2d_2, \dots ,s_{m-1}d_{m-1}\Bigr).
    \]
\end{construction}

\begin{example}
    If $A = (8,8):(1,8)$, then 
    \[\comp^\flat(A) = ():()\]
    is the empty layout. More generally, if $A$ is compact, then $\comp^\flat(A) = ():()$ is the empty layout.
\end{example}

\begin{example}
    If $A = (2,2):(2,8)$, then \[\comp^\flat(A) = (2,2):(1,4).\] 
\end{example}

\begin{example}
    If $A = (3,3,8):(16,96,1)$, then 
    \[\comp^\flat(A) = (2,2):(8,48).\]
\end{example}

\noindent Let's justify that $\comp^\flat(A)$ is, in fact, a complement of $A$. 

\begin{lemma}\label{complementsarecomplements}
    Suppose $A$ is a flat layout. If $A$ is a complementable, then 
    \[A \perp \comp^\flat(A).\]
\end{lemma}

\begin{proof}
    Lets write 
    \[\sort(\squeeze(A)) = (s_1,\dots,s_m):(d_1,\dots,d_m),\]
    so that $\comp^\flat(A) = \coalesce^\flat(C)$
    where
    \[C = \Bigl(d_1 , \dfrac{d_2}{s_1d_1} , \dfrac{d_3}{s_2d_2} , \dots ,\dfrac{d_m}{s_{m-1}d_{m-1}} \Bigr) : \Bigl(1 , s_1d_1 , s_2d_2, \dots ,s_{m-1}d_{m-1}\Bigr).
    \]
    By Proposition \ref{equivalentconditionsforcomplement}, it suffices to prove that $C$ is a complement of $\sort(\squeeze(A))$. This is the case since the concatenation 
    \[
    \sort(\squeeze(A)) \star C
    \]
    is equal to
    \begin{align*}
    \left(s_1,\dots,s_m,d_1,\dfrac{d_2}{s_1d_1},\dots,\dfrac{d_m}{s_{m-1}d_{m-1}}\right):(d_1,\dots,d_m,1,s_1d_1,\dots, s_{m-1}d_{m-1}),
    \end{align*}
    and its sorting is equal to 
    \[
    \left(d_1,s_1, \dfrac{d_2}{s_1d_1},\dots, \dfrac{d_m}{s_{m-1}{d_{m-1}}}, s_m \right):(1,d_1,s_1d_1, \dots ,s_{m-1}d_{m-1}, d_m)
    \]
    which is column-major.
\end{proof}

We have shown that if $A$ is complementable, then $A$ admits a complement. Next, we prove that the converse also holds. 

\begin{proposition}\label{complementableproposition}
    Suppose $A$ is a flat layout. Then there exists a complement $B$ of $A$ if and only if $A$ is complementable.
\end{proposition}

\begin{proof}
    If $A$ is complementable, then by Lemma \ref{complementsarecomplements} the layout $B = \comp^\flat(A)$ is a complement of $A$. Conversely, suppose there exists a complement $B$ of $A$, and consider the flat layout 
    \begin{align*} L & = \sort\bigl(\squeeze(A)\star \squeeze(B)\bigr) \\
    & = (s_1,\dots,s_m):(d_1,\dots,d_n).
    \end{align*}
    Since $\Phi_L(0) = 0$, and $\Phi_L$ is injective, we know that $d_1 \neq 0$. We will argue that $d_i = s_1 \cdots s_{i-1}$, i.e., that $L$ is column-major. Since 
    \[\Phi_L^{\cosize(L)} : [0,\size(L)) \to [0,\cosize(L))
    \]
    is a bijection, we know that $1$ is in the image of $\Phi_L$, which implies that $d_1 = 1$. Suppose $1 < i \leq m$, and suppose we have proved that $d_j = s_1 \cdots s_{j-1}$ for all $j < i$. Consider the stride $d_{i}$. We know that there is no $(x_1,\dots,x_{i-1},0,\dots,0)$ such that $(x_1,\dots,x_{i-1},0,\dots,0) \cdot (d_1,\dots,d_m) = s_1\cdots s_{i-1}$, since the largest possible value of such an expression is 
    \[
    \sum_{j = 1}^{i-1} (s_j-1)(s_1\cdots s_{j-1}) = s_1 \cdots s_i - 1. 
    \]
    Since $\Phi_L$ is surjective, and $d_i \leq  d_{i+1} \leq  \cdots \leq  d_m$, it follows that the next largest value of $\Phi_L$ is $d_i$, so we must have $d_i = s_1\cdots s_{i-1}$, as claimed.

    Returning to our main goal, consider the layout 
    \[\sort(\squeeze(A)) = (s_1',\dots,s_{m'}'):(d_1',\dots,d_{m'}').\] Then there exist $j_1 < \cdots < j_{m'}$ such that $s_i' = s_{j_i}$ and $d_i' = d_{j_i}$ for each $1 \leq i \leq m'$. If $1 \leq i < m'$, then
    \[
    s_i'd_i' = s_{j_i}d_{j_i} = s_{j_i} s_1 \cdots s_{j_i - 1}
    \]
    divides 
    \[
    d_{i+1}' = s_1 \cdots s_{j_{i+1}-1},
    \]
    so we conclude that $A$ is complementable. 
\end{proof}

\noindent Our next goal is to give an abstract characterization of the complement $\comp^\flat(A)$ of a flat layout $A$. In order to do so, we need the following lemma.

\begin{lemma}\label{sortedcomplementablelayoutsareincreasing}
    Suppose $A$ is a flat layout. If $A$ is complementable and sorted, then the layout function 
    \[
    \Phi_A:[0,\size(A)) \to \mathbb{Z}
    \]
    is increasing.
\end{lemma}

\begin{proof}
    Write 
    \[
    A = (s_1,\dots,s_m):(d_1,\dots,d_m).
    \]
    If $1 \leq k \leq m$, we claim that 
    \[
    d_1(s_1-1) + d_2(s_2-1) + \cdots + d_{k-1}(s_{k-1}-1) \leq d_k.
    \]
    If $k = 1$, this holds vacuously, and by induction on $k$, we have 
    \begin{align*}
        d_1(s_1-1) + \cdots + d_{k-2}(s_{k-2}-1) + d_{k-1}(s_{k-1}-1) & \leq d_{k-1} + d_{k-1}(s_k-1)\\
        & = d_{k-1}s_{k-1}\\
        & \leq d_k.
    \end{align*}
    Now, suppose we have $x,y \in [0,\size(A))$ with $x \leq y$. These integers correspond, under the colexicographic isomorphism, to tuples.
    \[(x_1,\dots,x_m),(y_1,\dots,y_m) \in [0,s_1) \times \cdots \times [0,s_m)\]
    Since $x \leq y$, we know there is some maximal $1 \leq k \leq m$ such that $x_k<y_k$, and $x_\ell = y_\ell$ for all $k < \ell \leq m$. Now we can compute 
    \begin{align*}
    \Phi_A(x) & = d_1x_1 + \cdots + d_{k-1}x_{k-1} + d_k x_k + d_{k+1}x_{k+1} + \cdots + d_m x_m \\
    & = d_1x_1 + \cdots + d_{k-1}x_{k-1} + d_k x_k + d_{k+1}y_{k+1} + \cdots + d_m y_m\\
    & \leq d_1(s_1-1) + \cdots + d_{k-1}(s_{k-1}-1) + d_k x_k + d_{k+1}y_{k+1} + \cdots + d_m y_m\\
    & \leq d_k + d_kx_k + d_{k+1}y_{k+1} + \cdots + d_m y_m\\
    & = d_k(x_k+1) + d_{k+1}y_{k+1} + \cdots + d_m y_m\\
    & \leq d_ky_k +  d_{k+1}y_{k+1} + \cdots + d_m y_m\\
    & \leq d_1y_1 + \dots d_my_m\\
    & = \Phi_A(y).
    \end{align*}
\end{proof}

\begin{proposition}\label{characterizationofflatcomplement}
    Suppose $A$ and $B$ are flat layouts. If
    \begin{enumerate}
        \item $A \perp B$,
        \item $\size(B) = \size(\comp^\flat(A))$,
        \item $B$ is coalesced, and 
        \item $B$ is sorted, 
    \end{enumerate}
    then $B = \comp^\flat(A)$. 
\end{proposition}

\begin{proof}
    Conditions 1 and 2 imply that $\Phi_{B}$ and $\Phi_{\comp^\flat(A)}$ have the same image. Since $B$ and $\comp^\flat(A)$ are sorted, we know by Lemma \ref{sortedcomplementablelayoutsareincreasing} that $\Phi_{B}$ and $\Phi_{\comp^\flat(A)}$ are increasing. Combining these two facts, it follows that $\Phi_{B} = \Phi_{\comp^\flat(A)}$. Proposition \ref{coalesceproposition} and condition 3 then imply that
    \[
    B = \coalesce^\flat(B) = \coalesce^\flat(\comp^\flat(A)) = \comp^\flat(A).
    \]
\end{proof}

\begin{definition}\label{definitionofNcomplementable}
    Suppose $A$ and $B$ are flat layouts, and $N$ is a positive integer. We say $B$ is a $N${\it -complement} of $A$ if $B$ is a complement of $A$ and 
    \[
    \size(A)\cdot \size(B) = N.
    \]
\end{definition}

\begin{definition}
    Suppose $A$ is a flat layout, and write 
    \[
    \sort(\squeeze(A)) = (s_1,\dots,s_m):(d_1,\dots,d_m).
    \]
    We say $A$ is $N$-complementable if 
    \begin{enumerate} 
    \item for all $1 \leq i < m$, the integer $s_id_i$ divides $d_{i+1}$, and 
    \item the integer $s_md_m$ divides $N$. 
    \end{enumerate}
\end{definition}

\begin{observation}
If $A$ is complementable, and $s_m:d_m$ is the last mode in the layout $\sort(\squeeze(A))$, then $A$ is $N$-complementable exactly when $N$ is a positive integer multiple of $s_md_m$. 
\end{observation}

\begin{observation}
    $N$-complements are not unique. For example, if $A = (2,2):(1,50)$ and $N = 100$, then each of the layouts $B_1 = (25):(2)$, and $B_2 = (5,5):(2,10)$ is a $N$-complement of $A$. As a more general example, if $B$ is a $N$-complement of $A$, then $\coalesce^\flat(B)$ is also a $N$-complement of $A$.
\end{observation}

\begin{remark}
    Suppose $A$ is a flat layout and $B_1$ and $B_2$ are $N$-complements of $A$. Then the layout functions $\Phi_{B_1}$ and $\Phi_{B_2}$ need not be equal, but they necessarily have the same image. For example, if $A = (4):(63)$ and $N = 252$ then $B_1 = (7,9):(1,7)$ and $B_2 = (9,7):(7,1)$ are $N$-complements of $A$, and $\Phi_{B_1} \neq \Phi_{B_2}$, since 
    \[
    \Phi_{B_1}(1) = 1 \neq 7= \Phi_{B_2}(1).
    \]
    As a more general example, if $B$ is a $N$-complement of $A$, then $\sort(B)$ is also a $N$-complement of $A$.
\end{remark}

\begin{construction}\label{Ncomplementconstruction}
    Suppose $A$ is a flat layout, $N$ is a positive integer, and $A$ is $N$-complementable. If we write 
    \[\sort(\squeeze(A)) = (s_1,\dots,s_m):(d_1,\dots,d_m),\]
    then we define a flat layout $\comp^\flat(A,N)$ by 
    \[
    \comp^\flat(A,N) = \coalesce^\flat(C)
    \]
    where
    \[C = \Bigl(d_1 , \dfrac{d_2}{s_1d_1} , \dfrac{d_3}{s_2d_2} , \dots ,\dfrac{N}{s_md_m} \Bigr) : \Bigl(1 , s_1d_1 , s_2d_2, \dots ,s_md_m \Bigr).
    \]
\end{construction}

\begin{example} 
If $A = (3,10):(80,4)$ and $N = 2400$, then 
\[
\comp^\flat(A,N) = (4,2,10):(1,40,240).
\]
\end{example}

\begin{lemma}\label{Ncomplementsarecomplements}
    Suppose $A$ is a flat layout, $N$ is a positive integer, and $A$ is $N$-complementable. Then $\comp^\flat(A,N)$ is a $N$-complement of $A$.
\end{lemma}

\begin{proof}
    Lets write 
    \[\sort(\squeeze(A)) = (s_1,\dots,s_m):(d_1,\dots,d_m),\]
    so that $\comp^\flat(A,N) = \coalesce^\flat(C)$ where
    \begin{align*}
    C
    & = \left(d_1 , \dfrac{d_2}{s_1d_1} , \dfrac{d_3}{s_2d_2} , \dots ,\dfrac{N}{s_md_m} \right) : \left(1 , s_1d_1 , s_2d_2, \dots ,s_md_m \right).
    \end{align*}
    First, we compute
    \begin{align*}
    \size(A) \cdot \size(B) & = \left(\prod_{i=1}^m s_i \right) \cdot \left( d_1 \cdot \left(\prod_{i=2}^m \dfrac{d_i}{s_{i-1}d_{i-1}}\right) \cdot \dfrac{N}{s_md_m} \right)\\
    & = \dfrac{\left(\displaystyle\prod_{i=1}^m s_i\right)\left(\displaystyle\prod_{i=1}^m d_i \right)}{\left(\displaystyle\prod_{i=1}^m s_id_i \right)} \cdot N \\
    & = N. 
    \end{align*}
    We need to check that $A \star B$ is compact. Equivalently, we need to check that $\Phi_{A \star B}^N$ is an isomorphism. By Lemma \ref{compactisomorphismproposition}, it suffices to prove that 
    \[\squeeze(A)\star \squeeze(B)\]
    is compact. This is the case since this layout is equal to
    \begin{align*}
    \left(s_1,\dots,s_m,d_1,\dfrac{d_2}{s_1d_1},\dots,\dfrac{N}{s_md_m}\right):(d_1,\dots,d_m,1,s_1d_1,\dots, s_md_m)
    \end{align*}
    and so its sorting
    \[\sort(\squeeze(A)\star \squeeze(B))\] is equal to 
    \[
    \left(d_1,s_1, \dfrac{d_2}{s_1d_1},\dots, \dfrac{d_m}{s_{m-1}{d_{m-1}}}, s_m , \dfrac{N}{s_md_m}\right):(1,d_1,s_1d_1, \dots ,s_{m-1}d_{m-1}, d_m,s_md_m)
    \]
    which is column-major.
\end{proof}

\begin{proposition}\label{characterizationofexistenceofNcomplements}
    Suppose $A$ is a flat layout and $N$ is a positive integer. Then there exists a $N$-complement $B$ of $A$ if and only if $A$ is $N$-complementable.
\end{proposition}

\begin{proof}
    If $A$ is $N$-complementable, then by Lemma \ref{Ncomplementsarecomplements} the layout $B = \comp^\flat(L,N)$
    is a $N$-complement of $A$.
    
    On the other hand, suppose there exists a $N$-complement $B$ of $A$. Consider the flat layout 
    \begin{align*} L & \coloneq \sort\bigl(\squeeze(A)\star \squeeze(B)\bigr) \\
    & = (s_1,\dots,s_m):(d_1,\dots,d_n).
    \end{align*}
    Since $\Phi_L(0) = 0$, and $\Phi_L$ is injective, we know that $d_1 \neq 0$. We will argue that $d_i = s_1 \cdots s_{i-1}$, i.e., that $L$ is column-major. Since 
    \[\Phi_L^{N} : [0,N) \to [0,N)\]
    is a bijection, we know that $1$ is in the image of $\Phi_L$, which implies that $d_1 = 1$. Suppose $1 < i \leq m$, and suppose we have proved that $d_j = s_1 \cdots s_{j-1}$ for all $j < i$. Consider the stride $d_{i}$. We know that there is no $(x_1,\dots,x_{i-1},0,\dots,0)$ such that $(x_1,\dots,x_{i-1},0,\dots,0) \cdot (d_1,\dots,d_m) = s_1\cdots s_{i-1}$, since the largest possible value of such an expression is 
    \[
    \sum_{j = 1}^{i-1} (s_j-1)(s_1\cdots s_{j-1}) = s_1 \cdots s_i - 1. 
    \]
    Since $\Phi_L$ is surjective, and $d_i \leq  d_{i+1} \leq  \cdots \leq  d_m$, it follows that the next largest value of $\Phi_L$ is $d_i$, so we must have $d_i = s_1\cdots s_{i-1}$, as claimed.

    Returning to our main goal, consider the layout 
    \[\sort(\squeeze(A)) = (s_1',\dots,s_{m'}'):(d_1',\dots,d_{m'}').\] Then there exist $j_1 < \cdots < j_{m'}$ such that $s_i' = s_{j_i}$ and $d_i' = d_{j_i}$ for each $1 \leq i \leq m'$. If $1 \leq i < m'$, then
    \[
    s_i'd_i' = s_{j_i}d_{j_i} = s_{j_i} s_1 \cdots s_{j_i - 1}
    \]
    divides 
    \[
    d_{i+1}' = s_1 \cdots s_{j_{i+1}-1}.
    \]
    If $i = m'$, then 
    \[
    s_{m'}'d_{m'}' = s_{j_{m'}}d_{j_{m'}} = s_{j_{m'}} s_1 \cdots s_{j_{m'} - 1}
    \]
    divides 
    \[
    N = s_1 \cdots s_m.
    \]
    We conclude that $A$ is $N$-complementable. 
\end{proof}

\begin{proposition}\label{characterizationofflatNcomplement}
    Suppose $N$ is a positive integer, and $A$ is a $N$-complementable flat layout. If $B$ is a flat layout such that 
    \begin{enumerate}
        \item $B$ is a $N$-complement of $L$,
        \item $B$ is coalesced, and 
        \item $B$ is sorted, 
    \end{enumerate}
    then $B = \comp^\flat(A,N)$. 
\end{proposition}

\begin{proof}
    Conditions 1 and 2 imply that $\Phi_{B}$ and $\Phi_{\comp^\flat(A,N)}$ have the same image. Since $B$ and $\comp^\flat(A,N)$ are sorted, we know by Lemma \ref{sortedcomplementablelayoutsareincreasing} that $\Phi_{B}$ and $\Phi_{\comp^\flat(A,N)}$ are increasing. Combining these two facts, it follows that $\Phi_{B} = \Phi_{\comp^\flat(A,N)}$. Proposition \ref{coalesceproposition} and condition 3 then imply that
    \[
    B = \coalesce^\flat(B) = \coalesce^\flat(\comp^\flat(A,N)) = \comp^\flat(A,N).
    \]
\end{proof}

\begin{lemma}\label{compatibilityofNcomplements}
    Suppose $A$ is a flat layout. If $N_1 \leq N_2$ are positive integers such that $A$ is $N_1$-complementable and $A$ is $N_2$-complementable, then 
    \[
    \Phi_{\comp^\flat(A,N_2)} \mid_{[0,N_1)} = \Phi_{\comp^\flat(A,N_1)}. 
    \]
\end{lemma}

\begin{proof}
    Write 
    \begin{align*}
    \sort(\squeeze(A)) & = (s_1,\dots,s_m):(d_1,\dots,d_m), \\
    C & = \Bigl(d_1 , \dfrac{d_2}{s_1d_1} , \dfrac{d_3}{s_2d_2} , \dots ,\dfrac{d_m}{s_{m-1}d_{m-1}} \Bigr) : \Bigl(1 , s_1d_1 , s_2d_2, \dots ,s_{m-1}d_{m-1}\Bigr)
    \end{align*}
    and write 
    \begin{align*}
    E_1 & = \left( \dfrac{N_1}{s_md_m}\right):(s_md_m), \\
    E_2 & = \left( \dfrac{N_2}{s_md_m}\right):(s_md_m), \\
    C_1 & = C \star E_1,\\
    C_2 & = C \star E_2,
    \end{align*}
    so that 
    \begin{align*}
        \comp^\flat (A) & = \coalesce^\flat(C)\\
        \comp^\flat(A,N_1) & = \coalesce^\flat(C_1)\\
        \comp^\flat(A,N_2) & = \coalesce^\flat(C_2).
    \end{align*}
    Then we have a commuting diagram 
    \[\begin{tikzcd} 
    {[}0,\size(C_1){)}\ar[rrr,"\text{colex}_{(\size(C),N_1)}^{-1}"] \ar[d,"\subseteq"] & & & {[}0,\size(C){)} \times {[}0,N_1 {)} \ar[rrr,"\Phi_C \times s_md_m "] \ar[d,"\id \times \subseteq"] & & & \mathbb{Z} \times \mathbb{Z} \ar[r,"+"] \ar[d,"\id" ] & \mathbb{Z} \ar[d,"\id"] \\
    {[}0,\size(C_2){)}\ar[rrr,"\text{colex}_{(\size(C),N_2)}^{-1}"] & & & {[}0,\size(C){)} \times {[}0,N_2 {)} \ar[rrr,"\Phi_C \times s_md_m "] & & & \mathbb{Z} \times \mathbb{Z} \ar[r,"+"] & \mathbb{Z} \\
    \end{tikzcd} \] 
    where, by Proposition \ref{layoutfunctionofconcatenatedlayout}, the composite of the top row is the layout function of $C_1 = C \star E_1$, and the composite of the bottom row is the layout function of $C_2 = C \star E_2$. This tells us that the restriction of $\Phi_{C_2}$ to $[0,\size(C_2))$ is $\Phi_{C_1}$, and the result follow from the fact that 
    \begin{align*}
    \Phi_{\comp^\flat (A,N_1)}&  = \Phi_{C_1}\\
    \Phi_{\comp^\flat (A,N_2)}&  = \Phi_{C_2}.
    \end{align*}
\end{proof}

\newpage 

\subsection{Further operations}

In this section, we define several further operations on flat layouts, namely {\it composition}, {\it flat division}, and {\it flat products}. These are the flattened variants of more natural operations on (nested) layouts. We do not often work with these operations, but include them anyway for completeness.

\subsubsection{Composition}
If $A$ and $B$ are flat layouts, then the composite $B \circ A$ is a flat layout whose layout function is the composite of the layout functions of $A$ and $B$. More precisely, we have the following definition.

\begin{definition}\label{definitionofflatlayoutcomposition} Suppose $A$ and $B$ are flat layouts. We say the flat layout $C$ is the {\it composition} of $A$ and $B$, and write $C = B \circ A$, if 
    \begin{enumerate}
    \item $C$ is non-degenerate,
    \item $\shape(A) = \shape(R)$,
    \item $\Phi_R = \Phi_B \circ \Phi_A^{\size(B)}$.
    \end{enumerate}
\end{definition}

\begin{remark}
    Note that condition 2 in our definition ensures that $\Phi_R$ and $\Phi_A$ have the same domain, and condition 3 implies $\cosize(A) \leq \size(B)$. 
\end{remark}

\begin{example}\label{flatcompositionexample1}
If $A = (2,3):(5,6)$ and $B = (80):(10)$, then 
\[
B \circ A = (2,3):(50,60).
\]
More generally, if 
\[
A = (s_1,\dots,s_m):(d_1,\dots,d_m)
\]
is a non-degenerate flat layout, and 
\[
B = (t):(e)
\]
is a rank $1$ flat layout with $t \geq \cosize(A)$, then $A$ and $B$ are composable, and 
\[
B \circ A = (s_1,\dots,s_m):(td_1,\dots,td_m).
\]
\end{example}

\begin{example}
    If $A = (128,128):(0,0)$ and $B = (64,32):(1,64)$, then 
    \[
    B \circ A = (128,128):(0,0).
    \]
    More generally, if $A$ is a flat layout each of whose stride entries is zero, and $B$ is any flat layout, then $A$ and $B$ are composable with $B \circ A = A$. 
\end{example}

\begin{example}
    If $A = (64,32):(2,256)$ and $B = (2048,2048):(1,2048)$, then 
    \[
    B \circ A = (64,32):(2,256).
    \]
    More generally, if $A$ is any flat layout, and $B$ is a column-major flat layout with $\cosize(A) \leq \size(B)$, then $B \circ A = A$. 
\end{example}

\begin{example}\label{flatcompositionexample2}
    If $A = (4):(2)$ and $B = (2,2,6):(12,6,1)$, then there is no flat layout $R$ with $R = B \circ A$.
\end{example}

\begin{remark}
    If $B'$ and $B$ have the same layout function, then $B \circ A = B' \circ A$.
\end{remark}

\begin{remark} Flat layouts are a special case of the more general notion of {\it layouts} (Definition \ref{definitionofnestedlayout}). It turns out that there are cases (such as Example \ref{flatcompositionexample2}) where there does not exist a flat layout $C$ with $C = B \circ A$, but there does exist a (nested) layout $C$ with $C = B \circ A$ (see Example \ref{compositionexample2}). For this reason, we postpone further discussion and analysis of composition until we have defined layouts in their full generality.
\end{remark}

\subsubsection{Flat division}
If $A$ and $B$ are flat layouts, then the flat division of $A$ by $B$ is a flattened version of the more natural {\it logical division} of layouts. See Section \ref{logicaldivisionsection} for details. 
\begin{definition}
Suppose $A$ and $B$ are flat layouts, and that $B$ is $\size(A)$-complementable, with 
\[
B^c = \comp^\flat(B,\size(A)).
\]
We define the {\it flat division} of $A$ by $B$ to be the flat layout
\[
A\oslash^\flat B = A \circ (B \star B^c).
\]
\end{definition}

\begin{example}
    If $A = (2,2,2,2):(1,4,2,8)$ and $B = (2,2):(4,2)$, then 
    \[
    A\oslash^\flat B = (2,2,2,2):(4,2,1,8).
    \]
\end{example}

\begin{example}
    If $A = (3,5,9,6):(54,0,6,1)$ and $B = (6,3):(135,1)$, then 
    \[
    A\oslash^\flat B = (6,3,5,9):(1,54,0,6).
    \]
\end{example}

\begin{example}
    If $A$ is any flat layout and $B = ():()$ is the empty layout, then
    \[
    A\oslash^\flat B =  A.
    \]
\end{example}

\subsubsection{Flat products}

If $A$ and $B$ are flat layouts, then the flat product $A \otimes^\flat B$ of $A$ and $B$ is a flattened version of the more natural {\it logical product} of layouts. See Section \ref{logicalproductsection} for details.
\begin{definition}
    Suppose $A$ and $B$ are flat layouts, and that $A$ is $\size(A) \cdot \cosize(B)$-complementable, with 
    \[
    A^c = \comp^\flat(A,\size(A)\cdot \cosize(B)).
    \]
    We define the flat product of $A$ and $B$ by 
    \[
    A \otimes^\flat B = A \star (A^c \circ B).
    \]
\end{definition}
\begin{example}
    If $A = (2,2,2):(1,2,4)$ and $B = (2,2,2):(1,2,4)$, then 
    \[
    A \otimes^\flat B = (2,2,2,2,2,2):(1,2,4,8,16,32).
    \]
\end{example}

\begin{example}
    If $A = (2,2,2):(1,2,4)$ and $B = (3,5):(5,1)$, then 
    \[
    A \otimes^\flat B = (2,2,2,3,5):(1,2,4,40,8).
    \]
\end{example}

\begin{example}
    If $A$ is any flat layout and $B = ():()$ is the empty layout, then 
    \[
    A \otimes^\flat B = A.
    \]
\end{example}

\subsection{Tractable flat layouts}

In this section we define an especially well-behaved class of flat layouts, called {\it tractable} flat layouts. Tractable flat layouts include the most important examples of interest, such as row-major, column-major, compact, and complementable layouts. Later on, we will see that tractable flat layouts are precisely the layouts which arise from a certain category $\catstyle{Tuple}$.

\begin{definition}\label{definitionoftractableflatlayouts}
    Suppose $L$ is a flat layout, and write
    \[
    \sort(L) = (s_1,\dots,s_m):(d_1,\dots,d_m).
    \]
    We say $L$ is {\it tractable} if for each $1 \leq i < m$, we have 
    \begin{enumerate}
        \item $d_i = 0$\text{, or}
        \item $s_id_i$ divides $d_{i+1}$. 
    \end{enumerate}
\end{definition}

\begin{example}
    The flat layout 
    \[L = (12):(17)\]
    is tractable. More generally, any flat layout of rank $1$ is tractable. 
\end{example}

\begin{example}
    The flat layout 
    \[L = (2,4,32):(1,2,8)\]
    is tractable. More generally, any column-major layout 
    \[
    L = (s_1,\dots,s_m):(1,s_1,\dots,s_1 \cdots s_{m-1})
    \]
    is tractable.
\end{example}
\begin{example}
    The flat layout 
    \[L = (2,4,32):(128,32,1)\]
    is tractable. More generally, any row-major layout 
    \[
    L = (s_1,\dots,s_m) : (s_2\cdots s_m, \dots, s_m , 1)
    \]
    is tractable.
\end{example}
\begin{example}
    The flat layout 
    \[L = (3,3,1,3,3,1,3):(81,1,0,9,3,0,27)\]
    is tractable. More generally, any compact flat layout is tractable.
\end{example}

\begin{example}
    The flat layout 
    \[L = (3,7,7):(0,15,0)\]
    is tractable. More generally, any flat layout with exactly one non-zero stride is tractable.
\end{example}

\begin{example}
    The flat layout 
    \[L = (2,2,2,2):(1,2048,16,64)\]
    is tractable. More generally, any complementable flat layout is tractable. 
\end{example}

\begin{example}
    Suppose $L$ is a flat layout. If $L$ is tractable and $I \subset \langle m \rangle$ is any subset, then the restriction $L \mid_I$ is tractable. In particular, if $L$ is tractable, then $\squeeze(L)$ and $\filter(L)$ are tractable.
\end{example}

\begin{example}
    The flat layout  
    \[L = (4,8):(3,3)\]
    is not tractable. In particular, this shows that the concatenation $L_1 \star L_2$ of tractable flat layouts $L_1$ and $L_2$ need not be tractable.
\end{example}

\begin{observation}
    If $L$ is a tractable flat layout and no entry of $\stride(L)$ is equal to $0$, then $L$ is complementable. In particular, if $L$ is tractable, then $\filter(L)$ is complementable.
\end{observation}

We conclude this section by enumerating a family of equivalent conditions for a flat layout $L$ to be tractable.

\begin{proposition}\label{equivalentconditionsfortractable}
    Suppose $L$ is a flat layout. Then the following conditions are equivalent.
    \begin{enumerate}
        \item $L$ is tractable. 
        \item $\sort(L)$ is tractable.
        \item $\filter(L)$ is tractable.
        \item $\filter(L)$ is complementable. 
    \end{enumerate}
\end{proposition}

\begin{proof} Suppose $L$ is a flat layout.
    \begin{itemize}
        \item (1 $\Leftrightarrow $ 2): This follows from the fact that 
        \[
        \sort(\sort(L)) = \sort(L).
        \]
        \item (1 $\Leftrightarrow $ 3): This follows from the fact that 
        \[
        \sort(\filter(L)) = \filter(\sort(L)).
        \]
        \item (3 $\Leftrightarrow $ 4): This follows from the fact that if
        \[
        L = (s_1,\dots,s_m):(d_1,\dots,d_m)
        \]
        is a flat layout such that each of the stride entries $d_i$ is nonzero, then the definition of tractability coincides with that of complementability.
    \end{itemize}
\end{proof}

\newpage

\section{Nested Tuples}\label{nestedtuplessection}

In this section, we introduce nested tuples, which are the generalization of tuples needed to define layouts in full generality. 

\subsection{Profiles}

A nested tuple $S$ is determined by its {\it flattening}, which is an ordinary tuple, and its {\it profile}, which describes parenthesization pattern on $S$. We define profiles precisely as follows.

\begin{definition}
    A {\it profile} $P$ is either
    \begin{enumerate}
        \item $P = *$, or
        \item a tuple $P = (P_1,\dots,P_r)$ of profiles $P_1,\dots,P_r$ for some $r \geq 0$. 
    \end{enumerate}
    We write $\mathterm{Profile}$ for the set of profiles. 
\end{definition}

\begin{example}
    Here are some examples of profiles.
    \[
    \begin{aligned}
    P_1 &= (*,*) \\
    P_2 &= (*,(*,*))\\
    P_3 & = ((*,*),(*,*))\\
    P_4 & = ((*,*,*),(*,()))\\
    P_5 & = ()\\
    P_6 & = *
\end{aligned}
\]
\end{example}
Let's define some important attributes of profiles.
\begin{definition}\label{definitionofnestedtupleattributes}
    Suppose $P$ is a profile.
    \begin{itemize} 
    \item The {\it rank} of $X$ is 
    \[
    \rank(P) = \begin{cases}
        1 & P = *\\
        r & P = (P_1,\dots,P_r) \text{ is a tuple of profiles.}
    \end{cases}.
    \]
    \item The {\it length} of $P$ is 
    \[
    \len(P) = 
    \begin{cases}
        1 & P = *\\
        \sum_{i=1}^r\len(P_i) & P = (P_1,\dots,P_r) \text{ is a tuple of profiles.}
    \end{cases}
    \]
    \item The {\it depth} of $P$ is 
    \[
    \depth(P) = 
    \begin{cases}
        0 & P = *\\
        1 + \displaystyle\max_{1 \leq i \leq r} (\depth(P_i)) & P = (P_1,\dots,P_r) \text{ is a tuple of profiles.}
    \end{cases}
    \]
    \end{itemize}
\end{definition}

\begin{example}
Here are some examples of profiles, together with their rank, length, and depth :
\[
\begin{aligned}
P &= * & \quad \rank(P) &= 1, \quad \len(P)  = 1, & \depth(P) &= 0\\
P &= (*,*,*) & \quad \rank(P) &= 3, \quad \len(P)  = 3, &\depth(P) &= 1 \\
P &= (((*,*),*,*),*,*), & \quad \rank(P) &= 3, \quad \len(P)  = 6, & \depth(P) &= 3 \\
P &= (((),()),(*,(*,*))), & \quad \rank(P) &= 2, \quad \len(P)  = 3, & \depth(P) &= 3 \\
\end{aligned}
\]
\end{example}

\begin{definition}\label{definitionofmodeofnestedtuple}
Suppose $P$ is a profile with $\rank(P) = r$. If $1 \leq i \leq r$, then the $i$th {\it mode of} $P$ is
\[
\mode_i(P) = \begin{cases}
    P & \depth(P) = 0 \text{ (hence }i=r=1\text{),}\\
    P_i & P = (P_1,\dots,P_{r}) \text{ has depth }\geq 1.
\end{cases}
\]
\end{definition}

\begin{example}
    If $P = ( (*,*),(()),((*,(*,*))))$ then the modes of $P$ are 
    \begin{align*}
        \mode_1(P) & = (*,*)\\
        \mode_2(P) & = (())\\
        \mode_3(P) & = (*,(*,*)).
    \end{align*}
\end{example}

The following notation will be useful.

\begin{notation}
    Suppose $P$ is a profile of depth $>0$. For any $1 \leq j \leq \rank(P)$, we write 
    \begin{align*}
    \len_j(X) & = \len(\mode_j(P))\text{, } \\
    \len_{<j}(P) & = \sum_{i=1}^{j-1} \len_i(X)\text{,}\\
    \len_{\leq j}(X) & = \len_{<j}(P) + \len_j(P)\\
    \end{align*}
\end{notation}
The most important operation supported by profiles is {\it substitution}: If $Q$ is a profile of length $m$, and $P_1,\dots,P_m$ are profiles, then we can obtain a new profile $(P_1,\dots,P_m)_Q$ by substituting the $i$th entry of $Q$ with the profile $P_i$, for each $1 \leq i \leq m$. More precisely, we have the following definition. 

\begin{definition}\label{definitionofPconcatenation}
    Suppose $Q$ is a profile of length $m$, and suppose $P_1,\dots,P_m$ are profiles. Then the $Q${\it -substitution} of $P_1,\dots,P_m$ is the profile
    \[
    (P_1,\dots,P_m)_Q
    \]
    defined as follows. Write $\depth(Q) = d$ and $\rank(Q) = r$. 
    \begin{itemize}[left = 0pt]
        \item If $d = 0$, then $m = 1$, and we define 
        \[
        (P_1)_Q = P_1.
        \]
        \item Suppose next that $d > 0$, and that we have defined $Q'$-substitution for all profiles $Q'$ of depth $<d$. We can write 
    \[
    Q = (Q_1,\dots, Q_r)
    \]
    where each mode $Q_i = \mode_i(Q)$ has depth $< d$. If for each $1 \leq i \leq r$, we set 
    \[\ell_i = \len(P_1) + \cdots + \len(P_{i-1}),\] 
    then we define
    \[
    (P_1,\dots,P_r)_Q = ((P_{1},\dots,P_{\ell_2})_{Q_1},\dots,(P_{\ell_r +1},\dots,P_{\ell_{r+1}})_{Q_r}).
    \]
    \end{itemize}
\end{definition}

\begin{example}
    If $Q = (*,*)$ and $P_1 = (*,*)$, $P_2 = (*,*,*)$, 
    \[
    (P_1,P_2)_Q = ( (*,*), (*,*,*)).
    \]
    More generally, if $Q = (*,\dots,*)$ is the profile with $\depth(Q) = 1$ and $\len(Q) = \rank(Q) = r$, then 
    \[
    (P_1,\dots,P_r)_Q = (P_1,\dots,P_r)
    \]
    is ordinary concatenation. 
\end{example}

\begin{aside}
    There is an operadic interpretation of $Q$-substitution. The set $\mathterm{Profile}$ of profiles has the structure of a (non-symmetric) operad: the set 
    \[
    \mathterm{Profile}(n) = \{P \in \mathterm{Profile} \mid \len(P) = n\}
    \]
    forms the collection of $n$-ary operations of $\mathterm{Profile}$, and if $n = m_1 + \cdots + m_r$, then the structure map  
    \[ \begin{tikzcd} 
     \mathterm{Profile}(m_1)\times \dots \times \mathterm{Profile}(m_r) \times \mathterm{Profile}(n) \ar[rr] & &  \mathterm{Profile}(m_1 + \cdots + m_r)\\
    (P_1,\dots,P_r),Q \ar[rr,mapsto] & & (P_1,\dots,P_r)_Q
    \end{tikzcd}\]
    is given by $Q$-substitution. One can also form the cofree symmetric operad on this non-symmetric operad, which amounts to endowing the sets of $n$-ary operations with trivial symmetric group action.
\end{aside}

\subsection{Basic definitions}

Having defined profiles and their basic properties, we can now define nested tuples. 

\begin{definition}\label{definitionofnestedtuple}
    If $V$ is a set, then a {\it nested tuple} $X$ with entries in $V$ is a pair $(X^\flat,P)$ consisting of 
    \begin{enumerate}
        \item a tuple $X^\flat = (x_1,\dots,x_m)$ with entries in $V$, called the {\it flattening} of $X$, and  
        \item a profile $\profile(X) = P$ of length $m$, called the {\it profile} of $X$. 
    \end{enumerate}
    We write $\mathterm{Nest}(V)$ for the set of all nested tuples with entries in a set $V$. 
\end{definition}

\begin{example} Here are some examples of nested tuples, together with their flattening and profile. 
\[
\begin{aligned}
X &= (2,(2,2)) \quad & X^\flat &= (2,2,2) \quad & \profile(X) & = (*,(*,*)) \\
X & = 25 \quad & X^\flat & = (25) \quad & \profile(X) & = * \\
X & = (((2,2,2),8),64) \quad & X^\flat & = (2,2,2,8,26) \quad & \profile(X) & = (((*,*,*),*),*)\\
X & = ((),(32,()),(4,8)) \quad & X^\flat & = (32,4,8) \quad & \profile(X) & = ((),(*,()),(*,*))
\end{aligned}
\]
\end{example}

\begin{notation}
    We sometimes write  
    \[X = (x_1,\dots,x_m)_P\]
    to denote a nested tuple with $X^\flat = (x_1,\dots,x_m)$ and profile $\profile(X) = P$. 
\end{notation}

\begin{observation}\label{nesttuplepullbacksquare}
    If $V$ is any set, then by definition, we have a pullback square
    \[ \begin{tikzcd} 
    \mathterm{Nest}(V) \ar[rr,"\profile(-)"] \ar[d,swap,"(-)^\flat"] \arrow[drr, phantom, "\lrcorner", very near start] & & \mathterm{Profile} \ar[d,"\len(-)"] \\
    \mathterm{Tuple}(V) \ar[rr,swap,"\len(-)"] & & \mathbb{N}.
    \end{tikzcd} \]
\end{observation}

\begin{remark}
    Given the recursive definition of profiles, we could equivalently define a nested tuple with entries in $V$ to be either
    \begin{enumerate}
        \item an element of $V$, or 
        \item a tuple of nested tuples with entries in $V$.
    \end{enumerate}
\end{remark}

Let's define some important attributes of nested tuples. Each such attribute of a nested tuple $X$ is inhereted by its flattening $X^\flat$ or its profile $\profile(X)$. 

\begin{definition}\label{definitionofnestedtupleattributes}
    Suppose $X$ is a nested tuple with entries in $V$. 
    \begin{itemize} 
    \item The {\it rank} of $X$ is 
    \[
    \rank(X) = \rank(P)
    \]
    \item The {\it length} of $X$ is 
    \[
    \len(X) = \len(P) = \len(X^\flat)
    \]
    \item The {\it depth} of $X$ is 
    \[
    \depth(X) = \depth(P)
    \]
    \item If $V = \mathbb{Z}$, then the {\it size} of $X$ is
    \[
    \size(X) = \size(X^\flat).
    \]
    \end{itemize}
\end{definition}

\begin{example}
Here are some examples of nested tuples of integers, together with their rank, length, depth, and size:
\[
\begin{aligned}
X &= 27 & \quad \rank(X) &= 1, \quad \len(X)  = 1, & \depth(X) &= 0, \quad \size(X) = 27 \\
X &= (2,10,5) & \quad \rank(X) &= 3, \quad \len(X)  = 3, &\depth(X) &= 1, \quad \size(X) = 100 \\
X &= (((3,4),2,2),8,9), & \quad \rank(X) &= 3, \quad \len(X)  = 6, & \depth(X) &= 3, \quad \size(X) = 3096 \\
X &= (((),()),(2,(5,5))), & \quad \rank(X) &= 2, \quad \len(X)  = 3, & \depth(X) &= 3, \quad \size(X) = 50 \\
\end{aligned}
\]
\end{example}

\begin{example}
    A nested tuple of integers with depth $0$ is simply an integer.
\end{example}

\begin{example}
A nested tuple of integers with depth $1$ is simply a tuple of integers. If $X$ is such a nested tuple, then $\rank(X) = \len(X)$.
\end{example}

\begin{definition}\label{definitionofmodeofnestedtuple}
Suppose $X = (x_1,\dots,x_m)_P$ is a nested tuple with $\rank(X) = r$. If $1 \leq i \leq r$, then the $i$th {\it mode of} $X$ to be the nested tuple 
\[
\mode_i(X) = (x_{\len_{<i}(P) + 1},\dots,x_{\len_{\leq i}(P)})_{\mode_i(P)}.
\]
\end{definition}

\begin{example}
    If 
    \[X = ((3),4,((10,10),12)),\]
    then the modes of $X$ are
    \begin{align*}
    \mode_1(X) & = (3)\\
    \mode_2(X) & = 4 \\
    \mode_3(X) & = ((10,10),12)
    \end{align*}
\end{example}

\begin{example}
    If $X = (32,5,6,64)$, then the modes of $X$ are 
    \begin{align*}
    \mode_1(X) & = 32\\
    \mode_2(X) & = 5 \\
    \mode_3(X) & = 6\\
    \mode_4(X) & = 64
    \end{align*}
\end{example}

\noindent It will be convenient to introduce the following notation. 

\begin{notation}
    Suppose $X$ is a nested tuple of integers with $\depth(X)>0$. For any $1 \leq j \leq \rank(X)$, we write 
    \begin{align*}
    \len_j(X) & = \len(\mode_j(X))\text{, } \\
    \len_{<j}(X) & = \sum_{i=1}^{j-1} \len_i(X)\text{,}\\
    \len_{\leq j}(X) & = \len_{<j}(X) + \len_j(X)
    \end{align*}
    and similarly, we write 
    \begin{align*}
    \size_j(X) & = \size(\mode_j(X))\text{, }\\
    \size_{< j} (X) & = \prod_{i=1}^{j-1} \size_j(X)\text{, and}\\
    \size_{\leq j} (X) & = \size_{<j}(X)\cdot\size_j(X).
    \end{align*}
\end{notation}

\begin{definition}
    \label{definitionofentryofnestedtuple}
    If $X = (x_1,\dots,x_m)_P$ is a nested tuple and $1 \leq i \leq m$, then the $i$th entry of $X$ is 
    \[
    \entry_i(X) = \entry_i(X^\flat) = x_i.
    \]
\end{definition}

% \begin{definition} \label{definitionofentryofnestedtuple}
% Suppose $X$ is a nested tuple with $\len(X) = \ell$. If $1 \leq i \leq \ell$, then the $i${\it th entry of} $X$ is defined recursively by
% \[
% \entry_i(X) = \begin{cases}
%     X &  \depth(X) = 0 \text{ (hence }i=\ell=1\text{)}\\
%  \entry_{i - N}(\mode_j(X)) & \begin{matrix} j \text{ is the largest integer such that }\\
%  N \coloneqq \len_{<j}(X) < i. \\
%  \end{matrix} 
% \end{cases}
% \]
% \end{definition}

\begin{example}
    If 
    \[X = ((3),4,((10,10),12)),\]
    then the entries of $X$ are 
    \begin{align*}
        \entry_1(X) & = 3\\
        \entry_2(X) & = 4\\
        \entry_3(X) & = 10\\
        \entry_4(X) & = 10\\
        \entry_5(X) & = 12.
    \end{align*}
\end{example}

\begin{example}
    If $X = (32,5,6,64)$, then the entries of $X$ are 
    \begin{align*}
        \entry_1(X) & = 32\\
        \entry_2(X) & = 5\\
        \entry_3(X) & = 6\\
        \entry_4(X) & = 4.
    \end{align*}
\end{example}

\begin{example}
    If $X$ is a nested tuple with depth $1$, then $\mode_i(X) = \entry_i(X)$ for all $1 \leq i \leq \rank(X) = \len(X)$.
\end{example}

\begin{observation}
    If $X$ is a nested tuple of integers, then the {\it entries} of $X$ are integers, while the {\it modes} of $X$ are themselves nested tuples of integers.
\end{observation}

Finally, we introduce the notion of {\it congruence} of nested tuples, which indicates when nested tuples have the same profile.

\begin{definition}
    If $X_1$ and $X_2$ are nested tuples, we say $X_1$ and $X_2$ are {\it congruent}, if 
    \[
    \profile(X_1) = \profile(X_2).
    \]
\end{definition}

% \begin{definition}
%     Suppose $X_1$ and $X_2$ are nested tuples. We say $X_1$ and $X_2$ are {\it congruent}, and write $X_1 \cong X_2$, if
%     \begin{enumerate}
%         \item $X_1$ and $X_2$ have depth $0$, or
%         \item $X_1$ and $X_2$ have depth $\geq 1$, and 
%         \[
%         \mode_i(X_1) \cong \mode_i(X_2)
%         \]
%         for each $1 \leq i \leq \rank(X_1) = \rank(X_2)$.
%     \end{enumerate}
% \end{definition}

\begin{example}
    Here are some examples of nested tuples $X_1$ and $X_2$, and whether or not they are congruent 
\[
\begin{aligned}
X_1 &= 27 \quad & X_2 &= 100 \quad & \text{congruent}\\
X_1 &= (2,2) \quad & X_2 &= (8,64) \quad & \text{congruent}\\
X_1 &= ((4,8),(4,8)) \quad & X_2 &= ((1,1),(5,10)) \quad & \text{congruent}\\
X_1 &= ((64,(8,8)),(25,(5,5))) \quad & X_2 &= ((2,(3,5)),(7,(11,13))) \quad & \text{congruent}\\
X_1 &= 27 \quad & X_2 &= (100) \quad & \text{not congruent}\\
X_1 &= (2,2) \quad & X_2 &= (8,64,128) \quad & \text{not congruent}\\
X_1 &= ((4,8),(4,8)) \quad & X_2 &= (((1,1),(5,10))) \quad & \text{not congruent}\\
\end{aligned}
\]
\end{example}

\subsection{Substitution}

Recall that if $Q$ is a profile of length $r$ and $P_1,\dots,P_r$ are profiles, then we defined a profile 
\[
(P_1,\dots,P_r)_Q
\]
called the $Q${\it-substitution} of $P_1,\dots,P_r$. This profile is obtained from $Q$ by replacing the $i$th entry of $Q$ with the profile $P_i$. We can extend this to an operation on nested tuples as follows. 

\begin{definition}
    Suppose $X_1,\dots,X_m$ are nested tuples with profiles $P_1,\dots,P_m$, and suppose $Q$ is a profile of length $m$. We define the $Q$-substitution
    \[
    (X_1,\dots,X_m)_Q
    \]
    of $X_1,\dots,X_m$ to be the nested tuple with flattening 
    \[
    (X_1,\dots,X_m)_Q^\flat = X_1^\flat \star \cdots \star X_m^\flat
    \]
    and profile 
    \[
    (P_1,\dots,P_m)_Q. 
    \]
    More generally, if $X_1,\dots,X_m$ are nested tuples and $Y$ is a nested tuple of length $m$, we define 
    \[
    (X_1,\dots,X_m)_Y = (X_1,\dots,X_m)_{\profile(Y)}.
    \]
\end{definition}

\begin{example}
    If $(X_1,X_2,X_3) = (64,16,4)$ and $Q = (*,(*,*))$, then 
    \[
    (X_1,X_2,X_3)_Q = (64,(32,4))
    \]
\end{example}

\begin{example}
    If $(X_1,X_2,X_3,X_4) = ( (2,2), (3,3), (5,5), (7,7))$ and $Q = ((*,*),(*,*))$, then 
    \[
    (X_1,X_2,X_3,X_4)_Q = (((2,2),(3,3)),((5,5),(7,7))).
    \]
\end{example}
\begin{example}
    If $X = (12)$ and $Q = *$, then 
    \[
    (X)_Q = 12.
    \]
\end{example}

\begin{example}\label{exampleofmultiplesubstitution1}
    If $X_1 = 2$, $X_2 = 2$, $X_3 = (5,5)$, and $Q = (*,*,*)$, then 
    \begin{align*}
        (X_1,X_2,X_3)_Q = (2,2,(5,5)) = (X_1,X_2,X_3).
    \end{align*}
    More generally, if $X_1,\dots,X_m$ are any nested tuples and $P = (* , \dots , *)$ then
    \[
    (X_1,\dots,X_m)_Q = (X_1,\dots,X_k)
    \]
    is the {\it concatenation} of $X_1,\dots,X_m$. 
\end{example}

\begin{aside}
    There is an {\it operadic} interpretation of substitutions of nested tuples. The set $\mathterm{Nest}(\mathbb{Z})$ of nested tuples of integers is an {\it algebra} over the operad  $\mathterm{Profile}$, with structure maps given by $Q$-substitution:
    \[ \begin{tikzcd} 
    \mathterm{Nest}(\mathbb{Z})\times \dots \times \mathterm{Nest}(\mathbb{Z})\times \mathterm{Profile}(n) \ar[rr] & & \mathterm{Nest}(\mathbb{Z})\\
    (X_1,\dots,X_m),Q \ar[rr,mapsto] & & (X_1,\dots,X_m)_Q.
    \end{tikzcd}\]
\end{aside}

\subsection{Refinement}\label{refinementsection}

In this section, we introduce an important relation on nested tuples called {\it refinement}. Intuitively, if $X'$ and $X$ are nested tuples of integers, we say $X'$ refines $X$ if $X'$ may be obtained from $X$ by replacing each entry of $X$ with some nested tuple of the same size. More precisely, we have the following definition.

\begin{definition}\label{definitionofrefinement}
If $X'$ and $X$ are nested tuples, then we say $X'$ {\it refines} $X$ if either 
    \begin{enumerate}
    \item $X = \size(X')$, or 
    \item 
    \begin{enumerate}
    \item $\depth(X'),\depth(X) > 0$,
    \item $\rank(X') = \rank(X)$, and 
    \item for each $1 \leq i \leq \rank(X)$, $\mode_i(X')$ refines $\mode_i(X)$.
    \end{enumerate}
    \end{enumerate} 
\end{definition}

\begin{notation}
We write 
\[X' \twoheadrightarrow X\]
to indicate that $X'$ refines $X$. 
\end{notation}

\begin{example}
    Here are some examples of refinements of nested tuples. 
    \[\begin{aligned}
        (2,(2,2))\twoheadrightarrow &  \hspace{0.05in} 8\\
        ((2,2),(3,3),(5,5)) \twoheadrightarrow & \hspace{0.05in} (4,9,25)\\
        (64) \twoheadrightarrow & \hspace{0.05in} 64\\
        (8,((2,2,2),((1,4),(2,2)))) \twoheadrightarrow & \hspace{0.05in} (8,(8,8))
    \end{aligned}\]
\end{example}

\begin{observation}
Refinement of nested tuples is reflexive, transitive, and antisymmetric, so refinement specifies a partial ordering on the collection of nested tuples of positive integers.
\end{observation}

% \begin{definition}[Reduction]\label{definitionofreduction}
%     Suppose $X$ is a nested tuple of rank $\rank(X) = r$, and suppose $d \geq 0$ is an integer. We define the {\it depth }$\leq d$ {\it reduction} $\red_{\leq d}(X)$ of $X$ as follows.
%     \begin{enumerate}
%     \item If $d \geq \depth(X)$, we define $\red_{\leq d}(X) = X$.
%     \item If $\depth(X) > d$ and $d = 0$, we define $\red_{\leq 0}(X) = \size(X)$.
%         \item If $\depth(X) > d$ and $d > 0 $, we set 
%         \[
%         \red_{\leq d}(X) = (\red_{\leq d-1}(\mode_1(X)),\dots, \red_{\leq d-1}(\mode_r(X))).
%         \]
%     \end{enumerate}
% \end{definition}

% \begin{example}
%     If $X = (((25)))$, then $\depth(X)= 3$, and 
%         \begin{align*}
%         \red_{\leq 0}(X) & = 25\\
%         \red_{\leq 1}(X) & = (25)\\
%         \red_{\leq 2}(X) & = ((25))\\
%         \red_{\leq 3}(X) & = (((25))) = X\\
%     \end{align*}
% \end{example}

% \begin{example}
%     If  
%     \[X = (((2,(2,2)),5),3,(10,2)),\]
%     then $\depth(X) = 4$, and 
%     \begin{align*}
%         \red_{\leq 0}(X) & = 2400\\
%         \red_{\leq 1}(X) & = (40,3,20)\\
%         \red_{\leq 2}(X) & = ((8,5),3,(10,2))\\
%         \red_{\leq 3}(X) & = (((2,4),5),3,(10,2))\\
%         \red_{\leq 4}(X) & = (((2,(2,2)),5),3,(10,2)) = X\\
%     \end{align*}
% \end{example}

% \begin{example}
%     If $\depth(X) > 0$, and $\rank(X) = r$, then the depth $\leq 1$ reduction of $X$ is the tuple 
%     \[
%     \red_{\leq 1}(X) = \Bigl(\size_1(X),\dots, \size_r(X)\Bigr).
%     \]
% \end{example}

If $X'$ refines $X$, then we can think of $X'$ as being obtained from $X$ by replacing each entry $x_i$ of $X$ with some nested tuple $X_i'$ of size $x_i$. We refer to the nested tuple $X_i'$ as the $i$th mode of $X'$ relative to $X$. More precisely, we have the following definition.

\begin{construction}
    Suppose $X$ is a nested tuple of integers of length $m$, and suppose $X'$ refines $X$. For any $1 \leq i \leq m$, we define a nested tuple
    \[
    X_i'=\mode_i(X',X),
    \]
    called the $i${\it th mode of} $X'$ {\it relative to} $X$, by the formula
    \[
\mode_i(X',X) = \begin{cases}
    X' &  \depth(X) = 0 \text{ (hence }i=\ell=1\text{)}\\
 \mode_{i - N}(\mode_j(X'),\mode_j(X)) & \begin{matrix} j \text{ is the largest integer such that }\\
 N \coloneqq \len_{<j}(X) < i. \\
 \end{matrix} 
\end{cases}
\]
\end{construction}

\begin{example}
    If $X = ((4,9),(25,36))$, and $X' = ( ( (2,2),(3,3) ) , ( 25, (6,(2,3)) ) )$, then $X'$ refines $X$ and the modes of $X'$ relative to $X$ are 
    \begin{align*}
    \mode_1(X',X) & = (2,2) \\
    \mode_2(X',X) & = (3,3) \\
    \mode_3(X',X) & = 25 \\
    \mode_4(X',X) & = (6,(2,3)) .
    \end{align*}
\end{example}

\begin{example}
    If $X$ is any nested tuple, then $X$ refines $X$, and for any $1 \leq i \leq \len(X)$ we have 
    \[
    \mode_i(X,X) = \entry_i(X).
    \]
\end{example}

\begin{example}
    If $X = X^\flat$ is a tuple, and $X'$ refines $X$, then for any $1 \leq i \leq \len(X)$, we have
    \[
    \mode_i(X',X) = \mode_i(X').
    \]
\end{example}

\begin{example}
    If $X'$ is a nested tuple with $\size(X') = N$, then $X'$ refines $N$, and the only mode of $X'$ relative to $N$ is 
    \[
    \mode_1(X',N) = X'. 
    \]
\end{example}

\begin{notation}
    If $X' \twoheadrightarrow X$ is a refinement and $1 \leq i \leq \len(X)$, then we write
    \begin{align*}
    \len_i(X',X) & = \len(\mode_i(X',X))\\
    \len_{<i}(X',X) & = \sum_{j < i} \len_j(X',X)\\
    \len_{\leq i}(X',X) & = \sum_{j \leq i} \len_j(X',X)
    \end{align*}
\end{notation}

% \begin{proposition}
%     Suppose $X'$ and $X$ are nested tuples. Then $X'$ refines $X$ if and only if 
%     \[
%     X' = \sub(X,(X_1',\dots,X_m')).
%     \]
%     for some uniquely determined nested tuples $X_1',\dots,X_m'$, namely 
%     \[
%     X_i' = \mode_i(X',X).
%     \]
% \end{proposition}

% \begin{proof}
%     INCOMPLETE.
% \end{proof}

\begin{definition}
    Suppose $X'$ refines $X$, and write $X_i' = \mode_i(X',X)$. Then the {\it flattening of }$X'$ {\it relative to }$X$ is the nested tuple 
    \[
    \flatten(X',X) = (X_1',\dots,X_m').
    \]
\end{definition}

\begin{example} 
If $X' = (((2,2),(3,3)),((5,5),(7,7)))$ and $X = ((4,9),(25,49))$, then 
\[
\flatten(X',X) = ((2,2),(3,3),(5,5),(7,7)).
\]
\end{example}

\begin{example}
    If $X$ is any nested tuple, then the flattening of $X$ relative to $X$ is
    \[
    \flatten(X,X) = X^\flat.
    \]
\end{example}

\begin{example}
    If $X = X^\flat$ is a tuple, and $X'$ refines $X$, then the flattening of $X'$ relative to $X$ is 
    \[
    \flatten(X',X) = X'. 
    \]
\end{example}

\begin{example}
    If $X'$ is a nested tuple with $\size(X') = N$, then $X'$ refines $N$, and the flattening of $X'$ relative to $N$ is 
    \[
    \flatten(X',N) = (N). 
    \]
\end{example}

\begin{observation}
    If $X'$ refines $X$, then $\flatten(X',X)$ refines $X^\flat$. 
\end{observation}

% Every refinement $X' \twoheadrightarrow X$ encodes a generalized colexicographic isomorphism. In order to make this precise, we introduce some notation. 

% \begin{definition}
%     If $X$ is a nested tuple, then we define $[0,X)$ to be the finite set whose elements are nested tuples $Y$ of non-negative integers such that
%     \begin{enumerate}
%         \item $Y$ is congruent to $X$, and 
%         \item if $1 \leq i \leq \text{length}(X)$, then $0 \leq \text{entry}_i(Y)  < \text{entry}_i(X)$.
%     \end{enumerate}
% \end{definition}

% \begin{example}
%     If $X = (8,(3,5))$, then 
%     \[
%     [0,X) =  [0,8) \times \left([0,3) \times [0,5) \right).
%     \]
%     Elements of $X$ include $(0,(0,0))$, $(1,(2,1))$, $(7,(2,4))$, etc.
% \end{example}

% \begin{observation}
%     We can give a recursive description of $[0,X)$. If we choose $\{()\}$ for our model of the empty product of sets, i.e.
%     \[\prod_{\varnothing} = \{()\},\]
%     then  
%     \[
%     [0,X) = \begin{cases}
% \{0,\dots,X-1\} & \text{depth}(X) = 0\\
% \prod_{i=1}^{\text{rank}(X)} [0,\text{mode}_i(X)) & \text{depth}(X) > 0. 
%     \end{cases}
%     \]
% \end{observation}

% \begin{construction}
%     Suppose $X' \twoheadrightarrow X$ is a refinement. We define the colexicographic isomorphism
%     \[
%     \colex_{X'}^{X}:[0,X') \to [0,X)
%     \]
%     as follows. 
%     \begin{itemize}
%         \item If $\depth(X) = 0$, we define 
%         \[
%         \colex_{X'}^X = \colex_X
%         \]
%     \end{itemize}
% \end{construction}

\newpage

\newpage

\section{Layouts}\label{layoutssection}

Having developed the necessary background on nested tuples, we turn our attention to {\it layouts}. These are a generalization of flat layouts in which shapes and strides are allowed to be nested tuples, rather than (flat) tuples.

\subsection{Basic definitions}

\begin{definition}\label{definitionofnestedlayout}
    A {\it layout} is a pair 
    \[
    L = S:D
    \]
    consisting of a nested tuple of positive integers 
    \[\shape(L) = S\] called the {\it shape} of $L$, and a nested tuple of non-negative integers 
    \[
    \stride(L) = D
    \]
    called the {\it stride} of $L$, such that $S$ and $D$ are congruent. 
\end{definition}

\begin{definition} If $L=S:D$ is a layout, then the {\it rank}, {\it length}, {\it depth}, {\it size}, and {\it profile} of $L$ are defined to be the rank, length, depth, size, and profile of $S$, respectively.
\end{definition}

\begin{example} The layout $L = (3,(3,2)):(3,(1,10))$ may be pictured as follows.
\[
\begin{tikzpicture}[x={(0cm,-1cm)},y={(1cm,0cm)},every node/.style={minimum size=1cm, outer sep=0pt}]

\node[fill=gray!20] at (0,0) {0};
\node[fill=gray!20] at (0,1) {1};
\node[fill=gray!20] at (0,2) {2};
\node[fill=gray!20] at (0,3) {10};
\node[fill=gray!20] at (0,4) {11};
\node[fill=gray!20] at (0,5) {12};
\node[fill=gray!20] at (1,0) {3};
\node[fill=gray!20] at (1,1) {4};
\node[fill=gray!20] at (1,2) {5};
\node[fill=gray!20] at (1,3) {13};
\node[fill=gray!20] at (1,4) {14};
\node[fill=gray!20] at (1,5) {15};
\node[fill=gray!20] at (2,0) {6};
\node[fill=gray!20] at (2,1) {7};
\node[fill=gray!20] at (2,2) {8};
\node[fill=gray!20] at (2,3) {16};
\node[fill=gray!20] at (2,4) {17};
\node[fill=gray!20] at (2,5) {18};
\draw[color=black,thick,shift={(-0.5,-0.5)}] (0,0) grid (3,6);

\node[anchor = east] at (1,-1) {$L = $};
\end{tikzpicture}
\]
\end{example}

\begin{example} The layout $L = ((2,2),(2,2)):((1,4),(2,8))$ may be pictured as follows.
\[
\begin{tikzpicture}[x={(0cm,-1cm)},y={(1cm,0cm)},every node/.style={minimum size=1cm, outer sep=0pt}]

\node[fill=gray!20] at (0,0) {0};
\node[fill=gray!20] at (0,1) {2};
\node[fill=gray!20] at (0,2) {8};
\node[fill=gray!20] at (0,3) {10};
\node[fill=gray!20] at (1,0) {1};
\node[fill=gray!20] at (1,1) {3};
\node[fill=gray!20] at (1,2) {9};
\node[fill=gray!20] at (1,3) {11};
\node[fill=gray!20] at (2,0) {4};
\node[fill=gray!20] at (2,1) {6};
\node[fill=gray!20] at (2,2) {12};
\node[fill=gray!20] at (2,3) {14};
\node[fill=gray!20] at (3,0) {5};
\node[fill=gray!20] at (3,1) {7};
\node[fill=gray!20] at (3,2) {13};
\node[fill=gray!20] at (3,3) {15};
\draw[color=black,thick,shift={(-0.5,-0.5)}] (0,0) grid (4,4);

\node[anchor = east] at (1.5,-1) {$L = $};
\end{tikzpicture}
\]
\end{example}

\begin{example}
    The layout 
    \[
    L = 10:4
    \]
    has $\rank(L) = 1$, $\len(L) = 1$, $\depth(L) = 0$, $\size(L) = 10$, and $\profile(L) = *$. 
\end{example}

\begin{example}
    The layout
    \[
    L = (7,(2,10,4),(3,7)) : (1,(7,14,140),(560,1680))
    \]
    has $\rank(L) = 3$, $\len(L) = 6$, $\depth(L) = 2$, $\size(L) = 11760$, and $\profile(L) = (*,(*,*,*),(*,*)).$ 
\end{example}

\begin{example}
    The layout
    \[
    L = ((2,2,2,(2,2))) : ((1,0,8,(0,16)))
    \]
    has $\rank(L)= 1$, $\len(L) = 5$, $\depth(L) = 3$, $\size(L) = 32$, and $\profile(L) = ((*,*,*,(*,*)))$.
\end{example}

\begin{example}
    The pair 
    \[
    S:D = (2,(2,2)):(1,2,4)
    \]
    is NOT a layout because $S$ and $D$ are not congruent.
\end{example}

\begin{definition}
    If $L = S:D$ is a layout, then for any $1 \leq i \leq \rank(L)$ we define the $i$th mode of $L$ to be the layout 
    \[
    \mode_i(L) = \mode_i(S) : \mode_i(D),
    \]
    and for any $1 \leq i \leq \len(L)$, we define the $i$th entry of $L$ to be the layout 
    \[
    \entry_i(L) = \entry_i(S) : \entry_i(D).
    \]
\end{definition}

\begin{example}
    If $L = \bigl((2,2),9\bigr):\bigl((3,6),12\bigr)$, then the modes of $L$ are 
    \begin{align*}
        \mode_1(L) & = (2,2):(3,6) \\
        \mode_2(L) & = 9 : 12 
    \end{align*}
    and the entries of $L$ are 
    \begin{align*}
    \entry_1(L) & = 2:3\\
      \entry_2(L) &=  2:6\\
      \entry_3(L)&=   9:12 .
    \end{align*} 
\end{example}

\begin{remark}
    If $L$ is a layout, then the modes of $L$ are also layouts, and the entries of $L$ are layouts of depth $0$. 
\end{remark}

\begin{remark}
A flat layout $L$ is precisely a layout of depth $1$. On the other hand, if $L$ is a layout, we may obtain a flat layout $L^\flat$ as follows.
\end{remark}

\begin{definition}
    If $L = S:D$ is a layout, we define the {\it flattening} of $L$ to be the flat layout 
    \[
    L^\flat = S^\flat : D^\flat.
    \]
\end{definition}

\begin{example}
    The flattening of $L = 10:4$ is $L^\flat = (10):(4)$. 
\end{example}

\begin{example}
The flattening of 
    \[
    L = \bigl((2,2,2,(2,2))\bigr) : \bigl((1,0,8,(0,16))\bigr)
    \]
    is 
    \[
    L^\flat = (2,2,2,2,2) : (1,0,8,0,16).
    \]
\end{example}

\begin{remark}
    If $L$ is a layout then $\len(L) = \rank(L^\flat)$, and for any $1 \leq i \leq \len(L)$, we have 
    \[\entry_i(L) = \mode_i(L^\flat).\]
\end{remark}

\noindent We can use the flattening construction above to extend many concepts from flat layouts to nested layouts. For example: 

\begin{construction}[Layout function]
    If $L$ is a nested layout, we define the layout function $\Phi_L$ of $L$ by
    \[
    \Phi_L = \Phi_{L^\flat},
    \]
    where $\Phi_{L^\flat}$ is the layout function of Construction \ref{constructionoflayoutfunctions}. Similarly, if $N$ is such that $\Image(\Phi_L) \subset [0,N)$, we define 
    \[
    \Phi_L^N = \Phi_{L^\flat}^N
    \]
    to be the factorization of $\Phi_L$ through the inclusion $[0,N) \subset \mathbb{Z}$. 
\end{construction}

\begin{example}
    If $L = ((2,2),2):((3,0),10)$, then the layout function 
    \[
    \Phi_L:[0,8) \to \mathbb{Z}
    \]
    of $L$ is given by 
    \[ \begin{tikzcd} [row sep = 2, column sep = 4]
    & 0 \ar[dd,mapsto] & 1 \ar[dd,mapsto] & 2 \ar[dd,mapsto] & 3 \ar[dd,mapsto] & 4 \ar[dd,mapsto] & 5 \ar[dd,mapsto] & 6  \ar[dd,mapsto] &  7 \ar[dd,mapsto]\\
    \Phi_L & & &  & & & & & & \\
    & 0 & 3 & 0 & 3 & 10 & 13 & 10 & 13
    \end{tikzcd} \]
\end{example}

Given a layout $L$, we can obtain a flat layout $L^\flat$, and a profile $P = \profile(L)$. Conversely, if we are given a flat layout $L$ and a profile $P$ with the same length as $L$, then we can construct a layout with flattening $L$ and profile $P$ as follows.

\begin{construction}\label{layoutfromflatlayoutandprofile}
If $L$ is a flat layout, and $P$ is a profile with $\len(P) = \len(L)$, then we can define 
\[
L = L_P
\]
to be the layout with shape 
\[
\shape(L) = \shape(L)_P
\]
and stride 
\[
\stride(L) = \stride(L)_P
\]
where $(-)_P$ is the $P$-substitution operation of Definition \ref{definitionofPconcatenation}.
\end{construction}

\begin{example}
    If $L = (8,8,8):(1,64,8)$ and $P = (*,(*,*))$, then 
    \[
    L_P = (8,(8,8)),(1,(64,8)).
    \]
\end{example}

\begin{example}
    If $L = (128):(2)$ and $P = *$, then 
    \[
    L_P = 128:2.
    \]
\end{example}

\begin{proposition}
    If $L'$ is a flat layout and $P$ is a profile with $\len(L') = \len(P)$, then there exists a unique layout $L$ whose flattening is $L^\flat = L'$ and whose profile is $\profile(L) = P$, namely $L = L'_P$.
\end{proposition}

\begin{proof}
    This follows from the definition of nested tuples, since a nested tuple is uniquely determined by its flattening and its profile.
\end{proof}

\begin{observation}
    The previous proposition tells us that we have a pullback square
    \[ \begin{tikzcd} 
    \mathterm{Layout}\ar[rr,"\profile(-)"] \ar[d,swap,"(-)^\flat"] \arrow[drr, phantom, "\lrcorner", very near start] & & \mathterm{Profile}\ar[d,"\len(-)"] \\
    \mathterm{FlatLayout} \ar[rr,swap,"\len(-)"] & & \mathbb{N}
    \end{tikzcd} \]
\end{observation}

We can extend the notion of non-degeneracy to the nested case as follows.

\begin{definition}\label{definitionofnondegeneratelayout}
    Suppose $L$ is a layout. We say $L$ is {\it non-degenerate} if for all $1 \leq i\leq \len(L)$, the following condition holds:
    \[
    \entry_i(\shape(L)) \quad \Rightarrow \quad \entry_i(\stride(L))
    \]
\end{definition}

\begin{example}
    The layouts 
    \begin{align*}
        L_1 & = ((2,2),1):((1,2),0)\\
        L_2 & = ((8,8),(1,16)):((2,32),(0,128))
    \end{align*}
    are non-degenerate, while the layouts 
    \begin{align*}
        L_3 & = ((2,2),1):((1,2),4)\\
        L_4 & = ((8,8),(1,16)):((2,32),(1024,128))
    \end{align*}
    are degenerate.
\end{example}

% \subsection{Linear layouts}

% Before moving on to our discussion of layout operations, we will take a closer look at an especially important family of layouts called {\it linear layouts}.

% \begin{definition}
%     Suppose $s_1,\dots,s_m$ are positive integers, where each $s_i = 2^{r_i}$ is some positive power of $2$. Let $r = r_1 + \cdots + r_m$, so the size of $(s_1,\dots,s_m)$ is $2^r$. If $M$ is a $\ell \times r$ matrix with entries in $\mathbb{F}_2$, we may construct a layout $L_M$ whose coordinate function is the composite 
%     \[ \begin{tikzcd}
%     {[}0,s_1{)} \times \cdots \times {[}0,s_m{)} \ar[dd,swap,"\colex^{-1} \times \cdots \times \colex^{-1}"] & & & & \\ 
%      & & & & \\
%     {[}0,2{)}^{\times r_1}  \times \cdots \times  {[}0,2{)}^{\times r_m} \cong \mathbb{F}_2^{\times r} \ar[rr,"M"] & & \mathbb{F}_2^{\times \ell} \cong {[}0,2{)}^{\times \ell} \ar[rr,"\colex"] & & {[}0,2^\ell{)}.
%     \end{tikzcd} \]
%     Explicitly, $L_M$ is the layout whose shape is 
%     \[
%     \shape(L_M) = (S_1,\dots,S_m)
%     \]
%     where 
%     \[
%     S_i = (2,\dots,2)
%     \]
% \end{definition}

\subsection{Basic operations}

Having established the basic vocabulary for layouts, we turn to the operations they support. In this section, we define basic operations that will be needed to construct more sophisticated operations such as {\it coalesce}, {\it complement},  {\it composition}, {\it logical division}, and {\it logical product}.

\subsubsection{Flattening} If $L$ is a layout, then we may obtain a flat layout $L^\flat$ by flattening the shape and stride of $L$.
\begin{definition}
    If $L = S:D$ is a layout, we define the {\it flattening} of $L$ to be the flat layout 
    \[
    L^\flat = S^\flat : D^\flat.
    \]
\end{definition}

\begin{example}
The flattening of 
    \[
    L = ((2,2,2,(2,2))) : ((1,0,8,(0,16)))
    \]
    is 
    \[
    L^\flat = (2,2,2,2,2) : (1,0,8,0,16).
    \]
\end{example}

\begin{example}
    The flattening of $L = 10:4$ is $L^\flat = (10):(4)$. 
\end{example}

\begin{example}
    Suppose $L$ is a layout. Then $\depth(L) = 1$ if and only if $L = L^\flat$.
\end{example}

\subsubsection{Concatenate}

We can concatenate layouts by concatenating their shapes and concatenating their strides.

\begin{definition}\label{definitionoflayoutconcatenation}
    If $L=S:D$ and $L'=S':D'$ are layouts, then the {\it concatenation} of $L$ and $L'$ is the layout $(L,L')$
    \[
    (L,L') = (S,S'):(D,D').
    \]
    More generally, if $L_1,\dots,L_k$ is any finite collection of layouts, with $L_i = S_i:D_i$, then the concatenation of $L_1,\dots,L_k$ is the layout 
    \[
    (L_1,\dots,L_k) = (S_1,\dots,S_k):(D_1,\dots,D_k).
    \]
\end{definition}

\begin{remark}
    Concatenation of nested tuples (and hence of layouts) is not associative. For example, take $L_1 = 3:4$, $L_2 = 2:2$, and $L_3 = 5:1$. Then 
    \[\bigl(L_1, (L_2, L_3)\bigr) = \bigl(3, (2, 5)\bigr):\bigl(4, (2, 1)\bigr) \neq \bigl((3, 2), 5\bigr):\bigl((4, 2), 1\bigr) = \bigl((L_1, L_2), L_3\bigr). \]
    Moreover, neither of these layouts is equal to the ``three-fold" concatenation $(L_1,L_2,L_3) = (3,2,5):(4,2,1)$. However, we see that each of these layouts has the same flattening, so each of these layouts has the same layout function.
\end{remark}

\begin{example} \label{concatenationoflayoutsexample}
If $L = (3,7,2):(1,3,6)$ and $L' = (2,(2,(4,3))):(5,3,(2,2))$, then 
\[
(L,L') = ( (3,7,2), (2,(2,(4,3)))) : ( (1,3,6), (5,(3,(2,2)))) 
\]
\end{example} 
\begin{remark}
    Concatenation increases the depth of layouts. More precisely, we have 
    \[
    \depth(L,L') = 1 + \text{max}(\depth(L),\depth(L')).
    \]
\end{remark}

\begin{remark}
    When $L$ and $L'$ are flat layouts, the concatenation of Definition \ref{definitionoflayoutconcatenation} does {\it NOT} agree with the concatenation of flat layouts of Definition \ref{definitionofflatconcatenate}. Instead, these operations are related by the formula
    \[
    L \star L' = (L,L')^\flat.
    \]
\end{remark}

\begin{remark}
    If $L$ is any layout with $\depth(L)>0$ and $\rank(L) = r$, then we may write 
    \[
    L = (\mode_1(L),\dots,\mode_{r}(L))
    \]
    as the concatenation of its modes. 
\end{remark}

\begin{example}
    If 
    \[
    L = ((5,(7,7)),2,(4,5)):((1,(35,5)),0,(1,8))
    \]
    then $L = (L_1,L_2,L_3)$ where 
    \begin{align*}
        L_1 & = \bigl(5,(7,7)\bigr) : \bigl(1,(35,5)\bigr),\\
        L_2 & = 2:0\text{, and }\\
        L_3 & = (4,5):(1,8).
    \end{align*}
\end{example}

\subsubsection{Substitution}
Recall that if $X_1,\dots,X_k$ are nested tuples and $P$ is a profile with $\len(P) = k$, then we may form the $P$-substitution
\[
(X_1,\dots,X_k)_P
\]
which is obtained by replacing the $i$the entry of $P$ with the nested tuple $X_i$. We can extend this construction from nested tuples to layouts as follows.

\begin{definition}
    Suppose $L = S:D$ is a layout, and suppose $P$ is a profile with $\len(P) = \rank(L)$. We define 
    \[
    L_P = S_P:D_P
    \]
    where $S_P$ and $D_P$ are the $P$-substitutions of (the modes of) $S$ and $D$. 
\end{definition}

\begin{example}
    If $P = (*,(*,*))$ and $L = (8,8,8):(1,8,64)$, then 
    \[
    L_P = (8,(8,8)):(1,(8,64)).
    \]
\end{example}
\begin{example}
    If $P = (*,(*,*)))$ and 
    \[L = ((2,2),(3,3),(5,5)):((2,1),(12,4),(180,36)),\]
    then 
    \[
    L_P = ((2,2),((3,3),(5,5))):((2,1),((12,4),(180,36))).
    \]
\end{example}
\begin{example}
    If $L = (16):(1)$ and $P = *$, then 
    \[
    L_P = 16:1.
    \]
\end{example}

\subsection{Coalesce}

Recall that if $L$ is a flat layout, then $\coalesce^\flat(L)$ is a the unique flat layout of minimal rank whose layout function is $\Phi_L$. We can make a similar construction in the setting of arbitrary (nested) layouts. We begin by defining the notion of a coalesced layout.

\begin{definition}\label{definitionofnestedlayoutcoalesce}
    Suppose $L$ is a layout. We say $L$ is {\it coalesced} if one of the following conditions holds. 
    \begin{enumerate} 
    \item $L = 1:0$,
    \item $\depth(L) = 0$ and $\shape(L) > 1$, or 
    \item $\depth(L) = 1$, $\rank(L) > 1$, and $L$ is coalesced in the sense of Definition \ref{definitionofflatlayoutcoalesce}.
    \end{enumerate}
\end{definition}

\begin{example}
    The layout 
    \[L = (2,(2,2)) : (1,(16,512))\]
    is not coalesced since $\depth(L)>1$.
\end{example}
\begin{example}
    The layout \[L = (64):(2)\] is not coalesced, while the layout \[L' = 64:2\] is coalesced. 
\end{example}
\begin{example}
    The layout \[L = 1:8\] is not coalesced, while the layout \[L' = 1:0\] is coalesced. 
\end{example}
\begin{example}
    The empty layout 
    \[E = ():()\]
    is not coalesced.
\end{example}

\begin{observation}
    Recall that a layout $L$ is non-degenerate if 
    \[
    \entry_i(\shape(L)) = 1 \quad \Rightarrow \quad  \entry_i(\stride(L)) = 0.
    \]
    If $L$ is coalesced, then $L$ is non-degenerate.
\end{observation}
If $L$ is any layout, we can obtain a coalesced layout $\coalesce(L)$ as follows.

\begin{construction}\label{constructionoflayoutcoalesce}
    Suppose $L$ is a layout, and write 
    \[\coalesce^\flat(L^\flat) = (s_1,\dots,s_m):(d_1,\dots,d_m).\]
    \begin{enumerate}
        \item If $m > 1$, we define 
        \[
        \coalesce(L) = \coalesce^\flat(L^\flat)
        \]
        \item If $m = 1$, we define 
        \[
        \coalesce(L) = s_1:d_1
        \]
        \item If $m = 0$, we define 
        \[ \coalesce(L) = 1:0.\]
    \end{enumerate}
\end{construction}

\begin{example}
    If $E = ():()$ is the empty layout, then 
    \[
    \coalesce(E) = 1:0.
    \]
\end{example}
\begin{example}
    If $L = (1,1):(2,4)$, then 
    \[
    \coalesce(L) = 1:0.
    \]
\end{example}
\begin{example}
    If $L = (512):(4)$, then 
    \[
    \coalesce(L) = 512:4.
    \]
\end{example}
\begin{example}
    If $L = (2,2,2):(1,2,4)$, then 
    \[
    \coalesce(L) = 8:1.
    \]
\end{example}
\begin{example}
    If $L = ((2,2,2),(5,5)):((1,2,4),(10,50))$, then 
    \[
    \coalesce(L) = (8,25):(1,10).
    \]
\end{example}

\begin{remark} If $L$ is a layout, then $\coalesce(L)$ has depth $0$ or $1$.
\end{remark}

\begin{proposition}\label{nestedcoalesceproposition}
    If $A$ and $B$ are layouts, then 
    \[\Phi_A =\Phi_B \quad \Leftrightarrow \quad \coalesce(A) = \coalesce(B).\]
\end{proposition}
\begin{proof}
    Using Proposition \ref{coalesceproposition}, we have 
    \begin{align*}
    \Phi_A = \Phi_B \hspace{0.2in} & \Leftrightarrow \hspace{0.2in} \Phi_{A^\flat} = \Phi_{B^\flat} \\
    & \Leftrightarrow \hspace{0.2in} \coalesce^\flat(A^\flat) = \coalesce^\flat ( B^\flat ) \\
    & \Leftrightarrow \hspace{0.2in} \coalesce(A) = \coalesce ( B ).
    \end{align*}
\end{proof}

\begin{definition}
    If $L$ is a layout, define the {\it complexity} of $L$ to be the integer 
    \[
    \complexity(L) = \len(L) + \depth(L).
    \]
\end{definition}

\begin{proposition}
    If $L$ is a layout and $\size(L) >1$, then $\coalesce(L)$ is the unique complexity minimizing layout whose layout function is $\Phi_L$.
\end{proposition}

\begin{proof}
Suppose $L'$ is a layout with the same layout function as $L$, and suppose $\coalesce(L') \neq 1:0$. Then
\[
\len(L') \geq \len(\coalesce(L')) = \len(\coalesce(L)). 
\]
There are two cases to consider.
\begin{itemize}
    \item (Case 1): Suppose $\len(L') > 1$. Then $\depth(L') \geq 1 \geq \depth(\coalesce(L))$. Combining these inequalities, we observe that 
    \[
    \complexity(L') \geq \complexity(\coalesce(L)),
    \]
    where equality holds if and only if $L' = \coalesce(L') = \coalesce(L)$.
    \item (Case 2): Suppose $\len(L') = 1$. Then $L' = (s):(d)$ or $L' = s:d$ for some integers $s>1$ and $d \geq 0$. In either case, we have $\coal(L') = s:d$, and 
    \[
    \complexity(L') \geq \complexity(\coalesce(L)),
    \]
    where equality holds if and only if $L' = s:d = \coalesce(L)$. 
\end{itemize}
\end{proof}

\begin{remark}
    The only reason that we need to exclude the case $\size(L) = 1$ is that if $\size(L) = 1$, then  $1:0$ and the empty layout $():()$ are distinct layouts with minimal complexity, and the same layout function as $L$ (namely the trivial layout function $0 \mapsto 0$).
\end{remark}

\subsection{Relative coalesce}
There is an important invariant of coalesce called {\it relative coalesce}, denoted $\coalesce(L,\bar{S})$. This operation receives as an additional input a nested tuple $\bar{S}$ which is refined by $\shape(L)$. In this case, the relative coalesce operation simplifies the layout $L$ has much as possible, while ensuring that the resulting shape still refines $\bar{S}$.

\begin{definition}
    Suppose $L = S:D$ is a layout, and suppose $\bar{S}$ is some nested tuple of length $m$ which is refined by $S$. Recall that for any $1 \leq i \leq m$, we may consider the $i$th mode of $S$ relative to $\bar{S}$, denoted
    \[
    \mode_i(S,\bar{S}).
    \]
    Since $S$ and $D$ are congruent, there is a nested tuple
    \[
    \mode_i(D,\bar{S})
    \]
    corresponding to $\mode_i(S,\bar{S})$, and we define the $i${\it th mode of} $L$ {\it relative to} $\bar{S}$ to be the layout
    \[
    \mode_i(L,\bar{S}) = \mode_i(S,\bar{S}):\mode_i(D,\bar{S}).
    \]
\end{definition}

\begin{example}
    If $\bar{S} = (4,(9,25))$ and 
    \[
    L = ((2,2),((3,3),(5,(1,5)))):((1,2),((6,18),(90,(0,450))))
    \]
    then 
    \begin{align*}
    \mode_1(L,\bar{S}) & = (2,2):(1,2)\\
    \mode_2(L,\bar{S}) & = (3,3):(6,18)\\
    \mode_3(L,\bar{S}) & = (5,(1,5)):(90,(0,450)).
    \end{align*}
\end{example}

\begin{observation}
    Suppose $L = S:D$ is a layout, and suppose $\bar{S}$ is a nested tuple of length $m$ and profile $P$ which is refined by $S$. If for any $1 \leq i \leq m$, we write 
    \[
    L_i = \mode_i(L,\bar{S}),
    \]
    then 
    \[
    L =(L_1,\dots,L_m)_P
    \]
    is the $P$-substitution of its relative modes.
\end{observation}

\begin{definition}\label{definitionofcoalescedoverSbar}
    Suppose $L = S:D$ is a layout, and suppose $\bar{S}$ is a nested tuple of length $m$ and profile $P$ which is refined by $S$. We say $L$ is {\it coalesced over }$\bar{S}$ if each relative mode 
    \[
    \mode_i(L,\bar{S})
    \]
    is coalesced.
\end{definition}

\begin{observation}\label{coalescedoverSbarimpliesnondegenerate}
    In the setting of Definition \ref{definitionofcoalescedoverSbar}, if $L$ is coalesced over $\bar{S}$, then $L$ is non-degenerate.
\end{observation}

\begin{example}
    If $L$ is a layout, then $L$ is coalesced over $\shape(L)$ if and only if $L$ is non-degenerate, i.e.
    \[
    \entry_i(\shape(L)) =1 \quad \Rightarrow \quad \entry_i(\stride(L)) = 0.
    \]
\end{example}

\begin{definition}[Relative coalesce]\label{definitionofrelativecoalesce}
    Suppose $L = S:D$ is a layout, and suppose $ \bar{S}$ is a nested tuple of length $m$ and profile $P$ which is refined by $S$. We define 
    \[
    \coalesce(L,\bar{S}) = (\coalesce(L_1),\dots,\coalesce(L_m))_P.
    \]
\end{definition}

\begin{remark}
     In the setting of Definition \ref{definitionofrelativecoalesce}, the shape of $\coalesce(L,\bar{S})$ refines $\bar{S}$. 
\end{remark}

\begin{lemma}\label{relativecoalescepreserveslayoutfunction}
    If $L = S:D$ is a layout and $S$ refines $\bar{S}$, then 
    \[
    \Phi_{\coalesce(L,\bar{S})} = \Phi_L.
    \]
\end{lemma}

\begin{proof}
    As above, let 
    \[
    L_i = \mode_i(L,S)
    \]
    denote the $i$th mode of $L$ relative to $S$, and set $\bar{L}_i = \coalesce(L_i)$. Then 
    \begin{align*}
    \Phi_{\coalesce(L,\bar{S})} & = \Phi_{(\bar{L}_1,\dots,\bar{L}_m)_{\bar{S}}}\\
    & = \Phi_{(\bar{L}_1,\dots,\bar{L}_m)}\\
    & = \Phi_{\coalesce((\bar{L}_1,\dots,\bar{L}_m))}\\
    & = \Phi_{\coalesce((L_1,\dots,L_m))}\\
    & = \Phi_{(L_1,\dots,L_m)}\\
    & = \Phi_{(L_1,\dots,L_m)_{\bar{S}}}\\
    & = \Phi_L.
    \end{align*}
\end{proof}

\begin{proposition}\label{relativecoalescedlayoutscharacterizedbylayoutfunction}
    Suppose $A$ and $B$ are layouts, and suppose $\bar{S}$ is a nested tuple of length $m$ such that $\shape(A)$ refines $\bar{S}$ and $\shape(B)$ refines $\bar{S}$. Then
    \[\Phi_A = \Phi_B \quad \Leftrightarrow \quad 
    \coalesce(A,\bar{S}) = \coalesce(B,\bar{S})
    \]
\end{proposition}

\begin{proof}
    If $\coalesce(A,\bar{S}) = \coalesce(B,\bar{S})$, the using 
    Lemma \ref{relativecoalescepreserveslayoutfunction}, we have 
    \[
    \Phi_A = \Phi_{\coalesce(A,\bar{S})} = \Phi_{\coalesce(B,\bar{S})} = \Phi_B.
    \]
    Conversely, suppose that $\Phi_A = \Phi_B$. We will argue that $\coalesce(A,\bar{S}) = \coalesce(B,\bar{S})$. Set $P = \profile(\bar{S})$, and for any $1 \leq i \leq m$, set
    \begin{align*}
    A_i & = \mode_i(A,\bar{S})\\
    B_i & = \mode_i(B,\bar{S}).
    \end{align*}
    Since 
    \[
    \coalesce(A,\bar{S}) =(\coalesce(A_1),\dots,\coalesce(A_m))_P
    \]
    and
    \[
    \coalesce(B,\bar{S}) = (\coalesce(B_1),\dots,\coalesce(B_m))_P
    \]
    it suffices to prove that $\coalesce(A_i) = \coalesce(B_i)$ for all $1 \leq i \leq m$. By the associativity of colexicographic isomorphisms, we can write the layout function $\Phi_A$ of $A$ as 
    \[
    \begin{tikzcd} 
    {[}0,\size(A){)} \ar[rr,"\colex^{-1}"] & & \prod_{j=1}^m {[}0,\size(A_j) {)} \ar[rr,"\prod \Phi_{A_j}"] & & \prod_{j=1}^m \mathbb{Z} \ar[rr,"+"] & & \mathbb{Z} 
    \end{tikzcd} 
    \]
    and we can write the layout function $\Phi_B$ of $B$ as 
    \[
    \begin{tikzcd} 
    {[}0,\size(B){)} \ar[rr,"\colex^{-1}"] & & \prod_{j=1}^m {[}0,\size(B_j) {)} \ar[rr,"\prod \Phi_{B_j}"] & & \prod_{j=1}^m \mathbb{Z} \ar[rr,"+"] & & \mathbb{Z} 
    \end{tikzcd} 
    \]
    For a fixed $1 \leq i \leq m$, consider the subset
    \[
    [0,\size(A_i)) \subset \prod_{j=1}^m [0,\size(A_j))
    \]
    and its image 
    \[
    \colex([0,\size(A_i))) \subset [0,\size(A)). 
    \]
    Since $\size(A_j) = \size(B_j)$ for all $1 \leq j \leq m$, this is the same as the image 
    \[
    \colex([0,\size(B_j))) \subset [0,\size(B)) = [0,\size(B)).
    \]
    The restriction of $\Phi_A$ to this subset is $\Phi_{A_i}$, and the restriction of $B$ to this subset is $\Phi_{B_i}$, so it follows that $\Phi_{A_i} = \Phi_{B_i}$, and by Proposition \ref{nestedcoalesceproposition}, we have $\coalesce(A_i) = \coalesce(B_i)$. We deduce that 
    \[
    \coalesce(A,\bar{S}) = \coalesce(B,\bar{S}),
    \]
    as desired. 
\end{proof}

% \color{purple}
% \begin{proposition}
%     Include characterization of relative coalesce!
% \end{proposition}
% \color{black}

\subsection{Compact layouts}
We can easily extend the concept of compact layouts to the nested case. Again, in terms of the standard grid diagrams depicting layouts, a layout $L$ is compact if each integer $0 \leq i < \size(L)$ appears exactly once. More preciesly, we have the following definition.

\begin{definition}\label{definitionofcompactnestedlayout}
Suppose $L$ is a layout. We say $L$ is {\it compact} if the layout function 
\[
\Phi_L^{\cosize(L)} : [0,\size(L)) \to [0,\cosize(L))
\]
is an isomorphism. 
\end{definition}

\begin{example} The layout
\[
\begin{tikzpicture}[x={(0cm,-1cm)},y={(1cm,0cm)},every node/.style={minimum size=1cm, outer sep=0pt}]

\node[fill=gray!20] at (0,0) {0};
\node[fill=gray!20] at (0,1) {2};
\node[fill=gray!20] at (0,2) {8};
\node[fill=gray!20] at (0,3) {10};
\node[fill=gray!20] at (1,0) {1};
\node[fill=gray!20] at (1,1) {3};
\node[fill=gray!20] at (1,2) {9};
\node[fill=gray!20] at (1,3) {11};
\node[fill=gray!20] at (2,0) {4};
\node[fill=gray!20] at (2,1) {6};
\node[fill=gray!20] at (2,2) {12};
\node[fill=gray!20] at (2,3) {14};
\node[fill=gray!20] at (3,0) {5};
\node[fill=gray!20] at (3,1) {7};
\node[fill=gray!20] at (3,2) {13};
\node[fill=gray!20] at (3,3) {15};
\draw[color=black,thick,shift={(-0.5,-0.5)}] (0,0) grid (4,4);

\node[anchor = east] at (1.5,-1) {$A = ((2,2),(2,2)):((1,4),(2,8))=$};
\end{tikzpicture}
\]
is compact, while the layouts 
\[
\begin{tikzpicture}[x={(0cm,-1cm)},y={(1cm,0cm)},every node/.style={minimum size=1cm, outer sep=0pt}]

\node[fill=gray!20] at (0,0) {0};
\node[fill=gray!20] at (0,1) {2};
\node[fill=gray!20] at (0,2) {32};
\node[fill=gray!20] at (0,3) {34};
\node[fill=gray!20] at (1,0) {1};
\node[fill=gray!20] at (1,1) {3};
\node[fill=gray!20] at (1,2) {33};
\node[fill=gray!20] at (1,3) {35};
\node[fill=gray!20] at (2,0) {4};
\node[fill=gray!20] at (2,1) {6};
\node[fill=gray!20] at (2,2) {36};
\node[fill=gray!20] at (2,3) {38};
\node[fill=gray!20] at (3,0) {5};
\node[fill=gray!20] at (3,1) {7};
\node[fill=gray!20] at (3,2) {37};
\node[fill=gray!20] at (3,3) {39};
\draw[color=black,thick,shift={(-0.5,-0.5)}] (0,0) grid (4,4);

\node[anchor = east] at (1.5,-1) {$B = ((2,2),(2,2)):((1,4),(2,32))=$};
\end{tikzpicture}
\]
and
\[
\begin{tikzpicture}[x={(0cm,-1cm)},y={(1cm,0cm)},every node/.style={minimum size=1cm, outer sep=0pt}]

\node[fill=gray!20] at (0,0) {0};
\node[fill=gray!20] at (0,1) {2};
\node[fill=gray!20] at (0,2) {0};
\node[fill=gray!20] at (0,3) {2};
\node[fill=gray!20] at (1,0) {1};
\node[fill=gray!20] at (1,1) {3};
\node[fill=gray!20] at (1,2) {1};
\node[fill=gray!20] at (1,3) {3};
\node[fill=gray!20] at (2,0) {4};
\node[fill=gray!20] at (2,1) {6};
\node[fill=gray!20] at (2,2) {4};
\node[fill=gray!20] at (2,3) {6};
\node[fill=gray!20] at (3,0) {5};
\node[fill=gray!20] at (3,1) {7};
\node[fill=gray!20] at (3,2) {5};
\node[fill=gray!20] at (3,3) {7};
\draw[color=black,thick,shift={(-0.5,-0.5)}] (0,0) grid (4,4);

\node[anchor = east] at (1.5,-1) {$C = ((2,2),(2,2)):((1,4),(2,0))=$};
\end{tikzpicture}
\]
are not compact. 
\end{example}

\begin{example}
    The following layouts are compact:
    \begin{align*}
        L_1 & = (2,(2,2)):(8,(1,4))\\
        L_2 & = ((8,1),(8,32)):((2,0),(16,128))\\
        L_3 & = 64:1
    \end{align*}
\end{example}

\begin{example}
    The layout 
    \[
    L = (2,(2,2)):(4,(8,16))
    \]
    is not compact since the integer $1 \in [0,29) = [0,\cosize(L))$ is not in the image of $\Phi_L$. More generally, if $\size(L) \neq \cosize(L)$, then $L$ is not compact. 
\end{example}

We conclude this section by listing some equivalent conditions for a layout $L$ to be compact. 
\begin{proposition} \label{equivalentconditionsfornestedcompact}
Suppose $L$ is a layout. Then the following conditions are equivalent.
\begin{enumerate}
    \item $L$ is compact.
    \item $L^\flat$ is compact.
    \item $\coalesce(L)$ is compact. 
\end{enumerate}
\end{proposition}
\begin{proof}
    The equivalence of these conditions follows from the fact that 
    \[
    \Phi_L = \Phi_{L^\flat} = \Phi_{\coalesce(L)}. 
    \]
\end{proof}

\subsection{Complements}

We can easily extend the concept of complement to the nested case as follows.

\begin{definition}
    Suppose $A$ and $B$ are layouts. We say $B$ is a {\it complement of } $A$, and write $A \perp B$, if the concatenated layout $(A,B)$ is compact. 
\end{definition}

\begin{lemma}\label{flatteningcomplementationlemma}
    Suppose $A$ and $B$ are layouts. Then 
    \[
    A \perp B \quad \Leftrightarrow \quad A^\flat \perp B^\flat.
    \]
\end{lemma}

\begin{proof}
    This follows from the observation that $(A,B)^\flat = A^\flat \star B^\flat$.
\end{proof}

\begin{definition}
    Suppose $A$ is a layout. We say $A$ is {\it complementable} if $A^\flat$ is complementable.
\end{definition}

\begin{lemma}
    Suppose $A$ is a layout. Then there exists a complement $B$ of $A$ if and only if $A$ is complementable. 
\end{lemma}
\begin{proof}
    If $A$ is complementable, then $A^\flat$ is complementable, so there exists a flat layout $B$ such that the flat concatenation $A^\flat \star B$ is compact. It follows that the concatenation $(A , B)$ is also compact, so $A$ admits a complement. Conversely, suppose there exists a layout $B$ such that $(A,B)$ is compact. Then $B^\flat$ is a complement of $A^\flat$, so by Proposition \ref{complementableproposition}, $A^\flat$ is complementable, hence, by definition, so is $A$.
\end{proof}

\begin{definition}\label{definitionoflayoutcomplements}
    Suppose $A$ is a layout. If $A$ is complementable, then we define 
    \[
    \comp(A) = \coalesce(\comp^\flat(A^\flat)),
    \]
    as in Construction \ref{complementconstruction}. If $N$ is a positive integer and $A$ is $N$-complementable, then we define 
    \[
    \comp(A,N) = \coalesce(\comp^\flat(A^\flat,N))
    \]
    as in Construction \ref{Ncomplementconstruction}.
\end{definition}
\begin{remark}
    Suppose $A$ is a complementable layout. Then we almost always have $\comp(A) = \comp^\flat(A^\flat)$. More precisely, if $\comp^\flat(A^\flat)$ has length $>1$, then 
    \[
    \comp(A) = \comp^\flat(A^\flat),
    \]
    if $\comp^\flat(A^\flat) = (s):(d)$ has length $1$, then 
    \[
    \comp(A) = s:d,
    \]
    and if $\comp^\flat(A^\flat) = ():()$, then 
    \[
    \comp(A) = 1:0.
    \]
\end{remark}

\begin{definition}
    Suppose $A$ is a layout and $N$ is a positive integer. We say a layout $B$ is a $N${\it -complement} of $A$ if $A \perp B$, and 
    \[
    \size(A) \cdot \size(B) = N.
    \]
\end{definition}

\begin{definition}
    Suppose $A$ is a layout and $N$ is a positive integer. We say $A$ is $N${\it -complementable} if the flat layout $A^\flat$ is $N$-complementable, as in Definition \ref{definitionofNcomplementable}.
\end{definition}

\begin{proposition}
    Suppose $A$ is a layout. Then there exists a $N$-complement of $A$ if and only if $A$ is $N$-complementable. 
\end{proposition}

\begin{proof}
If $B$ is a $N$-complement of $A$, then $B^\flat$ is a $N$-complement $A^\flat$, and so by Proposition \ref{characterizationofexistenceofNcomplements}, $A^\flat$ is $N$-complementable, hence, so is $A$. Conversely, if $A$ is $N$-complementable, then $\comp(A,N)$ is a $N$-complement of $A$. 
\end{proof}

\begin{example}
If $A = ((4,2),(2,2)) : ((3,24),(192,96)) $ and $N = 768$ then
\[\comp(A,N) = (3,2,2,2):(1,12,48,384).\]
\end{example}

\begin{example}
    If $A = ((16,4),64):((1,16),64)$ and $N = 4096$ then
    \begin{align*}
    \comp(A,N) & = \coalesce(():()) \\
    & = 1:0.
    \end{align*}
\end{example}

\begin{example}
        If $A = ((16,4),64):((1,16),64)$ and $N = 8192$ then
    \begin{align*}
    \comp(A,N) & = \coalesce((2):(4096)) \\
    & = 2:4096.
    \end{align*}
\end{example}

\begin{example}
    If $A = ((16,4),64):((8,1),128)$ and $N = 16384$, then 
    \begin{align*}
    \comp(A,N) & = \coalesce((2,2):(4,8192)) \\
    & = (2,2):(4,8192).
    \end{align*}
\end{example}

\subsection{Composition}

In this section, we discuss the most important operation on layouts, namely {\it composition}. If $A$ and $B$ are layouts, then the composition of $A$ and $B$ is a layout $B \circ A$ whose layout function is the composite of the layout functions of $A$ and $B$. More precisely, we have the following definition.

\begin{definition}[Composition of layouts]\label{definitionoflayoutcomposition}
    Suppose $A$ and $B$ are layouts. The {\it composite} of $A$ and $B$ is the unique layout $B \circ A$ satisfying the following properties. 
    \begin{enumerate}
    \item $\shape(B\circ A)$ refines $\shape(A)$,
    \item $B \circ A$ is coalesced over $\shape(A)$, and 
    \item $\Phi_{B \circ A} = \Phi_B \circ \Phi_A^{\size(B)}$.
    \end{enumerate}
\end{definition}

\begin{remark}
    In order for $B \circ A$ to exist, we must have 
    \[
    \Image(\Phi_A) \subseteq [0,\size(B)).
    \]
\end{remark}

\begin{remark}
    There is an implicit assertion in the definition of layout composition, namely that there is at most one layout satisfying the three conditions. This is justified by Proposition \ref{relativecoalescedlayoutscharacterizedbylayoutfunction}. We might define a {\it weak composite} of $A$ and $B$ to be a layout $C$ satisfying conditions 1. and 3. (but not necessarily 2.), in which case 
    \[
    B \circ A = \coalesce(C,\shape(A))
    \]
    We will see later on that when attempting to compute compositions of layouts, it is useful to compute any weak composite $C$ of $A$ and $B$, then coalesce over $\shape(A)$ to form the actual composite $B \circ A$. 
\end{remark}

\begin{remark}
Note that, by Observation \ref{coalescedoverSbarimpliesnondegenerate}, condition 2. in the definition of composition implies that $B \circ A$ is non-degenerate.
\end{remark}

\begin{example}\label{compositionexample1}
If $A = (3,5):(10,2)$ and $B = (100):(7)$, then 
\[
B \circ A = (3,5):(70,14).
\]
\end{example}

\begin{example}\label{compositionexample2}
    If $A = (4):(2)$ and $B = (2,2,6):(12,6,1)$, then the composition of $A$ and $B$ is 
    \[
    B \circ A = ((2,2)):((6,1)).
    \]
\end{example}

\begin{remark}
    Example \ref{compositionexample2} illustrates the fact that the composition of flat layouts $A$ and $B$ need not be flat.
\end{remark}

\begin{example}
    If $A = ((2,4),8):((4,8),8)$ and $B = (4,4,4,4):(2,4,8,16)$, then 
    \[
    B \circ A = ( (2,(2,2)),(2,4) ): ( (4,(8,8)) , (8,8) ).
    \]
\end{example}

\begin{example}
    If $A = ((3,(2,2)),24):((3,(9,18)),72)$ and $B = (9,8,3,8):(24,3,1,384) $ then
    \[
    B \circ A = ((3,(2,2)),(3,8) ) : ((72,(3,6)),(1,384) ) 
    \]
\end{example}

Next, we develop some useful properties for computing the composition of layouts. 

\begin{proposition}\label{postcomposewithextension}
    Suppose $A$ is a layout, and suppose $B$ and $\tilde{B}$ are layouts such that 
    \begin{itemize}
        \item $\size(B) \leq \size(\tilde{B})$, and 
        \item $\Phi_{\tilde{B}} \mid_{\size(B)} = \Phi_B$.
    \end{itemize}
    If $A$ and $B$ are composable, then 
    \[
    B \circ A = \tilde{B} \circ A.
    \]
\end{proposition}

\begin{proof}
    Suppose $A$ and $B$ are composable. Then $\cosize(A) \leq \size(B)$, and the fact that $B \circ A$ is the composite of $A$ and $\tilde{B}$ follows from the equality
    \begin{align*}
        \Phi_{\tilde{B}} \circ \Phi_A^{\size(\tilde{B})} & = \left( \Phi_{\tilde{B}} \right) \mid_{\size(B)} \circ \Phi_A^{\size(B)}\\
        & = \Phi_B \circ \Phi_A^{\size(B)}.
    \end{align*}
\end{proof}

\begin{corollary}\label{postcomposewithcoalesce}
    If $A$ and $B$ are layouts, then $A$ and $B$ are composable if and only if $A$ and $\coalesce(B)$ are composable, and 
    \[
    B \circ A = \coalesce(B) \circ A. 
    \]
\end{corollary}

Now that we have developed the basic properties of layout composition, we turn our attention to the two most important instances of composition, namely {\it logical division} and {\it logical products}. 

\subsection{Logical division}\label{logicaldivisionsection}

\noindent In this section, we define the {\bf logical division} of layouts. As a motivating example, consider the layout

\bigskip

\begin{centering}
    
\begin{tikzpicture}[x={(0cm,-1cm)},y={(1cm,0cm)},every node/.style={minimum size=1cm, outer sep=0pt}]

\node[fill=gray!20] at (0,0) {0};
\node[fill=gray!20] at (0,1) {4};
\node[fill=gray!20] at (0,2) {8};
\node[fill=gray!20] at (0,3) {12};
\node[fill=gray!20] at (0,4) {16};
\node[fill=gray!20] at (0,5) {20};
\node[fill=gray!20] at (0,6) {24};
\node[fill=gray!20] at (0,7) {28};
\node[fill=gray!20] at (1,0) {1};
\node[fill=gray!20] at (1,1) {5};
\node[fill=gray!20] at (1,2) {9};
\node[fill=gray!20] at (1,3) {13};
\node[fill=gray!20] at (1,4) {17};
\node[fill=gray!20] at (1,5) {21};
\node[fill=gray!20] at (1,6) {25};
\node[fill=gray!20] at (1,7) {29};
\node[fill=gray!20] at (2,0) {2};
\node[fill=gray!20] at (2,1) {6};
\node[fill=gray!20] at (2,2) {10};
\node[fill=gray!20] at (2,3) {14};
\node[fill=gray!20] at (2,4) {18};
\node[fill=gray!20] at (2,5) {22};
\node[fill=gray!20] at (2,6) {26};
\node[fill=gray!20] at (2,7) {30};
\node[fill=gray!20] at (3,0) {3};
\node[fill=gray!20] at (3,1) {7};
\node[fill=gray!20] at (3,2) {11};
\node[fill=gray!20] at (3,3) {15};
\node[fill=gray!20] at (3,4) {19};
\node[fill=gray!20] at (3,5) {23};
\node[fill=gray!20] at (3,6) {27};
\node[fill=gray!20] at (3,7) {31};
\draw[color=black,thick,shift={(-0.5,-0.5)}] (0,0) grid (4,8);

\node[anchor=east] at (1.5,-1) {$A = (4,8):(1,4) = $};
\end{tikzpicture}

\end{centering} 

\bigskip

\noindent For various purposes, we may want to {\it tile} the layout $A$. For example, here are the tilings of $A$ by various layouts $B$.

\bigskip

\begin{centering}

\begin{tikzpicture}[x={(0cm,-1cm)},y={(1cm,0cm)},every node/.style={minimum size=1cm, outer sep=0pt}]

\node[fill=colorx1!60] at (0,0) {0};
\node[fill=colorx1!60] at (0,1) {4};
\node[fill=colorx3!60] at (0,2) {8};
\node[fill=colorx3!60] at (0,3) {12};
\node[fill=colorx5!60] at (0,4) {16};
\node[fill=colorx5!60] at (0,5) {20};
\node[fill=colorx7!60] at (0,6) {24};
\node[fill=colorx7!60] at (0,7) {28};
\node[fill=colorx1!60] at (1,0) {1};
\node[fill=colorx1!60] at (1,1) {5};
\node[fill=colorx3!60] at (1,2) {9};
\node[fill=colorx3!60] at (1,3) {13};
\node[fill=colorx5!60] at (1,4) {17};
\node[fill=colorx5!60] at (1,5) {21};
\node[fill=colorx7!60] at (1,6) {25};
\node[fill=colorx7!60] at (1,7) {29};
\node[fill=colorx2!60] at (2,0) {2};
\node[fill=colorx2!60] at (2,1) {6};
\node[fill=colorx4!60] at (2,2) {10};
\node[fill=colorx4!60] at (2,3) {14};
\node[fill=colorx6!60] at (2,4) {18};
\node[fill=colorx6!60] at (2,5) {22};
\node[fill=colorx8!60] at (2,6) {26};
\node[fill=colorx8!60] at (2,7) {30};
\node[fill=colorx2!60] at (3,0) {3};
\node[fill=colorx2!60] at (3,1) {7};
\node[fill=colorx4!60] at (3,2) {11};
\node[fill=colorx4!60] at (3,3) {15};
\node[fill=colorx6!60] at (3,4) {19};
\node[fill=colorx6!60] at (3,5) {23};
\node[fill=colorx8!60] at (3,6) {27};
\node[fill=colorx8!60] at (3,7) {31};
\draw[color=black,thick,shift={(-0.5,-0.5)}] (0,0) grid (4,8);

\node at (1.5,-1) {$A = $};
\node at (1.5,9) {$B = $};

\node[fill=gray!20] at (1,10) {0};
\node[fill=gray!20] at (2,10) {1};
\node[fill=gray!20] at (1,11) {4};
\node[fill=gray!20] at (2,11) {5};
\draw[color=black,thick,shift={(-0.5,-0.5)}] (1,10) grid (3,12);

\node[fill=colorx1!60] at (5,0) {0};
\node[fill=colorx1!60] at (5,1) {4};
\node[fill=colorx1!60] at (5,2) {8};
\node[fill=colorx1!60] at (5,3) {12};
\node[fill=colorx5!60] at (5,4) {16};
\node[fill=colorx5!60] at (5,5) {20};
\node[fill=colorx5!60] at (5,6) {24};
\node[fill=colorx5!60] at (5,7) {28};
\node[fill=colorx3!60] at (6,0) {1};
\node[fill=colorx3!60] at (6,1) {5};
\node[fill=colorx3!60] at (6,2) {9};
\node[fill=colorx3!60] at (6,3) {13};
\node[fill=colorx7!60] at (6,4) {17};
\node[fill=colorx7!60] at (6,5) {21};
\node[fill=colorx7!60] at (6,6) {25};
\node[fill=colorx7!60] at (6,7) {29};
\node[fill=colorx1!60] at (7,0) {2};
\node[fill=colorx1!60] at (7,1) {6};
\node[fill=colorx1!60] at (7,2) {10};
\node[fill=colorx1!60] at (7,3) {14};
\node[fill=colorx5!60] at (7,4) {18};
\node[fill=colorx5!60] at (7,5) {22};
\node[fill=colorx5!60] at (7,6) {26};
\node[fill=colorx5!60] at (7,7) {30};
\node[fill=colorx3!60] at (8,0) {3};
\node[fill=colorx3!60] at (8,1) {7};
\node[fill=colorx3!60] at (8,2) {11};
\node[fill=colorx3!60] at (8,3) {15};
\node[fill=colorx7!60] at (8,4) {19};
\node[fill=colorx7!60] at (8,5) {23};
\node[fill=colorx7!60] at (8,6) {27};
\node[fill=colorx7!60] at (8,7) {31};
\draw[color=black,thick,shift={(-0.5,-0.5)}] (5,0) grid (9,8);

\node at (6.5,-1) {$A = $};
\node at (6.5,9) {$B = $};

\node[fill=gray!20] at (6,10) {0};
\node[fill=gray!20] at (7,10) {2};
\node[fill=gray!20] at (6,11) {4};
\node[fill=gray!20] at (7,11) {6};
\node[fill=gray!20] at (6,12) {8};
\node[fill=gray!20] at (7,12) {10};
\node[fill=gray!20] at (6,13) {12};
\node[fill=gray!20] at (7,13) {14};
\draw[color=black,thick,shift={(-0.5,-0.5)}] (6,10) grid (8,14);

\node[fill=colorx1!60] at (10,0) {0};
\node[fill=colorx1!60] at (10,1) {4};
\node[fill=colorx3!60] at (10,2) {8};
\node[fill=colorx3!60] at (10,3) {12};
\node[fill=colorx1!60] at (10,4) {16};
\node[fill=colorx1!60] at (10,5) {20};
\node[fill=colorx3!60] at (10,6) {24};
\node[fill=colorx3!60] at (10,7) {28};
\node[fill=colorx5!60] at (11,0) {1};
\node[fill=colorx5!60] at (11,1) {5};
\node[fill=colorx7!60] at (11,2) {9};
\node[fill=colorx7!60] at (11,3) {13};
\node[fill=colorx5!60] at (11,4) {17};
\node[fill=colorx5!60] at (11,5) {21};
\node[fill=colorx7!60] at (11,6) {25};
\node[fill=colorx7!60] at (11,7) {29};
\node[fill=colorx1!60] at (12,0) {2};
\node[fill=colorx1!60] at (12,1) {6};
\node[fill=colorx3!60] at (12,2) {10};
\node[fill=colorx3!60] at (12,3) {14};
\node[fill=colorx1!60] at (12,4) {18};
\node[fill=colorx1!60] at (12,5) {22};
\node[fill=colorx3!60] at (12,6) {26};
\node[fill=colorx3!60] at (12,7) {30};
\node[fill=colorx5!60] at (13,0) {3};
\node[fill=colorx5!60] at (13,1) {7};
\node[fill=colorx7!60] at (13,2) {11};
\node[fill=colorx7!60] at (13,3) {15};
\node[fill=colorx5!60] at (13,4) {19};
\node[fill=colorx5!60] at (13,5) {23};
\node[fill=colorx7!60] at (13,6) {27};
\node[fill=colorx7!60] at (13,7) {31};
\draw[color=black,thick,shift={(-0.5,-0.5)}] (10,0) grid (14,8);

\node at (11.5,-1) {$A = $};
\node at (11.5,9) {$B = $};

\node[fill=gray!20] at (11,10) {0};
\node[fill=gray!20] at (12,10) {2};
\node[fill=gray!20] at (11,11) {4};
\node[fill=gray!20] at (12,11) {6};
\node[fill=gray!20] at (11,12) {16};
\node[fill=gray!20] at (12,12) {18};
\node[fill=gray!20] at (11,13) {20};
\node[fill=gray!20] at (12,13) {22};
\draw[color=black,thick,shift={(-0.5,-0.5)}] (11,10) grid (13,14);
\end{tikzpicture}

\end{centering}

\bigskip 

When working with such tiled layouts, we would like to index into our layout with coordinates of the form $(\texttt{tile}\_\texttt{coordinate},\texttt{tile})$ where $\texttt{tile}$ specifies which tile we are working with, and $\texttt{tile}\_\texttt{coordinate}$ specifies a coordinate within the specified tile. For example, if both $A$ and $B$ have rank $2$, we would like to write $((i,j),(k,\ell))$ as the index of the $(i,j)$th entry of the $(k,\ell)$th tile of $A$. The logical division of $A \oslash B$ is precisely the layout which affords us this ability.

\begin{definition}\label{definitionoflogicaldivide}
Suppose $A$ and $B$ are layouts, and suppose 
\[B^c = \comp(B,\size(A))\]
is the complement of $B$ with respect to $\size(A)$. Then the {\it logical division} of $A$ by $B$ is the layout
\begin{align*}
A \oslash B & = A \circ (B,B^c)\\
& = (A \circ B, A \circ B^c).
\end{align*}
\end{definition}

\begin{example}
    If $A = (4,8):(1,4)$ and $B = (2,2):(1,4)$, then 
    \[A \oslash B = ((2,2),(2,4)):((1,4),(2,8)),\]
    as depicted below.

    \bigskip

\begin{centering}
\begin{tikzpicture}[x={(0cm,-1cm)},y={(1cm,0cm)},every node/.style={minimum size=1cm, outer sep=0pt}]

\node[fill=colorx1!60] at (0,0) {0};
\node[fill=colorx1!60] at (0,1) {4};
\node[fill=colorx3!60] at (0,2) {8};
\node[fill=colorx3!60] at (0,3) {12};
\node[fill=colorx5!60] at (0,4) {16};
\node[fill=colorx5!60] at (0,5) {20};
\node[fill=colorx7!60] at (0,6) {24};
\node[fill=colorx7!60] at (0,7) {28};
\node[fill=colorx1!60] at (1,0) {1};
\node[fill=colorx1!60] at (1,1) {5};
\node[fill=colorx3!60] at (1,2) {9};
\node[fill=colorx3!60] at (1,3) {13};
\node[fill=colorx5!60] at (1,4) {17};
\node[fill=colorx5!60] at (1,5) {21};
\node[fill=colorx7!60] at (1,6) {25};
\node[fill=colorx7!60] at (1,7) {29};
\node[fill=colorx2!60] at (2,0) {2};
\node[fill=colorx2!60] at (2,1) {6};
\node[fill=colorx4!60] at (2,2) {10};
\node[fill=colorx4!60] at (2,3) {14};
\node[fill=colorx6!60] at (2,4) {18};
\node[fill=colorx6!60] at (2,5) {22};
\node[fill=colorx8!60] at (2,6) {26};
\node[fill=colorx8!60] at (2,7) {30};
\node[fill=colorx2!60] at (3,0) {3};
\node[fill=colorx2!60] at (3,1) {7};
\node[fill=colorx4!60] at (3,2) {11};
\node[fill=colorx4!60] at (3,3) {15};
\node[fill=colorx6!60] at (3,4) {19};
\node[fill=colorx6!60] at (3,5) {23};
\node[fill=colorx8!60] at (3,6) {27};
\node[fill=colorx8!60] at (3,7) {31};
\draw[color=black,thick,shift={(-0.5,-0.5)}] (0,0) grid (4,8);

\node[anchor = east] at (1.5,-1) {$A = $};
\node[anchor = east] at (1.5,9.5) {$B = $};

\node[fill=gray!5] at (1,10) {0};
\node[fill=gray!20] at (2,10) {1};
\node[fill=gray!40] at (1,11) {4};
\node[fill=gray!60] at (2,11) {5};
\draw[color=black,thick,shift={(-0.5,-0.5)}] (1,10) grid (3,12);
\end{tikzpicture}

\end{centering}

\vspace{0.2in}

\begin{centering}
\begin{tikzpicture}[x={(0cm,-1cm)},y={(1cm,0cm)},every node/.style={minimum size=1cm, outer sep=0pt}]

\node[fill=colorx1!10]       at (0,0) {0};
\node[fill=colorx2!10]       at (0,1) {2};
\node[fill=colorx3!10]      at (0,2) {8};
\node[fill=colorx4!10]      at (0,3) {10};
\node[fill=colorx5!10]  at (0,4) {16};
\node[fill=colorx6!10]  at (0,5) {18};
\node[fill=colorx7!10]        at (0,6) {24};
\node[fill=colorx8!10]        at (0,7) {26};

\node[fill=colorx1!25]       at (1,0) {1};
\node[fill=colorx2!25]       at (1,1) {3};
\node[fill=colorx3!25]      at (1,2) {9};
\node[fill=colorx4!25]      at (1,3) {11};
\node[fill=colorx5!25]  at (1,4) {17};
\node[fill=colorx6!25]  at (1,5) {19};
\node[fill=colorx7!25]        at (1,6) {25};
\node[fill=colorx8!25]        at (1,7) {27};

\node[fill=colorx1!40]       at (2,0) {4};
\node[fill=colorx2!40]       at (2,1) {6};
\node[fill=colorx3!40]      at (2,2) {12};
\node[fill=colorx4!40]      at (2,3) {14};
\node[fill=colorx5!40]  at (2,4) {20};
\node[fill=colorx6!40]  at (2,5) {22};
\node[fill=colorx7!40]        at (2,6) {28};
\node[fill=colorx8!40]        at (2,7) {30};

\node[fill=colorx1!60]       at (3,0) {5};
\node[fill=colorx2!60]       at (3,1) {7};
\node[fill=colorx3!60]      at (3,2) {13};
\node[fill=colorx4!60]      at (3,3) {15};
\node[fill=colorx5!60]  at (3,4) {21};
\node[fill=colorx6!60]  at (3,5) {23};
\node[fill=colorx7!60]       at (3,6) {29};
\node[fill=colorx8!60]       at (3,7) {31};
\draw[color=black,thick,shift={(-0.5,-0.5)}] (0,0) grid (4,8);

\node[anchor=east] at (1.5,-1) {$ A \oslash B = $};

\end{tikzpicture}

\end{centering}
\end{example}

\begin{remark}
    The {\it color} of each entry in $A \oslash B$ indicates the tile to which it belongs, and the {\it opacity} of each entry in $A\oslash B$ indicates which entry of the tile it represents. This is why each column of $A\oslash B$ has the same color, and each row of $A \oslash B$ has the same opacity.
\end{remark}

\begin{example}
    If $A = (4,8):(1,4)$ and $B = (2,2):(4,1)$, then 
    \[A \oslash B = ((2,2),(2,4)):((4,1),(2,8)),\]
    as depicted below.

    \bigskip

\begin{centering}
\begin{tikzpicture}[x={(0cm,-1cm)},y={(1cm,0cm)},every node/.style={minimum size=1cm, outer sep=0pt}]

\node[fill=colorx1!60] at (0,0) {0};
\node[fill=colorx1!60] at (0,1) {4};
\node[fill=colorx3!60] at (0,2) {8};
\node[fill=colorx3!60] at (0,3) {12};
\node[fill=colorx5!60] at (0,4) {16};
\node[fill=colorx5!60] at (0,5) {20};
\node[fill=colorx7!60] at (0,6) {24};
\node[fill=colorx7!60] at (0,7) {28};
\node[fill=colorx1!60] at (1,0) {1};
\node[fill=colorx1!60] at (1,1) {5};
\node[fill=colorx3!60] at (1,2) {9};
\node[fill=colorx3!60] at (1,3) {13};
\node[fill=colorx5!60] at (1,4) {17};
\node[fill=colorx5!60] at (1,5) {21};
\node[fill=colorx7!60] at (1,6) {25};
\node[fill=colorx7!60] at (1,7) {29};
\node[fill=colorx2!60] at (2,0) {2};
\node[fill=colorx2!60] at (2,1) {6};
\node[fill=colorx4!60] at (2,2) {10};
\node[fill=colorx4!60] at (2,3) {14};
\node[fill=colorx6!60] at (2,4) {18};
\node[fill=colorx6!60] at (2,5) {22};
\node[fill=colorx8!60] at (2,6) {26};
\node[fill=colorx8!60] at (2,7) {30};
\node[fill=colorx2!60] at (3,0) {3};
\node[fill=colorx2!60] at (3,1) {7};
\node[fill=colorx4!60] at (3,2) {11};
\node[fill=colorx4!60] at (3,3) {15};
\node[fill=colorx6!60] at (3,4) {19};
\node[fill=colorx6!60] at (3,5) {23};
\node[fill=colorx8!60] at (3,6) {27};
\node[fill=colorx8!60] at (3,7) {31};
\draw[color=black,thick,shift={(-0.5,-0.5)}] (0,0) grid (4,8);

\node[anchor = east] at (1.5,-1) {$A = $};
\node[anchor = east] at (1.5,9.5) {$B = $};

\node[fill=gray!5] at (1,10) {0};
\node[fill=gray!40] at (2,10) {4};
\node[fill=gray!20] at (1,11) {1};
\node[fill=gray!60] at (2,11) {5};
\draw[color=black,thick,shift={(-0.5,-0.5)}] (1,10) grid (3,12);
\end{tikzpicture}

\end{centering}

\vspace{0.2in}

\begin{centering}
\begin{tikzpicture}[x={(0cm,-1cm)},y={(1cm,0cm)},every node/.style={minimum size=1cm, outer sep=0pt}]

\node[fill=colorx1!10] at (0,0) {0};
\node[fill=colorx2!10] at (0,1) {2};
\node[fill=colorx3!10] at (0,2) {8};
\node[fill=colorx4!10] at (0,3) {10};
\node[fill=colorx5!10] at (0,4) {16};
\node[fill=colorx6!10] at (0,5) {18};
\node[fill=colorx7!10] at (0,6) {24};
\node[fill=colorx8!10] at (0,7) {26};
\node[fill=colorx1!40] at (1,0) {4};
\node[fill=colorx2!40] at (1,1) {6};
\node[fill=colorx3!40] at (1,2) {12};
\node[fill=colorx4!40] at (1,3) {14};
\node[fill=colorx5!40] at (1,4) {20};
\node[fill=colorx6!40] at (1,5) {22};
\node[fill=colorx7!40] at (1,6) {28};
\node[fill=colorx8!40] at (1,7) {30};
\node[fill=colorx1!25] at (2,0) {1};
\node[fill=colorx2!25] at (2,1) {3};
\node[fill=colorx3!25] at (2,2) {9};
\node[fill=colorx4!25] at (2,3) {11};
\node[fill=colorx5!25] at (2,4) {17};
\node[fill=colorx6!25] at (2,5) {19};
\node[fill=colorx7!25] at (2,6) {25};
\node[fill=colorx8!25] at (2,7) {27};
\node[fill=colorx1!60] at (3,0) {5};
\node[fill=colorx2!60] at (3,1) {7};
\node[fill=colorx3!60] at (3,2) {13};
\node[fill=colorx4!60] at (3,3) {15};
\node[fill=colorx5!60] at (3,4) {21};
\node[fill=colorx6!60] at (3,5) {23};
\node[fill=colorx7!60] at (3,6) {29};
\node[fill=colorx8!60] at (3,7) {31};
\draw[color=black,thick,shift={(-0.5,-0.5)}] (0,0) grid (4,8);

\node[anchor=east] at (1.5,-1) {$ A \oslash B = $};

% \node at (0,-1) {\Large{\texttt{0}}};
% \node at (1,-1) {\Large{\texttt{1}}};
% \node at (2,-1) {\Large{\texttt{2}}};
% \node at (3,-1) {\Large{\texttt{3}}};
% \node at (-1,0) {\Large{\texttt{0}}};
% \node at (-1,1) {\Large{\texttt{1}}};
% \node at (-1,2) {\Large{\texttt{2}}};
% \node at (-1,3) {\Large{\texttt{3}}};
% \node at (-1,4) {\Large{\texttt{4}}};
% \node at (-1,5) {\Large{\texttt{5}}};
% \node at (-1,6) {\Large{\texttt{6}}};
% \node at (-1,7) {\Large{\texttt{7}}};
\end{tikzpicture}

\end{centering}
\end{example}

\begin{remark}
    Note the difference between the previous two examples. The tiling of $A$ in each of the two examples is identical, but the layout of each tile is different. In first example, the tiles have column-major layouts, while in the second example, the tiles have row-major layouts. This results in different layouts when one performs logical division.
\end{remark}

\begin{example}
    If $A = (4,8):(1,4)$ and $B = (2,4):(2,4)$, then 
    \[A \oslash B = ((2,4),(2,2)):((2,4),(1,16)).\]
\end{example}

\begin{example}
    If $A = (4,6):(1,40)$ and $B = 6:4$, then 
    \[
    A \oslash B = (6,4):(40,1).
    \]
\end{example}

\begin{example}
    If $A = (4,6,2,4,2,5) : (36,1,18,0,0,144 )$ and $B = (4,10):(1,192)$, then 
    \[
    A \oslash B = (((4,(2,5)),(6,2,4)):((36,(0,144)
    ),(1,18,0))
    \]
\end{example}

\begin{example}
    If $A = (8,(4,4))$ and $B = (2,(8,16))$, then 
    \[
    A \oslash B = ((2,2),(2,(4,4))):((4,8),(2,(8,16))).
    \]
\end{example}

\subsection{Logical product}\label{logicalproductsection}

In this section, we define the {\bf logical product} of layouts. 
% For example, consider the layouts 

%     \[
%     \begin{tikzpicture}[x={(0cm,-1cm)},y={(1cm,0cm)},every node/.style={minimum size=1cm, outer sep=0pt}]

% \node[fill=colorx2!60] at (0,0) {0};
% \node[fill=colorx4!60] at (0,1) {1};
% \node[fill=colorx6!60] at (1,0) {2};
% \node[fill=colorx8!60] at (1,1) {3};
% \draw[color=black,thick,shift={(-0.5,-0.5)}] (0,0) grid (2,2);

% \node[fill=gray!20] at (0,5) {0};
% \node[fill=gray!20] at (0,6) {0};
% \node[fill=gray!20] at (0,7) {0};
% \node[fill=gray!20] at (1,5) {0};
% \node[fill=gray!20] at (1,6) {0};
% \node[fill=gray!20] at (1,7) {0};
% \draw[color=black,thick,shift={(-0.5,-0.5)}] (0,5) grid (2,8);

% \node[anchor = east] at (0.5,-1) {$A = $ };
% \node[anchor = east] at (0.5,4) {$B = $ };

% \end{tikzpicture}
%     \]

\begin{definition} \label{definitionoflogicalproduct}
Suppose $A$ and $B$ are layouts, and suppose 
\[
A^c = \comp\bigl(A,\size(A)\cdot \cosize(B) \bigr)
\]
is the complement of $A$ with respect to $\size(A) \cdot \cosize(B)$. Then the {\it logical product} of $A$ and $B$ is the layout
\[
A \otimes B = (A , A^c \circ B).
\]
\end{definition}

\begin{observation}
    By Proposition \ref{postcomposewithextension} and Proposition \ref{postcomposewithcoalesce}, if we let
    \[
    \widetilde{A}^c = \comp(A,N)
    \]
    for any valid $N \geq \size(A) \cdot \cosize(B)$, then  
    \[
    A^c \circ B = \tilde{A}^c \circ B.
    \]
    This means that when computing $A \otimes B$, we can take $A^c$ to be any sufficiently large (sorted) complement of $A$.
\end{observation}

\begin{example}
    If $A = (2,2):(5,10)$ and $B = (3,5):(5,1)$ are the layouts
    \[
    \begin{tikzpicture}[x={(0cm,-1cm)},y={(1cm,0cm)},every node/.style={minimum size=1cm, outer sep=0pt}]

\node[fill=colorx2!60] at (0.5,0) {0};
\node[fill=colorx4!60] at (0.5,1) {10};
\node[fill=colorx6!60] at (1.5,0) {5};
\node[fill=colorx8!60] at (1.5,1) {15};
\draw[color=black,thick,shift={(0,-0.5)}] (0,0) grid (2,2);

\node[fill=gray!0] at (0,5) {0};
\node[fill=gray!5] at (0,6) {1};
\node[fill=gray!10] at (0,7) {2};
\node[fill=gray!15] at (0,8) {3};
\node[fill=gray!20] at (0,9) {4};
\node[fill=gray!25] at (1,5) {5};
\node[fill=gray!30] at (1,6) {6};
\node[fill=gray!35] at (1,7) {7};
\node[fill=gray!40] at (1,8) {8};
\node[fill=gray!45] at (1,9) {9};
\node[fill=gray!50] at (2,5) {10};
\node[fill=gray!55] at (2,6) {11};
\node[fill=gray!60] at (2,7) {12};
\node[fill=gray!65] at (2,8) {13};
\node[fill=gray!70] at (2,9) {14};
\draw[color=black,thick,shift={(-0.5,-0.5)}] (0,5) grid (3,10);

\node[anchor = east] at (1,-1) {$A = $ };
\node[anchor = east] at (1,4) {$B = $ };

\end{tikzpicture}
    \]
    then $A \otimes B$ is the layout
    \[
    A \otimes B = ((2,2),(3,5)):((5,10),(20,1))
    \]
as depicted below.
    
\[
\begin{tikzpicture}[x={(0cm,-1cm)},y={(1cm,0cm)},every node/.style={minimum size=1cm, outer sep=0pt}]
\node[fill=colorx2!5] at (0,0) {0};
\node[fill=colorx2!10] at (0,1) {20};
\node[fill=colorx2!15] at (0,2) {40};
\node[fill=colorx2!20] at (0,3) {1};
\node[fill=colorx2!25] at (0,4) {21};
\node[fill=colorx2!30] at (0,5) {41};
\node[fill=colorx2!35] at (0,6) {2};
\node[fill=colorx2!40] at (0,7) {22};
\node[fill=colorx2!45] at (0,8) {42};
\node[fill=colorx2!50] at (0,9) {3};
\node[fill=colorx2!55] at (0,10) {23};
\node[fill=colorx2!60] at (0,11) {43};
\node[fill=colorx2!65] at (0,12) {4};
\node[fill=colorx2!70] at (0,13) {24};
\node[fill=colorx2!75] at (0,14) {44};
\node[fill=colorx4!5] at (1,0) {5};
\node[fill=colorx4!10] at (1,1) {25};
\node[fill=colorx4!15] at (1,2) {45};
\node[fill=colorx4!20] at (1,3) {6};
\node[fill=colorx4!25] at (1,4) {26};
\node[fill=colorx4!30] at (1,5) {46};
\node[fill=colorx4!35] at (1,6) {7};
\node[fill=colorx4!40] at (1,7) {27};
\node[fill=colorx4!45] at (1,8) {47};
\node[fill=colorx4!50] at (1,9) {8};
\node[fill=colorx4!55] at (1,10) {28};
\node[fill=colorx4!60] at (1,11) {48};
\node[fill=colorx4!65] at (1,12) {9};
\node[fill=colorx4!70] at (1,13) {29};
\node[fill=colorx4!75] at (1,14) {49};
\node[fill=colorx6!5] at (2,0) {10};
\node[fill=colorx6!10] at (2,1) {30};
\node[fill=colorx6!15] at (2,2) {50};
\node[fill=colorx6!20] at (2,3) {11};
\node[fill=colorx6!25] at (2,4) {31};
\node[fill=colorx6!30] at (2,5) {51};
\node[fill=colorx6!35] at (2,6) {12};
\node[fill=colorx6!40] at (2,7) {32};
\node[fill=colorx6!45] at (2,8) {52};
\node[fill=colorx6!50] at (2,9) {13};
\node[fill=colorx6!55] at (2,10) {33};
\node[fill=colorx6!60] at (2,11) {53};
\node[fill=colorx6!65] at (2,12) {14};
\node[fill=colorx6!70] at (2,13) {34};
\node[fill=colorx6!75] at (2,14) {54};
\node[fill=colorx8!5] at (3,0) {15};
\node[fill=colorx8!10] at (3,1) {35};
\node[fill=colorx8!15] at (3,2) {55};
\node[fill=colorx8!20] at (3,3) {16};
\node[fill=colorx8!25] at (3,4) {36};
\node[fill=colorx8!30] at (3,5) {56};
\node[fill=colorx8!35] at (3,6) {17};
\node[fill=colorx8!40] at (3,7) {37};
\node[fill=colorx8!45] at (3,8) {57};
\node[fill=colorx8!50] at (3,9) {18};
\node[fill=colorx8!55] at (3,10) {38};
\node[fill=colorx8!60] at (3,11) {58};
\node[fill=colorx8!65] at (3,12) {19};
\node[fill=colorx8!70] at (3,13) {39};
\node[fill=colorx8!75] at (3,14) {59};
\draw[color=black,thick,shift={(-0.5,-0.5)}] (0,0) grid (4,15);
\end{tikzpicture}
\]
\end{example}

\begin{example} If $A = (3,3):(6,1)$ and $B = (10,12):(24,2)$, then 
    \[ 
    A \otimes B = \bigl((3,3),(10,12) \bigr) : \bigl((6,1),(216,18) \bigr).
    \]
\end{example}

\begin{example}
    If $A = (2,10):(1680,4) $ and $B = (4,9):(2,56)$, then
    \[
    A \otimes B = ( (2,10) , ((2,2),(3,3) ) ) : ((1680,4), ( ( 2 ,40 ) , (560,3360 )  ) ).
    \]
\end{example}

\begin{example}
    If $A = (4,(2,2)):(9,(1,3))$ and $B = ((2,4),8):((1,4),2)$, then 
    \[
    A \otimes B = ( ( 4,(2,2)),((2,4),8 ) ) : ( ( 9,(1,3)),((36,144),72 ) ).
    \]
\end{example}

\subsection{Tractable layouts}

In this section we define an especially well-behaved class of layouts, called {\it tractable} layouts. We will see that tractable layouts are precisely the layouts which arise from a certain category $\catstyle{Nest}$.

\begin{definition}\label{definitionoftractablelayouts}
    We say a layout $L$ is {\it tractable} if the flat layout $L^\flat$ is tractable, in the sense of Definition \ref{definitionoftractableflatlayouts}. Explicitly, $L$ is tractable if the flat layout  
    \[
    \sort(L^\flat) = (s_1,\dots,s_m):(d_1,\dots,d_m)
    \]
    is such that for each $1 \leq i < m$, we have 
    \begin{enumerate}
        \item $d_i = 0$\text{, or}
        \item $s_id_i$ divides $d_{i+1}$. 
    \end{enumerate}
\end{definition}

\begin{example}
    The layout 
    \[L = (((12))):(((17)))\]
    is tractable. More generally, any layout $L$ of length $1$ is tractable. 
\end{example}

\begin{example}
    The layout 
    \[L = ((2,4),32):((1,2),8)\]
    is tractable. More generally, any column-major layout is tractable.
\end{example}
\begin{example}
    The layout 
    \[L = (2,(4,32)):(128,(32,1))\]
    is tractable. More generally, any row-major layout $L$ is tractable.
\end{example}
\begin{example}
    The layout 
    \[L = ((3,3),(1,3),(3,1,3)):((81,1),(0,8),(3,0,27))\]
    is tractable. More generally, any compact layout is tractable.
\end{example}

\begin{example}
    The layout 
    \[L = ((3,7,7)):((0,15,0))\]
    is tractable. More generally, any layout with exactly one non-zero stride entry is tractable.
\end{example}

\begin{example}
    The layout 
    \[L = (2,(2,(2,2))):(1,(2048,(16,64)))\]
    is tractable. More generally, any complementable layout is tractable. 
\end{example}

\begin{example}
    The layout  
    \[L = ((8,8),(5,5)):((8,1),(10,2))\]
    is not tractable. In particular, this shows that the concatenation $(L_1,L_2)$ of tractable layouts $L_1$ and $L_2$ need not be tractable.
\end{example}

\chapter{Categories of layouts}\label{categorieschapter}

Having thoroughly explored the algebra of layouts, we now turn our attention to the mathematical heart of this work: realizing layouts as morphisms in suitably-defined categories. Along the way, we develop a \emph{graphical calculus of layout diagrams} that affords more straightforward computation of layout operations.

\section{The category $\catstyle{Tuple}$}\label{thecategoryTuplesection}

In this section, we define a category $\catstyle{Tuple}$ whose objects are tuples of positive integers, and whose morphisms we call {\it tuple morphisms}. Each tuple morphism $f:S \to T$ encodes a flat layout $L_f$. Composition of tuple morphisms is compatible with layout composition, in that if $f$ and $g$ are composable tuple morphisms, then
\[
L_{g \circ f}= L_g \circ L_f.
\]
We define a realization functor (\Cref{constructionofrealizationfunctoronC})
\[
| \cdot |: \catstyle{Tuple} \to \catstyle{FinSet}
\]
which recovers the layout function of $L_f$ via the formula
\[
|f| = \Phi_{L_f}^{\size(T)}.
\]
We develop an ``algebra of tuple morphisms" which includes operations such as {\it sort} (\Cref{subsubsection.sort}), {\it coalesce} (\Cref{subsubsection.coalesce}), {\it complement} (\Cref{subsubsection.complement}), {\it concatenate} (\Cref{subsubsection.concatenate}), {\it flat division} (\Cref{subsubsection.flatdivision}), and {\it flat products} (\Cref{subsubsection.flatproduct}), which are compatible with the corresponding operations on flat layouts.

\subsection{Basic definitions}

\begin{definition}
Let $\catstyle{Fin}_*$ denote the category whose objects are the pointed finite sets 
\[
\langle m \rangle_* = \{*,1,2,\dots,m\}
\]
for $m \geq 0$, and whose morphisms $\alpha: \langle m \rangle_* \to \langle n \rangle_*$ are functions satisfying $\alpha(*) = *$. We call these morphisms {\it pointed maps}, or simply {\it maps}. 
\end{definition}

\begin{aside}
    $\Fin_*$ is a \emph{skeleton} of the category $\FinSet_*$ of finite pointed sets.
\end{aside}

\begin{notation} If the codomain of a pointed map $\alpha:\langle m \rangle_* \to \langle n \rangle_*$ is understood, we sometimes write 
\[
\alpha = (\alpha(1),\dots,\alpha(m))
\]
as a tuple of length $m$ with entries in $\langle n \rangle_*$. 
\end{notation}
\begin{example}\label{morphisminFinexample1}
    There is a morphism $\alpha : \langle 4 \rangle_* \to \langle 6 \rangle_*$ in $\catstyle{Fin}_*$ given by 
    \[
    \alpha = (2,1,*,6),
    \]
    which we can visualize using the following diagram.
    
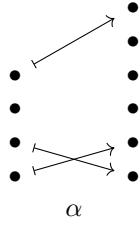
\[ \begin{tikzcd} [row sep = 1, column sep = 8]
          & & \bullet \\
     & & \bullet \\
    \bullet \ar[uurr,mapsto] & & \bullet \\
    \bullet & & \bullet \\
    \bullet  \ar[drr,mapsto] & & \bullet \\
    \bullet\ar[urr,mapsto] & & \bullet \\
    & \alpha & \\
    \end{tikzcd} \]
    Note that the bullet corresponding to entry 3 does not support an arrow, reflecting the fact that it gets sent to $*$. 
\end{example}

\begin{example}\label{morphisminFinexample2}
    There is a morphism $\beta:\langle 5 \rangle_* \to \langle 3 \rangle_*$ in $\catstyle{Fin}_*$ given by 
    \[
    \beta = (*,1,2,3,*)~,
    \]
    which we can visualize using the following diagram.
    \[ \begin{tikzcd} [row sep = 1, column sep = 8]
    \bullet & & \\
    \bullet \ar[drr,mapsto]& & \\
    \bullet \ar[drr,mapsto]& & \bullet\\
    \bullet \ar[drr,bend left = 10, mapsto] & & \bullet\\
    \bullet \ar[rr,mapsto] & & \bullet\\
    & \beta & \\
    \end{tikzcd} \]
\end{example}

\begin{example}\label{morphisminFinexample3}
For any $m \geq 0$, there is a unique morphism in $\catstyle{Fin}_*$ of the form $\pi:\langle m \rangle_* \to \langle 0 \rangle_*$, namely
\[
\pi = (*,\dots,*).
\]
\end{example}

\begin{example}
For any $n \geq 0$, there is a unique morphism in $\catstyle{Fin}_*$ of the form $\delta : \langle 0 \rangle_* \to \langle m \rangle_*$, namely 
    \[
    \delta = ().
    \]
\end{example}

\begin{aside}
    The category $\catstyle{Fin}_*$ is the category of operators for the commutative operad, so we sometimes write 
    \[\catstyle{Fin}_* = \catstyle{Comm}^\otimes.\]
\end{aside}

\noindent We are especially interested in {\it tractable} morphisms in $\catstyle{Fin}_*$, which we define below.

\begin{definition}\label{definitionofessentiallyinjective}
    We say a pointed map $\alpha: \langle m \rangle_* \to \langle n \rangle_*$ is {\it tractable} if for any $j \in \langle n \rangle \subset \langle n \rangle_*$, the preimage $\alpha^{-1}(j)$ is empty or consists of a single element.
\end{definition}

\begin{example}\label{essentiallyinjectivemorphismexample}
    The maps
    \[ \begin{tikzcd} [row sep = 1, column sep = 8]
            & &          & & & \bullet & &        & & & \bullet \ar[drr,mapsto] & & \bullet \\
            & &          & & & \bullet \ar[rr,mapsto] & & \bullet & & & \bullet \ar[urr,mapsto] & & \bullet \\
            & & \bullet   & & & \bullet \ar[drr,mapsto] & & \bullet & & & \bullet \ar[drr,mapsto] & & \bullet \\
    \bullet \ar[urr,mapsto] & & \bullet & & & \bullet & & \bullet & & & \bullet \ar[drr,mapsto]  & & \bullet \\
    \bullet \ar[rr,mapsto]& & \bullet & & & \bullet & & \bullet & & & \bullet \ar[uurr,mapsto] & & \bullet \\
    & \alpha_1 & & & & & \alpha_2 & & & & & \alpha_3 & & \\
    \end{tikzcd} \]
    are tractable, while the maps 
    \[ \begin{tikzcd} [row sep = 1, column sep = 8]
    & &          & & & \bullet & &        & & & \bullet \ar[drr,mapsto] & & \bullet \\
    & &          & & & \bullet \ar[rr,mapsto] & & \bullet & & & \bullet \ar[urr,mapsto] & & \bullet \\
    & & \bullet   & & & \bullet \ar[drr,mapsto] & & \bullet & & & \bullet \ar[drr,bend left = 10, mapsto] & & \bullet \\
    \bullet \ar[drr,bend left = 10, mapsto] & & \bullet & & & \bullet & & \bullet & & & \bullet \ar[rr,mapsto]  & & \bullet \\
    \bullet \ar[rr,mapsto]& & \bullet & & & \bullet \ar[urr,mapsto] & & \bullet & & & \bullet \ar[urr,bend right = 10, mapsto] & & \bullet \\
    & \beta_1 & & & & & \beta_2 & & & & & \beta_3 & & \\
    \end{tikzcd} \]
    are not tractable 
\end{example}

\begin{remark}
    If we represent a morphism $\alpha:\langle m \rangle_* \to \langle n \rangle_*$ in $\catstyle{Fin}_*$ as a tuple, i.e.
    \[
    \alpha = (\alpha(1),\dots,\alpha(m))
    \]
    then $\alpha$ is tractable if and only if no positive integer occurs more than once in $\alpha$. 
\end{remark}

\begin{aside} The wide subcategory 
\[\catstyle{E}_0^\otimes \subset \catstyle{Comm}^\otimes = \catstyle{Fin}_*\] on the  tractable pointed maps is the category of operators for the $\catstyle{E}_0$ operad.
\end{aside}

\begin{definition}\label{definitionofTuple} Let $\catstyle{Tuple}$ denote the category whose objects are tuples 
\[
S = (s_1,\dots,s_m)
\]
of positive integers, where a morphism 
\[f:(s_1,\dots,s_m) \to (t_1,\dots,t_n)\]
is specified by a tractable pointed map $\alpha:\langle m \rangle_* \to \langle n \rangle_*$
satisfying the property that
\begin{itemize} 
\item if $ 1 \leq i \leq m$ and $\alpha(i) \neq *$, then $s_i = t_{\alpha(i)}$~.
\end{itemize}
We say that such a morphism $f$ {\it lies over} $\alpha$, and refer to $f$ as a {\it tuple morphism}.
\end{definition}

\begin{notation}
    If $f:(s_1,\dots,s_m) \to (t_1,\dots,t_n)$ is a tuple morphism which lies over $\alpha$, then we sometimes depict $f$ as 
    \[
    \begin{tikzcd}
(s_1,\dots,s_m) \ar[rr,"f"] \ar[rr,swap,"\alpha"] & & (t_1,\dots,t_n). 
    \end{tikzcd}
    \]
\end{notation}

The graphical calculus of layouts we develop is based on the natural visualizations of morphisms in $\Tuple$, as exemplified below.

\begin{example} 
\label{morphismfinC} The tuple morphism 
\[
\begin{tikzcd} 
(3,128,128) \ar[rr,"f"] \ar[rr,swap,"(1{,}3{,}5)"] & &  (3,2,128,2,128)
\end{tikzcd} \]
can be visualized using the following diagram.
    \[ 
    \begin{tikzcd}[row sep = 1, column sep = 8]
     & & 128  \\
     & & 2   \\
     128 \ar[rruu, mapsto] & & 128  \\
     128  \ar[rru,mapsto] & & 2  \\
     3 \ar[rr,mapsto] & & 3  \\
     & f &  \\
    \end{tikzcd}
    \]
    \end{example}

\begin{example}
\label{morphismginC}
The tuple morphism 
\[\begin{tikzcd} 
(3,128,128) \ar[rr,"g"] \ar[rr,swap,"(*{,}2{,}1)"] & &  (128,128)
\end{tikzcd} \] 
can be visualized using the following diagram.

\[ 
    \begin{tikzcd}[row sep = 1, column sep = 8]
     128 \ar[rrdd, mapsto] & &  \\
     128 \ar[rr,mapsto] & & 128 \\
     3 & & 128 \\
     & g & \\
    \end{tikzcd}
    \]
\end{example}
\begin{example}\label{morphismhinC}
    The tuple morphism 
    \[ \begin{tikzcd}
    (16,16,16,1,32) \ar[rr,"h"] \ar[rr,swap,"(*{,}*{,}1{,}*{,}2)"] & &  (16,32,1,1)
    \end{tikzcd} \]
    can be visualized using the following diagram. 
    
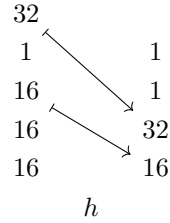
\[
    \begin{tikzcd} [row sep = 1, column sep = 8]
    32 \ar[dddrr,mapsto] & & \\
    1 & & 1 \\
    16 \ar[ddrr,mapsto] & & 1 \\
    16 & & 32 \\
    16 & & 16 \\
    & h & \\
    \end{tikzcd}
    \]
\end{example}

\begin{observation}
    We can relate the category $\catstyle{Tuple}$ to some well-known operads as follows. Let $\mathbb{Z}_{>0}^{\mathterm{div}}$ denote the poset of positive integers under the divisibility relation, considered as a symmetric monoidal category with product given by multiplication of integers. Let $(\mathbb{Z}_{>0}^{\mathterm{div}})^\otimes$ denote the category of operators of $\mathbb{Z}_{>0}^{\mathterm{div}}$. Then there are evident functors \[\catstyle{Tuple} \to (\mathbb{Z}_{>0}^{\mathterm{div}})^\otimes,\] and \[\catstyle{Tuple} \to \catstyle{E}_0^\otimes,\] such that the diagram
    \[ \begin{tikzcd} 
    \catstyle{Tuple} \ar[r] \ar[d] & (\mathbb{Z}_{>0}^{\mathterm{div}})^\otimes \ar[d] \\
    \catstyle{E}_0^\otimes \ar[r] & \catstyle{Comm}^\otimes
    \end{tikzcd} \]
    commutes. This exhibits $\catstyle{Tuple}$ as the wide subcategory of the pullback operad
    \[
    \catstyle{Tuple} \subset \catstyle{E}_0^\otimes \times_{\catstyle{Comm}^\otimes}(\mathbb{Z}_{>0}^{\mathterm{div}})^\otimes\]
    on the morphisms 
    \[\begin{tikzcd} (s_1,\dots,s_m) \ar[r,"f"] \ar[r,swap,"\alpha"] & (t_1,\dots,t_n)
    \end{tikzcd}\]
    satisfying
    \[
    \alpha(i) \neq 1 \quad \Rightarrow \quad  s_i = t_{\alpha(i)}. 
    \]
\end{observation}

\subsection{From tuple morphisms to flat layouts} 

\noindent The impetus for working with the category $\Tuple$ is that each tuple morphism $f$ encodes a flat layout $L_f$. Moreover, each tractable layout $L$ gives rise to a tuple morphism $f_L$. We prove as \Cref{tractableflatlayoutscomefromtuplemorphisms} that these constructions are in some sense inverses, and that tractable layouts are precisely those encoded by tuple morphisms. 

\begin{construction} \label{layoutfrommorphisminTuple}
Suppose 
\[\begin{tikzcd} (s_1,\dots,s_m) \ar[rr,"f"] \ar[rr,swap,"\alpha"] & &  (t_1,\dots,t_n)\end{tikzcd} \]
is a tuple morphism. We define $L_f$ to be the flat layout whose shape \[\shape(L_f) = (s_1,\dots,s_m)\]
is the domain of $f$, and whose stride \[\stride(L_f) = (d_1,\dots,d_m)\] is defined by the formula
\[d_i = \begin{cases} 
0 & \alpha(i) = * \\
\prod_{j < \alpha(i)} t_j & \alpha(i) \neq *.
\end{cases}
\]
We refer to $L_f$ as the {\it layout encoded by} $f$ or the \emph{layout associated to} $f$. 
\end{construction}

\begin{example} The tuple morphism 
\[ 
\begin{tikzcd}[row sep = 1, column sep = 8]
 & & 128  \\
 & & 2   \\
 128 \ar[rruu, mapsto] & & 128  \\
 128  \ar[rru,mapsto] & & 2  \\
 3 \ar[rr,mapsto] & & 3  \\
 & f &  \\
\end{tikzcd}
\]
of Example \ref{morphismfinC} encodes the layout 
\[
L_f  = (3,128,128) : (1,6,1536).
\]
Note that computing the stride via the formula in \Cref{layoutfrommorphisminTuple} amounts to following the arrow from a specific shape entry to its target entry and multiplying together all entries below that one (taking the empty product to equal $1$).
    % \[ 
    % \begin{tikzcd}[row sep = 1, column sep = 8]
    % & & & 128 & & & & &  \\
    % & & & 2 & & & & &  \\
    % & 128 \ar[rruu, mapsto] & & 128 & & & & & \\
    % & 128 \ar[rru,mapsto] & & 2 & & & \rightsquigarrow & & L_f = (3,128,128) : (1,6,1536)\\
    % & 3 \ar[rr,mapsto] & & 3 & & & & & \\
    % & & f & & & & & &  \\
    % \end{tikzcd}
    % \]
\end{example}

\begin{example} The tuple morphism
\[\begin{tikzcd} 
(3,128,128) \ar[rr,"g"] \ar[rr,swap,"(*{,}2{,}1)"] & &  (128,128)
\end{tikzcd} \] 
of Example \ref{morphismginC} encodes the layout
\[
L_g = (3,128,128):(0,128,1).
\]
% \[ 
%     \begin{tikzcd}[row sep = 1, column sep = 8]
%     & 128 \ar[rrdd, mapsto] & & & & & & & \\
%     & 128 \ar[rr,mapsto] & & 128 & & & \rightsquigarrow & & L_g = (3,128,128):(0,128,1)\\
%     & 3 & & 128 & & & & & \\
%     & & g & & & & & & \\
%     \end{tikzcd}
% \]
\end{example}

\begin{example} The tuple morphism 
    \[ \begin{tikzcd}
    (16,16,16,1,32) \ar[rr,"h"] \ar[rr,swap,"(*{,}*{,}1{,}*{,}2)"] & &  (16,32,1,1)
    \end{tikzcd} \]
of Example \ref{morphismhinC} encodes the layout
\[
L_h = (16,16,16,1,32):(0,0,1,0,16).
\]
% \[ 
%     \begin{tikzcd}[row sep = 1, column sep = 8]
%     & 32\ar[rrddd,mapsto] & & & & & & & \\
%     & 1 & & 1 & & & & & \\
%     & 16 \ar[rrdd, mapsto] & & 1 & & & & & \\
%     & 16  & & 32 & & & \rightsquigarrow & & L_h = (16,16,16,1,32):(0,0,1,0,16)\\
%     & 16 & & 16 & & & & & \\
%     & & h & & & & & & \\
%     \end{tikzcd}
% \]
\end{example}

We have seen how to compute the flat layout $L_f$ encoded by a tuple morphism $f$. On the other hand, if $L$ is {\it tractable}, then we can go in the other direction, constructing a tuple morphism $f$ which encodes $L$. Recall from Definition \ref{definitionoftractableflatlayouts} that a flat layout $L$ is {\it tractable} if 
\[\sort(L) = (s_1,\dots,s_m):(d_1,\dots,d_m)\]
satisfies the following property:

\[ \text{If }1 \leq i < m\text{, then }d_i = 0\text{, or }s_id_i \text{ divides }d_{i+1}. \] 

\begin{construction}\label{tuplemorphismfromlayout}
Suppose $L = (s_1,\dots,s_m):(d_1,\dots,d_m)$ is tractable, and set
\[
\sort(L) = (s_1',\dots,s_m'):(d_1',\dots,d_m'),
\]
so there is some permutation $\sigma \in \Sigma_m$ such that $\sort(L) = L^\sigma$. In other words, $s_i' = s_{\sigma(i)}$ and $d_i' = d_{\sigma(i)}$ for each $1 \leq i \leq m$. If each $d_i'$ is nonzero, then let $k = 0$. Otherwise, let $k$ be the largest integer such that $d_k' = 0$. Let $\ell = 2(m-k)$, and let
\begin{align*}
 (t_1',\dots,t_{\ell}') = \left(d_{k+1}',s_{k+1}', \dfrac{d_{k+2}'}{s_{k+1}'d_{k+1}'} , s_{k+2}' , \dfrac{d_{k+3}'}{s'_{k+2}d'_{k+2}}, \dots,\dfrac{d'_m}{s'_{m-1}d'_{m-1}}, s'_m \right).
\end{align*}
We define 
\[f_L':(s_1,\dots,s_m) \to (t_1',\dots,t_\ell')\]
to be the tuple morphism lying over the map $\alpha : \langle m \rangle_* \to \langle \ell \rangle_*$ given by 
\[
\alpha'(i) = \begin{cases}
    * & \sigma^{-1}(i) \leq k\\
    2(\sigma^{-1}(i)-k) & k+1 \leq \sigma^{-1}(i) \leq m.
\end{cases}
\]
Let $J = \{j_1 < \dots < j_n\} \subset \langle \ell \rangle$ denote the collection of indices such that $j_i$ is even or $t_{j_i} \neq 1$. Let 
\[
(t_1,\dots,t_n) = (t_{j_1}',\dots,t_{j_n}'),
\]
and let $\iota:\langle n \rangle_* \to \langle \ell \rangle_*$ be the inclusion map $i \mapsto j_i$. Then by construction, the map $\alpha'$ factors as $\alpha' = \iota \circ \alpha$, and we define the {\it standard representation} of $L$ to be the tuple morphism
\[
\begin{tikzcd}
    (s_1,\dots,s_m) \ar[rr,"f_L"] \ar[rr,swap,"\alpha"] & & (t_1,\dots,t_n). \\
\end{tikzcd}
\]
\end{construction}

\begin{example}
    If 
    \[L = (2,2):(3,30),\]
    then $L$ is tractable, and the standard representation of $L$ is the tuple morphism
% \[\begin{tikzcd}
%     (64) \ar[rr,"f_L"] \ar[rr,swap,"(2)"]& & (7,64).
% \end{tikzcd}\]
    \[ \begin{tikzcd} [row sep = 1, column sep = 8] 
     & & 2\\
     & & 5\\
    2 \ar[uurr,mapsto] & & 2 \\
    2 \ar[urr,mapsto] & & 3 \\
    & f_L & \\
    \end{tikzcd} \]
Note that, informally, computing $f_L$ via \Cref{tuplemorphismfromlayout} amounts to 
\begin{itemize}
    \item initializing the codomain as $ ()$,
    \item traversing the non-zero strides of $L$ in increasing order,
    \item if $d_j$ is the current stride, and $d_i$ is the previously visited stride, appending 
    \begin{itemize}
    \item $(s_j)$ if $s_{i}d_{i} = d_j$, or
    \item $\left(\frac{d_j}{s_id_i},s_j\right)$ if $s_id_i < d_j$,
    \end{itemize}
    and
    \item mapping $s_j \mapsto s_j$. 
\end{itemize}
\end{example}

\begin{example}
    If 
    \[L = (128,128):(128,1),\]
    then $L$ is tractable, and the standard representation of $L$ is the tuple morphism
%     \[\begin{tikzcd}
%     (128,128) \ar[rr,"f_L"] \ar[rr,swap,"(4{,}2)"]& & (1,128,1,128).
% \end{tikzcd}\]
\[
\begin{tikzcd}[row sep = 1, column sep = 8]
128 \ar[drr,mapsto] & & 128 \\
128 \ar[urr,mapsto] & & 128 \\
& f_L & 
\end{tikzcd} 
\]
\end{example}

\begin{example}
    If 
    \[L = (2,2,2,2):(24,0,3,480),\]
    then $L$ is tractable, and the standard representation of $L$ is the tuple morphism
%     \[\begin{tikzcd}
%     (2,2,2,2) \ar[rr,"f_L"] \ar[rr,swap,"(4{,}*{,}2{,}6)"]& & (3,2,4,2,10,2).
% \end{tikzcd}\]
    
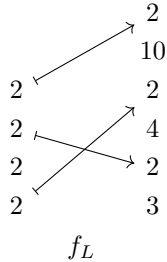
\[ \begin{tikzcd} [row sep = 1, column sep = 8] 
     & & 2 \\
     & & 10 \\
    2 \ar[uurr,mapsto]& & 2 \\
    2 \ar[drr,mapsto]& & 4 \\
    2 & & 2 \\
    2 \ar[uuurr,mapsto] & & 3 \\
    & f_L & \\
    \end{tikzcd} \]
\end{example}

Let's justify that the tuple morphism $f_L$ of \Cref{tuplemorphismfromlayout} does, in fact, encode the layout $L$. 

\begin{lemma}\label{tuplemorphismfromlayoutagreement}
Suppose $L$ is a tractable flat layout, and $f = f_L$ is the standard representation of $L$. Then the layout encoded by $f$ is
\[
L_{f} = L.
\]
\end{lemma} 
\begin{proof} Suppose $L = (s_1,\dots,s_m):(d_1,\dots,d_m)$ is tractable, and let 
\[ \begin{tikzcd} 
(s_1,\dots,s_m) \ar[r,"f"] \ar[r,swap,"\alpha"] &  (t_1,\dots,t_n)
\end{tikzcd} \]
be the standard representation of $L$. Clearly 
\[
\shape(L_f) = (s_1,\dots,s_m) = \shape(L).
\]
We need to check that $\stride(L_f) = \stride(L)$. In other words, we need to check that for any $1 \leq i \leq m$, we have 
\[
d_i = 
\begin{cases} 
0 & \alpha(i) = *\\
\prod_{j < \alpha(i)} t_j & \alpha(i) \neq *.
\end{cases}
\]
We borrow the notation of \Cref{tuplemorphismfromlayout}. If $\alpha(i) = *$, then $\alpha'(i) = *$, and so $\sigma^{-1}(i) \leq k$. This implies
\[
d_i = d'_{\sigma^{-1}(i)} = 0.
\]
Suppose otherwise that $\alpha(i) \neq *$. Then $\alpha'(i) \neq *$, and so $k+1 \leq \sigma^{-1}(i) \leq m$. We compute 
\begin{align*}
\prod_{j < \alpha(i)} t_j = \prod_{\substack{j' < \alpha'(i)\\ t_{j'}' \neq 1} } t_{j'}' = \prod_{j' < \alpha'(i)} t_{j'}' = \prod_{j' < 2(\sigma^{-1}(i)-k)}t_{j'}' & = d_{k+1}' \cdot \left( \prod_{v = 1}^{\sigma^{-1}(i)-(k+1)} s'_{k+v} \dfrac{d'_{k+v+1}}{s_{k+v}'d_{k+v}'} \right)\\
& = d_{\sigma^{-1}(i)}'\\
& = d_i.
\end{align*}
\end{proof}

We have proved that if $L$ is a tractable flat layout, then there exists a tuple morphism $f$ which encodes $L$. Next, we prove the converse, which implies that tractable flat layouts are precisely the layouts encoded by tuple morphisms.

\begin{proposition}\label{tractableflatlayoutscomefromtuplemorphisms}
Suppose $L$ is a flat layout. Then there exists a tuple morphism $f$ encoding $L$ if and only if $L$ is tractable.
\end{proposition} 

\begin{proof}
    First, suppose $L$ is a flat layout, and $f:(s_1,\dots,s_m) \to (t_1,\dots,t_n)$ is a tuple morphism with $L_f = L$. We want to show that $L_f$ is tractable. Let 
    \[
    \sort(L) = (s_1',\dots,s_m'):(d_1,\dots,d_m)
    \]
    be the sorting of $L$, and suppose that $ 1 \leq i < m$. We will argue that $d_i = 0$, or $s_i'd_i$ divides $d_{i+1}$. If $d_i = 0$, then we are done. Suppose otherwise that $d_i \neq 0$. Then 
    \[d_i = \prod_{j< k}t_j\]
    for some $1 \leq k \leq n$ with $s_i' = t_k$. Since $d_{i+1} \geq d_i$, we know that $d_{i+1} \neq 0$, so $d_{i+1}$ has the form 
    \[d_{i+1} = \prod_{j<\ell} t_j\]
    for some $1 \leq \ell \leq n$. There are two cases to consider:
    \begin{itemize}
        \item (Case 1) If $\ell > k$, then
        \begin{align*}
            d_{i+1} = \prod_{j < \ell} t_j & = \left(\prod_{j \leq k} t_j \right) \left(\prod_{k < j < \ell} t_j \right) = s_i'd_i\left( \prod_{k<j<\ell} t_j \right),
        \end{align*}
        so $s_i'd_i$ divides $d_{i+1}$. 
        \item (Case 2) If $\ell \leq k$, then since 
        \[\prod_{j < \ell} t_j = d_{i+1}  \geq d_i =   \prod_{j < k} t_j,\]
        we must have 
        \[t_\ell  = \cdots = t_{k-1} = 1,\]
        and 
        \[
        d_{i+1} = d_i.
        \]
        In particular, we have $s_{i+1}' = t_\ell = 1$. But since $\sort(L_f)$ is sorted and $d_{i+1} = d_i$, we have $s_i' \leq s_{i+1}' = 1$, so $s_i' = 1$. We deduce that 
        \[s_i' d_i  = d_{i+1},\] 
        so in particular, $s_i'd_i$ divides $d_{i+1}$. 
    \end{itemize}
    We conclude that $L$ is tractable.

    Next, suppose that $L$ is tractable. Then we can take $f = f_L$ to be the standard representation of $L$ (see Construction \ref{tuplemorphismfromlayout}), in which case, by Lemma \ref{tuplemorphismfromlayoutagreement}, we have $L = L_f$.
\end{proof}

\begin{remark}\label{standardformmotivation} It is important to note that there are many different tuple morphisms which give rise to the same layout. For example, each of the tuple morphisms shown below

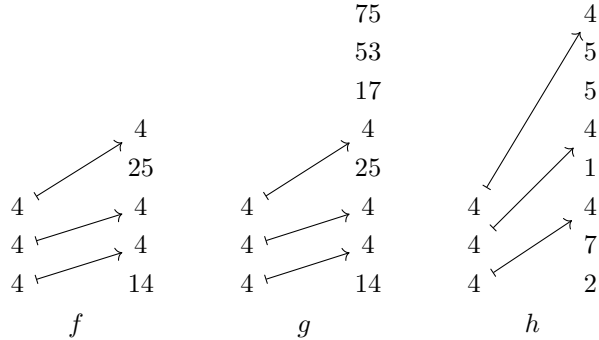
\[
\begin{tikzcd}[row sep=1, column sep=8]
  & &    & & &   & & 75 & & &   & & 4\\
  & &    & & &   & & 53 & & &   & & 5\\
  & &    & & &   & & 17 & & &   & & 5\\
  & &  4 & & &   & & 4  & & &   & & 4\\
  & & 25 & & &   & & 25 & & &   & & 1\\
4 \ar[uurr,mapsto] & & 4 & & & 4 \ar[uurr,mapsto] & & 4 & & & 4 \ar[uuuuurr,mapsto] & & 4\\
4 \ar[urr,mapsto] & & 4 & & & 4 \ar[urr,mapsto] & & 4 & & & 4 \ar[uuurr, mapsto] & & 7\\
4 \ar[urr,mapsto] & &14 & & & 4 \ar[urr,mapsto] & &14 & & & 4 \ar[uurr,mapsto] & & 2\\
  &f & & & & & g & & & & & h & \\
\end{tikzcd}
\]
encodes the layout 
\[
L_{f} = L_{g} = L_{h} = (4,4,4):(14,56,5600).
\]
\noindent Among these, $f$ is the simplest: There are no extraneous entries lying above the image of $f$ (unlike $g$), and the entries not hit by $f$ are condensed (unlike $h$). To make precise the simplicity of $f$ among these morphisms, we introduce the notion of {\it standard form}.
\end{remark}

\begin{definition}\label{definitionofstandardform}
    Suppose 
    \[ \begin{tikzcd} 
    (s_1,\dots,s_m) \ar[rr,swap,"\alpha"] \ar[rr,"f"] & & (t_1,\dots,t_n)
    \end{tikzcd} \]
    is a tuple morphism. We say $f$ has {\it standard form} if the following conditions hold:
    \begin{enumerate}
        \item If $n > 1$, then $n \in \Image(\alpha)$.
        \item If $1 \leq j < n$, then 
        \[
         j \notin \Image(\alpha) \quad \Rightarrow \quad \begin{matrix} t_j \neq 1\text{, and}\\
        j+1 \in \Image(\alpha)
         \end{matrix} 
        \]
    \end{enumerate}
\end{definition}

\begin{example}
    The tuple morphisms $f$ of Remark \ref{standardformmotivation} has standard form, while $g$ and $h$ do not.
\end{example}

\begin{example}
    The tuple morphisms 

    \[ \begin{tikzcd} [row sep = 1, column sep = 8]
 & & & & & & 64 & &  & &  \\
 & &  & & & & 2 & & & & 128\\ 
& &  & & 64 \ar[rruu,mapsto] & & 64 & &  128 \ar[urr,mapsto] & & 512\\
& &  & & 64 \ar[rr,mapsto] & & 64 & & 128  & & 128\\
8 \ar[rr,mapsto] & & 8 & & 64 \ar[rruu,mapsto] & & 2 & & 128 \ar[rru,mapsto]  & & 3\\
& f_1 & & & & f_2 & & & & f_3 & \\
\end{tikzcd} \]
    have standard form, while the tuple morphisms

    
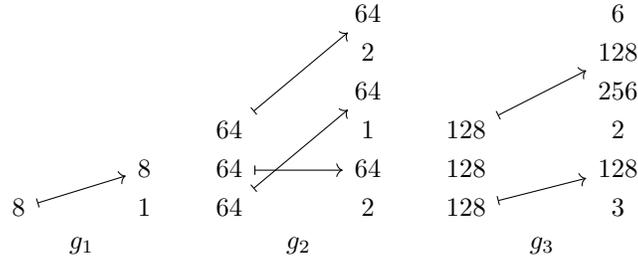
\[ \begin{tikzcd} [row sep = 1, column sep = 8]
 & & & & & & 64 & & & & 6\\
 & & & & & & 2 & &  & & 128\\
 & &  & & & & 64 & & & & 256\\ 
& &  & & 64 \ar[rruuu,mapsto] & & 1 & &  128 \ar[uurr,mapsto] & & 2\\
& & 8 & & 64 \ar[rr,mapsto] & & 64 & & 128  & & 128\\
8 \ar[rru,mapsto] & & 1 & & 64 \ar[rruuu,mapsto] & & 2 & & 128 \ar[rru,mapsto]  & & 3\\
& g_1 & & & & g_2 & & & & g_3 & \\
\end{tikzcd} \]
    do not. 
\end{example}

\begin{example}
    If $L$ is a tractable layout, then by construction, the standard representation $f_L$ of $L$ has standard form.
\end{example}

If we restrict to tuple morphisms of standard form, then there is almost a one-to-one correspondence with tractable layouts. However, there is one problematic case we need to exclude, as explicated in the following example.

\begin{example}\label{badlayouts}
Consider the tuple morphisms $f$ and $g$ shown below.
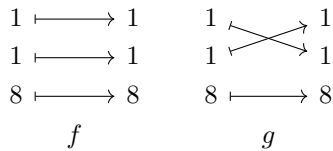
\[ \begin{tikzcd}[row sep = 1, column sep = 8]
1 \ar[rr,mapsto] & & 1 & & 1\ar[drr,mapsto]  & & 1\\
1 \ar[rr,mapsto] & & 1 & & 1\ar[urr,mapsto]  & & 1\\
8 \ar[rr,mapsto] & & 8 & & 8 \ar[rr,mapsto] & & 8\\
 & f & & & & g & 
\end{tikzcd}\]
Both $f$ and $g$ have standard form, and 
\[
L_f = (8,1,1):(1,8,8) = L_g.
\]
This example illustrates that the presence of entries of the form $s_i = 1$ and $\alpha(i) \neq *$ can lead to non-uniqueness of a representing tuple morphism of standard form. On the layout side, this corresponds to shape entries $s_i = 1$ with stride $d_i \neq 0$. In order to exclude such pathological examples, we introduce the notion of {\it non-degeneracy}.
\end{example}

\begin{definition}\label{definitionofflatnondegenerate}
    Suppose
    \[
    \begin{tikzcd} 
(s_1,\dots,s_m) \ar[rr,"f"] \ar[rr,swap,"\alpha"] & & (t_1,\dots,t_n)
    \end{tikzcd} 
    \]
    is a tuple morphism and 
    \[
    L = (s_1,\dots,s_m):(d_1,\dots,d_m)
    \]
    is a flat layout.
    \begin{enumerate} 
    \item We say $f$ is {\it non-degenerate} if 
    \[
    s_i = 1 \quad \Rightarrow \quad \alpha(i) = *.
    \]
    \item We say $L$ is {\it non-degenerate} if 
    \[
    s_i = 1 \quad \Rightarrow \quad d_i = 0.
    \]
    \end{enumerate}
\end{definition}

\begin{observation}
    If $f$ is a non-degenerate tuple morphism, then the layout $L_f$ encoded by $f$ is non-degenerate. Conversely, if $L$ is a non-degenerate flat layout, then the standard representation $f_L$ of $L$ is non-degenerate.
\end{observation}

\begin{observation}
    Restricting to non-degenerate flat layouts is no real loss of generality. If $L$ is an arbitrary flat layout, then $\filter(L)$ is a non-degenerate flat layout with the same coordinate function and layout function as $L$.
\end{observation}

The essential property of non-degenerate tuple morphisms of standard form is that they are characterized by the layouts which they encode. This is made precise as follows.

\begin{lemma}\label{nondegeneratestandardformtuplemorphisms}
Suppose $f$ and $g$ are non-degenerate tuple morphisms of standard form. If $L_f = L_g$, then $f = g$. 
\end{lemma}

\begin{proof}
    Suppose
    \[\begin{tikzcd}
        (s_1,\dots,s_m) \ar[rr,"f"] \ar[rr,swap,"\alpha"] & & (t_1,\dots,t_n)
    \end{tikzcd}
    \]
    and 
    \[\begin{tikzcd}
        (s_1,\dots,s_m) \ar[rr,"g"] \ar[rr,swap,"\beta"] & & (u_1,\dots,u_p)
    \end{tikzcd}
    \]
    are non-degenerate tuple morphisms of standard form with
    \[
    L_f = (s_1,\dots,s_m):(d_1,\dots,d_m) = L_g.
    \]
    We want to show that $f = g$. First, we will argue that $(t_1,\dots,t_n) = (u_1,\dots,u_p)$. Let 
    \begin{align*}
    X & = \{ t_1 \cdots t_j \mid 1 \leq j \leq n\}\\
    Y & = \{u_1\cdots u_k \mid 1 \leq k \leq p\}
    \end{align*}
    denote the sets of prefix products of $(t_1,\dots,t_n)$ and $(u_1,\dots,u_p)$, respectively. We claim $X = Y$, since each of these sets is equal to 
    \[
    Z = \{d_i, s_id_i \mid 1 \leq i \leq m \text{ and }d_i \neq 0\}.
    \]
    Lets argue that $X = Z$. Suppose $1 \leq j \leq n$. If there exists some $i \in \langle m \rangle$ with $\alpha(i) = j$, then $t_1 \cdots t_j = s_id_i$. On the other hand, if $j$ is not in the image of $\alpha$, then since $f$ has standard form, there exists some $i \in \langle m \rangle$ such that $\alpha(i) = j+1$, in which case $t_1 \cdots t_j = d_i$. This proves that $X \subseteq Z$. Conversely, if $1 \leq i \leq m$ and $d_i \neq 0$, then $d_i = t_1 \cdots t_{\alpha(i)-1}$ and $s_id_i = t_1 \cdots t_{\alpha(i)}$, which proves $Z \subseteq X$. We deduce that $X = Z$. The same argument proves $Y = Z$.

    Since $f$ and $g$ are non-degenerate of standard form, we know that each $t_j$ and each $u_k$ is greater than $1$, which implies 
    \begin{align*}
    t_1  < t_1t_2 < & \cdots < t_1\cdots t_n,\\
    u_1  < u_1u_2 < & \cdots < u_1\cdots u_p,
    \end{align*}
    and since $X = Y$, it follows that $n = p$, and $t_1 \cdots t_j = u_1 \cdots u_j$ for each $1 \leq j \leq n$. We deduce that $(t_1,\dots,t_n) = (u_1,\dots,u_p)$. 

    Next, we need to argue that $\alpha = \beta$. Suppose for contradiction that there exists some $i \in \langle m \rangle$ with $\alpha(i) \neq \beta(i)$. There are two cases to consider. 
    \begin{itemize}
        \item If $\alpha(i) = * \neq \beta(i)$, then 
        \[
        0 = d_i = t_1 \cdots t_{\beta(i)-1},
        \]
        a contradiction. The case $\alpha(i) \neq * = \beta(i)$ is analogous.
        \item If $\alpha(i) \neq * \neq \beta(i)$, then without loss of generality we may assume $\alpha(i) < \beta(j)$, in which case 
        \[
        d_i = t_1 \cdots t_{\alpha(i)-1} < t_1 \cdots t_{\beta(i)-1} = d_i,
        \]
        a contradiction.
    \end{itemize}
    We deduce that $\alpha = \beta$, so $f = g$.
\end{proof}

We are now ready to prove our correspondence theorem, which identifies non-degenerate tuple morphisms of standard form with non-degenerate tractable flat layouts. 

\begin{theorem}\label{flatcorrespondence}
    The maps 
    \[ \begin{tikzcd} 
    f \ar[r,mapsto] & L_f \\
    \begin{Bmatrix}
    \text{Non-degenerate} \\
        \text{tuple morphisms of}\\
        \text{standard form}
    \end{Bmatrix}
    \ar[r] & \ar[l] 
    \begin{Bmatrix}
    \text{Non-degenerate}\\
        \text{tractable flat} \\
        \text{layouts}
    \end{Bmatrix} \\
    f_L & \ar[l,mapsto] L
    \end{tikzcd} 
    \]
    of Constructions \ref{layoutfrommorphisminTuple} and \ref{tuplemorphismfromlayout} determine a one-to-one correspondence between non-degenerate tuple morphisms of standard form, and non-degenerate tractable flat layouts.
\end{theorem}

\begin{proof}
    We want to show that the constructions $f \mapsto L_f$ and $L \mapsto f_L$ are inverses, when restricted to tuple morphisms and layouts of the stated form. If $L$ is a non-degenerate tractable flat layout, then by Lemma \ref{tuplemorphismfromlayoutagreement} we have $L_{f_L} = L$. Suppose next that $f$ is a non-degenerate tuple morphism of standard form and $L = L_f$ is the layout encoded by $f$. Since $f$ and $f_{L_f}$  are non-degenerate tuple morphisms of standard form, and the layouts encoded by these tuple morphsims are equal, it follows from Lemma \ref{nondegeneratestandardformtuplemorphisms} that $f = f_{L_f}$.
\end{proof}

\subsection{Examples}

In this section, we introduce some important families of tuple morphisms, and describe the flat layouts to which they give rise.

\begin{example}[Identity morphisms]
    We say a tuple morphism $f$ is an {\it identity morphism} if $f = \id_S$ for some tuple $S$. If $f = \id_S$ is an identity morphism, then $L_{f}$ is the column-major layout with shape $S$. For instance, here is an example of an identity morphism $f$
    together with its associated layout $L_f$.
    \[ 
    \begin{tikzcd}[row sep = 1, column sep = 8]
    & 4 \ar[rr,mapsto] & & 4 & & & & & \\
    & 4 \ar[rr,mapsto] & & 4 & & & & &  \\
    & 2 \ar[rr, mapsto] & & 2 & & & & & \\
    & 2 \ar[rr,mapsto] & & 2 & & & \rightsquigarrow & & L_f = (2,2,2,4,4):(1,2,4,8,32)\\
    & 2 \ar[rr,mapsto] & & 2 & & & & & \\
    & & f & & & & & & \\
    \end{tikzcd}
\] 
\end{example}

\begin{example}[Isomorphisms]
A tuple morphism $f:S \to T$ is an {\it isomorphism} if there is a tuple morphism $g:T \to S$ such that $g \circ f = \id_S$ and $f \circ g = \id_T$. If $f$ is an isomorphism, then its associated layout $L_f$ is compact. For instance, here is an isomorphism $f$ together with its associated layout $L_f$. 
    \[ 
    \begin{tikzcd}[row sep = 1, column sep = 8]
    & 4 \ar[rrd,mapsto] & & 2 & & & & & \\
    & 4 \ar[rrd,mapsto] & & 4 & & & & &  \\
    & 2 \ar[rruu, mapsto] & & 4 & & & & & \\
    & 2 \ar[rrd,mapsto] & & 2 & & & \rightsquigarrow & & L_f = (2,2,2,4,4):(2,1,64,4,16)\\
    & 2 \ar[rru,mapsto] & & 2 & & & & & \\
    & & f & & & & & & \\
    \end{tikzcd}
\]. 
\end{example}

% \begin{example}[Permutations] \label{definitionofpermutation} Suppose $S = (s_1,\dots,s_m)$ is a tuple of positive integers, and suppose $\sigma \in \Sigma_m$ is a permutation. Then there is a tuple morphism 
% \[\begin{tikzcd} (s_1,\dots,s_m) \ar[rr,"f_\sigma"] \ar[rr,swap,"\sigma_*"] & &  (s_{\sigma(1)},\dots,s_{\sigma(m)})
%     \end{tikzcd} \]
% lying over the pointed map $\sigma_*: \langle m \rangle_* \to \langle m \rangle_*$ associated to $\sigma$. We call a tuple morphism of this form a {\it permutation morphism}, or simply a {\it permutation}. If $f_\sigma$ is permutation, then the layout encoded by $f_\sigma$ is compact. For instance, if $\sigma = (1 \; 2)(3 \; 4) \in \Sigma_4$ and $S = (4,1,10,10)$, then $f_\sigma$ and $L_{f_\sigma}$ are as shown below. 
%     \[ 
%     \begin{tikzcd}[row sep = 1, column sep = 8]
%     & 10 \ar[rrd, mapsto] & & 10 & & & &  \\
%     & 10 \ar[rru, mapsto] & & 10 & & & & \\
%     & 32 \ar[rrd, mapsto] & & 4 & & \rightsquigarrow & & L_{f_\sigma} = (4,32,10,10):(32,1,128,1280) \\
%     & 4 \ar[rru, mapsto] & & 32 & & & & \\
%         & & f_\sigma & & & & & \\
%     \end{tikzcd}
%     \]
% \end{example}

\begin{observation}\label{isomorphismsvspermutations}\label{definitionofpermutation}
    Note that if a tuple morphism
    \[\begin{tikzcd} (s_1,\dots,s_m) \ar[rr,"f"] \ar[rr,swap,"\alpha"] & &  (t_1,\dots,t_n)
    \end{tikzcd} \]
    is an isomorphism, then $\alpha : \langle m \rangle_* \to \langle m \rangle_*$ is a bijection, and so $\alpha \mid_{\langle m \rangle} \in \Sigma_m$ is a permutation. Conversely, if $\sigma \in \Sigma_m$  is a permutation, and $(s_1,\dots,s_m)$ is a tuple of positive integers, then we may construct the isomorphism
    \[\begin{tikzcd} (s_1,\dots,s_m) \ar[rr,"f"] \ar[rr,swap,"\sigma_*"] & &  (s_{\sigma(1)},\dots,s_{\sigma(m)}).
    \end{tikzcd} \]
    We conclude that there is a one-to-one correspondence between tuple isomorphisms $f$ with domain $(s_1,\dots,s_m)$, and permutations in $\Sigma_m$. 
\end{observation}

\begin{example}[Projections] \label{definitionofprojection} Suppose $S = (s_1,\dots,s_m)$ is a shape, and suppose 
\[\{i_1<\dots<i_r\} \subset \langle m \rangle\]
is some subset. Let 
\[\begin{tikzcd} (s_1,\dots,s_m) \ar[rr,"p"] \ar[rr,swap,"\alpha"] & &  (s_{i_1},\dots,s_{i_r})
    \end{tikzcd} \]
be the tuple morphism lying over the map $\alpha$ with 
\[
\alpha(x) = 
\begin{cases}
    j & x = i_j\\
    * & \text{else.}
\end{cases}
\]
We call $p$ the {\it projection of } $(s_1,\dots,s_m)$ {\it onto } $(s_{i_1},\dots,s_{i_r})$. The layout encoded by $p$ is 
\[L_p = (s_1,\dots,s_m):(d_1,\dots,d_m),
\]
where
\[
d_i = \begin{cases}
    s_{i_1}\cdots s_{i_{j-1}} & i = i_j \text{ for some }1 \leq j \leq r\\
    0 & \text{otherwise.}\\
\end{cases}
\]
For instance, here is a projection $p$ of $(64,64,3,8)$ onto $(64,3)$, together with its associated layout.
    \[ 
    \begin{tikzcd}[row sep = 1, column sep = 8]
    & 8 & &  &  & & & & \\
    & 3 \ar[rrd, mapsto] & & & & & & & \\
    & 64 & & 3 & & & \rightsquigarrow & & L_p = (64,64,3,8):(1,0,64,0) \\
    & 64 \ar[rr,mapsto] & & 64 & & & & & \\
    & & p & & & & & & \\
    \end{tikzcd}
\]
\end{example}

\begin{example}[Dilations] \label{definitionofscaling} Suppose $S = (s_1,\dots,s_m)$ is a shape, and suppose $c_1,\dots,c_m$ are positive integers. The tuple morphism
\[\begin{tikzcd} (s_1,\dots,s_m) \ar[rrr,"f"] \ar[rrr,swap,"(*{,}2{,}*{,}4{,}\dots{,}*{,}2m)"] & & &  (c_1,s_1,\dots,c_m,s_m)
    \end{tikzcd} \]
is called the {\it dilation of }$(s_1,\dots,s_m)$ {\it by} $(c_1,\dots, c_m)$. The layout $L_f$ associated to this morphism is $L_f = (s_1,\dots,s_m):(d_1,\dots,d_m)$, where 
\[
d_i = \prod_{j < i} c_js_j.
\]
For instance, here is the dilation $f$ of $(512,512)$ by $(2,4)$, together with its associated layout.
    \[ 
    \begin{tikzcd}[row sep = 1, column sep = 8]
    &  & & 512 & & & & &  \\
    &  & & 4 & & & & & \\
    & 512 \ar[rruu,mapsto] & & 512 & & & \rightsquigarrow & & L_f = (512,512):(2,4096)\\
    & 512 \ar[rru,mapsto] & & 2 & & & & & \\
        & & f & & & & & & \\
    \end{tikzcd}
    \]
\end{example}

\begin{example}[Expansions] \label{definitionofcodomainexpansion} Suppose $S = (s_1,\dots,s_m)$ is a tuple of positive integers, and suppose $1 \leq i \leq m'$, so that $S' = (s_1,\dots,s_{m'})$ divides $S$. Then the tuple morphism 
\[\begin{tikzcd} (s_1,\dots,s_{m'}) \ar[rr,"e"] \ar[rr,swap,"(1{,}2{,}\dots{,}m')"] & &  (s_1,\dots,s_{m'},\dots,s_m)
    \end{tikzcd} \]
is called the {\it expansion} of  $S'$ {\it to} $S$. The layout encoded by $e$ is the column-major layout with shape $(s_1,\dots,s_{m'})$. For instance, here is the expansion of $S' = (4,4)$ to $S = (4,4,8,8)$. 
    \[ 
    \begin{tikzcd}[row sep = 1, column sep = 8]
    &  & & 8 & & & & &  \\
    &  & & 8 & & & & & \\
    & 4 \ar[rr,mapsto] & & 4 & & & \rightsquigarrow & & L_e = (4,4):(1,4)\\
    & 4 \ar[rr,mapsto] & & 4 & & & & & \\
        & & e & & & & & & \\
    \end{tikzcd}
    \]
An important property of expansions is that if $f:S \to T$ is any tuple morphism and $e: T \to T'$ is an expansion, then 
\[
L_{e \circ f} = L_f.
\]
In other words, post-composing $f$ with an expansion does not change the layout encoded by $f$.
\end{example}

\begin{example}[Restrictions]\label{definitionofrestriction}
Suppose 
\[\begin{tikzcd} (s_1,\dots,s_m) \ar[rr,"f"] \ar[rr,swap,"\alpha"] & &  (t_1,\dots,t_n)
    \end{tikzcd} \]
is a tuple morphism, and suppose  
\[I = \{i_1 < \cdots < i_r \} \subset \langle m \rangle\]
is a subset of indices. Then the tuple morphism 
\[\begin{tikzcd} (s_{i_1},\dots,s_{i_r}) \ar[rr,"f\mid_I"] \ar[rr,swap,"\alpha \circ \iota"] & &  (t_1,\dots,t_n)
    \end{tikzcd} \]
 is called the {\it restriction of }$f$ {\it to} $I$. If the layout encoded by $f$ is 
 \[L_f = (s_1,\dots,s_m):(d_1,\dots,d_m),\]
 then the layout encoded by $f \mid_I$ is 
\[L_{f \mid_I} = (s_{i_1},\dots,s_{i_r}):(d_{i_1},\dots,d_{i_r}).\]
For instance, here is the restriction $f \mid_I$ of a tuple morphism $f$, where $I = \{2,4\}$. 
    \[ \begin{tikzcd} [row sep = 1, column sep = 8]
    4 \ar[rrd,mapsto] & & & & & & \\
    8 \ar[ddrr,mapsto] & & 4 & & & & \\
    16 \ar[rr,mapsto] & & 16 & & \rightsquigarrow & & L_f = (2,16,8,4):(0,8,1,128)\\
    2 & & 8 & & & & \\
    & f & & & & & \\
    & & & \vspace{5.0in} & & & \\
    & & & & & &  \\
    & & 4 & & & & \\
    4 \ar[rru,mapsto] & & 16 & & \rightsquigarrow & & L_{f\mid_I} = (16,4):(8,128)\\
    16 \ar[urr,mapsto] & & 8 & &  & & \\
    & f\mid_{I} & & & & & \\
    \end{tikzcd} \]
\end{example}

\begin{example}[Entry inclusions]\label{definitionofentriesoftuplemorphism}
    An important special case of the previous construction is as follows. If $f:(s_1,\dots,s_m) \to (t_1,\dots,t_m)$ is a tuple morphism and $1 \leq i \leq m$, then the $i${\it th entry} $f_i$ of $f$ is 
\[\begin{tikzcd} (s_i) \ar[rr,"f_i"] \ar[rr,swap,"(i)"] & &  (t_1,\dots,t_n)
    \end{tikzcd} \]
    If the layout encoded by $f$ is 
    \[L_f = (s_1,\dots,s_m):(d_1,\dots,d_m),\]
    then the layout encoded by $f_i$ is 
    \[
    L_{f_i} = (s_i):(d_i).
    \]
    For instance, here is a tuple morphism $f$, and its fourth entry $f_4$.
    \[ \begin{tikzcd} [row sep = 1, column sep = 8]
    4 \ar[rrd,mapsto] & & & & & & \\
    8 \ar[ddrr,mapsto] & & 4 & & & & \\
    16 \ar[rr,mapsto] & & 16 & & \rightsquigarrow & & L_f = (2,16,8,4):(0,8,1,128)\\
    2 & & 8 & & & & \\
    & f & & & & & \\
    & & & \vspace{5.0in} & & & \\
    & & & & & &  \\
    & & 4 & & & & \\
    & & 16 & & \rightsquigarrow & & L_{f_4} = (4):(128)\\
    4 \ar[uurr,mapsto] & & 8 & &  & & \\
    & f_4 & & & & & \\
    \end{tikzcd} \]
\end{example}

\begin{remark}
    Given an $\n_* \in \Fin_*$, there is a morphism $\varphi_i: \langle 1 \rangle_* \to \n_*$ for each $i \in \n$ sending $* \mapsto *$ and $1 \mapsto i$. For a tuple morphism $f: (s_1, \dots, s_m) \to (t_1, \dots, t_n)$ lying over $\alpha: \m_* \to \n_*$, the $i$-th entry lies over the composite $\alpha \circ \varphi_i: \langle 1 \rangle_* \to \n_*$.
\end{remark}

\begin{example}[Factorizations]\label{definitionoffactorization}
Suppose 
\[\begin{tikzcd} (s_1,\dots,s_m) \ar[rr,"f"] \ar[rr,swap,"\alpha"] & &  (t_1,\dots,t_n)
    \end{tikzcd} \]
is a tuple morphism, and suppose 
\[J = \{j_1< \cdots < j_\ell\} \subset \langle n \rangle\]
is a subset such that $\Image(\alpha) \subseteq J \cup \{*\}$. If we write $\iota:\langle \ell \rangle_* \to \langle n \rangle_*$ for the map $k\mapsto j_k$, then $\alpha$ factors as $\alpha = \iota \circ \bar{\alpha}$ for a unique map $\bar{\alpha}: \langle m \rangle_* \to \langle \ell \rangle_*$, and we define the {\it factorization of }$f$ {\it through} $J$ to be the tuple morphism
\[\begin{tikzcd} (s_1,\dots,s_m) \ar[rr,"f\mid^J"] \ar[rr,swap,"\bar{\alpha}"] & &  (t_{j_1},\dots,t_{j_\ell}).
    \end{tikzcd} \]
If the layout encoded by $f$ is 
\[L_f = (s_1,\dots,s_m):(d_1,\dots,d_m),\]
then the layout encoded by $f \mid^J$ is 
\[L_{f \mid^J} = (s_1,\dots,s_m):(d_1',\dots,d_m'),\]
where 
\[
d_i' = \dfrac{d_i}{\left(\displaystyle\prod_{ k < \alpha(i)\text{ and } k \notin J }t_j\right)}.
\]
For instance, here is the factorization $f \mid^J$ of a tuple morphism $f$, where $J = \{2,4,5\}$. 
    \[ \begin{tikzcd} [row sep = 1, column sep = 8]
    & & 10 & & & & \\
    & & 8 & & & & \\
    & & 2 & & \rightsquigarrow & & L_f = (8,8):(32,2) \\
    8 \ar[rr,mapsto] & & 8 & & & & \\
    8 \ar[uuurr,mapsto] & & 2 & & & &  \\
    & f & & & & & \\
    & \vspace{5.0in} & & & & & \\
    & & 10 & & & & \\
    8 \ar[drr,mapsto] & & 8 & & \rightsquigarrow & & L_{f\mid^J} = (8,8):(8,1)\\
    8 \ar[urr,mapsto] & & 8 & & & & \\ 
    & f \mid^J & & & & & \\
    \end{tikzcd}\]
\end{example}

\begin{remark}
    There is a categorical interpretation of factorizations. Borrowing the notation of Example \ref{definitionoffactorization}, we may observe that there is a tuple morphism $i:(t_{j_1},\dots,t_{j_\ell}) \to (t_1,\dots,t_n)$ lying over $\iota$, and $f \mid^J$ is the pullback of $f$ along $i$:
    \[ \begin{tikzcd} 
    (s_1,\dots,s_m) \ar[r,"f \mid^J"] \ar[d,swap,"\id"] \arrow[dr, phantom, "\lrcorner", very near start] & (t_{j_1},\dots,t_{j_\ell}) \ar[d,"i"] \\
     (s_1,\dots,s_m) \ar[r,swap,"f"] & (t_1,\dots,t_n) \\
    \end{tikzcd} \]
\end{remark}

\subsection{Realization of tuple morphisms}
\noindent As we have seen, a tuple morphism $f:S \to T$ encodes a flat layout $L_f$. In this section, we will construct a realization functor
\[
| \cdot |:\catstyle{Tuple} \to \catstyle{FinSet}.
\]
which makes this encoding explicit. The realization functor $| \cdot |$ sends a tuple morphism $f$ to the layout function $|f|$ of $L_f$. In order to construct our realization functor $| \cdot |$, we first construct an auxiliary functor 
\[
F:\catstyle{Tuple} \to \catstyle{FinSet}
\]
which we will use in our construction. 
\begin{construction}\label{constructionofrealizationfunctoronC} We define a  functor 
\[F:\catstyle{Tuple} \to \catstyle{FinSet}\]
as follows. 
\begin{itemize} 
\item For an object $S = (s_1,\dots,s_m)$ in $\catstyle{Tuple}$, we define 
\[
FS = [0,S) = \prod_{i=1}^m  [0,s_i).
\]
\item For a morphism $f:(s_1,\dots,s_m) \to (t_1,\dots,t_n)$ in $\catstyle{Tuple}$ lying over $\alpha$, we define $Ff$ to be the map
\[\begin{tikzcd} 
{[}0,S{)} \ar[rr,"Ff"] & & {[}0,T{)}
\end{tikzcd} \]
given by
\[
(Ff)(x_1,\dots,x_m) = (y_1,\dots,y_n)
\]
where 
\[
y_j= \begin{cases}
    x_i & \text{there exists }1 \leq i \leq m\text{ with }\alpha(i) = j,\\
    0 & \text{else}.\\
\end{cases}
\]
\end{itemize}
\end{construction}

\noindent One may easily verify that $F(g \circ f) = Fg \circ Ff$ and $F \id_S = \id_{FS}$, so $F$ is in fact a functor. 

\begin{example} 
    Suppose $f:(4,4) \to (4,4,4)$ is the tuple morphism lying over $\alpha = (1,3)$. Then 
    \[
    Ff : [0,(4,4)) 
    \to [0,(4,4,4))
    \]
    is given by 
    \[
    (Ff)(x_1,x_2) = (x_1,0,x_2).
    \]
\end{example}

\begin{example}
    Suppose $g:(3,256,256,512) \to (3,256,256)$ is the tuple morphism lying over $\beta = (*,3,2,*)$. Then 
    \[
    Fg : [0,(3,256,256,512))
    \to [0,(3,256,256))
    \]
    is given by 
    \[
    (Fg)(x_1,x_2,x_3,x_4) = (0,x_3,x_2).
    \]
\end{example}

\begin{construction}\label{constructionofrealizationfunctoronC} We define a  functor 
\[|\cdot |:\catstyle{Tuple} \to \catstyle{FinSet}\]
as follows. 
\begin{itemize} 
\item For an object $S = (s_1,\dots,s_m)$ in $\catstyle{Tuple}$, we define 
\[
|S| = [0,\size(S)) = \{0,1,\dots,\size(S)-1\}.
\]
\item For a tuple morphism $f:S \to T$, we define
\[
|f| = \colex_T \circ Ff \circ \colex_S^{-1}
\]
(recall \Cref{definitionofcolex}).
\end{itemize}
\end{construction}

\noindent If $f:S \to T$ and $g:T \to U$ are composable tuple morphisms then 
\begin{align*}
    |g \circ f| & = \colex_U \circ F(g \circ f) \circ \colex_S^{-1}\\ 
    & = \colex_U \circ Fg \circ Ff \circ \colex_S^{-1} \\
    & = \colex_U \circ Fg \circ \colex_T^{-1} \circ \colex_T \circ Ff \circ \colex_S^{-1} \\
    & = |g| \circ |f|
\end{align*}
and if $f = \id_S$ is an identity morphism, then 
\begin{align*}
|\id_S| & = \colex_S \circ F \id_S \circ \colex_S^{-1}\\
& = \colex_S \circ \id_{S} \circ \colex_S^{-1}\\
& = \colex_S \circ \colex_S^{-1}\\
& = \id_{|S|}
~,
\end{align*}
so $| \cdot |$ does in fact specify a functor. Next, we observe that for a morphism $f$ in $\catstyle{Tuple}$, the map $|f|$ is the layout function of $L_f$. This allows us to easily deduce that composition of morphisms in $\catstyle{Tuple}$ is compatible with composition of flat layouts (see Corollary \ref{compatibilityofcompositioninC}).

\begin{lemma} \label{layoutfunctionlemma}
    If $f:S \to T$ is a tuple morphism, then the realization $|f|$ of $f$ is the layout function of $L_f$:
    \[ |f| = \Phi_{L_f}^{\size(T)}
    \]
\end{lemma}

\begin{proof}
    Let $S = (s_1,\dots,s_m)$, $T = (t_1,\dots,t_n)$, and let 
    \[L_f = (s_1,\dots,s_m):(d_1,\dots,d_m)\]
    denote the layout associated to $f$, whose strides $d_i$ are defined by the formula 
    \[d_i = \begin{cases} 
0 & \alpha(i) = * \\
\prod_{j < \alpha(i)} t_j & \text{else.}
\end{cases}
\]
    By precomposing with $\colex_S:\prod_{i=1}^m[0,s_i) \to [0,\size(S))$, it suffices to prove that for any $(x_1,\dots,x_m) \in \prod_{i=1}^m [0,s_i)$, we have 
    \[(\colex_T \circ Ff) (x_1,\dots,x_m) = (x_1,\dots,x_m) \cdot (d_1,\dots,d_m).\]
    For a general input $(x_1,\dots,x_m) \in \prod_{i=1}^m [0,s_i)$, we have 
    \[
    (Ff) (x_1,\dots,x_m) = (y_1,\dots,y_n)
    \]
    where $y_j$ is equal to $x_i$ if $\alpha(i) = j$, and $0$ otherwise. It follows that 
    \begin{align*}
    (\colex_T \circ Ff)(x_1,\dots,x_m) & = (y_1, \cdots ,y_n) \cdot (1,t_1 , \dots, t_1 \cdots t_{n-1}) \\
    & = \sum_{j=1}^n y_j \cdot t_1 \cdots t_{j-1} \\
    & = \sum_{i=1}^m x_i d_i \\
    & = (x_1,\dots,x_m) \cdot (d_1,\dots,d_m),
    \end{align*}
    as desired.
\end{proof}

\begin{corollary}\label{compatibilityofcompositioninC}
    If $f$ and $g$ are non-degenerate composable tuple morphisms, then 
    \[
    L_{g \circ f} = L_g \circ L_f
    \]
\end{corollary}

\begin{proof}
    Suppose $f:S \to T$ and $g:T \to U$ are morphisms in $\catstyle{Tuple}$ lying over $\alpha$ and $\beta$, respectively. Write $S = (s_1,\dots,s_m)$ and $T = (t_1,\dots,t_n)$. We need to check that 
    \begin{enumerate}
        \item {\it  $\shape(L_{g \circ f})$ refines $\shape(L_f)$}: This holds since the shape of $L_f$ and $L_{g \circ f}$ are both equal to $S$.
        \item {\it $L_{g \circ f}$ is coalesced over $\shape(L_f)$}: This holds since the tuple morphism $g \circ f$ is non-degenerate, hence so is the layout $L_{g \circ f}$.
        \item {\it $\Phi_{L_{g\circ f}} = \Phi_{L_g} \circ \Phi_{L_f}^{\size(L_g)}$}: Using Lemma \ref{layoutfunctionlemma}, we have 
        \begin{align*}
        \Phi_{L_{g \circ f}}^{\size(U)} &  = |g \circ f|\\
        & = |g| \circ |f|\\
        & = \Phi_{L_g}^{\size(U)} \circ \Phi_{L_f}^{\size(T)}.
        \end{align*}
        and by postcomposing with the inclusion $[0,\size(U)) \subset \mathbb{Z}$, and observing that $\size(T) = \size(L_{g})$, the result follows.
    \end{enumerate}
\end{proof}

\subsection{Operations on tuple morphisms}

\noindent Our next goal is to develop an ``algebra of tuple morphisms", which includes operations such as {\it coalesce}, {\it complement}, {\it composition}, {\it flat division}, and {\it flat products}. We will prove that each of these operations is compatible with a corresponding operation on flat layouts.

\subsubsection{Sum}
The sum $f \oplus g$ of tuple morphisms $f$ and $g$ is obtained by concatenating the domains and codomains of $f$ and $g$. In order to define this operations precisely, we first define a corresponding operation on morphisms in $\catstyle{Fin}_*$.

\begin{definition}
    Suppose $\alpha: \langle m \rangle_* \to \langle n \rangle_*$ and $\beta: \langle p \rangle_* \to \langle q \rangle_*$ are morphisms in $\catstyle{Fin}_*$. We define the {\it sum} of $\alpha$ and $\beta$ to be the morphism 
    \[
    \alpha \oplus \beta : \langle m + p \rangle_* \to \langle n + q \rangle_*
    \]
    given by
    \[
(\alpha\oplus\beta)(x) = \begin{cases}
    \alpha(x) & 1 \leq x \leq m\\
    n + \beta( x - m) & m+1 \leq x \leq m + p\\
    * & x = *.
\end{cases}
\]
This operation is associative, so we can consider the sum $\alpha_1 \oplus \cdots \oplus \alpha_k$ for any finite collection of morphisms $\alpha_1,\dots,\alpha_k$ in $\catstyle{Fin}_*$.
\end{definition}

\begin{remark}
    If $\alpha$ and $\beta$ are tractable pointed maps, then $\alpha \oplus \beta$ is tractable.
\end{remark}

\noindent Now we can define the sum of morphisms in $\catstyle{Tuple}$.

\begin{definition}\label{sumofmorphismsinC} Suppose $f:S \to T$ and $g:U \to V$ are tuple morphisms lying over $\alpha$ and $\beta$, respectively. We define the {\it sum} of $f$ and $g$ to be the tuple morphism
\[f \oplus g: S \star U \to T \star V\]
lying over $\alpha \oplus \beta$. This operation is associative, so we can consider the sum $f_1 \oplus \cdots \oplus f_k$ for any finite collection of morphisms $f_1,\dots,f_k$ in $\catstyle{Tuple}$.
\end{definition}

\begin{example} \label{concatenationexample}
    Here is an example of the sum $f \oplus g$ of tuple morphisms $f$ and $g$.
        \[ 
    \begin{tikzcd}[row sep = 1, column sep = 8]
    & & & & & & & & 4  & & 4 \\
      & & & & & & & &  4 \ar[rru,mapsto] & & 2\\
    32 \ar[rr, mapsto] & & 32 & & 4  & & 4 & & 32 \ar[rr,mapsto] & & 32\\
    16 \ar[rr,mapsto] & & 16 & & 4 \ar[rru ,mapsto]& & 2 & & 16 \ar[rr,mapsto] & & 16\\
    & f & & & & g & & & & f \oplus g  &  \\
    \end{tikzcd}
    \]
\end{example}

\begin{example}
    Here is another example of the sum $f \oplus g$ of tuple morphisms $f$ and $g$.
            \[ 
    \begin{tikzcd}[row sep = 1, column sep = 8]
    & & & & & & & & 64 \ar[rr,mapsto] & & 64 \\
    & & & & & & & & 512 \ar[rrd,mapsto]   & & 256 \\
      & & & & 64 \ar[rr,mapsto] & & 64& &  256 \ar[rru,mapsto] & & 512\\
    128 \ar[rrd, mapsto] & & 128 & & 512  \ar[rrd,mapsto] & & 256 & & 128 \ar[rrd,mapsto] & & 128\\
    128 \ar[rru,mapsto] & & 128 & & 256 \ar[rru ,mapsto]& & 512 & & 128 \ar[rru,mapsto] & & 128\\
    & f & & & & g & & & & f \oplus g  &  \\
    \end{tikzcd}
    \]
\end{example}

\begin{remark}
    There is a categorical interpretation of the sum of tuple morphisms: if $f:S \to T$ and $g:U \to V$ are tuple morphisms, then 
    \[
    f \oplus g: S \star U \to T \star V
    \]
    is the coproduct of $f$ and $g$ in the arrow category $\mathterm{Ar}(\catstyle{Tuple})$. 
\end{remark}

\subsubsection{Squeeze}

It is often the case that we want to remove any instances of the integer $1$ from our tuples. This is accomplished by the {\it squeeze} functor.

\begin{definition}\label{squeeze} 
We define a functor 
\[ \begin{tikzcd} 
\catstyle{Tuple} \ar[rr,"\squeeze(-)"] & & \catstyle{Tuple}
\end{tikzcd} \]
as follows. If $S = (s_1,\dots,s_m)$ is an object in $\catstyle{Tuple}$, we define
\[
\squeeze(S) = (s_{i_1},\dots,s_{i_k})
\]
where $\{i_1<\dots<i_k\} \subset  \langle m \rangle$ are the indices with $s_{i_j} \neq 1$. If $f:(s_1,\dots,s_m) \to (t_1,\dots,t_n)$ is a tuple morphism, we define 
\[
\squeeze(f):\squeeze(S) \to \squeeze(T)
\]
to be the tuple morphism
\[
\squeeze(f) = (f \mid_I )\mid^J
\]
where $f \mid_I$ is the restriction of $f$ to 
\[
I = \{ i \in \langle m \rangle \mid s_i \neq 1\}
\]
as in Definition \ref{definitionofrestriction},
and where $(f\mid_I)\mid^J$ is be the factorization of $f \mid_I$ through 
\[
J = \{ j \in \langle n \rangle \mid t_j \neq 1\},
\]
as in Definition \ref{definitionoffactorization}.
\end{definition}

\begin{example}\label{squeezeexample} Here is an example of a morphism $f$ and the corresponding morphism $\squeeze(f)$. 

\[ \begin{tikzcd}[row sep = 1, column sep = 12]
 & & 256 & & & & & & \\
256 \ar[rru ,mapsto] & & 128   & & & & & & 256 \\
128 \ar[rru,mapsto] & & 1& & & & & & 128 \\
1 & & 32 & & & & 256 \ar[uurr,mapsto] & & 32\\
8 \ar[drr,mapsto] & & 32 & & \rightsquigarrow & & 128 \ar[uurr,mapsto]  & & 32 \\
1 & & 8 & & & &  8 \ar[rr,mapsto] & & 8 \\
& f & & & & & & \mathclap{\squeeze(f)} & 
\end{tikzcd} \]
    
\end{example}

\begin{example}
    If $f:(s_1,\dots,s_m) \to (t_1,\dots,t_n)$ is a tuple morphism, then 
    \[
    f = \squeeze(f) \hspace{0.2in} \Leftrightarrow \hspace{0.2in} \text{no }s_i \text{, }t_j \text{ is equal to }1.
    \]
\end{example}

\begin{proposition}\label{squeezecompatibility}
    If $f$ is a tuple morphism, then 
    \[
    L_{\squeeze(f)} = \squeeze(L_f).
    \]
\end{proposition}

\begin{proof}
    Suppose $f:(s_1,\dots,s_m) \to (t_1,\dots,t_m)$ is a tuple morphism, and let 
    \[
    L_f = (s_1,\dots,s_m):(d_1,\dots,d_m)
    \]
    be the flat layout associated to $f$. Let $I = \{i_1<\cdots < i_{m'}\} \subset \langle m \rangle$ denote the subset of indices with $s_{i_k} \neq 1$. Then 
    \begin{align*}
    L_{f \mid_I} & = (s_{i_1},\dots,s_{i_k}):(d_{i_1},\dots,d_{i_k})\\
    & = \squeeze(L_f).
    \end{align*}
    
    \noindent Let $J = \{j_1<\dots<j_{n'}\} \subset \langle n \rangle$ denote the subset of indices with $t_{j_k} \neq 1$, so that $\squeeze(f) = \left(f \mid_I \right)\mid^J$. Let $\beta$ denote the map over which $\squeeze(f)$ lies. Then 
    \begin{align*}
        L_{\squeeze(f)} = L_{\left( f \mid_I \right)\mid^J} & = (s_{i_1},\dots,s_{i_k}) : (d_{i_1}',\dots,d_{i_k}')
    \end{align*}
    where 
    \begin{align*}
    d_{i_k}' & = \dfrac{d_{i_k}}{\left( \displaystyle\prod_{\ell < \beta(k) \text{ and }\ell \notin J } t_{\ell} \right)}\\
    & = d_{i_k}
    \end{align*}
    since $t_\ell = 1$ for any $\ell \notin J$. We conclude that 
    \begin{align*}
    L_{\squeeze(f)} & = (s_{i_1},\dots,s_{i_k}):(d_{i_1},\dots,d_{i_k})\\
    & = \squeeze(L_f).
    \end{align*}
\end{proof}

\begin{observation}
    If $f$ is a tuple morphism, then 
    \[
    \squeeze(\squeeze(f)) = \squeeze(f),
    \]
    so 
    \[
    \begin{tikzcd} 
    \catstyle{Tuple} \ar[rr,"\squeeze(-)"] & & \catstyle{Tuple}
    \end{tikzcd} \]
    is an idempotent functor.
\end{observation}

\subsubsection{Sort}
\label{subsubsection.sort}

The sort operation $f \mapsto \sort(f)$ permutes the domain of $f$ so that the resulting morphism is {\it sorted}, in the following sense.

\begin{definition}\label{definitionofsortedmorphisminC}
    We say a tuple morphism 
    \[\begin{tikzcd} (s_1,\dots,s_m) \ar[r,"f"] \ar[r,swap,"\alpha"] &  (t_1,\dots,t_n) \end{tikzcd} \]
    is {\it sorted} if for any $1 \leq i,j \leq m$, the following conditions hold.
    \begin{enumerate}
        \item If $\alpha(i) = * \neq \alpha(j)$, then $i < j$. 
        \item If $\alpha(i) = * = \alpha(j)$, then 
        \[
        i \leq j \quad \Rightarrow \quad s_i \leq s_j.
        \]
        \item If $\alpha(i) \neq * \neq \alpha(j)$, then 
        \[
        i \leq j \quad \Rightarrow \quad \alpha(i) \leq \alpha(j).
        \]

    \end{enumerate}
\end{definition}

\begin{example}\label{sortedmorphismsinCexample}
    The morphisms $f_1$, $f_2$, and $f_3$ shown below
    
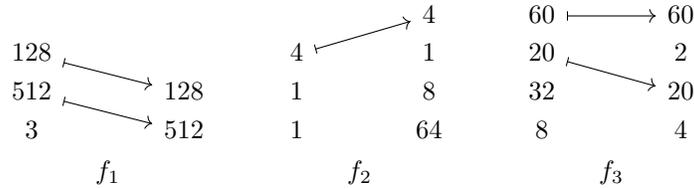
\[ \begin{tikzcd} [row sep = 1, column sep = 8]
     & & & & &  & & 4 & & & 60 \ar[rr,mapsto] & & 60\\
    128 \ar[drr,mapsto] & & & & & 4 \ar[urr,mapsto] & & 1 & & & 20 \ar[rrd,mapsto] & & 2\\
    512 \ar[drr,mapsto] & & 128 & & & 1 & & 8 & & & 32 & & 20\\
    3 & & 512 & & & 1 & & 64 & & & 8 & & 4 \\
     & f_1 & & & & & f_2 & & & & & f_3 & \\
    \end{tikzcd} \]
    are sorted, while the morphisms $g_1$, $g_2$, and $g_3$ shown below
    
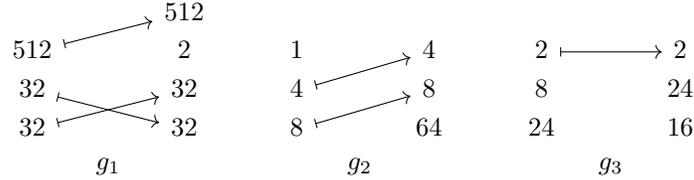
\[ \begin{tikzcd}[row sep = 1, column sep = 8]
     & & 512 & & & & & & & & & & \\
    512 \ar[urr,mapsto] & & 2 & & & 1  & & 4 & & & 2 \ar[rr,mapsto] & & 2 \\
    32 \ar[drr,mapsto] & & 32 & & & 4 \ar[urr,mapsto] & & 8 & & & 8 & & 24\\
    32 \ar[urr,mapsto] & & 32 & & & 8 \ar[urr,mapsto] & & 64 & & & 24 & & 16 \\
    & g_1 & & & & & g_2 & & & & & g_3 & \\
    \end{tikzcd} \]
    are not sorted. The morphisms $g_1$, $g_2$, and $g_3$ violate conditions $3$, $1$, and $2$, respectively.
\end{example}

\begin{proposition}\label{sortmorphisminCproposition}
    If $f$ is a sorted tuple morphism, then the flat layout $L_f$ is sorted. 
\end{proposition}

\begin{proof}
    Suppose  
    \[\begin{tikzcd} (s_1,\dots,s_m) \ar[r,"f"] \ar[r,swap,"\alpha"] &  (t_1,\dots,t_n) \end{tikzcd} \]
    is sorted, and consider the layout 
    \[L_f = (s_1,\dots,s_m):(d_1,\dots,d_m).\] Suppose $1 \leq i < m$. We want to show that $d_i < d_{i+1}$, or $d_i = d_{i+1}$ and $s_i \leq s_{i+1}$. There are two cases to consider.
    \begin{itemize} 
    \item (Case 1) Suppose that $\alpha(i) = *$, so that $d_i = 0$. If $\alpha(i+1) = *$, then $d_{i+1} = 0$ and since $f$ is sorted we have $s_i \leq s_{i+1}$. If $\alpha(i+1) \neq *$, then $d_{i+1} \geq 1 > 0 = d_i$. 
    
    \item (Case 2) Suppose that  $\alpha(i) \neq *$, in which case $\alpha(i+1) \neq *$ and $\alpha(i) < \alpha(i+1)$. Then 
        \[
        d_i = \prod_{j<\alpha(i)}t_j \leq \prod_{j < \alpha(i+1)} = d_{i+1},
        \]
    where equality holds only if $ s_i = t_{\alpha(i)} = 1$, which implies $s_i \leq s_{i+1}$.
    \end{itemize} 
    We conclude that $L_f$ is sorted.
    % Suppose next that $L_f = (s_1,\dots,s_m):(d_1,\dots,d_m)$ is sorted. We will argue that $f$ is sorted. Let $1 \leq i \leq m$. There are two cases to consider.
    % \begin{itemize}
    %     \item (Case 1) Suppose $i > 1$ and $\alpha(i) = *$. Then $d_i = 0$, and since $L_f$ is sorted we have $d_{i-1} \leq d_i = 0$. This implies that $d_{i-1} = 0$, so  $\alpha(i-1) = *$. Again, since $L_f$ is sorted, we must have $s_{i-1} \leq s_i$. 
    %     \item (Case 2) Suppose $i < m$ and $\alpha(i) \neq *$. Then $d_i > 0$ which implies that $d_{i+1} > 0$, and so $\alpha(i+1) \neq *$. We need to argue that $\alpha(i) < \alpha(i+1)$. This follows from the fact that 
    % \[
    % \prod_{j < \alpha(i)} t_j = d_i \leq d_{i+1} = \prod_{j<\alpha(i+1)} t_j,
    % \]
    % and $t_\alpha(i) = s_i > 1$.
    % \end{itemize} 
    % We conclude that $f$ is sorted, as desired.     
\end{proof}

\noindent Next, we define our $\sort(-)$ operation on $\catstyle{Tuple}$. If $f$ is a tuple morphism, then $\sort(f)$ will be obtained by precomposing $f$ with an appropriate permutation $g$.

\begin{construction}\label{sortofmorphisminC}
    Suppose 
    \[ \begin{tikzcd} 
    (s_1,\dots,s_m) \ar[r,"f"] \ar[r,swap,"\alpha"] &  (t_1,\dots,t_n)
    \end{tikzcd}\]
    is a tuple morphism. We define a permutation $\sigma \in \Sigma_m$ as follows. Set
    \begin{align*} P & = \{i \in \langle m \rangle \mid \alpha(i) = *\},\\
Q & = \{ i \in \langle m \rangle \mid \alpha(i) \neq *\},
    \end{align*}
    so $\langle m \rangle$ is the disjoint union of $P$ and $Q$. We define a linear ordering of $P$ by $i_1 \preceq_P i_2$ if 
    \begin{enumerate} 
    \item $s_{i_1} < s_{i_2}$, or 
    \item $s_{i_1} = s_{i_2}$ and $i_1 \leq i_2$.
    \end{enumerate}
    We define a linear ordering on $Q$ by $j_1 \preceq_Q j_2$ if $\alpha(i_1) \leq \alpha(i_2)$. We define a linear ordering on $\langle m \rangle$ by $i_1 \preceq i_2$ if 
    \begin{enumerate} 
    \item $i_1 \in P$ and $i_2 \in Q$,
    \item $i_1,i_2 \in P$ and $i_1 \preceq_P i_2$, or 
    \item $i_1,i_2 \in Q$ and $i_1 \preceq_Q i_2$. 
    \end{enumerate}
    Let $\sigma$ be permutation associated to the linear ordering $\preceq$ of $\langle m \rangle$, and let $\sigma^{-1}$ be its inverse.  The map $\sigma^{-1}_*:\langle m \rangle_* \to \langle m \rangle_*$ is covered by a tuple morphism 
    \[
    g: (s_{\sigma^{-1}(1)},\dots,s_{\sigma^{-1}(m)}) \to (s_1,\dots,s_m),
    \]
    and we define $\sort(f)$ to be the composite 
    \[
    \sort(f) = f \circ g.
    \]
\end{construction}

\begin{example} The sortings of the morphisms $g_1$, $g_2$, and $g_3$ of Example \ref{sortedmorphismsinCexample} are shown below.
    \[ \begin{tikzcd} [row sep = 1, column sep = 16]
    & & 512 & & & & & & 512\\
    512 \ar[urr,mapsto] & & 2 & & & & 512 \ar[urr,mapsto] & & 2\\
    32 \ar[drr,mapsto] & & 32 & & \rightsquigarrow & & 32 \ar[rr,mapsto] & & 32\\
    32 \ar[urr,mapsto] & & 32 & & & & 32 \ar[rr,mapsto] & & 32\\
    & g_1 & & & & & & \mathclap{\sort(g_1)} & \\
 & \vspace{0.2in} & & &  & & & & \\
        1 & & 4 & & & & 4\ar[rr,mapsto]  & & 4\\
    4 \ar[urr,mapsto] & & 8 & & \rightsquigarrow & & 8 \ar[rr,mapsto] & & 8\\
    8  \ar[urr,mapsto] & & 64 & & & & 1 & & 64\\
    & g_2 & & & & & & \mathclap{\sort(g_2)} & \\
     & \vspace{0.2in} & & &  & & & & \\
     & & & & & & & & \\
        2  \ar[rr,mapsto] & & 2 & & & & 2 \ar[rr,mapsto]  & & 2\\
    8  & & 24 & & \rightsquigarrow & & 24  & & 24\\
    24    & & 16 & & & & 8 & & 16\\
    & g_3 & & & & & & \mathclap{\sort(g_3)} & \\
    \end{tikzcd} \]
\end{example}

\begin{lemma}
    Suppose $f:S \to T$ is a tuple morphism. Then $f$ is sorted if and only if $\sort(f) = f$.
\end{lemma}
\begin{proof}
    Our construction of $\sort(-)$ guarantees that $\sort(f)$ is sorted for any tuple morphism $f$. In particular, if $f = \sort(f)$, then $f$ is sorted. Conversely, if $f$ is sorted, then the permutation $\sigma \in \Sigma_m$ from Construction \ref{sortofmorphisminC} is the identity permutation, so $g = \id_S$, and so
    \[
    \sort(f) = f \circ \id_S = f.
    \]
\end{proof}

\begin{proposition}\label{sortofmorphisminCcompatibility} If $f$ is a tuple morphism, then 
\[
L_{\sort(f)} = \sort(L_f).
\]
\end{proposition}

\begin{proof}
    Borrowing our notation form Construction \ref{sortofmorphisminC}, we have $\sort(f) = f \circ g$ where 
    \[
    g : (s_{\sigma^{-1}(1)},\dots,s_{\sigma^{-1}(m)}) \to (s_1,\dots,s_m)
    \]
    lies over $\sigma^{-1}_*: \langle m \rangle_* \to \langle m \rangle_*$. If $L_f = (s_1,\dots,s_m):(d_1,\dots,d_m)$, then 
    \begin{align*}
    L_{\sort(f)} & = (s_1',\dots,s_m'):(d_1',\dots,d_m') \\
    & = (s_{\sigma^{-1}(1)},\dots,s_{\sigma^{-1}(m)}):(d_{\sigma^{-1}(1)},\dots,d_{\sigma^{-1}(m)}). 
    \end{align*}
    Since the modes of $L_{\sort(f)}$ are a permutation of the modes of $L_f$, it suffices to prove that $L_{\sort(f)}$ is sorted. Suppose $1 \leq i <m$. Suppose first that $\sigma^{-1}(i) \in P$, so that $d_i' = d_{\sigma^{-1}(i)} = 0$. If $\sigma^{-1}(i+1) \in P$, then $d'_{i+1} = d_{\sigma^{-1}(i+1)} = 0$. By construction of $\sigma$, we have $s_i' = s_{\sigma^{-1}(i)} \leq s_{\sigma^{-1}(i+1)} = s_{i+1}'$. If instead $\sigma^{-1}(i+1) \in Q$, then $d_{i+1}' = d_{\sigma^{-1}(i+1)}' > 0 = d_i'$. Suppose next that $\sigma^{-1}(i) \in Q$. Then by construction of $\sigma$, we have $\sigma^{-1}(i+1) \in Q$ and $\alpha(\sigma^{-1}(i)) < \alpha(\sigma^{-1}(i+1))$, and we have 
    \begin{align*}
    d_i' = d_{\sigma^{-1}(i)} & = \prod_{j < \alpha(\sigma^{-1}(i))} t_j \\
    & \leq \prod_{j<\alpha(\sigma^{-1}(i+1))} t_j \\
    & = d_{\sigma^{-1}(i+1)}\\
    & = d_{i+1}',
    \end{align*}
    where equality holds if and only if $t_{\alpha(\sigma^{-1}(i))} =  \cdots = t_{\alpha(\sigma^{-1}(i+1))-1} = 1$. In particular, we have $s_i = s_{\sigma^{-1}(i)} = t_{\alpha(\sigma^{-1}(i))} = 1$, and so $s_i' \leq s_{i+1}'$. We conclude that $L_{\text{sort(f)}}$ is sorted, so $L_{\sort(f)} = \sort(L_f)$. 
\end{proof}

\begin{remark}
    The operation $\sort(-)$ is not functorial. For example, consider the tuple morphisms 
    \[ \begin{tikzcd} (2,3)\ar[r,"f"] \ar[r,swap,"(2{,}1)"] & (3,2) & \text{ and } & (10,25)\ar[r,"g"] \ar[r,swap,"(2{,}1)"] & (25,10)\end{tikzcd} \]
    Then $f$ and $g$ are composable with $g \circ f = \id_{(25,10)}$, but the sorted morphisms 
    \[ \begin{tikzcd} (10,25)\ar[rr,"\sort(f)"] \ar[rr,swap,"(1{,}2)"] & & (10,25) & \text{ and } & (25,10)\ar[rr,"\sort(g)"] \ar[rr,swap,"(1{,}2)"] & & (25,10)\end{tikzcd} \]
    are not composable.
\end{remark}

\subsubsection{Coalesce}
\label{subsubsection.coalesce}

\noindent We begin by introducing the notion of a coalesced tuple morphism.
\begin{definition}\label{definitionofcoalescedmorphisminC}
    Suppose $f:S \to T$ is a tuple morphism lying over $\alpha$. We say $f$ is {\it coalesced} if 
    \begin{enumerate}
        \item $S = \squeeze(S)$ and 
        \item for any $ 1 \leq i < \len(S)$, exactly one of the following conditions holds: 
        \begin{enumerate}
            \item $\alpha(i) = * \neq \alpha(i+1)$,
            \item $\alpha(i) \neq *=\alpha(i+1)$,
            \item $\alpha(i) > \alpha(i+1)$, or 
            \item $\alpha(i)<\alpha(i+1)$, and there exists $ \alpha(i) < j < \alpha(i+1)$ with $t_j > 1$. 
        \end{enumerate}
    \end{enumerate}
\end{definition}

\begin{example}
    If there exists some $1 \leq i < \len(S)$ with $\alpha(i+1) = \alpha(i)+1$, then $f$ is not coalesced.
\end{example}

\begin{remark}
    If $f:S \to T$ is a tuple morphism such that $f = \squeeze(f)$, then $f$ is coalesced if and only if for any $1 \leq i < \len(S)$, one of the following conditions holds:
    \begin{enumerate}
            \item $\alpha(i) = * \neq \alpha(i+1)$,
            \item $\alpha(i) \neq * = \alpha(i+1)$,
            \item $\alpha(i) > \alpha(i+1)$, or 
            \item $\alpha(i+1) \neq \alpha(i) + 1$. 
    \end{enumerate}
\end{remark}

\begin{example}\label{coalescedmorphismsinCexample}
    The morphisms
    \[ \begin{tikzcd} [row sep = 1, column sep = 8]
    & & & & & & & 4 & & & & & \\
     & & 48 & & & 4 \ar[urr,mapsto] & & 32 & & & & 2 \\
    48 \ar[rru,mapsto] & & 8 & & & 2 & & 8 & & 2 \ar[rrdd,mapsto] & & 8 \\
    32 & & 16 & & & 32 \ar[uurr,mapsto] & & 16 & & 32 \ar[rr,mapsto] & & 32 \\
    16 \ar[rru,mapsto] & & 2 & & & 16 \ar[urr,mapsto] & & 8 & & 16 & & 2 \\
    & f_1 & & & & & f_2 & & & &  f_3 & \\
    \end{tikzcd} \]
    are coalesced, while the morphisms 
    \[ \begin{tikzcd} [row sep = 1, column sep = 8]
    & & & & & & & 4 & & & & & \\
     & & 8 & & & 4 \ar[urr,mapsto] & & 32 & & & & 2 \\
    32 & & 48 & & & 2 & & 1 & & 1 & & 8 \\
    48 \ar[urr,mapsto] & & 16 & & & 32 \ar[uurr,mapsto] & & 16 & & 32 \ar[rr,mapsto] & & 32 \\
    16 \ar[rru,mapsto] & & 2 & & & 16 \ar[urr,mapsto] & & 8 & & 16 & & 1 \\
    & g_1 & & & & & g_2 & & & &  g_3 & \\
    \end{tikzcd} \]
    are not coalesced.
\end{example}

\begin{proposition}\label{coalescedmorphismcoalescedlayoutproposition}
    Suppose $f$ is a tuple morphism. Then $f$ is coalesced if and only if $L_f$ is coalesced. 
\end{proposition}
\begin{proof} Suppose $f:(s_1,\dots,s_m) \to (t_1,\dots,t_n)$ is a tuple morphism, and let 
\[L_f = (s_1,\dots,s_m):(d_1,\dots,d_m)\]
be the layout encoded by $f$. 

Suppose first that $f$ is coalesced. Then no entry of $ \shape(L_f) = \domain(f)$ is equal to $1$. Suppose $1 \leq i < m$. We want to show that $s_id_i$ is not equal to $d_{i+1}$. If $d_i = 0$, then $\alpha(i) = *$, and we have 
\begin{align*}
    s_id_i = d_{i+1} & \hspace{0.2in} \Leftrightarrow \hspace{0.2in} d_{i+1} = 0\\
    & \hspace{0.2in} \Leftrightarrow \hspace{0.2in} \alpha(i+1) = *
\end{align*}
but by our assumption that $f$ is coalesced, we have $\alpha(i+1) \neq *$, hence $s_id_i \neq d_{i+1}$. If $d_i \neq 0$, then $\alpha(i) \neq *$. If $\alpha(i+1) = *$, then $d_{i+1} = 0$, so $s_id_i \neq d_{i+1}$. If $\alpha(i+1) < \alpha(i)$, then $d_i \geq d_{i+1}$, and since $s_i \neq 1$, we have $s_id_i > d_{i+1}$. Finally, if $\alpha(i) < \alpha(i+1)$, then 
\begin{align*}
s_id_i = s_i \cdot \left( \prod_{j< \alpha(i)} t_j\right) & = \prod_{j \leq \alpha(i)} t_j \\
& < \prod_{j < \alpha(i+1)} t_j\\
& = d_{i+1}.
\end{align*}
We conclude that $L_f$ is coalesced. 

Suppose next that the layout $L_f$ is coalesced. Then no entry in $\domain(f) = \shape(L_f)$ is equal to $1$. Suppose $1 \leq i < m$. If $\alpha(i) = *$, then $d_i = 0$, and since $L_f$ is coalesced, we must have $d_{i+1} \neq s_id_i = 0$, hence $\alpha(i+1) \neq *$. Suppose $\alpha(i) \neq *$, and $\alpha(i)<\alpha(i+1)$. Since $L_f$ is coalesced, we have $s_id_i \neq d_{i+1}$. But if we write  
\[
s_i d_i = \prod_{j \leq \alpha(i)} t_j,
\]
and 
\[
d_{i+1} = \prod_{j<\alpha(i+1)} t_j,
\]
this implies that $\prod_{\alpha(i) < j < \alpha(i+1)}t_j \neq 1$. In particular, there exists some $\alpha(i)<j<\alpha(i+1)$ with $t_j >1$. We conclude that $f$ is coalesced. 
\end{proof}

Next, we define our $\coalesce(-)$ operation on tuple morphisms.

\begin{construction}\label{constructionofcoalesceofmorphisminC}
    Suppose $f$ is a tuple morphism. We define a morphism $\coalesce(f)$ as follows:
    \begin{enumerate}
    \item First, we set $g = \squeeze(f)$, and we write $\beta:\langle m \rangle_* \to \langle n \rangle_*$ for the map over which $g$ lies. 
        \item Next, we define an equivalence relation $\sim$ on $\langle m \rangle$ where $i \sim i'$ if either
        \begin{enumerate}
            \item $\beta(i'') = *$ for $i \leq i'' \leq i'$, or 
            \item $\beta(i'') = \beta(i) + (i'' - i)$ for $i \leq i'' \leq i'$. 
        \end{enumerate}
        The quotient $\langle m \rangle / \sim$ is ordered by $[i_1] \leq [i_2]$ if $i_1 \leq i_2$, so we can identify this quotient with  $\langle \bar{m} \rangle$ where $\bar{m}$ is the size of $\langle m \rangle / \sim$. 
        \item Next, define an equivalence relation $\sim$ on $\langle n \rangle$ where $j \sim j'$ if there exists $i \in \langle m \rangle$ such that 
        \[\beta(i + (j''-j)) = \beta(i) + (j'' - j)\]
        for all $j \leq j'' \leq j'$. The quotient $\langle n \rangle / \sim$ is ordered by $[j_1] \leq [j_2]$ if $j_1 \leq j_2$, so we can identify this quotient with $\langle \bar{n} \rangle$ where $\bar{n}$ is the size of $\langle n \rangle / \sim$. 
        \item Next, we observe that the map $\beta:\langle m \rangle_* \to \langle n \rangle_*$ descends to a map 
        \[
        \bar{\beta}: \langle \bar{m} \rangle_* \to \langle \bar{n} \rangle_*
        \]
        given by $\bar{\beta}([i]) = [\beta(i)]$.
        \item The domain $\bar{S} = (\bar{s}_1,\dots,\bar{s}_{\bar{m}})$ of $\coalesce(f)$ is defined by setting 
        \[
        \bar{s}_i = \prod_{i' \in I} s_{i'}
        \]
        if $i \in \langle \bar{m} \rangle$ corresponds to the equivalence class $I \in \langle m \rangle / \sim$. The codomain 
        $\bar{T} = (\bar{t}_1,\dots,\bar{t}_{\bar{n}})$ of $\coalesce(f)$ is defined by setting 
        \[
        \bar{t}_j = \prod_{j' \in J} t_{j'}
        \]
        if $j \in \langle \bar{n} \rangle$ corresponds to the equivalence class $J \in \langle n \rangle / \sim$. We then define \[\coalesce(f):\bar{S} \to \bar{T}\] to be the tuple morphism lying over $\bar{\beta}$. 
    \end{enumerate}
\end{construction}
\color{black}

\begin{example}
    Here is an example of a tuple morphism $f$ and the coalesced morphism $\coalesce(f)$.
    \[ \begin{tikzcd} [row sep = 1, column sep = 8]
    7 \ar[rr,mapsto] & & 7 & & & & & & \\
    5 & & 2 & & & & & & \\
    5 & & 3 & & & & & & \\
    3 \ar[urr,mapsto] & & 3 & & & & & & 7\\
    3 \ar[urr,mapsto] & & 2 & & & & 7 \ar[urr,mapsto] & & 2\\
    2 \ar[rr,mapsto] & & 2 & & & & 25 & & 9\\
    2 \ar[rr,mapsto] & & 2 & & \rightsquigarrow  & & 9 \ar[urr,mapsto] & & 2\\
    2 \ar[rr,mapsto] & & 2 & & & & 8 \ar[rr,mapsto] & & 8\\
    & f &  & & & & & \mathclap{\coalesce(f)} &\\
    \end{tikzcd} \]
\end{example}

\begin{example} We can coalesce the morphism $f$ of Example \ref{squeezeexample} as follows
\[ \begin{tikzcd}[row sep = 1, column sep = 8]
 & & 256 & & & & & & \\
256 \ar[rru ,mapsto] & & 128   & & & & & & 256 \\
128 \ar[rru,mapsto] & & 1& & & & & & 128 & & & & & & 32768 \\
1 & & 32 & & & & 256 \ar[uurr,mapsto] & & 32 & & & & & & 32\\
8 \ar[drr,mapsto] & & 32 & & \rightsquigarrow & & 128 \ar[uurr,mapsto]  & & 32 & & \rightsquigarrow & & 32768 \ar[uurr,mapsto] & & 32\\
1 & & 8 & & & &  8 \ar[rr,mapsto] & & 8 & & & & 8 \ar[rr,mapsto] & & 8\\
& f & & & & & & \mathclap{\squeeze(f)} & & & & & & \mathclap{\coalesce(f)} & 
\end{tikzcd} \]
\end{example}

\begin{proposition}\label{coalesceinCproposition}
    If $f$ is a tuple morphism, then 
    \begin{enumerate}
    \item $\coalesce(f)$ is coalesced, and 
    \item $ L_{\coalesce(f)} = \coalesce(L_f).$
    \end{enumerate}
\end{proposition}

\begin{proof}
First, we will argue that $\coalesce(f)$ is coalesced. This is immediate from our construction, since applying $\squeeze$ eliminates all modes equal to $1$, and passing to the quotient in our construction consolidated all adjacent modes with $\alpha(i+1) = \alpha(i)+1$. 

    Next, we will prove that $L_{\coalesce(f)} = \coalesce(L_f)$. In light of Proposition \ref{coalesceproposition} and Proposition \ref{coalescedmorphismcoalescedlayoutproposition}, it suffices to prove that $\Phi_{\coalesce(f)} = \Phi_f$. Certainly applying $\squeeze(-)$ to $f$ has no impact on the associated layout function, so we need to argue that passing to the quotient in our construction does not change the layout function of the associated layout. This follows from the fact that forming our quotient can be formed in steps, where in each step we combine adjacent modes with either $\alpha(i) = * = \alpha(i+1)$, or $\alpha(i+1) = \alpha(i) + 1$. These correspond to replacing adjacent modes of the form $s_i,s_{i+1}:0,0$ with $s_is_{i+1}:0$, and $s_i,s_{i+1}:d_i,s_id_i$ with $s_is_{i+1}:d_i$, respectively. Neither such operation changes the layout function of a layout, and so we conclude that $\Phi_{L_{\coalesce(f)}} = \Phi_{\coalesce(L_f)}$, as desired.
\end{proof}

\subsubsection{Concatenate}
\label{subsubsection.concatenate}

Next, we will define a concatenation operation on tuple morphisms. This operation may be performed on tuple morphisms satisfying a ``disjointness" condition, which we specify below.

\begin{definition}\label{definitionofdisjointimagesinFin*}
    Suppose $\alpha:\langle m \rangle_* \to \langle n \rangle_*$ and $\beta:\langle p \rangle_* \to \langle n \rangle_*$ are morphisms in $\catstyle{Fin}_*$ with the same codomain. We say $\alpha$ and $\beta$ have {\it disjoint images} if 
    \[
    \Image(\alpha) \cap \Image(\beta) = \{*\}.
    \]
\end{definition}

\begin{construction}
    If $\alpha:\langle m \rangle_* \to \langle n \rangle_*$ and $\beta: \langle p \rangle_* \to \langle n \rangle_*$ have disjoint images, then we have a well-defined morphism
    \[
    \alpha \star \beta: \langle m + p \rangle_* \to \langle n\rangle_*
    \]
    given by 
    \[
    (\alpha \star \beta)(i) = \begin{cases}
    * & i = *\\
        \alpha(i) & 1 \leq i \leq m\\
        \beta(i-m) & m+1 \leq i \leq m+p.
    \end{cases}
    \]
    This operation is associative, so we can consider $\alpha_1 \star \cdots \star \alpha_k$ for any collection of morphisms $\alpha_1,\dots,\alpha_k$ in $\catstyle{Fin}_*$ with pairwise disjoint images.
\end{construction}

\begin{remark}
    If $\alpha$ and $\beta$ are tractable pointed maps and $\alpha$ and $\beta$ have disjoint images, then $\alpha \star \beta$ is tractable. 
\end{remark}

\begin{definition}\label{definitionofdisjointimagesinC} Suppose 
\[
f:S \to T
\]
and 
\[g:U \to T\]
are tuple morphisms lying over $\alpha$ and $\beta$, respectively. We say $f$ and $g$ have {\it disjoint images} if the morphisms $\alpha$ and $\beta$ have disjoint images. 
\end{definition}

\begin{example}\label{exampleofdisjointimages}
    Consider the tuple morphisms $f$, $g$, and $h$ shown below. 
    
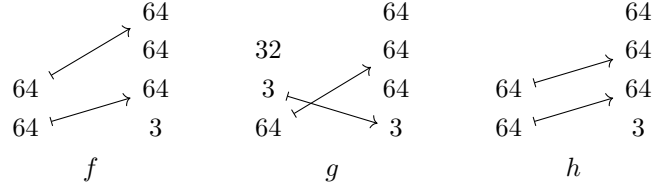
\[ \begin{tikzcd}[row sep = 1, column sep = 8]
    & & 64 & & & & & 64 & & & & & 64\\
     & & 64 & & & 32 & & 64 & & & & & 64\\
    64 \ar[rruu,mapsto] & & 64 & & &  3 \ar[rrd, mapsto] & & 64 & & & 64 \ar[urr,mapsto] & & 64\\
    64 \ar[rru,mapsto]& & 3 & & &  64 \ar[rruu,mapsto] & & 3 & & & 64 \ar[urr,mapsto] & & 3\\
    & f & & & & & g & & & & & h & \\
    \end{tikzcd} \]
    Then $f$ and $g$ have disjoint images, while $h$ and $g$ do not have disjoint images.
\end{example}

\begin{construction}\label{constructionofflatconcatenation}
    Suppose 
    \[f:S \to T \text{, and }g:U \to T\]
    are tuple morphisms lying over $\alpha$ and $\beta$, respectively, and that $f$ and $g$ have disjoint images. We define the {\it concatenation} of $f$ and $g$ to be the morphism
    \[
    f \star g: S \star U \to T
    \]
    lying over $\alpha \star \beta$. This operation is associative, so we can consider $f_1\cdots f_k$ for any finite collection of morphisms $f_i$ with pairwise disjoint images. 
\end{construction}

\begin{example}\label{exampleofconcatenation}
    If $f$ and $g$ are the morphisms in $\catstyle{Tuple}$ from Example \ref{exampleofdisjointimages}, then the concatenation of $f$ and $g$ is the morphism shown below. 
    \[ \begin{tikzcd} [row sep = 1, column sep = 24]
    32 & & \\
    3 \ar[dddrr,mapsto]& & 64\\
    64 \ar[rr,mapsto] & & 64\\
    64 \ar[uurr,mapsto] & & 64 \\
    64 \ar[urr,mapsto] & & 3 \\
    \end{tikzcd} 
    \]
    \[ f \star g \]
\end{example}

\begin{example}
Suppose $f:(s_1,\dots,s_m) \to (t_1,\dots,t_n)$ is a tuple morphism, and for any $1 \leq i \leq m$, let 
\[f_i:(s_i) \to (t_1,\dots,t_n)\]
denote the $i$th entry of $f$, as in Example \ref{definitionofentriesoftuplemorphism}. Then we can write 
\[
f = f_1 \star \cdots \star f_m
\]
as the concatenation of its entries.
\end{example}

\begin{lemma}\label{compositiondistributesoverconcatenation}
    Suppose $f_1:S_1 \to T$ and $f_2: S_2 \to T$ are tuple morphisms with disjoint images. If $g:T \to U$ is any tuple morphism, then 
    \[
    g \circ (f_1 \star f_2) = (g \circ f_1) \star (g \circ f_2).
    \]
\end{lemma}

\begin{proof}
    Suppose $f_1$, $f_2$, and $g$ lie over $\alpha_1:\langle m_1 \rangle_* \to \langle n \rangle$, $\alpha_2:\langle m_2 \rangle_* \to \langle n \rangle$, and $\beta: \langle n \rangle \to \langle p \rangle$, respectively. The two maps in question have the same domains and the same codomains, so it suffices to prove that
    \[
    \beta \circ (\alpha_1 \star \alpha_2) = (\beta \circ \alpha_1) \star (\beta \circ \alpha_2). 
    \]
    We compute 
    \begin{align*}
        (\beta \circ (\alpha_1 \star \alpha_2))(i) & = \beta((\alpha_1 \star \alpha_2)(i))\\
        & = \begin{cases} 
        \beta(*) & i = *\\
        \beta(\alpha_1(i)) & 1 \leq i \leq m_1 \\
        \beta (\alpha_2(i-m_1)) & m_1 + 1 \leq i \leq m_1 + m_2
        \end{cases}\\
        & = \begin{cases} 
        * & i = *\\
        (\beta\circ \alpha_1)(i) & 1 \leq i \leq m_1 \\
        (\beta \circ \alpha_2)(i-m_1) & m_1 + 1 \leq i \leq m_1 + m_2
        \end{cases}\\
        & = ((\beta \circ \alpha_1)\star (\beta \circ \alpha_2))(i).
    \end{align*}
\end{proof}

\begin{proposition}\label{concatenateinCproposition}
    Suppose $f_1,\dots,f_k$ are morphisms in $\catstyle{Tuple}$ with the same codomain and with pairwise disjoint images. Then the layouts $L_{f_1},\dots,L_{f_k}$ satisfy 
    \[
    L_{f_1 \star \cdots \star  f_k} = 
    L_{f_1} \star \cdots \star  L_{f_k}.
    \]
\end{proposition}
\begin{proof}
    First, we prove the result for $k = 2$. Suppose 
    \[f = (s_1,\dots,s_m) \to (t_1,\dots,t_n)\text{, and } g:(u_1,\dots,u_p) \to (t_1,\dots,t_n)\]
    have disjoint images, and write 
    \[L_f = (s_1,\dots,s_m):(d_1,\dots,d_m)\text{, and }L_g = (u_1,\dots, u_p) :(d_1',\dots,d_{p}').\]
    Then the layout $L_{f \star g}$ is given by 
    \[
    L_{f \star g} = (s_1,\dots,s_m,u_1,\dots,u_p): (e_1,\dots,e_{m+m'})
    \]
    where 
    \begin{align*}
    e_i & = \prod_{j< (\alpha \star \beta)(i)} t_j \\
    & = \begin{cases}
        \displaystyle\prod_{j<\alpha(i)}t_j & 1 \leq i \leq m\\
        \displaystyle\prod_{j<\beta(i-m)} t_j & m+1 \leq i \leq m+m'.
    \end{cases}\\
    & = \begin{cases}
        d_i & 1 \leq i \leq m\\
        d_{i-m}' & m+1 \leq i \leq m+m'.
    \end{cases}
    \end{align*}
    This concludes the proof of the result when $k = 2$. The general case follows from the associativity of concatenation of tuple morphisms, and the associativity of concatenation of flat layouts.
\end{proof}

\subsubsection{Complement}
\label{subsubsection.complement}

We begin by defining the notion of complementary tuple morphisms.

\begin{definition}\label{definitionofcomplementarymorphismsinC} Suppose $f:S \to T$ and $g:U \to T$ are tuple morphisms. We say $g$ is a {\it complement} of $f$ if 
\begin{enumerate} 
\item $f$ and $g$ have disjoint images, and 
\item the concatenation
\[ \begin{tikzcd} f \star g:S \star U \ar[r,"\cong"] &  T \end{tikzcd} \] is an isomorphism.
\end{enumerate}
\end{definition}

\begin{example}
    If $f$ and $g$ are the morphisms shown below
    \[ \begin{tikzcd} [row sep = 1, column sep = 8]
     & & 16 & & & & & & 16 \\
     & & 32 & & & &  & & 32\\ 
    32 \ar[rru,mapsto] & & 32 & & & & 10 \ar[rrd, bend left = 20, mapsto] & & 32\\
    32 \ar[rru,mapsto] & & 10 & & & & 16 \ar[rruuu, mapsto]  & & 10 \\
     & f & & & & & & g & 
    \end{tikzcd} \]
    then $g$ is a complement of $f$.
\end{example}

\begin{example}
    If $f$ is the morphism shown below
    
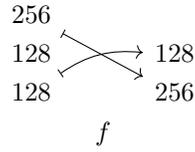
\[\begin{tikzcd}[row sep = 1, column sep = 8]
    256 \ar[rrdd,mapsto] & & \\
    128 & & 128\\ 
    128 \ar[rru,bend left = 20,mapsto] & & 256 \\
    & f & \\
    \end{tikzcd} \]
    then $f$ does not admit a complement. 
\end{example}

\noindent Next, we prove that complementary tuple morphisms give rise to complementary flat layouts.

\begin{proposition}\label{complementofmorphisminCproposition}
    If $f:S \to T$ is a tuple morphism and $g$ is a complement of $f$, then $L_{g}$ is a $\size(T)$-complement of $L_f$.
\end{proposition}

\begin{proof}
    Write $S = \domain(f)$, $U = \domain(g)$, and $T = \codomain(f) = \codomain(g)$. First, we note that 
    \begin{align*}
    \size(L_f) \cdot \size(L_{g}) & = \size(L_f\star L_{g})\\
    & = \size(L_{f \star g})\\
    & = \size(S \star U)\\
    & = \size(T).
    \end{align*}
    Next, we note that $f \star g$ is an isomorphism, hence so is
    \[
    |f \star g| = \Phi_{L_{f \star g}}^{\size(T)} 
    \]
    where we have used the identification of $\Phi^{\size(T)}_{L_{f \star g}}$ of Lemma \ref{layoutfunctionlemma}.
\end{proof}

\begin{proposition}\label{coalescedlayoutoftuplemorphismcomplement}
    If $f$ is an injective tuple morphism, then 
    \[
    \coalesce^\flat(L_{f^c}) = \comp^\flat(L_f,\size(T)).
    \]
\end{proposition}

\begin{proof}
    By Proposition \ref{complementofmorphisminCproposition}, we know that $L_{f^c}$ is a $\size(T)$-complement of $L_f$. Since $f^c$ is sorted, so is $L_{f^c}$ and it follows from Proposition \ref{characterizationofflatNcomplement}, it follows that 
    \[
    \coalesce^\flat(L_{f^c}) = \comp^\flat(L_f,\size(T)),
    \]
    since both of these layouts are flat, sorted, coalesced complements of $L_f$ of the same size. 
\end{proof}

\begin{proposition}
    If $f:(s_1,\dots,s_m) \to (t_1,\dots,t_n)$ is an injective tuple morphism of standard form, then 
    \[
    L_{f^c} = \comp^\flat(L_f).
    \]
\end{proposition}

\begin{proof}
    Write 
    \[
    L_f = (s_1,\dots,s_m):(d_1,\dots,d_m)
    \]
    for the layout encoded by $f$. By Proposition \ref{complementofmorphisminCproposition}, we know that $L_{f^c}$ is a $\size(T)$-complement of $L_f$. Where 
    \begin{align*}
    \size(T) = t_1 \cdots t_n & = (t_1 \cdots t_{n-1})t_n\\
    & = d_m s_m.
    \end{align*}
    By construction, $f^c$ is sorted, hence so is $L_{f^c}$. Moreover, since $f$ has standard form, it follows that $f^c$ is coalesced. By Proposition \ref{characterizationofflatcomplement}, we deduce that 
        \[
    L_{f^c} = \comp^\flat(L_f).
    \]
\end{proof}

\begin{definition}\label{definitionofcomplementabletuplemorphism}
    Suppose $f$ is a tuple morphism lying over $\alpha:\langle m \rangle_* \to \langle n \rangle_*$. We say $f$ is ${\it complementable}$ if $\alpha$ is injective.
\end{definition}

\begin{construction}\label{complementofmorphisminC}
    Suppose $f:(s_1,\dots,s_m) \to (t_1,\dots,t_n)$ is a complementable tuple morphism. Let $j_1 < \cdots < j_{n-m}$ denote the collection of indices in $\langle n \rangle$ which are not in the image of $\alpha$. We define the {\it complement of }$f$ to be the tuple morphism
    \[
    f^c :(t_{j_1},\dots,t_{j_k}) \to (t_1,\dots,t_n)
    \]
    lying over the map $\complement(\alpha): \langle n-m \rangle_* \to \langle n \rangle_*$ given by $k \mapsto j_k$. By construction, we may observe that $f^c$ is a complement of $f$, in the sense of Definition \ref{definitionofcomplementarymorphismsinC}
\end{construction}

\begin{example}
    Below is an example of a morphism $f$ and its complement $ f^c$. 
        \[ \begin{tikzcd} [row sep = 1, column sep = 16]
     & & 512 & & & & & & 512 \\
     & & 512 & & & &  & & 512\\ 
    512 \ar[rru,mapsto] & & 256 & & & & 512 \ar[rruu, mapsto] & & 256\\
    256 \ar[rru,mapsto] & & 10 & & & & 10\ar[rr, mapsto]  & & 10 \\
     & f & & & & & & f^c & 
    \end{tikzcd} \]
\end{example}

\begin{proposition}
    If $f$ is a tuple morphism and $g$ is a complement of $f$, then 
    \[
    \sort(g) = f^c.
    \]
\end{proposition}
\begin{proof}
    Suppose $f$ lies over $\alpha:\langle m \rangle_* \to \langle n \rangle_*$, $\sort(g)$ lies over $\beta: \langle n - m \rangle_* \to \langle n \rangle_*$ and $f^c$ lies over $\alpha^c:\langle n - m \rangle_* \to \langle n \rangle_*$. Then $\beta$ and $\alpha^c$ are increasing maps with the same image, namely 
    \[\Image(\beta) = \langle n \rangle \setminus \Image(\alpha) = \Image(\alpha^c),\]
    hence $\beta = \alpha^c$, and hence $\sort(g) = f^c$. 
\end{proof}

\begin{proposition}
    Suppose $f$ is a tuple morphism. Then $f$ admits a complement if and only if $f$ is complementable, in the sense of Definition \ref{definitionofcomplementabletuplemorphism}.
\end{proposition}

\begin{proof}
    If $f$ lies over a map $\alpha$ which is not injective, then for any morphism $f^*$ such that $f$ and $f^*$ have disjoint images, the morphism $f \star f^*$ lies over a map which is not injective, hence $f \star f^*$ is not an isomorphism. Conversely, if $f$ lies over an injective map, then the morphism $f^c$ of Construction \ref{complementofmorphisminC} is a complement of $f$.
\end{proof}

\begin{proposition}
    If $f$ is a complementable tuple morphism, then 
    \[
    \sort(f) = (f^c)^c.
    \]
\end{proposition}

\begin{proof}
    Both maps are increasing, injective, and have the same image, so they are equal.
\end{proof}

\subsubsection{Flat division}
\label{subsubsection.flatdivision}

In this section, we define a division operation on tuple morphisms. 

\begin{definition}
    If $f$ and $g$ are tuple morphisms, we say $g$ {\it divides } $f$ if $g$ and $f$ are composable. In other words,
    \[
    \codomain(g) = \domain(f).
    \]
\end{definition}

\begin{definition}
    Suppose $g: S \to T$ and $f:T \to U$ are tuple morphisms. The {\it flat division} of $f$ by $g$ is the tuple morphism 
    \[
    f \oslash^\flat g = f \circ (g \star g^c).
    \]
\end{definition}

\begin{example}
Here is an example of tuple morphisms $f$ and $g$ together with their flat quotient $f \oslash^\flat g$.
    \[ \begin{tikzcd} [row sep = 1, column sep = 8]
     & & 128& & & 128 \ar[rr,mapsto] & & 128 & & & 2 \ar[drr,mapsto] & & 128\\
    128 \ar[rru,mapsto] & & 2 & & &  2 \ar[rr,mapsto] & & 2 & & & 128 \ar[urr,mapsto] & & 2\\
    & g & & & & & f & & & & & f \oslash^\flat g & 
    \end{tikzcd} \]
\end{example}

\begin{example}
    Here is an example of tuple morphisms $f$ and $g$ together with their flat quotient $f \oslash^\flat g$.
    \[ \begin{tikzcd} [row sep = 1, column sep = 8]
     & &  5 \ar[rr,mapsto] & & 5 & & & 2 \ar[ddrr,mapsto] & & 5\\
     & &  5 \ar[rr,mapsto] & & 5 & & & 2 \ar[ddrr,mapsto] & & 5\\
    5 \ar[uurr,mapsto] & &  2 \ar[rr,mapsto] & & 2 & & & 5 \ar[uurr,mapsto] & & 2\\
    5 \ar[uurr,mapsto] & &  2 \ar[rr,mapsto] & & 2 & & & 5 \ar[uurr,mapsto] & & 2\\
    & g & & f & & & & & f\oslash^\flat g & 
    \end{tikzcd} \]
\end{example}

\begin{example}
   Here is an example of tuple morphisms $f$ and $g$ together with their flat quotient $f \oslash^\flat g$.
    \[ \begin{tikzcd} [row sep = 1, column sep = 8]
    & & 512 \ar[drr,mapsto] & & & & & 512 \ar[drr,mapsto] & & \\
     & & 2 & & 512 & & & 2 & & 512\\
     & & 8  & & 4 & & & 2 \ar[ddrr,mapsto] & & 4\\
    8 \ar[urr,mapsto] & & 2  \ar[drr,mapsto]  & & 8 & & & 8 & & 8\\
    8 \ar[rr,mapsto] & & 8 \ar[urr,mapsto] & &  2 & & & 8 \ar[urr,mapsto]  & & 2\\
    & g & & f & & & &  & f \oslash^\flat g & 
    \end{tikzcd} \]
\end{example}

% \begin{proposition}
%     Suppose $f$, $g$, and $h$ are tuple morphisms such that $h$ divides $g$ and $g$ divides $f$. Then $g / h$ divides $f$, $h$ divides $f/g$, and 
%     \[f/(g/h) = (f/g)/h.\]
% \end{proposition}
% \begin{proof}
%     First, we will argue that $g/h$ divides $f$. We have 
%     \[g/h = g \circ (\bar{h} \star \bar{h}^c),\]
%     which is the composite of injective functions, so $g/h$ is injective. Moreover, the codomain of $g/h$ is the codomain of $g$, which by assumption, divides the domain of $f$. We conclude that $g/h$ divides $f$.

%     Next, we will argue that $h$ divides $f/g$. By assumption, we have that $h$ is injective. Moreover, the codomain of $h$ divides the domain of $g$, which divides the domain of $\bar{g} \star \bar{g}^c$. This is the same as the domain of $f/g$, so be the transitivity of tuple divisibility, it follows that the codomain of $h$ divides the domain of $f/g$. We conclude that $h$ divides $g/f$. 

%     Finally, we will argue that $f/(g/h) = (f/g)/h$. In order to do so, let's fix some notation. Let's write 
%     \[
%     h: S \to T',\hspace{ 0.1in} g:T \to U',\text{ and } f:U \to V.
%     \]
%     with $U = U'\star U''$ and $T = T' \star T''$. Let's also write
%     \[
%     \bar{h}:S \to T \text{, and } \bar{g}:T \to U
%     \]
%     for the relevant codomain expansions of $h$ and $g$. Then 
%     \begin{align*}
%         (f/g)/h & = (f/g) \circ (\bar{h} \circ \bar{h}^c)
%     \end{align*}
%     INCOMPLETE
% \end{proof}

\begin{proposition}
    If $f$ and $g$ are non-degenerate composable tuple morphisms, then 
    \[
    \coalesce^\flat(L_{f \oslash^\flat g}) = \coalesce^\flat(L_f \oslash^\flat L_g)
    \]
\end{proposition}

\begin{proof}
By Proposition \ref{coalescedlayoutoftuplemorphismcomplement}, we have
\[
\coalesce^\flat(L_{g^c}) = \comp^\flat(L_g,\size(L_f)),
\]
and we compute
    \begin{align*}
    \coalesce^\flat(L_f \oslash^\flat L_g) & = \coalesce^\flat(L_f \circ \left(L_g \star \comp(L_g,\size(L_f))\right))\\
    & = \coalesce(L_f \circ \left( L_{g} \star L_{g^c} \right))\\
    & = \coalesce(L_f \circ L_{g\star g^c})\\
    & = \coalesce(L_{f \circ (g \star g^c)})\\
    & = \coalesce(L_{f \oslash^\flat g}).
    \end{align*}
\end{proof}

\subsubsection{Flat products}
\label{subsubsection.flatproduct}

In this section we define a product operation on tuple morphisms. 
\begin{definition}
Suppose $f$ and $g$ are tuple morphisms. We say $f$ and $g$ are {\it product admissible} if $\codomain(g) = \domain(f^c)$. If $f$ and $g$ are product admissible, then we define {\it flat product} of $f$ and $g$ to be  
    \[
    f \otimes^\flat g = f \star (f^c \circ g).
    \]
\end{definition}

\begin{example}
    If $f$ and $g$ are the tuple morphisms shown below
    \[ \begin{tikzcd} [row sep = 1, column sep = 8]
      & & & & & & & 16\\
      & & & & & & & 16\\
     16 \ar[rr,mapsto] & & 16 & & & 8 \ar[rr,mapsto] & & 8\\
    16 \ar[rr,mapsto] & & 16 & & & 8 \ar[rr,mapsto] & & 8\\
    & g & & & & & f & \\
    \end{tikzcd} \]
    then 
    $f$ and $g$ are product-admissible, and $f \otimes^\flat g$ is the tuple morphism shown below.
    \[ \begin{tikzcd}[row sep = 1, column sep = 8]
    16 \ar[rr,mapsto] & & 16\\
    16 \ar[rr,mapsto] & & 16 \\
    8 \ar[rr,mapsto] & & 8\\
    8 \ar[rr,mapsto] & & 8 \\
    & f \otimes^\flat g & 
    \end{tikzcd} \]
\end{example}

\begin{example}
    If $f$ and $g$ are the tuple morphisms shown below
    \[ \begin{tikzcd} [row sep = 1, column sep = 8]
      & & & & & & & 128\\
      & & & & & & & 128\\
      & & 32 & & & 128 \ar[uurr,mapsto] & & 32\\
    32 \ar[urr,mapsto] & & 32 & & & 128 \ar[uurr,mapsto] & & 32\\
    & g & & & & & f & \\
    \end{tikzcd} \]
    then 
    $f$ and $g$ are product-admissible, and $f \otimes^\flat g$ is the tuple morphism shown below.
    \[ \begin{tikzcd}[row sep = 1, column sep = 8]
     & & 128\\
    32 \ar[drr,mapsto] & & 128 \\
    128 \ar[uurr,mapsto] & & 32\\
    128 \ar[uurr,mapsto] & & 32 \\
    & f \otimes^\flat g & 
    \end{tikzcd} \]
\end{example}

\begin{lemma}\label{complementdistributesoverproducts}
    If $f$ and $g$ are product admissible and $g$ is injective, then $f \otimes^\flat g$ is injective and 
    \[
    (f \otimes^\flat g)^c = f^c \circ g^c.
    \]
\end{lemma}
\begin{proof}
    The tuple morphisms $(f \otimes^\flat g)^c$ and $f^c \circ g^c$ are injective, increasing, and have the same codomain, so it suffices to show that they have the same image. The image of $(f \otimes^\flat g)^c = (f \star (f^c \circ g) )^c$ consists of those entries which are not in the image of $f$, and not in the image of $f^c \circ g$. The image of $f^c$ consists of those entries which are not in the image of $f$, and so the image of the composition $f^c \circ g^c$ consists of those entries which are not in the image of $f$, and not in the image of $f^c \circ g$. 
\end{proof}

\begin{proposition}
    Suppose $f$ and $g$ are product admissible, and $g$ and $h$ are product admissible. Then
    \begin{enumerate}
        \item $f \otimes^\flat g$ and $h$ are product admissible,
        \item $f$ and $g \otimes^\flat h$ are product admissible, and
        \item $(f \otimes^\flat g) \otimes^\flat h = f \otimes^\flat (g \otimes^\flat h)$.
    \end{enumerate}
\end{proposition}

\begin{proof}
    Using Lemma \ref{compositiondistributesoverconcatenation} and Lemma \ref{complementdistributesoverproducts}, we compute 
    \begin{align*}
    f \otimes^\flat (g \otimes^\flat h) & = f \star (f^c \circ (g \otimes^\flat h))\\
    & = f \star (f^c \circ (g \star (g^c \circ h) ) )\\
    & = f \star ( (f^c \circ g) \star (f^c \circ (g^c \circ h) ) ) \\
    & = f \star ( (f^c \circ g) \star ((f^c \circ g^c) \circ h ) ) \\
    & = f \star (f^c \circ g) \star ( (f \otimes^\flat g)^c \circ h) \\
    & = (f \otimes^\flat g) \star ( (f \otimes^\flat g)^c \circ h) \\
    & = (f \otimes^\flat g) \otimes^\flat h.
    \end{align*} 
\end{proof}

\begin{proposition}
    Suppose $f$ and $g$ are non-degenerate tuple morphisms and that $f$ and $g$ are product admissible. Then 
    \[
    L_{f \otimes^\flat g} = L_f \otimes^\flat L_g.
    \]
\end{proposition}

\begin{proof} Suppose $f:S \to T$ and $g:U \to V$ are product admissible, and set
    \[
    L_f^* = \comp^\flat(L_f,\size(L_f)\cdot \cosize(L_g)).
    \]
    Since $f$ is injective and the codomain of $g$ is the domain of $f^c$, it follows that 
    \[\size(L_f) \cdot \cosize(L_g) \leq \size(S) \cdot \size(V) = \size(T).\]
    Using this fact, and the fact that \[\Phi_{\comp(L_f,\size(T))} = \Phi_{L_{f^c}},\]
    we have
    \begin{align*}
    L_f^* \circ L_g & = \comp(L_f,\size(T)) \circ L_g\\
    & = L_{f^c} \circ L_g.
    \end{align*}
    Using this fact, we compute 
    \begin{align*}
        L_f \otimes^\flat L_g & = L_f \star (L_f^* \circ L_g)\\
        & = L_f \star (L_{f^c} \circ L_g)\\
        & = L_f \star L_{f^c \circ g}\\
        & = L_{f \star (f^c \circ g)}\\
        & = L_{f \otimes^\flat g}
    \end{align*}
\end{proof}

\newpage

\newpage
\section{The category $\catstyle{Nest}$}\label{thecategoryNestTuplesection}

In the previous section, we introduced a category $\catstyle{Tuple}$, whose morphisms encode flat tractable layouts. In this section, we introduce a category $\catstyle{Nest}$, whose morphisms encode tractable layouts with arbitrary nesting.

\subsection{Basic definitions} 

Recall that for a nested tuple $S$, we write $S^\flat$ for the flattening of $S$. For example, if $S = (64,(8,8))$, then $S^\flat = (64,8,8)$.

\begin{definition}\label{definitionofNest}
    Let $\catstyle{Nest}$ denote the category whose objects are nested tuples of positive integers, and in which a morphism 
    \[f:S \to T\]
    in $\catstyle{Nest}$ is specified by a tuple morphism
    \[
    f^\flat:S^\flat \to T^\flat.
    \]
In other words, 
\[
\Hom_{\catstyle{Nest}}(S,T) = \Hom_{\catstyle{Tuple}}(S^\flat,T^\flat).
\]
Explicitly, a morphism $f:S \to T$ in $\catstyle{Nest}$ is specified by a tractable pointed map $\alpha:\langle \len(S) \rangle_* \to \langle \len(T) \rangle_*$ satisfying the following property:
\begin{itemize}
    \item If $1 \leq i \leq \len(S)$ and $\alpha(i) \neq *$, then $\entry_i(S) = \entry_{\alpha(i)}(T)$. 
\end{itemize}
    We say such a morphism $f$ {\it lies over} $\alpha$, and refer to $f$ as a {\it nested tuple morphism}.
\end{definition}

\begin{notation}
    If $f:S \to T$ is a nested tuple morphism which lies over $\alpha$, we depict $f$ as 
    \[ \begin{tikzcd} 
    S \ar[rr,"f"] \ar[rr,swap,"\alpha"] & & T
    \end{tikzcd} \]
\end{notation}

\begin{example}
    Here are some examples of nested tuple morphisms. 
    \[ \begin{tikzcd} 
    (64, (8,8)) \ar[rr,"f"] \ar[rr,swap,"(1{,}2{,}3)"] & &  (64, 8,8) \\
    ((2,2),2) \ar[rr,"g"] \ar[rr,swap,"(*{,}5{,}2)"] & &  (10,2,2,(3,2,3)) \\
    64 \ar[rr,"h"] \ar[rr,swap,"(2)"] & &  ((64,64),512).
    \end{tikzcd} \]
\end{example}

\begin{observation}
    If $X$ is a set, lets write $X^{\mathterm{ind}}$ for the indiscrete category on $X$. This is the category whose objects are the elements of $X$, and in which there is a unique (iso)morphism between any two objects. Then by definition of $\catstyle{Nest}$, we have a pullback square 
    \[\begin{tikzcd}
        \catstyle{Nest} \ar[rr,"\profile(-)"] \ar[d,swap,"(-)^\flat"] \arrow[drr, phantom, "\lrcorner", very near start]& &  \mathterm{Profile}^{\mathterm{ind}}\ar[d,"\len(-)"] \\
        \catstyle{Tuple} \ar[rr,swap,"\len(-)"] & & \mathbb{N}^{\mathterm{ind}} 
    \end{tikzcd}\]
    We may view this as a categorification of the pullback square \ref{nesttuplepullbacksquare}. 
\end{observation}

\begin{example}
    Suppose $S$ is a nested tuple of length $m$. If $ 1 \leq i \leq m$ then there is a nested tuple morphism 
    \[
    \entry_i(S) \to S
    \]
    lying over the map $\langle 1 \rangle_* \to \langle m \rangle_*$ given by $1 \mapsto i$. For instance, if $S = (64,(8,8))$ and $i = 1$, then we have a nested tuple morphism
    \[ \begin{tikzcd} 
    64 \ar[rr,swap,"(1)"] & &  (64,(8,8)).
    \end{tikzcd} \]
\end{example}

\begin{example}
    Suppose $S$ is a nested tuple of rank $r$. If $1 \leq i \leq r$, then there is a canonical nested tuple morphism 
    \[
    \mode_i(S) \to S
    \]
    lying over the map $\langle \len_i(S) \rangle_* \to \langle \len(S) \rangle_*$ given by $j \mapsto j + \len_{<i}(S)$. For instance, if $S = (64,(8,8))$, then we have a nested tuple morphism 
    \[ \begin{tikzcd} 
    (8,8) \ar[rr,swap,"(2{,}3)"] & & (64,(8,8)).
    \end{tikzcd} \]
\end{example}

\begin{observation} There are functors relating the categories $\catstyle{Nest}$ and $\catstyle{Tuple}$. First, there is an inclusion functor
    \[ \begin{tikzcd}
    \catstyle{Tuple} \ar[r,"\subset"] &  \catstyle{Nest}
    \end{tikzcd} \]
    which considers a tuple morphism $f:S \to T$ as a nested tuple morphism. Next, there is a {\it flattening} functor
    \[ \begin{tikzcd} 
    \catstyle{Nest} \ar[r,"(-)^\flat"] & \catstyle{Tuple}
    \end{tikzcd} \]
    which sends a nested tuple morphism $f:S \to T$ to the underlying tuple morphism $f^\flat:S^\flat \to T^\flat$. The composite 
    \[ \begin{tikzcd} 
    \catstyle{Tuple} \ar[r,"\subset"] & \catstyle{Nest} \ar[r,"(-)^\flat"] & \catstyle{Tuple}
    \end{tikzcd} \]
    is the identity functor on $\catstyle{Tuple}$, so $\catstyle{Tuple}$ is a retractive subcategory of $\catstyle{Nest}$. Moreover, these functors form an adjoint equivalence of categories.
\end{observation}

\begin{remark}
    One might wish to consider some category $\catstyle{C}$ whose morphisms encode tractable layouts, but which is {\it not} equivalent to $\catstyle{Tuple}$. The authors have considered several such examples, but leave their investigation to future work.
\end{remark}

\subsection{From nested tuple morphisms to layouts}

The key feature of the category $\catstyle{Nest}$ is that if $f:S \to T$ is a nested tuple morphism, then $f$ encodes a layout $L_f$. This layout is obtained by equipping the flat layout $L_{f^\flat}$ with the nesting profile of $S$. More precisely, we have the following construction.

\begin{construction}\label{layoutfromnestedtuplemorphism}
    Suppose
    \[
    f:S \to T
    \]
    is a nested tuple morphism, and suppose $P = \profile(S)$. We define $L_f$
    to be the layout
    \[
    L_f = (L_{f^\flat})_P
    \]
    where $(-)_P$ is the $P$-substitution operation of Definition \ref{layoutfromflatlayoutandprofile}. We refer to $L_{f}$ as the {\it layout encoded by} $f$.
\end{construction}

\begin{construction}\label{layoutfromnestedtuplemorphism}
    Suppose 
\[\begin{tikzcd} (s_1,\dots,s_m)_P \ar[rr,"f"] \ar[rr,swap,"\alpha"] & &  (t_1,\dots,t_n)_Q\end{tikzcd} \]
is a nested tuple morphism. We define $L_f$ to be the layout whose shape \[\shape(L_f) = (s_1,\dots,s_m)_P\]
is the domain of $f$, and whose stride \[\stride(L_f) = (d_1,\dots,d_m)_P\] has entries defined by the formula
\[d_i = \begin{cases} 
0 & \alpha(i) = * \\
\prod_{j < \alpha(i)} t_j & \alpha(i) \neq *.
\end{cases}
\]
We refer to $L_f$ as the {\it layout encoded by} $f$.
\end{construction}

\begin{example}
    The layout encoded by 
    \[ \begin{tikzcd} 
    ((8,8),(4,4)) \ar[rr,"f"] \ar[rr,swap,"(1{,}4{,}3{,}2)"] & & (8,4,4,8)
    \end{tikzcd} \]
    is 
    \[
    L_f = ((8,8),(4,4)):((1,128),(32,8)).
    \]
\end{example}

\begin{example}
    The layout encoded by 
    \[ \begin{tikzcd} 
    (128,(4,4,2)) \ar[rr,"g"] \ar[rr,swap,"(3{,}1{,}2{,}*)"] & & ((4,4),128)
    \end{tikzcd} \]
    is 
    \[
    L_g = (128,(4,4,2)):(16,(1,4,0)).
    \]
\end{example}

\begin{observation}\label{flatlayoutcompatibility}
The flattening functor 
\[ \begin{tikzcd}
\catstyle{Nest} \ar[r,"(-)^\flat"] & \catstyle{Tuple}
\end{tikzcd} \]
is compatible with flattening of layouts, in that if $f$ is a nested tuple morphism, then
\[
(L_f)^\flat = L_{f^\flat}.
\]
\end{observation}

If $L$ is a tractable layout, then we can construct a nested tuple morphism which encodes $L$ as follows. 

\begin{construction}\label{nestedtuplemorphismfromlayout}
    Suppose $L$ is a tractable layout. We define the {\it standard representation} of $L$ to be the nested tuple morphism 
    \[f_L:S \to T\]
    where $(f_L)^\flat  = f_{L^\flat}$ is the standard representation of $L^\flat$, $S = \shape(L)$ is the shape of $L$, and $T$ is the codomain of $f_{L^\flat}$.
\end{construction}

\begin{example}
    If 
    \[
    L = (32,(2,2)):(192,(24,3))
    \]
    then the standard representation of $L$ is 
    \[\begin{tikzcd}
    (32,(2,2)) \ar[rr,"f_L"] \ar[rr,swap,"(6{,}4{,}2)"] & & (3,2,4,2,4,32).
    \end{tikzcd}\]
\end{example}

\begin{lemma}\label{standardrepresentationlemma}
    If $L$ is a tractable layout, and $f = f_L$ is the standard representation of $L$, then 
    \[
    L_f = L.
    \]
\end{lemma}

\begin{proof}
    We have 
    \[
    (L_f)^\flat = L_{f^\flat} = L^\flat
    \]
    and 
    \begin{align*}
        \shape(L_f) = \shape(L).
    \end{align*}
\end{proof}

\begin{proposition}\label{tractablelayoutscomefromnestedtuplemorphisms}
Suppose $L$ is a layout. Then there exists a nested tuple morphism $f$ encoding $L$ if and only if $L$ is tractable.
\end{proposition}

\begin{proof}
    Suppose first that $L = L_f$ for some nested tuple morphism $f$. Then $(L_f)^\flat = L_{f^\flat}$, and by Proposition \ref{tractableflatlayoutscomefromtuplemorphisms}, we know that $L^\flat$ is tractable, hence so is $L$. Conversely, if $L$ is tractable, then we can take $f = f_L$ to be the standard representation of $L$, and by Lemma \ref{standardrepresentationlemma}, we have $L_f = L$. 
\end{proof}

In order to establish a one-to-one correspondence between tractable layouts and certain nested tuple morphisms, we introduce the notion of {\it standard form} for nested tuple morphisms. 

\begin{definition}\label{definitionofnestedstandardform}
    Suppose $f:S \to T$ is a nested tuple morphism. We say $f$ has {\it standard form} if 
    \begin{enumerate}
        \item $f^\flat$ has standard form, as in Definition \ref{definitionofstandardform}, and 
        \item $T$ is flat.
    \end{enumerate}
\end{definition}

\begin{example}
    The nested tuple morphism
    \[\begin{tikzcd} 
    ((2,2),(3,3)) \ar[rr,"f"] \ar[rr,swap,"(4{,}6{,}2{,}3"] & & (10,3,3,2,10,2)
    \end{tikzcd} \]
    has standard form.
\end{example}
\begin{example}
The nested tuple morphism
    \[\begin{tikzcd} 
    ((2,2),(3,3)) \ar[rr,"f"] \ar[rr,swap,"(4{,}6{,}2{,}3"] & & ((10,3,3),(2,10,2))
    \end{tikzcd} \]
    does not have standard form since the codomain of $g$ is not flat. 
\end{example}

Just as in the flat case, we need to exclude non-degenerate nested tuple morphisms and non-degenerate layouts in order to obtain a one-to-one correspondence between nested tuple morphisms of standard form and tractable layouts. To this end, we make the following definition.

\begin{definition}
    Suppose
    \[
    \begin{tikzcd} 
S \ar[rr,"f"] \ar[rr,swap,"\alpha"] & & T
    \end{tikzcd} 
    \]
    is a nested tuple morphism, and suppose 
    \[
    L = S:D
    \]
    is a layout. 
    \begin{enumerate}
        \item We say $f$ is {\it non-degenerate} if 
    \[
    \entry_i(S) = 1 \quad \Rightarrow \quad \alpha(i) = *.
    \]
    \item We say $L$ is {\it non-degenerate} if 
    \[
    \entry_i(S) = 1 \quad \Rightarrow \quad \entry_i(D) = 0.
    \]
    \end{enumerate}
\end{definition}

\begin{remark}
    If $f$ is a nested tuple morphism, then $f$ is non-degenerate if and only if $f^\flat$ is non-degenerate. If $L$ is a layout, then $L$ is non-degenerate if and only if $L^\flat$ is non-degenerate.
\end{remark}

\begin{proposition}\label{nestedonetoonecorrespondence}
    The maps 
    \[ \begin{tikzcd} 
    f \ar[r,mapsto] & L_f \\
    \begin{Bmatrix}
    \text{Non-degenerate}\\
        \text{nested tuple morphisms }\\
        \text{ of standard form}
    \end{Bmatrix}
    \ar[r] & \ar[l] 
    \begin{Bmatrix}
    \text{Non-degenerate}\\
    \text{tractable layouts}
    \end{Bmatrix} \\
    f_L & \ar[l,mapsto] L
    \end{tikzcd} 
    \]
    of Constructions \ref{layoutfromnestedtuplemorphism} and \ref{nestedtuplemorphismfromlayout} determine a one-to-one correspondence between nested tuple morphisms of standard form, and tractable layouts.
\end{proposition}

\begin{proof}
    We have already shown in Proposition \ref{tractablelayoutscomefromnestedtuplemorphisms} that if $L$ is a tractable layout and  $f = f_L$ is the standard form of $L$, then $L_f = L$. Suppose next that $f$ has standard form, and let $L = L_f$ be the layout encoded by $f$. We want to show that $f$ is equal to the standard representation $f_L$ of $L$. By Proposition \ref{flatcorrespondence}, we know that $f^\flat$ is equal to the standard representation $f_{L^\flat}$ of $L^\flat$, and since 
    \[
    \domain(f) = \shape(L) = \domain(f_L),
    \]
    and 
    \[
    \codomain(f) = \codomain(f^\flat) = \codomain(f_{L^\flat}) = \codomain(f_L),
    \]
    we deduce that $f = f_L$. 
\end{proof}

\subsection{Examples}

In this section, we list some important families of nested tuple morphisms.

\begin{example}[Reparenthesizations]\label{definitionofreparenthesizationisomorphisms}
    Suppose $S_1$ and $S_2$ are nested tuples with the same flattening
    \[
    S_1^\flat = S_2^\flat.
    \]
    Then there is a {\it reparenthesization isomorphism} 
    \[\begin{tikzcd} 
    \id_{S_1}^{S_2} : S_1 \ar[r,"\cong"] & S_2
    \end{tikzcd} \]
    lying over the identity. These morphisms are transitive, in that 
    \[
    \id_{S_2}^{S_3} \circ \id_{S_1}^{S_2} = \id_{S_1}^{S_3},
    \]
    and compatible with identities, in that
    \[
    \id_S^S = \id_S.
    \]
    If $f = \id_{S_1}^{S_2}$ is a reparenthesization isomorphism, then $L_f$ is the column major layout with shape $S_1$. 
\end{example}

\begin{example}[Flattenings]\label{definitionofflatteningisomorphisms}
    As a special case of the previous example, if $S$ is any nested tuple, then we have a flattening isomorphism
    \[\begin{tikzcd} 
    \id_S^{S^\flat} : S \ar[r,"\cong"] & S^\flat
    \end{tikzcd} \]
    and an unflattening isomorphism
    \[\begin{tikzcd}
    \id_{S^\flat}^{S} : S^\flat \ar[r,"\cong"] & S
    \end{tikzcd}\]
\end{example}

\begin{observation}
    If $f:S \to T$ is a nested tuple, then $f$ is equal to the composite 
    \[ \begin{tikzcd} 
    S \ar[r,"\id_S^{S^\flat}"] \ar[rrr,swap,bend right = 30, "f"] & S^\flat \ar[r,"f^\flat"] & T^\flat \ar[r,"\id_{T^\flat}^T"] & T.
    \end{tikzcd} \]
    In other words, we have a canonical factorization 
    \[
    f = \id_{T^\flat}^T \circ f^\flat \circ \id_S^{S^\flat}.
    \]
\end{observation}

\begin{example}[Entries]\label{definitionofentriesofnestedtuplemorphisms}
    Suppose
    \[ \begin{tikzcd}
S \ar[rr,"f"] \ar[rr,swap,"\alpha"] & & T 
    \end{tikzcd}\]
    is a nested tuple morphism. Suppose $1 \leq i \leq \len(S)$, and write $j = \alpha(i)$. 
    Then we refer to the nested tuple morphism 
    \[ \begin{tikzcd} 
    \entry_i(S) \ar[rr,"\entry_i(f)"] \ar[rr,swap,"(j)"]& & T
    \end{tikzcd} \]
    as the $i$th entry of $f$. The layout encoded by $\entry_i(f)$ is 
    \[
    L_{\entry_i(f)} = \entry_i(L_f) .
    \]
\end{example}

\begin{example}[Entry inclusions]
    As a special case of the previous example, if $S$ is a nested tuple and $1 \leq i \leq \len(S)$, we can take $f = \id_S$, in which case 
    \[ \begin{tikzcd}
\entry_i(\id_S):\entry_i(S) \ar[rr] & & S
    \end{tikzcd}\]
    is the inclusion of the $i$th entry of $S$.
\end{example}

\begin{example}[Modes]\label{definitionofmodesofnestedtuplemorphisms}
    Suppose
        \[ \begin{tikzcd}
S \ar[rr,"f"] \ar[rr,swap,"\alpha"] & & T 
    \end{tikzcd}\]
    is a nested tuple morphism. Suppose $1 \leq i \leq \rank(S)$ and, write 
    \begin{align*}
    N & = \len_{<i}(S)\\
    \ell & = \len_i(S).
    \end{align*}
    Then we refer to the nested tuple morphism
    \[ \begin{tikzcd} [column sep = 30] 
    \mode_i(S) \ar[rr,"\mode_i(f)"] \ar[rr,swap,"(N+1{,}\dots{,}N+\ell)"]& & T
    \end{tikzcd} \]
    as the $i$th mode of $S$. The layout encoded by $\mode_i(L_f)$ is 
    \[
    L_{\mode_i(f)} = \mode_i(L_f).
    \]
\end{example}

\begin{example}[Mode inclusions]\label{definitionofmodeinclusions} As a special case of the previous example we may take $f = \id_S$, in which case
\[
\mode_i(\id_S) :\mode_i(S) \to S
\]
is the inclusion of the $i$th mode of $S$. We sometime denote this map by
\[
\incl_i(S) = \mode_i(\id_S).
\]
\end{example}

\subsection{Realization of nested tuple morphisms}\label{realizationofnestedtuplemorphisms}

In the flat case, we constructed a realization functor 
\[\begin{tikzcd} 
\catstyle{Tuple} \ar[r,"{|} \; \cdot \; {|}"] & \catstyle{FinSet}
\end{tikzcd} \]
which sends a tuple morphism $f$ to the layout function of $L_f$. We can extend this to a realization functor 
\[ \begin{tikzcd}
\catstyle{Nest} \ar[r,"{|}\; \cdot \;{|}"] & \catstyle{FinSet}
\end{tikzcd} \]
by precomposing with the flattening functor $\catstyle{Nest} \to \catstyle{Tuple}$. 

\begin{definition}
    We define the {\it realization functor}
    \[\begin{tikzcd}
        \catstyle{Nest}\ar[r,"{|}\;\cdot \;{|}"] & \catstyle{FinSet}
    \end{tikzcd}
    \]
    to be the composite 
    \[\begin{tikzcd}
\catstyle{Nest}\ar[r,"(-)^\flat"] & \catstyle{Tuple} \ar[r,"{|}\; \cdot \; {|}"] & \catstyle{FinSet}
    \end{tikzcd}
    \]
\end{definition}

\begin{lemma} \label{nestedlayoutfunctionlemma}
    If $f:S \to T$ is a nested tuple morphism, then the realization $|f|$ of $f$ is the layout function of $L_f$:
    \[
    |f| = 
    \Phi_{L_f}^{\size(T)}. 
    \]
\end{lemma}

\begin{proof}
This follows immediately from \ref{layoutfunctionlemma}, since
\[
|f| = |f^\flat| = \Phi_{L_{f^\flat}}^{\size(T)} = \Phi_{L_f}^{\size(T)}
\]
\end{proof}

\newpage 

\subsection{Refinements}\label{categoryrefinements}
In this section, we revisit the refinement of nested tuples from a categorical perspective. Recall from section \ref{refinementsection} that a nested tuple $S'$ {\it refines} $S$, denoted
\[ \begin{tikzcd} S' \ar[r,two heads] &  S\end{tikzcd} \]
if $S'$ may be obtained from $S$ by replacing each entry of $S$ with some nested tuple of the same size. For example, 
\[
(2,(2,2)) \twoheadrightarrow 8,
\]
and 
\[
((2,2),(3,3),(5,5)) \twoheadrightarrow (4,9,25).
\]
If $\len(S) = m$ and $\profile(S) = P$, then we can write 
\[
S' = (S_1',\dots,S_m')_P
\]
as the $P$-substitution of the relative modes
\[
S_i' = \mode_i(S',S).
\]
We refer to the ordinary concatenation
\[
(S_1',\dots,S_m') = \flatten(S',S)
\]
as the flattening of $S'$ relative to $S$.

Let $\catstyle{Ref}$ denote the poset category of nested tuples of positive integers under refinement, so that a morphism in $\catstyle{Ref}$ is a refinement $S' \twoheadrightarrow S$. If $S$ is a nested tuple, let 
\[
\catstyle{Ref}(S) = \{ S' \mid S'\text{ refines }S \}
\]
denote the poset of nested tuples refining $S$. Equivalently, $\catstyle{Ref}(S)$ is the slice category $\catstyle{Ref}_{/S}$.

\begin{construction}\label{definitionofrelativemodeinclusions}[Relative mode inclusions]
Suppose $S' \twoheadrightarrow S$ is a refinement, and write 
\[
S_i' = \mode_i(S',S)
\]
for the modes of $S'$ relative to $S$. Then $S'$ and $(S_1',\dots,S_m')$ have the same flattening, so we have a reparenthesization isomorphism 
\[
\begin{tikzcd} 
\id_{(S_1',\dots,S_m')}^{S'}:(S_1',\dots,S_m') \ar[r,"\cong"] & S'
\end{tikzcd}
\]
and we define 
\[
\incl_i(S',S): S_i' \to S'
\]
to be the composite 
\[
\begin{tikzcd} 
S_i' \ar[rrr,"\incl_i((S_1'{,}\dots{,}S_m'))"] & & & (S_1',\dots,S_m') \ar[rr,"\id_{(S_1',\dots,S_m')}^{S'}"] & & S'
\end{tikzcd} 
\]
of the $i$th mode inclusion of $(S_1',\dots,S_m')$ with the reparenthesization isomorphism $(S_1',\dots,S_m') \cong S'$. 
\end{construction}

\begin{example}
    If $S = (4,(9,25))$ and $S' = ((2,2),((3,3),25))$, then $S'$ refines $S$, and $\incl_2(S',S)$ is the nested tuple morphism 
    \[ \begin{tikzcd} 
    (3,3) \ar[rr,"\incl_2(S'{,}S)"] \ar[rr,swap,"(3{,}4)"] & & ((2,2),((3,3),25)).
    \end{tikzcd} \]
\end{example}

\begin{construction}\label{definitionofrelativemodes}[Relative modes]
    Suppose $f':S' \to T'$ is a nested tuple morphism, and suppose $S'$ refines $S$. We define the $i${\it th mode of} $f'$ {\it relative to} $S$, denoted 
    \[
    \mode_i(f',S) = f' \circ \incl_i(S',S): S_i' \to T'
    \]
    to be the composite 
    \[
    \begin{tikzcd} S_i' \ar[rr ,"\incl_i(S'{,}S)"] & & S' \ar[r,"f'"] &  T'
    \end{tikzcd} 
    \]
    In particular, we have 
    \[
    \mode_i(\id_{S'},S) = \incl_i(S',S).
    \]
\end{construction}

\begin{example}
    Suppose $S = (4,(9,25))$ and $S' = ((2,2),((3,3),25))$, so that $S'$ refines $S$. If $f'$ is the nested tuple morphism 
    \[ \begin{tikzcd} 
    ((2,2),((3,3),25)) \ar[rr,"f'"] \ar[rr,swap,"(1{,}3{,}2{,}*{,}4)"] & & (2,3,2,25).
    \end{tikzcd} \]
    then $\mode_2(f',S)$ is the nested tuple morphism 
    \[ \begin{tikzcd} 
    (3,3) \ar[rr,"\mode_2(f'{,}S)"] \ar[rr,swap,"(2{,}*)"] & & (2,3,2,25).
    \end{tikzcd} \]
\end{example}

\begin{construction}[Pullbacks]
    Suppose $f:S \to T$ is a nested tuple morphism lying over $\alpha$, and suppose $T' \twoheadrightarrow T$ is a refinement. Let 
    \[
    T_j' = \mode_j(T',T)
    \]
    denote the $j$th mode of $T'$ relative to $T$, and for any $1 \leq i \leq \len(S)$, set
    \[
    S_i' = 
    \begin{cases} 
    \entry_i(S) & \alpha(i) = *\\
    T'_j & \alpha(i) = j.
    \end{cases}
    \] 
    We define the {\it pullback of} $T'$ {\it along} $f$ to be the nested tuple
    \[
    S' = f^*T' = \sub(S,(S_1',\dots,S_m')).
    \]
    For any $1 \leq i \leq m$, let 
    \[
    f_i' :S_i' \to T'
    \]
    be the trivial map if $\alpha(i) = *$, and the inclusion
    \[
    \incl_j(T',T):S_i' = T_j' \to T'
    \]
    if $\alpha(i) = j$. The maps $f_1',\dots,f_m'$ have disjoint images, so we form the concatenation
    \[
    (f_1',\dots,f_m') : (S_1',\dots,S_m') \to T'.
    \]
    We define $f' = T'^*f$ to be the composite 
    \[ \begin{tikzcd}
    S' \ar[rrr,"\id_{S'}^{(S_1'{,}\dots{,}S_m')}"] & & &  (S_1',\dots,S_m') \ar[rr,"(f_1'{,}\dots{,}f_m')"] & & T'.
    \end{tikzcd} \]
     We refer to $f'$ as the pullback of $f$ along $T$, and depict such a pullback as a square 
    \[ \begin{tikzcd} 
    S' \ar[r,"f'"] \ar[d, two heads] \arrow[dr, phantom, "\lrcorner", very near start] & T' \ar[d,two heads] \\
    S \ar[r,swap,"f"] & T.
    \end{tikzcd} \]
\end{construction}

\begin{example}
    Suppose $f:(64,32) \to (4,64,4,32)$ lies over $\alpha = (2,4)$. Then we have a pullback square 
    \[ \begin{tikzcd} 
    ((16,4),(16,2)) \ar[d,two heads] \ar[r,"f'"] \arrow[dr, phantom, "\lrcorner", very near start] & ((2,2),(16,4),(2,2),(16,2)) \ar[d, two heads] \\
    (64,32) \ar[r,swap,"f"] & (4,64,4,32) 
    \end{tikzcd} \]
    where $f'$ lies over $\alpha' = (3,4,7,8)$.
\end{example}

\begin{example}
    Suppose $S$ is a nested tuple with flattening 
    \[
    S^\flat = (s_1,\dots,s_m),
    \]
    and suppose $S' \twoheadrightarrow S$ is a refinement with relative flattening 
    \[
    (S_1',\dots,S_m').
    \]
    Then the pullback of $S' \twoheadrightarrow S$ along the unflattening isomorphism 
    \[
    \id_{(s_1,\dots,s_m)}^S :(s_1,\dots,s_m) \to S 
    \]
    is the reparenthesization isomorphism 
    \[ \begin{tikzcd} 
    (S_1',\dots,S_m') \ar[rr,"\id_{(S_1'{,}\dots{,}S_m')}^{S'}"] \ar[d, two heads] \arrow[drr, phantom, "\lrcorner", very near start] & & S' \ar[d,two heads] \\
    (s_1,\dots,s_m) \ar[rr,swap,"\id_{(s_1{,}\dots{,}s_m)}^S"] & & S.
    \end{tikzcd}
    \]
\end{example}

\begin{example}
    Suppose $S' \twoheadrightarrow S$ is a refinement, and consider the $i$th entry inclusion
    \[
    s_i \to S.
    \]
    Then the pullback of $S' \twoheadrightarrow S$ along $s_i \to S$ is the $i$th relative mode inclusion 
        \[ \begin{tikzcd} 
    S_i' \ar[rr,"\incl_i(S'{,}S)"] \ar[d, two heads] \arrow[drr, phantom, "\lrcorner", very near start] & & S' \ar[d,two heads] \\
    s_i \ar[rr,swap] & & S.
    \end{tikzcd}
    \]
\end{example}

\begin{observation}
    The pullback construction above specifies a contravariant functor 
    \[ \begin{tikzcd} 
\catstyle{Nest}^\text{op} \ar[r] & \catstyle{Cat} \\
    S \ar[r,mapsto] \ar[d,swap,"f"] & \catstyle{Ref}(S)  & f^*T' \leftarrow f^*T'' \\
    T \ar[r,mapsto]  & \catstyle{Ref}(T) \ar[u,swap,"f^*"] & T' \leftarrow T'' \ar[u,mapsto] 
    \end{tikzcd} \]
\end{observation}

The key property of pullbacks is that the layout function of $f'$ is equal to that of $f$.

\begin{lemma}\label{pullbacklemma}
    Suppose
    \[ \begin{tikzcd} 
    S' \ar[r,"f'"] \ar[d, two heads] \arrow[dr, phantom, "\lrcorner", very near start] & T' \ar[d,two heads] \\
    S \ar[r,swap,"f"] & T
    \end{tikzcd} \]
    is a pullback square, where $f$ lies over $\alpha$. Let 
    \[f_i' : S_i' \to T\]
    denote the $i$th mode of $f'$ relative to $S$,
    and let 
    \[(L_f)^\flat = (s_1,\dots,s_m):(d_1,\dots,d_m).\]
    Then for any $1 \leq i \leq m$, we have 
    \[
    \coalesce(L_{f_i'}) = s_i:d_i.
    \]
\end{lemma}
\begin{proof}
    Suppose $1 \leq i \leq m$. If $\alpha(i) = *$, then $f_i'$ is the trivial map, so 
    \[
    L_{f_i'} = s_i : 0 = s_i : d_i.
    \]
    In particular, $\coalesce(L_{f_i'}) = s_i:0 = s_i:d_i$. Suppose next that $\alpha(i) = j \neq *$. By construction of $f'$, we have that 
    \[f_i' = \incl_j(T',T):T_j' \to T'.\]
    which lies over the map $\alpha'_i$ given by $t \mapsto \len_{<j}(T',T) + t$. For each $1 \leq t < \len(T'_j)$, we have $\alpha'_i(t) = \alpha'_i(t+1)$, so $L_{f_i'}$ is a column major layout with size $\size(T_j') = t_j = s_i$. This implies that $\coalesce(L_{f_i'})$ is a depth $0$ layout of the form 
    \[
    \coalesce(L_{f_i'}) = s_i : e
    \]
    for some integer $e \geq 0$. We claim that $e = d_i$. If we write $t_{j'}' = \entry_{j'}(T')$, then we have 
    \begin{align*}
    e = \entry_1(\stride(L_{f_i'})) & = \prod_{j' < \alpha_i'(1)} t_{j'}'\\
    & = \prod_{j' \leq  \len_{<j}(T',T)} t_{j'}'\\
    & = \prod_{j' < j} \size(T_{j'}')\\
    & = \prod_{j' < j} t_{j'}\\
    & = d_i .
    \end{align*}
\end{proof}

\begin{proposition}\label{pullbackproposition}
    If 
    \[ \begin{tikzcd} 
    S' \ar[r,"f'"] \ar[d, two heads] \arrow[dr, phantom, "\lrcorner", very near start] & T' \ar[d,two heads] \\
    S \ar[r,swap,"f"] & T
    \end{tikzcd} \]
    is a pullback square, then $\Phi_{L_f} = \Phi_{L_{f'}}$.
\end{proposition}

\begin{proof}
    We begin by fixing notation. Let $m = \len(S)$, and let 
    \begin{align*}
    S^\flat & = (s_1,\dots,s_m),\\
    S_i' & = \mode_i(S',S),\\
    T_j' & = \mode_j(T',T),\\
    (L_f)^\flat & = (s_1,\dots,s_m):(d_1,\dots,d_m).
    \end{align*}
    Consider the reparenthesization isomorphism 
    \[
    \id_{(S_1',\dots,S_m')}^{S'}:(S_1',\dots,S_m') \to S'
    \]
    The composite of this map with $f'$ is the concatenation $(f_1',\dots,f_m')$ where $f_i'$ is the trivial map if $\alpha(i) = *$, and the relative mode inclusion
    \[
    \incl_i(T',T):S_i' = T_j' \to T'
    \]
    otherwise. Using Lemma \ref{pullbacklemma}, and the fact that $L_{f'} = L_{(f_1',\dots,f_m')}$, we compute
    \begin{align*}
    \coalesce(L_{f'}) & = \coalesce(L_{(f_1',\dots,f_m')})\\
    & = \coalesce((L_{f_1'},\dots,L_{f_m'}))\\
    & = \coalesce((\coalesce(L_{f_1'}),\dots,\coalesce(L_{f_m'}) ) )\\
    & = \coalesce((s_1,\dots,s_m):(d_1,\dots,d_m))\\
    & = \coalesce(L_f).
    \end{align*}
    By Proposition \ref{nestedcoalesceproposition}, we deduce that $\Phi_{L_{f'}} = \Phi_{L_f}$.
\end{proof}

\begin{construction}[Pushforwards]
    Suppose $f:S \to T$ is a nested tuple morphism lying over $\alpha$, and suppose $S' \twoheadrightarrow S$ is a refinement. Let 
    \[
    S_i' = \mode_i(S',S)
    \]
    denote the $i$th mode of $S'$ relative to $S$, and for any $1 \leq j \leq \len(T)$, set 
    \[
    T_j' = 
    \begin{cases} 
    \entry_j(T) & j \notin \Image(\alpha) \\
    S'_i & \alpha(i) = j.
    \end{cases}
    \]
    We define the {\it pushforward of} $S'$ {\it along} $f$ to be the nested tuple
    \[
    T' = f_*S' = \sub(T,(T_1',\dots,T_n')).
    \]
    For any $1 \leq i \leq m$, let 
    \[
    f_i':S_i' \to T' 
    \]
    be the trivial map if $\alpha(i) = *$, and the relative mode inclusion 
    \[
    \incl_j(T',T):S_i' = T_j' \to T'
    \]
    if $\alpha(i) = j$. The morphisms $f_1',\dots,f_m'$ have disjoint images, so we can form the concatenation 
    \[
    (f_1',\dots,f_m'):(S_1',\dots,S_m') \to T'.
    \]
    We define $f' = S'_*f$ to be the composite 
    \[ \begin{tikzcd}
    S' \ar[rrr,"\id_{S'}^{(S_1'{,}\dots{,}S_m')}"] & & &  (S_1',\dots,S_m') \ar[rr,"(f_1'{,}\dots{,}f_m')"] & & T'.
    \end{tikzcd} \]
    We refer to $f'$ as the pushforward of $f$ along $T$. We depict such a pushforward as 
    \[ \begin{tikzcd} 
    S' \ar[r,"f'"] \ar[d,two heads] & \arrow[dl, phantom, "\llcorner", very near start]T' \ar[d,two heads] \\
    S \ar[r,swap,"f"] & T\\
    \end{tikzcd} \]
\end{construction}

\begin{example}
    If $f:(64,32) \to (4,64,4,32)$ lies over $\alpha = (2,4)$, then we have a pushforward square 
    \[ \begin{tikzcd} 
    ((16,4),(16,2)) \ar[r,"f'"] \ar[d,two heads] & \arrow[dl, phantom, "\llcorner", very near start](4,(16,4),4,(16,2)) \ar[d, two heads] \\
    (64,32) \ar[r,swap,"f"] & (4,64,4,32)\\
    \end{tikzcd} \]
\end{example}

The key property of pullbacks is that the layout function of $f'$ is equal to that of $f$.

\begin{lemma}\label{pushforwardlemma}
    Suppose
    \[ \begin{tikzcd} 
    S' \ar[r,"f'"] \ar[d, two heads] & T' \ar[d,two heads]\arrow[dl, phantom, "\llcorner", very near start] \\
    S \ar[r,swap,"f"] & T
    \end{tikzcd} \]
    is a pushforward square, where $f$ lies over $\alpha$. Let 
    \[f_i' : S_i' \to T\]
    denote the $i$th mode of $f'$ relative to $S$,
    and let 
    \[(L_f)^\flat = (s_1,\dots,s_m):(d_1,\dots,d_m).\]
    Then for any $1 \leq i \leq m$, we have 
    \[
    \coalesce(L_{f_i'}) = s_i:d_i.
    \]
\end{lemma}
\begin{proof}
The proof is identical to that of Lemma \ref{pullbacklemma}
\end{proof}

\begin{proposition}\label{pushforwardproposition}
    If 
    \[ \begin{tikzcd} 
    S' \ar[r,"f'"] \ar[d, two heads]  & T' \ar[d,two heads]\arrow[dl, phantom, "\llcorner", very near start] \\
    S \ar[r,swap,"f"] & T
    \end{tikzcd} \]
    is a pushforward square, then $\Phi_{L_f} = \Phi_{L_{f'}}$.
\end{proposition}

\begin{proof}
    The proof is identical to that of Proposition \ref{pullbackproposition}.
\end{proof}

\begin{observation}
    The pushforward construction defined above specifies a covariant functor
    \[ \begin{tikzcd} 
    \catstyle{Nest} \ar[r] & \catstyle{Cat} \\
    S \ar[r,mapsto] \ar[d,swap,"f"] & \catstyle{Ref}(S) \ar[d,"f_*"] & S'' \rightarrow S'\ar[d,mapsto]  \\
    T \ar[r,mapsto]  & \catstyle{Ref}(T)  & f_*S'' \rightarrow f_*S'   
    \end{tikzcd} \]
\end{observation}

\begin{observation}
    If $f:S \to T$ is an isomorphism of nested tuples, then 
    \[\begin{tikzcd}
        \catstyle{Ref}(T) \ar[r,"f^*"] & \catstyle{Ref}(S)
    \end{tikzcd}
    \]
    and
    \[\begin{tikzcd}
        \catstyle{Ref}(S) \ar[r,"f_*"] & \catstyle{Ref}(T)
    \end{tikzcd}
    \]
    are inverse isomorphisms of categories. Specifically, 
    \[
    (f^{-1})^* = f_* \hspace{0.2in} \text{and} \hspace{0.2in} (f^{-1})_* = f^*.
    \]
\end{observation}

\begin{observation}
    If $S_1$ and $S_2$ are nested tuples with $\flatten(S_1) = \flatten(S_2)$, then there is a canonical nested tuple isomorphism $S_1 \cong S_2$, and hence, a canonical isomorphism of categories
    \[
    \catstyle{Ref}(S_1) \cong \catstyle{Ref}(S_2).
    \]
\end{observation}

There is one more concept we need to specify, called {\it mutual refinements}. The importance of this concept will be come clear in Chapter \ref{computationschapter}, when we use this concept in our layout composition algorithm.

\begin{definition}
Suppose $T$ and $U$ are nested tuples. A {\it mutual refinement} of $(T,U)$ is a diagram of the form
\[ \begin{tikzcd} 
T' \ar[r,rightarrowtail] \ar[d, two heads] & U' \ar[d,two heads] \\
T & U
\end{tikzcd} \]
Explicitly, this is a pair of nested tuples $(T',U')$ such that \begin{enumerate}
    \item $T'$ refines $T$,
    \item $U'$ refines $U$, and
    \item $T'$ divides $U'$. 
\end{enumerate}
\end{definition}

In addition to the definition of mutual refinements, we need the following fact.

\begin{lemma}\label{mutualrefinementsofflattenings}
    Suppose $T$ and $U$ are nested tuples. Then there is a one-to-one correspondence between mutual refinements of $(T,U)$, and mutual refinements of $(T^\flat,U^\flat)$. 
\end{lemma}
\begin{proof}
    If $(T',U')$ is a mutual refinement of $(T,U)$, then pulling back along the unflattening isomorphisms $\id_{T^\flat}^T$ and $\id_{U^\flat}^U$ yields a mutual refinement 
    \[ \begin{tikzcd} 
 (\id_{T^\flat}^T)^*T' \ar[r,rightarrowtail] \ar[d, two heads] & (\id_{U^\flat}^U)^*U' \ar[d,two heads] \\
T^\flat & U^\flat
\end{tikzcd} \]
    of $(T^\flat, U^\flat)$. Conversely, if $((T^\flat)',(U^\flat)')$ is a mutual refinement of $T^\flat,U^\flat$, then pulling back along the flattening isomorphisms $\id^{T^\flat}_T$ and $\id^{U^\flat}_U$ yields a mutual refinement 
    \[ \begin{tikzcd} 
 (\id^{T^\flat}_T)^*(T^\flat)' \ar[r,rightarrowtail] \ar[d, two heads] & (\id^{U^\flat}_U)^*(U^\flat)' \ar[d,two heads] \\
T & U
\end{tikzcd} \]
    of $(T^\flat, U^\flat)$.
\end{proof}
% \begin{construction}
%     Suppose $S'$ is a refinement of $S$, and suppose $\len(S') = m'$ and $\len(S) = m$. We may define a pointed map 
%     \[
%     \mu:\langle m' \rangle_* \to \langle m \rangle_*
%     \]
%     by the formula 
%     \[
%     \mu(i') = \text{max}\{i \in \langle m \rangle \mid \len_i(S',S) < i' \}.
%     \]
% \end{construction}

% \begin{example}
%     If $S'=((2,2),(3,3),(5,5))$ and $S = (4,9,25)$, then $S' \twoheadrightarrow S$ lies over the map $\mu = (1,1,2,2,3,3)$.
% \end{example}

\newpage 

\subsection{Operations on nested tuple morphisms}
Our next task is to develop an ``algebra of nested tuple morphisms". Since we have already developed such an ``algebra" for tuple morphisms, we can extend to the nested case by equipping the outputs of our various operations with an appropriate profile.

\subsubsection{Concatenate}

Next, we define a concatenation operation on nested tuple morphisms, which is compatible with concatenation of layouts, in that 
\[
L_{(f,g)} = (L_f,L_g).
\]
We concatenate nested tuple morphisms $f$ and $g$ by concatenating the domains of $f$ and $g$. In order for this to be well-defined, we need $f$ and $g$ to satisfy a disjointness condition, which we specify below.

\begin{definition}\label{definitionofdisjointimagesinD} Suppose $f$ and $g$ are nested tuple morphisms with the same codomain. We say $f$ and $g$ have {\it disjoint images} if $f^\flat$ and $g^\flat$ have disjoint images, as in Definition \ref{definitionofdisjointimagesinC}. 
\end{definition}

\begin{example}\label{exampleofdisjointimagesinD1}
    If 
    \[
    f:(3,(512,512)) \to (2,512,2,512)
    \]
    lies over $(*,2,4)$ and 
    \[
    g:(2,2) \to (2,512,2,512)
    \]
    lies over $(1,3)$, then $f$ and $g$ have disjoint images.
\end{example}

\begin{example}\label{exampleofdisjointimagesinD2} If 
\[f : (2,(32,64)) \to (32,(2,2,2),64)\]
lies over $\alpha = (3,1,5)$ and 
\[g:((2,2)) \to (32,(2,2,2),64)\] 
lies over $\beta = (2,4)$, then $f$ and $g$ have disjoint images.
\end{example}

\begin{construction}\label{concatenationofmorphismsinD}
    Suppose $f:S \to T$ and $g:U \to T$ are nested tuple morphisms lying over $\alpha$ and $\beta$, respectively, and that $f$ and $g$ have disjoint images. We define the {\it concatenation} of $f$ and $g$ to be the nested tuple morphism 
    \[
    (f,g) : (S,U) \to T
    \]
    with 
    \begin{align*}
        \flatten((f,g)) & = f^\flat \star g^\flat.
    \end{align*}
    More generally, if $f_i:S_i \to T$ are nested tuple morphisms for $1 \leq i \leq k$, and $f_1,\dots,f_k$ have pairwise disjoint images, then we define the {\it concatenation} 
    \[
    (f_1,\dots,f_k) : (S_1 , \dots , S_k) \to T.
    \]
    to be the nested tuple morphism with 
    \begin{align*}
        (f_1, \dots , f_k)^\flat & = f_1^\flat \star \cdots \star f_k^\flat.
    \end{align*}
    
\end{construction}

\begin{example} The concatenation of the morphisms $f$ and $g$ of Example \ref{exampleofdisjointimagesinD1} is the nested tuple morphism
\[(f,g): \left((3,(512,512)),(2,2)\right) \to (2,512,2,512)\]
lying over $\alpha \star \beta = (*,2,4,1,3)$.
\end{example}

\begin{example} The concatenation of the morphisms $f$ and $g$ of Example \ref{exampleofdisjointimagesinD2} is the nested tuple morphism
\[(f,g): ( (2,(32,64)),((2,2))) \to (32,(2,2,2),64)\]
lying over $\alpha \star \beta = (3,1,5,2,4)$.
\end{example}

\begin{example}
    If 
    \[
    f:(2,2) \to (2,3,5,2,3,5)
    \]
    lies over $\alpha = (1,4)$
    \[
    g:(3,3) \to (2,3,5,2,3,5)
    \]
    lies over $\beta = (2,5)$, and 
    \[
    h: (5,5) \to (2,3,5,2,3,5)
    \]
    lies over $\gamma = (3,6)$, then $f$, $g$ and $h$ have pairwise disjoint images, and the concatenation
    \[
    (f,g,h):((2,2),(3,3),(5,5)) \to (2,3,5,2,3,5)
    \]
    lies over $\alpha \star \beta \star \gamma = (1,4,2,5,3,6)$. 
\end{example}

\begin{example}
Suppose $f:S \to T$ is a nested tuple morphism, and suppose 
\[
S^\flat = (s_1,\dots,s_m).
\]
Recall from example \ref{definitionofentriesofnestedtuplemorphisms} that for any $1 \leq i \leq m$, there is a nested tuple morphism
\[
f_i:s_i \to T.
\]
called the $i$th entry of $f$. These morphisms have pairwise disjoint images, and the concatenation 
\[
(f_1,\dots,f_m):S^\flat \to T
\]
is the composite
\[
(f_1,\dots,f_m) = f \circ \id_{S^\flat}^S 
\]
of Example \ref{definitionofflatteningisomorphisms}
\end{example}

\begin{example}
Suppose $f:S \to T$ is a nested tuple morphism, and suppose 
\[
S = (S_1,\dots,S_r).
\]
Recall from example \ref{definitionofmodesofnestedtuplemorphisms} that for any $1 \leq i \leq r$, there is a nested tuple morphism
\[
f_i:S_i \to T.
\]
called the $i$th mode of $f$. These morphisms have pairwise disjoint images, and the concatenation 
\[
(f_1,\dots,f_r):S \to T
\]
is equal to $f$. In other words, every nested tuple morphism $f$ may be written as the concatenation of its modes:
\[f = (f_1,\dots,f_r).\]
\end{example}

\begin{proposition}\label{compatibilityofconcatenateinD}
    If $f_1,\dots,f_k$ are nested tuple morphisms with the same codomain and with pairwise disjoint images, then 
    \[
    L_{(f_1,\dots,f_k)} = (L_{f_1},\dots,L_{f_k}).
    \]
\end{proposition}

\begin{proof}
By construction, we have 
\begin{align*}
 \shape((L_{f_1},\dots,L_{f_k})) & = (\shape(L_{f_1}),\dots,\shape(L_{f_k}))\\
& = \shape(L_{(f_1,\dots,f_k)}).
\end{align*}
and using Proposition \ref{concatenateinCproposition}, we have 
\begin{align*}
(L_{f_1},\dots,L_{f_k})^\flat & = L_{f_1}^\flat \star \cdots \star L_{f_k}^\flat\\
& = L_{f_1^\flat} \star \cdots \star L_{f_k^\flat} \\
& = L_{f_1^\flat \star \cdots \star f_k^\flat }\\
& = L_{(f_1,\dots, f_k)^\flat}\\
& = (L_{(f_1,\dots,f_k)})^\flat.
\end{align*}
\end{proof}

\subsubsection{Coalesce}

If $f$ is a nested tuple morphism, then we might  define $\coalesce(f)$ to be $\coalesce^\flat(f^\flat)$. Theoretically, this is a sound definition. However, in order to make our definitions compatible with the \texttt{cute} implementation, we make a small modification to our definition of $\coalesce(f)$. 

\begin{definition}\label{definitionofcoalesceofnestedtuplemorphism}
    Suppose $f:S \to T$ is a nested tuple morphism, and write 
    \[
    \coalesce^\flat(f^\flat):(s_1,\dots,s_m) \to (t_1,\dots,t_n).
    \]
    \begin{itemize}
    \item (Case 1): If $m >1$, we define 
    \[\coalesce(f) = \coalesce^\flat(f^\flat).\]
    \item (Case 2): If $m = 1$, we define $\coalesce(f)$ to be the composite 
    \[
    \begin{tikzcd} 
    s_1 \ar[r,swap,"(1)"] &  (s_1) \ar[rr,"\coalesce^\flat(f^\flat)"] & & (t_1,\dots,t_n).
    \end{tikzcd} 
    \]
    \item (Case 3): If $m = 0$, we define $\coalesce(f)$ to be the composite 
    \[
    \begin{tikzcd} 
    1 \ar[r,swap,"(*)"] &  () \ar[rr,"\coalesce^\flat(f^\flat)"] & & (t_1,\dots,t_n).
    \end{tikzcd} 
    \]
    \end{itemize} 
\end{definition}

\begin{example}
    If 
    \[f:((2,2),(3,3),(5,5)) \to (5,5,3,3,2,2)\]
    lies over $\alpha = (5,6,3,4,1,2)$, then 
    \[
    \coalesce(f):(4,9,25) \to (25,9,4)
    \]
    lies over $\alpha' = (3,2,1)$. 
\end{example}

\begin{proposition}\label{coalesceagreement}
    If $f:S \to T$ is a nested tuple morphism, then 
    \[
    \coalesce(L_f) = L_{\coalesce(f)}.
    \]
\end{proposition}

\begin{proof}
Let's again write 
    \[ \begin{tikzcd}
    (s_1,\dots,s_m) \ar[rr,"\coalesce^\flat(f^\flat)"] \ar[rr,swap,"\alpha"] & & (t_1,\dots,t_n).
    \end{tikzcd} \]
    There are three cases to consider. 
    \begin{itemize}
        \item (Case 1): Suppose $m > 1$. Then 
        \begin{align*}
        L_{\coalesce(f)} & = L_{\coalesce^\flat(f^\flat)}\\
        & = \coalesce^\flat(L_{f^\flat})\\
        & = \coalesce((L_f)^\flat)\\
        & = \coalesce(L_f).
    \end{align*}
    \item (Case 2): Suppose $m = 1$. Then 
    \begin{align*}
        L_{\coalesce(f)} & = s_1:t_1 \dots,t_{\alpha(1)-1}\\
        & = \coalesce((s_1):(t_1\cdots 
        t_{\alpha(1)-1}))\\
        &= \coalesce(L_{\coalesce^\flat(f^\flat)})\\
        & = \coalesce(\coal^\flat(L_{f^\flat}))\\
        & = \coalesce((L_f)^\flat)\\
        & = \coalesce(L_f).
    \end{align*}
    \item (Case 3): Suppose $m = 0$. Then 
    \begin{align*}
        L_{\coalesce(f)} & = 1:0 \\
        & = \coalesce(():())\\
        &= \coalesce(L_{\coalesce^\flat(f^\flat)})\\
        & = \coalesce(\coal^\flat(L_{f^\flat}))\\
        & = \coalesce((L_f)^\flat)\\
        & = \coalesce(L_f).
    \end{align*}
    \end{itemize}

\end{proof}

\subsubsection{Complement}

In this section, we define the notion of complementary nested tuple morphisms.

\begin{definition} \label{definitionofcomplementarymorphismsinD}
Suppose $f:S \to T$ and $g:U \to T$ are nested tuple morphisms with disjoint images. We say $g$ is a {\it complement} of $f$ if
\[
(f,g):(S,U) \to T
\]
is an isomorphism.
\end{definition}

\begin{remark}If $f:S \to T$ and $g:U \to T$ are nested tuple morphisms, then $g$ is a complement of $f$ if and only if $g^\flat$ is a complement of $f^\flat$, since $(f,g)^\flat = f^\flat \star g^\flat$. 
\end{remark}

\begin{proposition}\label{complementofmorphisminDproposition}
    If $f:S \to T$ is a nested tuple morphism and $g:U \to T$ is a complement of $f$, then $L_{g}$ is a $\size(T)$-complement of $L_{f}$.
\end{proposition}

\begin{proof}
    Observation \ref{flatlayoutcompatibility} implies that
    \begin{align*} (L_{f})^\flat & = L_{f^\flat}\text{, and }\\
    (L_{g})^\flat & = L_{g^\flat} 
    \end{align*}
    and Lemma \ref{flatteningcomplementationlemma} allows us to reduce to the flat case (Proposition \ref{complementofmorphisminCproposition}). 
\end{proof}

\begin{construction}\label{complementofmorphisminD}
    Suppose $f:S \to T$ is a nested nested tuple morphism. We define the {\it complement of }$f$ to be the composite 
    \[ \begin{tikzcd} 
    U \ar[dr,swap,"(f^\flat)^c"] \ar[rr,"f^c"] & & T \\
    & T^\flat \ar[ur,swap,"\id_{T^\flat}^T"]
    \end{tikzcd} \]
    where $(f^\flat)^c$, is as defined in Construction \ref{complementofmorphisminC}, and $\id_{T^\flat}^T:T^\flat \cong T$ is the unflattening isomorphism.
\end{construction}

\begin{example}
    The complement of the nested tuple morphism
    \[\begin{tikzcd} ((2,2),(5,5)) \ar[rr,"f"] \ar[rr,swap,"(1{,}4{,}2{,}5)"]& &  ((2,5,7),(2,5,7))\end{tikzcd} \] 
    is
    \[\begin{tikzcd} (7,7) \ar[rr,"f^c"] \ar[rr,swap,"(3{,}6)"]& &  ((2,5,7),(2,5,7)).\end{tikzcd} \] 
\end{example}

\begin{proposition}
    Suppose $f:S \to T$ and $g:U \to T$ are nested tuple morphisms. If $f$ is injective and $g$ is a complement of $f$, then $L_{g}$ is a $\size(T)$-complement of $L_f$.
\end{proposition}
\begin{proof}
    This follows from Proposition \ref{complementofmorphisminCproposition} and Lemma \ref{flatteningcomplementationlemma} since 
    \begin{align*}
    (L_f)^\flat & = L_{f^\flat}\\
    (L_{g})^\flat & = L_{g^\flat}.
    \end{align*}
\end{proof}

\begin{proposition}\label{coalescedlayoutoftuplemorphismcomplement}
    If $f:S \to T$ is an injective nested tuple morphism, then 
    \[
    \coalesce(L_{f^c}) = \comp(L_f,\size(T)).
    \]
\end{proposition}
\begin{proof} Since $f^c$ is obtained from $(f^\flat)^c$ by post-composing with a reparenthesization isomorphism, it follows that 
\[
L_{f^c} = L_{(f^\flat)^c}
\]
so by Proposition \ref{coalescedlayoutoftuplemorphismcomplement}, it follows that 
\[
\coal^\flat(L_{f^c})  = \comp^\flat(L_f,\size(T)). 
\]
Applying $\coal(-)$ to both sides yields the result. 
\end{proof}

\subsubsection{Composition}

We can use the realization functor of Section \ref{realizationofnestedtuplemorphisms} to prove that composition of nested tuple morphisms is compatible with composition of the associated layouts.
\begin{theorem}\label{compatibilityofcompositioninD}
    If $f$ and $g$ are non-degenerate composable nested tuple morphisms, then 
    \[
    L_{g\circ f} = L_{g} \circ L_{f}.
    \]
\end{theorem}

\begin{proof}
    Suppose $f:S \to T$ and $g: T \to U$ are non-degenerate nested tuple morphisms. We need to check that 
    \begin{enumerate}
        \item {\it $\shape(L_{g \circ f})$ refines $\shape(L_{f})$}: This holds since 
        \[\shape(L_f) = S = \shape(L_{g \circ f}).\]
        \item {\it $L_{g \circ f}$ is coalesced over $\shape(L_{f})$}: This holds since the nested tuple morphism $g \circ f$ is non-degenerate, hence so is the layout $L_{g \circ f}$.
        \item {\it $\Phi_{L_{g \circ f}} = \Phi_{L_{g}} \circ \Phi_{L_{f}}^{\size(L_{g})}$}: Using Lemma \ref{nestedlayoutfunctionlemma}, we have 
        \begin{align*}
        \Phi_{L_{g \circ f}}^{\size(U)} &  = |g \circ f| \\
        &  = |g| \circ |f| \\
        & = \Phi_{L_g}^{\size(U)} \circ \Phi_{L_f}^{\size(T)}\\
        \end{align*}
        and by postcomposing with the inclusion $[0,\size(U)) \subset \mathbb{Z}$, and observing that $\size(T) = \size(L_{g})$, the result follows.
    \end{enumerate}
\end{proof}

\subsubsection{Logical division}

Next, we introduce logical division of nested tuple morphisms. This construction is obtained from flat division by introducing nesting profiles, with no compatibility constraints.

\begin{definition}
     Suppose $f$ and $g$ are nested tuple morphisms. We say $g$ {\it divides} $f$ if $g$ and $f$ are composable. In other words,
     \[
     \codomain(g) = \domain(f).
     \]
\end{definition}

% \begin{definition}
%     Suppose $g:U \to S'$ and $f:S \to T$ are nested tuple morphisms, and that $g$ divides $f$. Let $\bar{g}$ be the composite
%     \[\begin{tikzcd}
%         U  \ar[r,"\cong"] & \flatten(U) \ar[rr,"g^\flat"]& & \flatten(S') \ar[r,rightarrowtail]  & S^\flat
%     \end{tikzcd}\]
%     and let $\bar{f}$ be the composite
%     \[ \begin{tikzcd} 
%     S^\flat \ar[rr,"f^\flat"] & &  T^\flat \ar[r,"\cong"] & T.
%     \end{tikzcd}\]
%     We define the {\it logical division} 
%     \[
%     f \oslash g: (U,U^*) \to T
%     \]
%     of $f$ by $g$ to be the nested tuple morphism 
%     \[
%     f \oslash g = \bar{f} \circ (\bar{g},\bar{g}^c).
%     \]
% \end{definition}

% \begin{remark}
%     If $S' = S$, then we can define 
%     \[
%     f \oslash g = f \circ (g,g^c).
%     \]
% \end{remark}

\begin{definition}\label{definitionoflogicaldivisionofnestedtuplemorphisms}
    Suppose $g:S \to T$ and $f:T \to U$ are nested tuple morphisms. We define the logical division of $f$ by $g$ to be the nested tuple morphism 
    \[
    f \oslash g = f \circ (g,g^c).
    \]
\end{definition}

\begin{example}
    The logical division of 
    \[\begin{tikzcd} ((2,2),2) \ar[rr,"f"] \ar[rr,swap,"(2{,}4{,}*)
"] & &  ((4,2),(4,2)) \end{tikzcd} \]
    by 
    \[ \begin{tikzcd} (2,2) \ar[rr,"g"] \ar[rr,swap,"(1{,}3)"] & &  ((2,2),2) \end{tikzcd} \]
    is  
    \[ \begin{tikzcd} ((2,2),2) \ar[rr,"f \oslash g"] \ar[rr,swap,"(2{,}*{,}4)"] & &  ((4,2),(4,2)). \end{tikzcd} \]
\end{example}

\begin{example}
    The logical division of 
    \[\begin{tikzcd} (8,8,512,512,512) \ar[rr,"f"] \ar[rr,swap,"(*{,}*{,}1{,}2{,}3)"] & &  (512,512,512) \end{tikzcd} \]
    by 
    \[ \begin{tikzcd} (8,512) \ar[rr,"g"] \ar[rr,swap,"(1{,}5)"] & &  (8,8,512,512,512) \end{tikzcd} \]
    is  
    \[ \begin{tikzcd} ((8,512),(8,512,512)) \ar[rr,"f \oslash g"] \ar[rr,swap,"(*{,}1{,}*{,}2{,}3)"] & &  ((4,2),(4,2)). \end{tikzcd} \]
\end{example}

\begin{proposition}\label{coalesceoflogicaldivision}
    If $g:S \to T$ and $f:T \to U$ are non-degenerate nested tuple morphisms, then 
    \[
    \coalesce(L_{f \oslash g}) = \coalesce(L_f \oslash L_g).
    \]
\end{proposition}

\begin{proof}
By Proposition \ref{coalescedlayoutoftuplemorphismcomplement}, we have
\[
\coalesce(\comp(L_g,\size(L_f))) = \coalesce(L_{g^c})
\]
and we compute 
    \begin{align*}
        \coalesce(L_f \oslash L_g) & = \coalesce(L_f \circ \left(L_g , \comp(L_g,\size(L_f))\right))\\
        & = \coalesce( L_{f} \circ (L_g, L_{g^c}))\\
        & = \coalesce(L_{f} \circ L_{(g,g^c)})\\
        & = \coalesce(L_{f} \circ L_{(g,g^c)})\\
        & = \coalesce(L_{f \circ (g,g^c)})\\
        & = \coalesce(L_{f \oslash g}).
    \end{align*}
\end{proof}

\begin{proposition}
    If $f$ and $g$ are nested tuples and $g$ divides $f$, then 
    \[
    (f \oslash g)^\flat = f^\flat\oslash^\flat g^\flat.
    \]
\end{proposition}

\begin{proof}
    We compute 
    \begin{align*}
    (f \oslash g)^\flat & = (f \circ (g,g^c))^\flat\\
    & = f^\flat \circ (g,g^c)^\flat\\
    & = f^\flat \circ (g^\flat \star (g^c)^\flat)\\
    & = f^\flat \circ (g^\flat \star (g^\flat)^c)\\
    & = f^\flat \oslash^\flat g^\flat.
    \end{align*}
\end{proof}

\subsubsection{Logical products}
In this section, we define the logical product of nested tuple morphisms.
\begin{definition}\label{definitionoflogicalproductofnestedtuplemorphisms}
Suppose $f$ and $g$ are nested tuple morphisms. We say $f$ and $g$ are {\it product admissible} if $\codomain(g) = \domain(f^c)$. If $f$ and $g$ are product admissible we define the {\it logical product} of $f$ and $g$ to be the nested tuple morphism 
\[
f \otimes g = (f , f^c \circ g).
\]
\end{definition}

\begin{example}
    The nested tuple morphisms 
    \[\begin{tikzcd} 
    (8,8) \ar[rr,"f"] \ar[rr,swap,"(1{,}2)"] & &  (8,8,16,16)
    \end{tikzcd} \]
    and 
    \[\begin{tikzcd} 
    (16,16) \ar[rr,"g"] \ar[rr,swap,"(1{,}2)"] & &  (16,16)
    \end{tikzcd} \]
    are product admissible, and their logical product is
    \[\begin{tikzcd} 
    ((8,8),(16,16))  \ar[rr,"f\otimes g"] \ar[rr,swap,"(1{,}2{,}3{,}4)"] & &  (8,8,16,16).
    \end{tikzcd} \]
\end{example}
\begin{example}
    The nested tuple morphisms 
    \[\begin{tikzcd} 
    (128,128) \ar[rr,"f"] \ar[rr,swap,"(3{,}4)"] & &  (32,32,128,128)
    \end{tikzcd} \]
    and 
    \[\begin{tikzcd} 
    (32) \ar[rr,"g"] \ar[rr,swap,"(2)"] & &  (32,32)
    \end{tikzcd} \]
    are product admissible, and their logical product is
    \[\begin{tikzcd} 
    ((128,128),(32))  \ar[rr,"f\otimes g"] \ar[rr,swap,"(3{,}4{,}2)"] & &  (32,32,128,128).
    \end{tikzcd} \]
\end{example}

\begin{proposition}\label{logicalproductcompatibility}
    Suppose $f$ and $g$ are non-degenerate nested tuple morphisms and that $f$ and $g$ are product-admissible. Then 
    \[
    L_{f \otimes g} = L_f \otimes L_g.
    \]
\end{proposition}

\begin{proof} Suppose $f:S \to T$ and $g:U \to V$ are product admissible, and set
    \[
    L_f^* = \comp(L_f,\size(L_f)\cdot \cosize(L_g))
    \]
    Since $f$ is injective and $\codomain(g) = \domain(f^c)$, it follows that 
    \[\size(L_f) \cdot \cosize(L_g) \leq \size(S) \cdot \size(V) = \size(T).\]
    Using this fact, and the fact that \[\Phi_{\comp(L_f,\size(T))} = \Phi_{L_{f^c}},\]
    we have
    \begin{align*}
    L_f^* \circ L_g & = \comp(L_f,\size(T)) \circ L_g\\
    & = L_{f^c} \circ L_g.
    \end{align*}
    Using this fact, we compute 
    \begin{align*}
        L_f \otimes L_g & = (L_f , L_f^* \circ L_g)\\
        & = (L_f , L_{f^c} \circ L_g)\\
        & = (L_f , L_{f^c \circ g})\\
        & = L_{(f , f^c \circ g)}\\
        & = L_{f \otimes g}
    \end{align*}
\end{proof}

\newpage

\chapter{Computations}\label{computationschapter}

The categories $\catstyle{Tuple}$ and $\catstyle{Nest}$ offer a powerful framework for computing with tractable layouts. It is frequently the case that in practice, however, one comes across tractable layouts $A$ and $B$ that are composable in the context of \verb|cute| but whose standard representations are neither composable in $\catstyle{Tuple}$ nor $\catstyle{Nest}$. This chapter is dedicated to the explication of how one may nevertheless use the categories $\catstyle{Tuple}$ and $\catstyle{Nest}$ to compute the {\it composition}, {\it logical division}, and {\it logical product} of tractable layouts, using the notion of {\bf mutual refinement}. We introduce this notion in \Cref{subsectionmutualrefinements}, present an algorithm for computing mutual refinements in Algorithm \ref{mutualrefinementalgorithm}, and work through many explicit examples.

\section{Composition of tractable layouts}
Suppose we want to compute the composition $B \circ A$ of the tractable layouts
\begin{align*}
    A & = (6,6):(6,1),\\
    B & = (12,3,6):(1,72,12).
\end{align*}
We might try to compute $B \circ A$ by computing the composite of the standard representations $f$ and $g$ of $A$ and $B$:
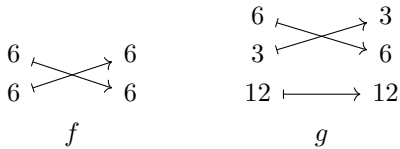
\[ \begin{tikzcd} [row sep = 1, column sep = 8]
         & & & & & & 6 \ar[drr,mapsto] & & 3\\
        6 \ar[drr,mapsto] & & 6 & & & & 3 \ar[urr,mapsto] & & 6\\
        6 \ar[urr,mapsto] & & 6 & & & & 12 \ar[rr,mapsto] & & 12\\
        & f & & & & & & g & \\
        \end{tikzcd} \]
However, these morphisms are not composable, since the codomain $(6,6)$ of $f$ is not equal to the domain $(12,3,6)$ of $g$. This means that we can not use the morphisms $f$ and $g$ to compute the composite $B \circ A$ directly. We can, however, proceed with our computation by finding a {\it mutual refinement} of $(6,6)$ and $(12,3,6)$, as depicted below

        
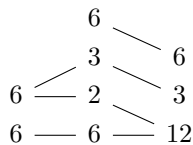
\[ \begin{tikzcd} [row sep = 1, column sep = 8]
           & & 6 \ar[-,drr] & &  \\
          & & 3 \ar[-,drr] & & 6 \\
         6  \ar[-,urr] \ar[-,rr]  & & 2 \ar[-,drr] & & 3 \\
        6 \ar[rr,-] & & 6 \ar[-,swap,rr] & & 12 \\
           &  & &  &  \\
        \end{tikzcd} \]

\noindent This is a device which converts $f$ and $g$ into composable morphisms $f'$ and $g'$:

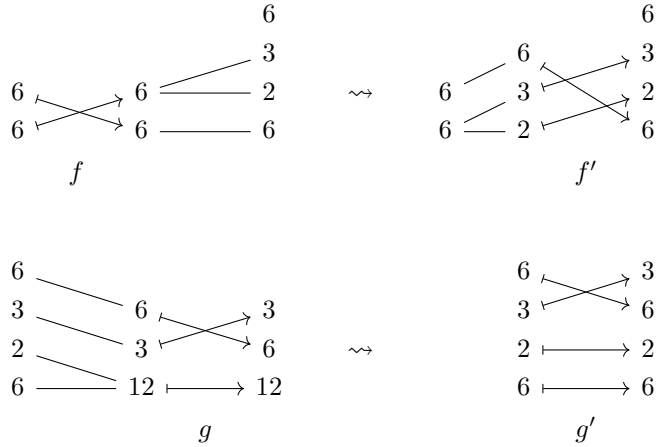
\[
\begin{tikzcd}[row sep = 1,column sep = 8] 
  & &   & & 6 & & & & & & & & 6\\
  & &   & & 3 & & & &  & & 6 \ar[ddrr,mapsto] & & 3 \\
6 \ar[drr,mapsto] & & 6 \ar[rr,-] \ar[rru,-] & & 2 & & \rightsquigarrow & & 6 \ar[urr,-]& & 3 \ar[urr,mapsto] & & 2\\
6 \ar[urr,mapsto] & & 6 \ar[rr,-]  & & 6 & & & & 6 \ar[rr,-] \ar[urr,-] & & 2 \ar[urr,mapsto] & & 6\\
 & f & & & & & & & & & & f' &  \\
 & & & & & & & & & & & & \\[2em]
6 \ar[drr,-] & &    & &    & & & & & & 6 \ar[drr,mapsto] & & 3\\
3 \ar[drr,-] & & 6  \ar[drr,mapsto] & & 3  & & & & & & 3 \ar[urr,mapsto] & & 6\\
2 \ar[drr,-] & & 3  \ar[urr,mapsto] & & 6  & & \rightsquigarrow & & & & 2 \ar[rr,mapsto] & & 2\\
6 \ar[rr,-] & & 12 \ar[rr,mapsto] & & 12 & & & & & & 6 \ar[rr,mapsto] & & 6\\
  & &    &g&    & & & & & &   &g'& \\
\end{tikzcd} 
\]
The morphisms $f'$ and $g'$ are composable, so we may form the composite 
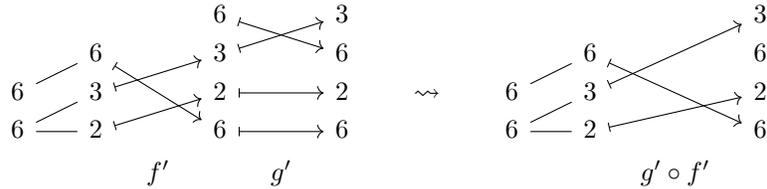
\[ \begin{tikzcd} [row sep = 1, column sep = 8]
              & &        & &  6 \ar[drr,mapsto] & & 3 & & & & & &  & & 3\\
        & &  6 \ar[ddrr,mapsto] & &  3 \ar[urr, mapsto] & & 6 & & & & & & 6 \ar[ddrr,mapsto] & & 6\\
       6 \ar[rru,-] & & 3 \ar[urr,mapsto] & &   2 \ar[rr,mapsto] & & 2  & & \rightsquigarrow & & 6 \ar[urr,-] & & 3 \ar[uurr,mapsto] & & 2\\
       6 \ar[rr,-] \ar[urr,-] & &  2 \ar[urr,mapsto] & &  6 \ar[rr,mapsto] & & 6 & & & & 6 \ar[rr,-] \ar[urr,-] & & 2 \ar[urr,mapsto] & & 6\\
        & & & f' &  & g' & & & & & & & & g' \circ f'& \\
        \end{tikzcd} \]
and computing the encoded layout yields
        \begin{align*}
        B \circ A & = L_{g' \circ f'} = ((2,3),6):((6,72),1).
        \end{align*}

\noindent The goal of this section is to formalize this computational process into an {\it algorithm} for computing the composite of tractable layouts $A$ and $B$. As we saw in our example, the non-trivial steps in our computation were
\begin{enumerate}
    \item finding a mutual refinement of certain (nested) tuples, and 
    \item using the mutual refinement to convert $f$ and $g$ into composable morphisms $f'$ and $g'$.
\end{enumerate}
We dedicate the following two sections to the explication of these steps.

\subsection{Mutual refinements}
\label{subsectionmutualrefinements}

Before giving a precise definition of mutual refinements using the categorical framework of Chapter \ref{categorieschapter}, we give an informal overview. Consider the tuples $(6,6)$ and $(12,3,6)$ of our motivating example. We asserted that the diagram 
        
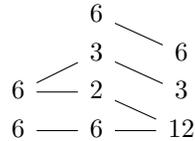
\[ \begin{tikzcd} [row sep = 1, column sep = 8]
           & & 6 \ar[-,drr] & &  \\
          & & 3 \ar[-,drr] & & 6 \\
         6  \ar[-,urr] \ar[-,rr]  & & 2 \ar[-,drr] & & 3 \\
        6 \ar[rr,-] & & 6 \ar[-,swap,rr] & & 12 \\
           &  & &  &  \\
        \end{tikzcd} \]
is a mutual refinement of $(6,6)$ and $(12,3,6)$. We can give a more precise description of this mutual refinement as follows. The left half of the diagram represents the refinement $(6,6) \twoheadleftarrow (6,(2,3))$, and the right half of the diagram represents the refinement $((6,2),3,6) \twoheadrightarrow (12,3,6)$:
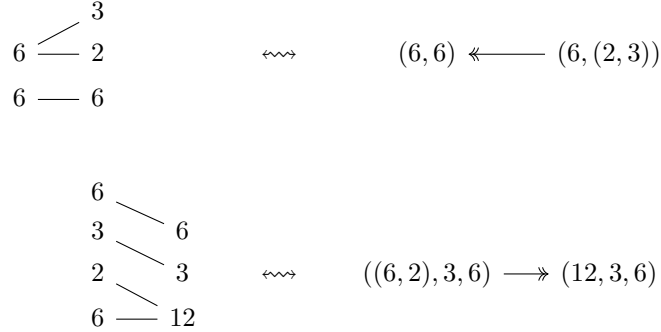
\[ \begin{tikzcd} [row sep = 1, column sep = 8]
 & & 3 & & & & & & & & \\
6 \ar[rr,-] \ar[urr,-] & & 2 & & & & \leftrightsquigarrow & & (6,6) & & \ar[ll, two heads] (6,(2,3)) \\
6 \ar[rr,-] & & 6 & & & & & & & & \\
         & &  & & & & & &  & & \\[2em]
 & & 6 \ar[drr,-] & & & & & & & & \\
 & & 3 \ar[drr,-]& & 6 & & & & & & \\
 & & 2 \ar[drr,-] & & 3 & & \leftrightsquigarrow & & ((6,2),3,6) \ar[rr,two heads] & & (12,3,6) \\
 & & 6 \ar[rr,-] & & 12 & & & & & & \\
\end{tikzcd} \]
The fact that the two halves of the diagram may be glued together corresponds to the fact that the nested tuple $(6,(2,3))$ {\it divides} $((6,2),3,6)$, which we denote 
\[
\begin{tikzcd} 
(6,(2,3)) \ar[r,rightarrowtail] & ((6,2),3,6).
\end{tikzcd}
\]
Putting these observations together, we may express our mutual refinement precisely as 
        
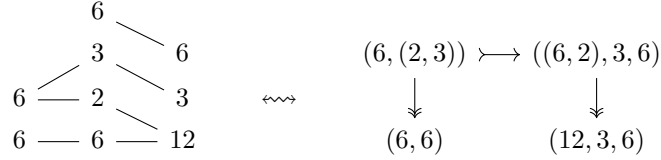
\[ \begin{tikzcd} [row sep = 1, column sep = 8]
           & & 6 \ar[-,drr] & & & & & &  & & \\
          & & 3 \ar[-,drr] & & 6 & & & & (6,(2,3)) \ar[dd, two heads] \ar[rr,rightarrowtail] & & ((6,2),3,6) \ar[dd, two heads] \\
         6  \ar[-,urr] \ar[-,rr]  & & 2 \ar[-,drr] & & 3 & & \leftrightsquigarrow & & & & \\
        6 \ar[rr,-] & & 6 \ar[-,swap,rr] & & 12 & & & & (6,6) & & (12,3,6) \\
           &  & &  &  & & & & \\
        \end{tikzcd} \]
where we opt to depict the refinements $(6,6) \twoheadleftarrow (6,(2,3))$ and $((6,2),3,6) \twoheadrightarrow (12,3,6)$ vertically. We can now give a precise definition of mutual refinements.

\begin{definition}
Suppose $T$ and $U$ are nested tuples. A {\it mutual refinement} of $(T,U)$ is a diagram of the form
\[ \begin{tikzcd} 
T' \ar[r,rightarrowtail] \ar[d, two heads] & U' \ar[d,two heads] \\
T & U
\end{tikzcd} \]
Explicitly, this is a pair of nested tuples $(T',U')$ such that \begin{enumerate}
    \item $T'$ refines $T$,
    \item $U'$ refines $U$, and
    \item $T'$ divides $U'$. 
\end{enumerate}
\end{definition}

\begin{example}
    A mutual refinement of $T = (6,6)$ and $U = (2,6,3)$ is given by
    \[ \begin{tikzcd} 
((2,3),(2,3)) \ar[r,rightarrowtail] \ar[d, two heads] & (2,(3,2),3) \ar[d,two heads] \\
(6,6) & (2,6,3)
\end{tikzcd} \]
We depict this mutual refinement as follows.
\[ \begin{tikzcd} [row sep = 1, column sep = 8] 
  & & 3 & & \\
6 \ar[rr,-] \ar[rru,-] & & 2 & & \ar[ull,-] 3\\
 & & 3 & & \ar[ll,-] \ar[ull,-] 6\\
6 \ar[rr,-] \ar[rru,-] & & 2 & & \ar[ll,-] 2\\
\end{tikzcd} \]
\end{example}

\begin{example}
    A mutual refinement of $T = (8,8,8)$ and $U = (2,8,8,8)$ is given by
    \[ \begin{tikzcd} 
 ((2,4),(2,4),(2,4))\ar[r,rightarrowtail] \ar[d, two heads] & (2,(4,2),(4,2),(4,2)) \ar[d,two heads] \\
(8,8,8) & (2,8,8,8)
\end{tikzcd} \]
We depict this mutual refinement as follows.
\[ \begin{tikzcd} [row sep = 1, column sep = 8] 
  & & 2 & & \\
  & & 4 & &\ar[ll,-] \ar[llu,-]  8\\
8 \ar[rr,-] \ar[rru,-]& & 2 & & \\
  & & 4 & &\ar[ll,-] \ar[llu,-]  8\\
8 \ar[rr,-] \ar[rru,-]& & 2 & & \\
  & & 4 & &\ar[ll,-] \ar[llu,-]  8\\
8 \ar[rr,-] \ar[rru,-] & & 2 & & \ar[ll,-] 2\\
\end{tikzcd} \]
\end{example}

\begin{example}
A mutual refinement of $T = (4, 2, 2, 32)$ and $U = (32, 32)$ is given by
    \[ \begin{tikzcd} 
 (4,2,2,(2,16))\ar[r,rightarrowtail] \ar[d, two heads] & ((4,2,2,2),(16,2)) \ar[d,two heads] \\
(4,2,2,32) & (32,32)
\end{tikzcd} \]
We depict this mutual refinement as follows.

\[ \begin{tikzcd} [row sep = 1, column sep = 8]
 & & 2 & & \\
 & & 16 & & \ar[ll,-]\ar[ull,-]32\\
32\ar[rr,-] \ar[urr,-]  & & 2 & & \\
2\ar[rr,-]  & & 2 & & \\
2\ar[rr,-]  & & 2 & & \\
4 \ar[rr,-] & & 4 & & \ar[ll,-]\ar[ull,-]\ar[uull,-]\ar[uuull,-] 32\\
\end{tikzcd} \]
\end{example}

\begin{example}
    If $T = (8,8)$ and $U = (3,8,8)$, then there does not exist a mutual refinement of $T$ and $U$.
\end{example}

\begin{example}
    If $T$ and $U$ are tuples with $\size(T) = 2^k$ and $\size(L) = 2^\ell$ with $k \leq \ell$, then there exists a mutual refinement of $T$ and $U$.  More generally, if $T$ and $U$ are tuples where $\size(T) \leq \size(U)$ are powers of some fixed integer, then there exists a mutual refinement of $T$ and $U$.
\end{example}

\begin{observation} In each of the previous examples, we have considered mutual refinements of {\it flat} tuples $T$ and $U$. The definition of mutual refinement, however, allows $T$ and $U$ to be any {\it nested} tuples. In any case, restricting to the flat case is no loss of generality, because there is a one-to-one correspondence between mutual refinements of a pair of nested tuples $(T,U)$, and mutual refinements of their flattenings $(T^\flat,U^\flat)$ (see Lemma \ref{mutualrefinementsofflattenings}). In particular, there exists a mutual refinement of $(T,U)$ if and only if there exists a mutual refinement of $(T^\flat,U^\flat)$.  
\end{observation}

Having made the appropriate definitions, we provide an algorithm for computing a mutual refinement of $(T,U)$. 

\begin{BreakableAlgorithm}{Mutual refinement algorithm}\label{mutualrefinementalgorithm}
\begin{algorithmic}[1]
\State \textbf{Input:} Nested tuples $T$ and $U$.
\State \textbf{Output:} A mutual refinement $(T',U')$ of $(T,U)$, if one exists, else $\mathbf{None}$.
\algsep

\State $X \gets T$; \ $Y \gets U$
\State $X',\, Y',\, X_{\mode},\, Y_{\mode} \gets ()$
\State $i \gets 1$; \ $j \gets 1$

\While{$i \le \len(X)$ \textbf{ and } $j \le \len(Y)$}
  \If{$\entry_i(X) = \entry_j(Y)$}
    \State append $\entry_i(X)$ to $X_{\mode}$; 
    append $X_{\mode}$ to $X'$; 
    $X_{\mode}\gets()$
    \State append $\entry_j(Y)$ to $Y_{\mode}$; 
    \State append $Y_{\mode}$ to $Y'$; 
    \State $Y_{\mode}\gets()$
    \State $i \gets i+1$;
    \State $j \gets j+1$
  \ElsIf{$\entry_i(X)$ divides $\entry_j(Y)$}
    \State append $\entry_i(X)$ to $X_{\mode}$; 
    \State append $X_{\mode}$ to $X'$; 
    \State $X_{\mode}\gets()$
    \State append $\entry_i(X)$ to $Y_{\mode}$
    \State $\entry_j(Y) \gets \entry_j(Y)/\entry_i(X)$; 
    \State $i \gets i+1$
  \ElsIf{$\entry_j(Y)$ divides $\entry_i(X)$}
    \State append $\entry_j(Y)$ to $X_{\mode}$; 
    \State append $\entry_j(Y)$ to $Y_{\mode}$; 
    \State append $Y_{\mode}$ to $Y'$; 
    \State $Y_{\mode}\gets()$
    \State $\entry_i(X) \gets \entry_i(X)/\entry_j(Y)$; 
    \State $j \gets j+1$
  \Else
    \State \Return $\mathbf{None}$
  \EndIf
\EndWhile

\If{$Y_{\mode} \neq ()$}
  \State append $\entry_j(Y)$ to $Y_{\mode}$; 
  \State append $Y_{\mode}$ to $Y'$; 
  \State $j \gets j+1$
\EndIf
\While{$j < \len(Y)$}
  \State append $\entry_j(Y)$ to $Y'$; 
  \State $j \gets j+1$
\EndWhile

\State $T' \gets (X')_{\profile(T)}$; 
\State $U' \gets (Y')_{\profile(U)}$
\State \Return $(T',U')$
\end{algorithmic}
\end{BreakableAlgorithm}

\subsection{From mutual refinements to composable morphisms}
Recall that in order to compute the composition $B \circ A$ of 
\begin{align*}
    A & = (6,6):(6,1) \text{ and }\\
    B & = (12,3,6):(1,72,12),
\end{align*}
we constructed tuple morphisms 
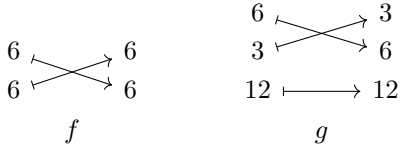
\[ \begin{tikzcd} [row sep = 1, column sep = 8]
         & & & & & & 6 \ar[drr,mapsto] & & 3\\
        6 \ar[drr,mapsto] & & 6 & & & & 3 \ar[urr,mapsto] & & 6\\
        6 \ar[urr,mapsto] & & 6 & & & & 12 \ar[rr,mapsto] & & 12\\
        & f & & & & & & g & \\
        \end{tikzcd} \]
and a mutual refinement.
        
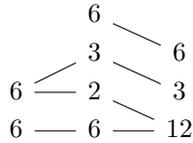
\[ \begin{tikzcd} [row sep = 1, column sep = 8]
           & & 6 \ar[-,drr] & &  \\
          & & 3 \ar[-,drr] & & 6 \\
         6  \ar[-,urr] \ar[-,rr]  & & 2 \ar[-,drr] & & 3 \\
        6 \ar[rr,-] & & 6 \ar[-,swap,rr] & & 12 \\
           &  & &  &  \\
        \end{tikzcd} \]

\noindent The next step in our computation is to use our mutual refinement to convert $f$ and $g$ into composable morphisms $f'$ and $g'$. Before giving a formal, categorical definition of this process, let's illustrate the process with an example.

We construct $f'$ from $f$ and the left half of our mutual refinement:
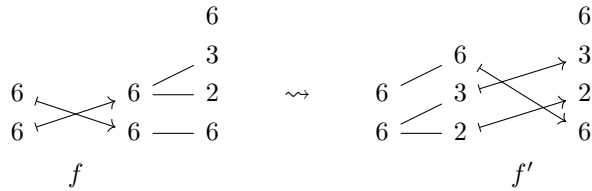
\[
\begin{tikzcd}[row sep = 1,column sep = 8] 
  & &   & &  6 & & & & & & & & 6 \\
  & &   & & 3 & & & &  & & 6 \ar[ddrr,mapsto] & & 3 \\
6 \ar[drr,mapsto] & & 6 \ar[rr,-] \ar[rru,-] & & 2 & & \rightsquigarrow & & 6 \ar[urr,-]& & 3 \ar[urr,mapsto] & & 2\\
6 \ar[urr,mapsto] & & 6 \ar[rr,-]  & & 6 & & & & 6 \ar[rr,-] \ar[urr,-] & & 2 \ar[urr,mapsto] & & 6\\
 & f & & & & & & & & & & f' &  \\
\end{tikzcd} 
\]
This construction is made by making the replacement
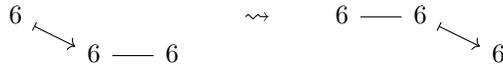
\[ \begin{tikzcd} [row sep = 1, column sep = 8]
                  % & &              & &   & &                  & &               & &       & &  \\
6 \ar[drr,mapsto] & &              & &   & & \rightsquigarrow & & 6 \ar[rr,-]  & &   6 \ar[drr,mapsto]                 & &  \\
                  & & 6  \ar[rr,-] & & 6 & &                  & &               & &                    & & 6\\
\end{tikzcd} \]
and making the replacement
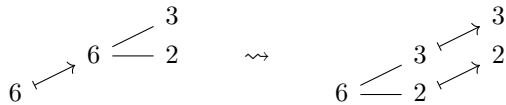
\[ \begin{tikzcd} [row sep = 1, column sep = 8]
 & & & & 3 & & & & & &  & & 3\\
 & & 6 \ar[urr,-] \ar[rr,-] & & 2 & & \rightsquigarrow & &  & & 3 \ar[urr,mapsto] & & 2\\
6 \ar[urr,mapsto] & & & & & & & & 6 \ar[rr,-] \ar[urr,-] & & 2 \ar[urr,mapsto] & & \\
\end{tikzcd} \]
More generally, we make the replacement
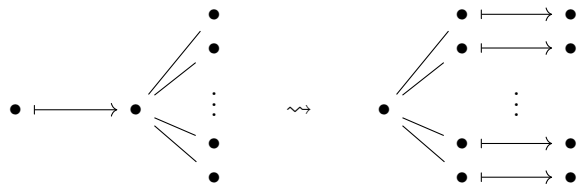
\[ \begin{tikzcd} [row sep = 1, column sep = 8]
 & & & & & & \bullet & & & & & & \bullet \ar[rr,mapsto] & & \bullet \\
 & & & & & & \bullet & & & &  & & \bullet \ar[rr,mapsto] & & \bullet \\
 \bullet \ar[rrrr,mapsto] & & & & \bullet \ar[rruu,-] \ar[rru,-]\ar[rrd,-] \ar[rrdd,-] & & \vdots  & & \rightsquigarrow & & \bullet \ar[uurr,-] \ar[urr,-] \ar[drr,-] \ar[ddrr,-]  & &  & \vdots &   \\
 & & & & & & \bullet & &  & & & &  \bullet \ar[rr,mapsto] & & \bullet \\
  & & & & & & \bullet & & & &  & &  \bullet \ar[rr,mapsto] & & \bullet \\
\end{tikzcd} \]

The process for constructing $g'$ from $g$, and the right half of our mutual refinement is similar. 
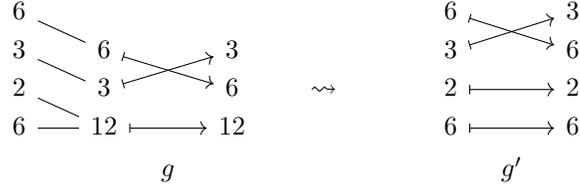
\[
\begin{tikzcd}[row sep = 1,column sep = 8] 
6 \ar[drr,-] & &    & &    & & & & & & 6 \ar[drr,mapsto] & & 3\\
3 \ar[drr,-] & & 6  \ar[drr,mapsto] & & 3  & & & & & & 3 \ar[urr,mapsto] & & 6\\
2 \ar[drr,-] & & 3  \ar[urr,mapsto] & & 6  & & \rightsquigarrow & & & & 2 \ar[rr,mapsto] & & 2\\
6 \ar[rr,-] & & 12 \ar[rr,mapsto] & & 12 & & & & & & 6 \ar[rr,mapsto] & & 6\\
  & &    &g&    & & & & & &   &g'& \\
\end{tikzcd} 
\]
This construction is made by making the replacements
\[ \begin{tikzcd} [row sep = 1, column sep = 8]
6 \ar[rr,-] & & 6  \ar[rr,mapsto] & & 6 & & \rightsquigarrow & & 6 \ar[rr,mapsto] & & 6\\
& & & & & & & & & & \\[2em]
3 \ar[rr,-] & & 3  \ar[rr,mapsto] & & 3 & & \rightsquigarrow & & 3 \ar[rr,mapsto] & & 3\\
& & & & & & & & & & \\[2em]
2 \ar[rrd,-] & & & & & & & &2 \ar[rr,mapsto] & & 2\\
6 \ar[rr,-] & & 12 \ar[rr,mapsto] & & 12 & & \rightsquigarrow & & 6 \ar[rr,mapsto] & & 6\\
\end{tikzcd} \]

More generally, we make the replacement
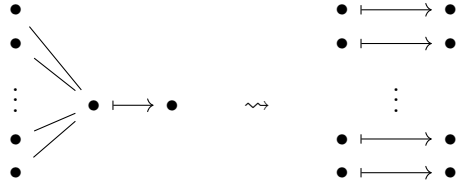
\[ \begin{tikzcd} [row sep = 1, column sep = 8] 
\bullet \ar[ddrr,-] & & & & & & & & \bullet \ar[rr,mapsto] & & \bullet \\
\bullet \ar[drr,-]  & & & & & & & & \bullet \ar[rr,mapsto] & & \bullet \\
\vdots              & & \bullet \ar[rr,mapsto] & & \bullet & & \rightsquigarrow & & & \vdots & \\
\bullet \ar[urr,-]  & &  & & & & & &  \bullet \ar[rr,mapsto] & & \bullet \\
\bullet \ar[uurr,-] & & & & & & & & \bullet \ar[rr,mapsto]  & & \bullet \\
\end{tikzcd} \]

Having given an informal description of our procedure, we make things precise as follows.

\begin{construction}\label{constructionofcomposablemorphismsfrommutualrefinement}
    Suppose $f:S \to T$ and $g:U \to V$ are nested tuple morphisms, and $(T',U')$ is a mutual refinement of $(T,U)$. Then we may use the {\it pullback} and {\it pushforward} constructions of section \ref{categoryrefinements} to form the diagram:
    \[ \begin{tikzcd} 
    S' \ar[r,"\tilde{f}"] \ar[d,two heads] \arrow[dr, phantom, "\lrcorner", very near start] & T' \ar[d,two heads] \ar[r,"i",rightarrowtail] & U' \ar[r,"\tilde{g}"] \ar[d,two heads] & V' \ar[d,two heads]\arrow[dl, phantom, "\llcorner", very near start] \\
    S \ar[r,swap,"f"] & T & U \ar[r,swap,"g"] & V
    \end{tikzcd} \]
    If we set $f' = i \circ \tilde{f}$ and $g' = \tilde{g}$, then 
    \[ \begin{tikzcd} 
    S' \ar[rr,"f'"] & & U' \ar[rr,"g'"] & & V'
    \end{tikzcd} \]
    are composable nested tuple morphisms.
\end{construction}

\subsection{The composition algorithm} 

\begin{BreakableAlgorithm}{Tractable Layout Composition Algorithm}
\label{tractablelayoutcompositionalgorithm}
\begin{algorithmic}[1]
\State \textbf{Input:} Tractable layouts $A$ and $B$.
\State \textbf{Output:} A weak composite $C$ of $A$ and $B$, if one exists, else $\mathbf{None}$..
\algsep

\State Take the standard representations
\[
\begin{tikzcd} 
S \ar[rr,"f"] & & T 
\qquad\qquad
U \ar[rr,"g"] & & V
\end{tikzcd}
\]
of $A$ and $\coalesce(B)$, respectively.

\State Use Algorithm~\ref{mutualrefinementalgorithm} to produce a mutual refinement
\[
\begin{tikzcd}
T' \ar[r,rightarrowtail] \ar[d,two heads] & U' \ar[d,two heads] \\
T & U
\end{tikzcd}
\]
of $(T,U)$. If there does not exist a mutual refinement of $(T,U)$, return $\mathbf{None}$. 

\State Use Construction~\ref{constructionofcomposablemorphismsfrommutualrefinement} to obtain the composable nested tuple morphisms
\[
\begin{tikzcd}
S' \ar[r,"f'"] & U' \ar[r,"g'"] & V'
\end{tikzcd}
\]

\State Compose $f'$ and $g'$, and compute the encoded layout
\[
C = L_{g' \circ f'}
\]

\State \Return $C$
\end{algorithmic}
\end{BreakableAlgorithm}

\begin{theorem}
    If $A$ and $B$ are tractable layouts, then the output $C$ of the previous algorithm is a weak composite of $A$ and $B$. Consequently,
    \[
    B \circ A = \coalesce(C,\shape(A)).
    \]
\end{theorem}

\begin{proof}
    Proposition \ref{pushforwardproposition} and tells us that 
    \[\Phi_{L_{g'}} = \Phi_{L_g} = \Phi_{\coalesce(B)} = \Phi_B,\]
    and Proposition \ref{pullbackproposition} and Example \ref{definitionofcodomainexpansion} tell us that 
    \[\Phi_{L_{f'}} = \Phi_{L_f}  = \Phi_A.\]
    Theorem \ref{compatibilityofcompositioninD} then implies that 
    \begin{align*}
    \Phi_C = \Phi_{L_{g' \circ f'}} & = \Phi_{g'} \circ \Phi_{f'}^{\size(U')}\\
    & = \Phi_B \circ \Phi_A^{\size(B)}.
    \end{align*}
    By construction, the shape $S'$ of $L_{f'}$ refines the shape $S$ of $A$, so we conclude that $C$ is a weak composite of $A$ and $B$. 
\end{proof}

% \begin{algorithm}[H]
% \caption{Tractable layout composition algorithm}
% \KwIn{Tractable layouts $A$ and $B$}
% \KwOut{The composite $B \circ A$, if it exists, else $\mathbf{None}$.}

% initialize $f = f_A:S \to T$\;
% initialize $g = f_{\coalesce(B)}:U \to V$.\;
% Use Algorithm \ref{mutualrefinementalgorithm} to produce a mutual refinement $T',U'$ of $T,U$, if one exists, else return $\mathbf{None}$.\;
% set $f' = T'^*f$\;
% set $i:T' \rightarrowtail U'$\;
% set $g' = U'_*f$ \;
% set $L = \coalesce(L_{g' \circ i \circ  f'},S)$\;
% \Return{$L$}\;
% \end{algorithm}

\subsection{Examples}

In this section we illustrate how Algorithm \ref{tractablelayoutcompositionalgorithm} may be used to compute the composition $B \circ A$ of tractable layouts $A$ and $B$.

\begin{example}
    Suppose $A = (4):(1)$, and $B = (2,2):(2,1)$.
    \begin{enumerate} 
        \item Take the standard representations of $A$ and $\coalesce(B) = B$.
        % Present $A$ as $L_f$ and $\coalesce(B)$ as $L_g$, where $f$ and $g$ are as shown below.
        % \[ \begin{tikzcd} 
        % f:(4) \ar[r,"(1)"] & (4)\\
        % g:(2,2) \ar[r,"(2{,}1)"] & (2,2).
        % \end{tikzcd} \]
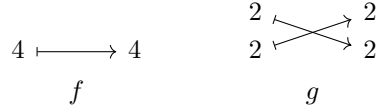
        \[ \begin{tikzcd} [row sep = 1, column sep = 8]
         & & & & & & 2 \ar[drr,mapsto] & & 2\\
        4 \ar[rr,mapsto] & & 4 & & & & 2 \ar[urr,mapsto] & & 2\\
        & f & & & & & & g & \\
        \end{tikzcd} \]
        \item Apply Algorithm \ref{mutualrefinementalgorithm} to obtain the mutual refinement
        % \[ \begin{tikzcd} [row sep = 1, column sep = 8]
        %     T = (4) & & \rightsquigarrow & & T' = ((2,2))\\
        %     U = (2,2) & & \rightsquigarrow & & U' = (2,2)\\
        % \end{tikzcd} 
        % \]
        
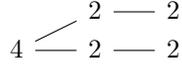
\[ \begin{tikzcd} [row sep = 1, column sep = 8]
          & & 2 \ar[-,rr]& & 2  \\
        4 \ar[-,rr] \ar[-,urr] & & 2 \ar[-,rr] & & 2 \\
        \end{tikzcd} \]
        \item Form  the diagram 
        % \[ \begin{tikzcd} 
        % ((2,2)) \ar[r,"{(}1{,}2{)}"] \ar[d,two heads] \arrow[dr, phantom, "\lrcorner", very near start]& ((2,2)) \ar[r,"{(}1{,}2{)}",rightarrowtail] \ar[d,two heads]& (2,2) \ar[r,"{(}2{,}1{)}"] \ar[d,two heads]& (2,2)\arrow[dl, phantom, "\llcorner", very near start] \ar[d,two heads]\\
        % (4) \ar[r,swap,"(1)"] & (4) & (2,2) \ar[r,swap,"{(}2{,}1{)}"] & (2,2)
        % \end{tikzcd} \]
        
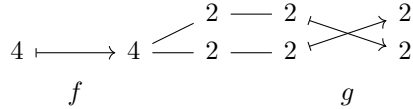
\[ \begin{tikzcd} [row sep = 1, column sep = 8]
         & & & & 2 \ar[-,rr]& & 2 \ar[drr,mapsto] & & 2\\
        4 \ar[rr,mapsto] & & 4 \ar[-,rr] \ar[-,urr] & & 2 \ar[-,rr] & & 2 \ar[urr,mapsto] & & 2\\
        & f & & & & & & g & \\
        \end{tikzcd} 
        \]
        \item Resolve the diagram 
        
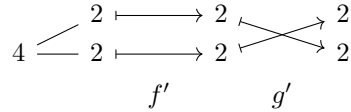
\[
        \begin{tikzcd}[row sep = 1, column sep = 8]
        & &  2 \ar[rr,mapsto] & &  2 \ar[drr,mapsto ]& & 2 & & \\
        4 \ar[rr,-] \ar[urr,-] & & 2 \ar[rr,mapsto] & & 2 \ar[urr,mapsto] & & 2 & & \\
        & & & f' &  & g' & & & \\
        \end{tikzcd} \]
        \item Compose $f'$ and $g'$ to obtain 
        
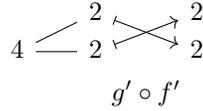
\[ \begin{tikzcd}[row sep = 1, column sep = 8]
        & & 2 \ar[drr,mapsto] & & 2\\
        4 \ar[rr,-] \ar[urr,-] & & 2 \ar[urr,mapsto] & & 2\\
        & & & \mathclap{g' \circ f'} & 
        \end{tikzcd} \]
        % \[ \begin{tikzcd}
        % ((2,2)) \ar[r,"(2{,}1)"] &(2,2)
        % \end{tikzcd} 
        % \]
        \item Compute the associated layout 
        \[
        L_{g' \circ f'} = ((2,2)):((2,1)).
        \]
        \item $L_{g' \circ f'}$ is coalesced over $(4)$, so 
        \[
        B \circ A = ((2,2)):((2,1)).
        \]
    \end{enumerate} 
\end{example}

\begin{example}
    Suppose $A = (6,6):(6,1)$, and $B = (12,3,6):(1,72,12)$.
    \begin{enumerate} 
        \item Take the standard representations of $A$ and $\coalesce(B) = B$.
        % Present $A$ as $L_f$ and $\coalesce(B)$ as $L_g$, where $f$ and $g$ are as shown below.
        % \[ \begin{tikzcd} 
        % f:(4) \ar[r,"(1)"] & (4)\\
        % g:(2,2) \ar[r,"(2{,}1)"] & (2,2).
        % \end{tikzcd} \]
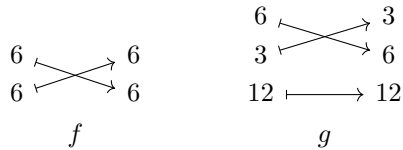
        \[ \begin{tikzcd} [row sep = 1, column sep = 8]
         & & & & & & 6 \ar[drr,mapsto] & & 3\\
        6 \ar[drr,mapsto] & & 6 & & & & 3 \ar[urr,mapsto] & & 6\\
        6 \ar[urr,mapsto] & & 6 & & & & 12 \ar[rr,mapsto] & & 12\\
        & f & & & & & & g & \\
        \end{tikzcd} \]
        \item Apply Algorithm \ref{mutualrefinementalgorithm} to obtain the mutual refinement
        
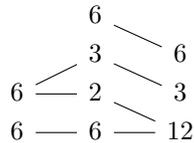
\[ \begin{tikzcd} [row sep = 1, column sep = 8]
           & & 6 \ar[-,drr] & &  \\
          & & 3 \ar[-,drr] & & 6 \\
         6  \ar[-,urr] \ar[-,rr]  & & 2 \ar[-,drr] & & 3 \\
        6 \ar[rr,-] & & 6 \ar[-,swap,rr] & & 12 \\
           &  & &  &  \\
        \end{tikzcd} \]
        \item Form the diagram 
        
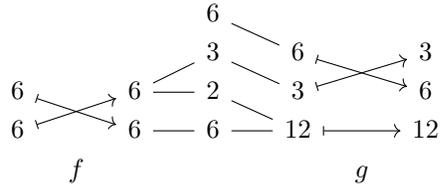
\[
\begin{tikzcd}[row sep = 1,column sep = 8] 
  & &   & & 6  \ar[drr,-] & &    & &\\
  & &   & & 3 \ar[drr,-]  & & 6  \ar[drr,mapsto] & & 3\\
6 \ar[drr,mapsto] & & 6 \ar[rr,-] \ar[rru,-] & & 2 \ar[drr,-]& & 3  \ar[urr,mapsto] & & 6\\
6 \ar[urr,mapsto] & & 6 \ar[rr,-]  & & 6 \ar[rr,-] & & 12 \ar[rr,mapsto] & & 12 \\
 & f & &  & & &    &g& \\
\end{tikzcd} 
\]
        
        \item Resolve the diagram to obtain
        
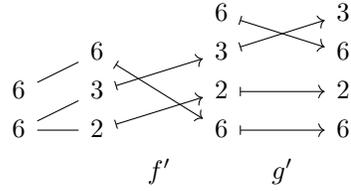
\[ \begin{tikzcd} [row sep = 1, column sep = 8]
              & &        & &  6 \ar[drr,mapsto] & & 3 \\
        & &  6 \ar[ddrr,mapsto] & &  3 \ar[urr, mapsto] & & 6 \\
       6 \ar[rru,-] & & 3 \ar[urr,mapsto] & &   2 \ar[rr,mapsto] & & 2 \\
       6 \ar[rr,-] \ar[urr,-] & &  2 \ar[urr,mapsto] & &  6 \ar[rr,mapsto] & & 6 \\
        & & & f' &  & g' &  \\
        \end{tikzcd} \]
%         \[
% \begin{tikzcd}[row sep = 1,column sep = 8] 
%   & &   & & 6 & & & & & & & & 6\\
%   & &   & & 3 & & & &  & & 6 \ar[ddrr,mapsto] & & 3 \\
% 6 \ar[drr,mapsto] & & 6 \ar[rr,-] \ar[rru,-] & & 2 & & \rightsquigarrow & & 6 \ar[urr,-]& & 3 \ar[urr,mapsto] & & 2\\
% 6 \ar[urr,mapsto] & & 6 \ar[rr,-]  & & 6 & & & & 6 \ar[rr,-] \ar[urr,-] & & 2 \ar[urr,mapsto] & & 6\\
%  & f & & & & & & & & & & f' &  \\
%  & & & & & & & & & & & & \\[2em]
% 6 \ar[drr,-] & &    & &    & & & & & & 6 \ar[drr,mapsto] & & 3\\
% 3 \ar[drr,-] & & 6  \ar[drr,mapsto] & & 3  & & & & & & 3 \ar[urr,mapsto] & & 6\\
% 2 \ar[drr,-] & & 3  \ar[urr,mapsto] & & 6  & & \rightsquigarrow & & & & 2 \ar[rr,mapsto] & & 2\\
% 6 \ar[rr,-] & & 12 \ar[rr,mapsto] & & 12 & & & & & & 6 \ar[rr,mapsto] & & 6\\
%   & &    &g&    & & & & & &   &g'& \\
% \end{tikzcd} 
% \]
        \item Compose $f'$ and $g'$ to obtain 
        
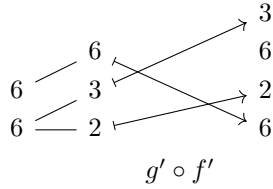
\[ \begin{tikzcd} [row sep = 1, column sep = 8]
         & &  & & 3\\
         & & 6 \ar[ddrr,mapsto] & & 6\\
        6 \ar[urr,-] & & 3 \ar[uurr,mapsto] & & 2\\
        6 \ar[rr,-] \ar[urr,-] & & 2 \ar[urr,mapsto] & & 6\\
        & & & g' \circ f'& \\
        \end{tikzcd} \]
        \item Compute the associated layout 
        \[
        L_{g' \circ f'} = ((2,3),6):((6,72),1).
        \]
        \item $L_{g' \circ f'}$ is coalesced over $(6,6)$, hence
        \[
        B \circ A = ((2,3),6):((6,72),1).
        \]
    \end{enumerate} 
\end{example}

\begin{example}
       Suppose $A =(8,8):(8,1)$, and $B = (16,16):(16,1)$.
    \begin{enumerate} 
        \item Take the standard representations of $A$ and $\coalesce(B) = B$.
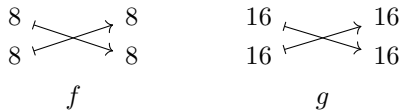
        \[ \begin{tikzcd} [row sep = 1, column sep = 8]
        8 \ar[drr,mapsto] & & 8 & & & & 16 \ar[drr,mapsto] & & 16\\
        8 \ar[urr,mapsto] & & 8 & & & & 16 \ar[urr,mapsto] & & 16\\
        & f & & & & & & g & \\
        \end{tikzcd} \]
        \item Apply Algorithm \ref{mutualrefinementalgorithm} to obtain the mutual refinement
        
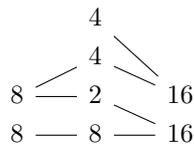
\[ \begin{tikzcd} [row sep = 1, column sep = 8]
         & & 4 & & \\
         & & 4 & & \\
        8 \ar[rr,-] \ar[urr,-] & & 2 & & \ar[ull,-] \ar[uull,-] 16\\
        8 \ar[rr,-] & & 8 & & \ar[ll,-] \ar[ull,-] 16\\
        \end{tikzcd} \]
        \item Form the diagram 
        
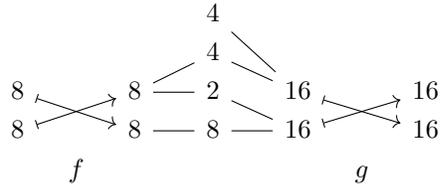
\[
\begin{tikzcd}[row sep = 1,column sep = 8] 
  & &   & & 4 \ar[ddrr,-] & &    & & \\
  & &   & & 4 \ar[drr,-] & &    & & \\
8 \ar[drr,mapsto] & & 8 \ar[rr,-] \ar[urr,-] & & 2 \ar[drr,-] & & 16 \ar[drr,mapsto] & & 16\\
8 \ar[urr,mapsto] & & 8 \ar[rr,-] & & 8 \ar[rr,-] & & 16 \ar[urr,mapsto] & & 16\\
  &f&   & &   & &   &g& \\
\end{tikzcd} 
\]
        
        \item Resolve the diagram to obtain
        
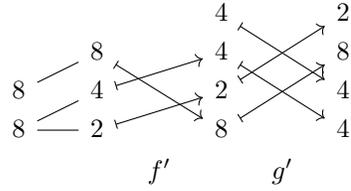
\[ \begin{tikzcd} [row sep = 1, column sep = 8]
        & & & & 4 \ar[ddrr,mapsto] & & 2\\
         & & 8 \ar[ddrr,mapsto] & & 4 \ar[ddrr,mapsto] & & 8\\
        8 \ar[urr,-] & & 4 \ar[urr,mapsto] & & 2 \ar[uurr,mapsto] & & 4\\
        8 \ar[rr,-] \ar[urr,-] & & 2 \ar[urr,mapsto] & & 8 \ar[uurr,mapsto] & & 4\\
        & & & f' & & g' & \\
        \end{tikzcd} \]

        \item Compose $f'$ and $g'$ to obtain 
        
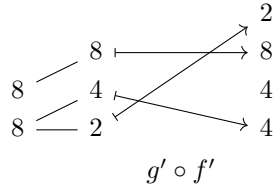
\[ \begin{tikzcd} [row sep = 1, column sep = 8]
         & & & & 2\\
         & & 8 \ar[rr,mapsto] & & 8\\
        8 \ar[urr,-] & & 4 \ar[drr,mapsto] & & 4\\
        8 \ar[rr,-] \ar[urr,-] & & 2 \ar[uuurr,mapsto] & & 4\\
         & & & g' \circ f' & \\
        \end{tikzcd} \]
        \item Compute the associated layout 
        \[L_{g' \circ f'} = ((2,4),8) : ((128,1),16)\]
        \item $L_{g' \circ f'}$ is coalesced over $(8,8)$, hence
        \[B \circ A = ((2,4),8) : ((128,1),16)\]
    \end{enumerate} 
\end{example}

\begin{example}
    Suppose $A =(16,16):(16,1)$, and $B = (8,8,8):(64,8,1)$.
    \begin{enumerate} 
        \item Take the standard representations of $A$ and $\coalesce(B) = B$.
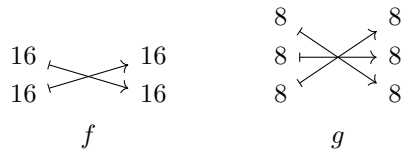
        \[ \begin{tikzcd} [row sep = 1, column sep = 8]
         & & & & & & 8 \ar[ddrr,mapsto] & & 8\\
        16 \ar[drr,mapsto] & & 16 & & & & 8 \ar[rr,mapsto] & & 8\\
        16 \ar[urr,mapsto] & & 16 & & & & 8 \ar[uurr,mapsto] & & 8\\
        & f & & & & & & g & \\
        \end{tikzcd} \]
        \item Apply Algorithm \ref{mutualrefinementalgorithm} to obtain the mutual refinement
        
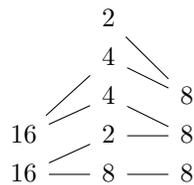
\[ \begin{tikzcd} [row sep = 1, column sep = 8]
        & & 2 \ar[-,ddrr]& & \\
           & & 4 \ar[-,drr] & &  \\
          & & 4 \ar[-,drr] & & 8 \\
         16  \ar[-,urr] \ar[-,uurr]  & & 2 \ar[-,rr] & & 8 \\
        16 \ar[rr,-] \ar[urr,-]& & 8 \ar[-,swap,rr] & & 8 \\
           &  & &  &  \\
        \end{tikzcd} \]
        \item Form the diagram 
        
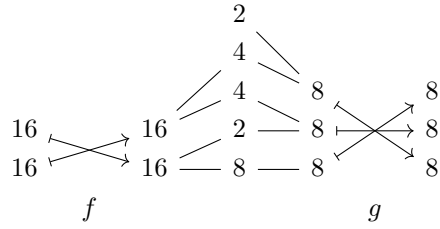
\[
\begin{tikzcd}[row sep = 1,column sep = 8] 
 & & & & 2 \ar[ddrr,-]& & & & \\
  & &   & & 4  \ar[drr,-] & &    & &\\
  & &   & & 4 \ar[drr,-]  & & 8  \ar[ddrr,mapsto] & & 8\\
16 \ar[drr,mapsto] & & 16 \ar[urr,-] \ar[rruu,-] & & 2 \ar[rr,-]& & 8  \ar[rr,mapsto] & & 8\\
16 \ar[urr,mapsto] & & 16 \ar[rr,-] \ar[urr,-] & & 8 \ar[rr,-] & & 8 \ar[uurr,mapsto] & & 8 \\
 & f & &  & & &    &g& \\
\end{tikzcd} 
\]
        
        \item Resolve the diagram to obtain
        
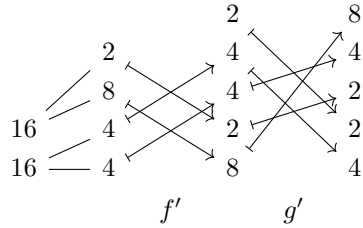
\[ \begin{tikzcd} [row sep = 1, column sep = 8]
         & & & & 2 \ar[rrddd,mapsto] & & 8 \\
         & & 2 \ar[rrdd,mapsto] & & 4 \ar[rrddd,mapsto] & & 4 \\
         & & 8 \ar[ddrr,mapsto] & & 4 \ar[urr,mapsto] & & 2 \\
        16 \ar[urr,-] \ar[uurr,-] & & 4 \ar[uurr,mapsto] & & 2 \ar[urr,mapsto] & & 2\\
        16 \ar[rr,-] \ar[urr,-] & & 4 \ar[uurr,mapsto] & & 8 \ar[uuuurr,mapsto] & & 4\\
        & & & f' & & g' & \\
        \end{tikzcd} \]

        \item Compose $f'$ and $g'$ to obtain 
        
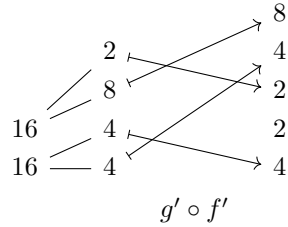
\[ \begin{tikzcd} [row sep = 1, column sep = 8]
        & & & & 8\\
        & & 2 \ar[drr,mapsto] & & 4 \\
        & & 8 \ar[uurr,mapsto] & & 2\\
        16 \ar[urr,-] \ar[uurr,-] & & 4 \ar[drr,mapsto] & & 2\\
        16 \ar[rr,-] \ar[urr,-] & & 4 \ar[uuurr,mapsto] & & 4\\
        & & & g' \circ f'& \\
        \end{tikzcd} \]
        \item Compute the associated layout 
        \[
        L_{g' \circ f'} = ((4,4),(8,2)):((16,1),(64,8)).
        \]
        \item $L_{g' \circ f'}$ is coalesced over $(16,16)$, hence
        \[
        B \circ A = ((4,4),(8,2)):((16,1),(64,8)).
        \]
    \end{enumerate} 
\end{example}

\begin{example}
    Suppose $A = (6,6):(5,60)$, and $B = (10,360):(2,60)$.
    \begin{enumerate} 
        \item Take the standard representations of $A$ and $\coalesce(B) = B$. 
        % \[ \begin{tikzcd} 
        % f:(6,6) \ar[r,"(2{,}4)"] & (5,6,2,6)\\
        % g:(10,360) \ar[r,"(2{,}4)"] & (2,10,3,360).
        % \end{tikzcd} \]
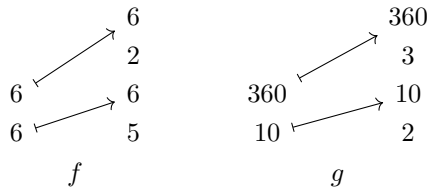
        \[ \begin{tikzcd}[row sep = 1, column sep = 8]
        & & 6 & & & & & & 360\\
 & & 2 & & & & & & 3\\
6 \ar[uurr,mapsto] & & 6 & & & & 360 \ar[uurr,mapsto] & & 10\\
6 \ar[urr,mapsto] & & 5 & & & & 10 \ar[urr,mapsto] & & 2\\        
 & f & & & & & & g & \\
        \end{tikzcd}\]
        \item Apply algorithm \ref{mutualrefinementalgorithm} to obtain the mutual refinement
        % \[ \begin{tikzcd} [row sep = 1, column sep = 8]
        %     T = (5,6,2,6) & & \rightsquigarrow & & T' = (5,(2,3),2,6)\\
        %     U = (10,360) & & \rightsquigarrow & & U' = ((5,2),(3,2,6,10))\\
        % \end{tikzcd} 
        % \]
                        
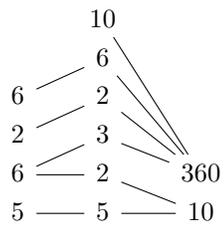
\[ \begin{tikzcd}[row sep = 1, column sep = 8]
   & & 10 & & \\
   & & 6 & & \\
 6 \ar[urr,-] & & 2 & & \\
 2 \ar[urr,-] & & 3& & \\
 6 \ar[rr,-] \ar[urr,-] & & 2 & & \ar[ull,-] \ar[uull,-] \ar[uuull,-] \ar[uuuull,-] 360 \\
 5 \ar[rr,-]& & 5 & & \ar[ull,-] \ar[ll,-] 10 \\        
& & & & \\
        \end{tikzcd}\]
        \item Form the diagram 
                
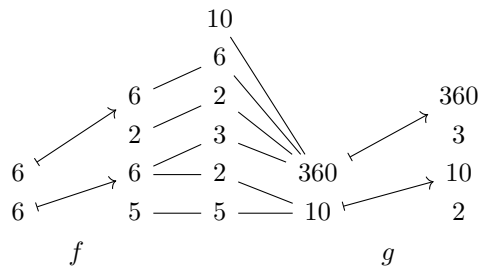
\[ \begin{tikzcd}[row sep = 1, column sep = 8]
& &   & & 10 & & & & \\
& &   & & 6 & & & & \\
& & 6 \ar[urr,-] & & 2 & & & & 360\\
 & & 2 \ar[urr,-] & & 3& & & & 3\\
6 \ar[uurr,mapsto] & & 6 \ar[rr,-] \ar[urr,-] & & 2 & & \ar[ull,-] \ar[uull,-] \ar[uuull,-] \ar[uuuull,-] 360 \ar[uurr,mapsto] & & 10\\
6 \ar[urr,mapsto] & & 5 \ar[rr,-]& & 5 & & \ar[ull,-] \ar[ll,-] 10 \ar[urr,mapsto] & & 2\\        
 & f & & & & & & g & \\
        \end{tikzcd}\]
        \item Resolve the diagram to obtain
        % \[ \begin{tikzcd}[column sep = 12]
        % ((2,3),6) \ar[rrr,"(2{,}3{,}5)"] \ar[d,two heads] \arrow[drrr, phantom, "\lrcorner", very near start] & & &  (5,(2,3),2,6) \ar[r,rightarrowtail] \ar[d,two heads]& ((5,2),(3,2,6,10)) \ar[rrr,"(2{,}3{,}5{,}6{,}7{,}8)"] \ar[d,two heads]& & & (2,(5,2),3,(3,2,6,10) ) \ar[d,two heads] \arrow[dlll, phantom, "\llcorner", very near start]\\
        % (6,6) \ar[rrr,swap,"(2{,}4)"] & & & (5,6,2,6) & (10,360) \ar[rrr,swap,"(2{,}4)"] & & & (2,10,3,360)
        % \end{tikzcd} \]
        
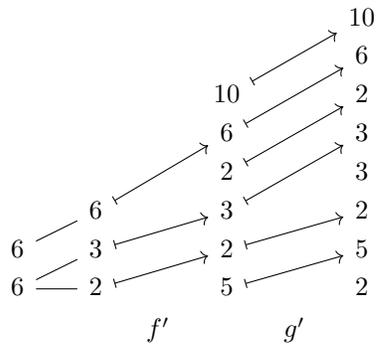
\[ \begin{tikzcd} [row sep = 1, column sep = 8]
         & &  & &   & & 10\\
         & &  & &   & & 6 \\
         & &  & &10 \ar[uurr,mapsto]& & 2 \\
        & &   & & 6 \ar[uurr,mapsto]& & 3 \\
        & &   & & 2 \ar[uurr,mapsto]& & 3 \\
        & &  6 \ar[uurr,mapsto] & & 3 \ar[uurr,mapsto] & & 2 \\
       6 \ar[rru,-]& &  3 \ar[urr,mapsto] & & 2 \ar[urr,mapsto] & & 5 \\
       6  \ar[rr,-] \ar[urr,-] & & 2 \ar[urr,mapsto] & & 5 \ar[urr,mapsto]  & & 2 \\
        & & & f' & & g' & 
        \end{tikzcd} \]
        \item Compose $f'$ and $g'$ to obtain 
        
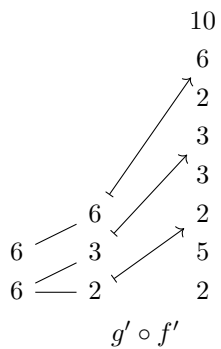
\[ \begin{tikzcd} [row sep = 1, column sep = 8]
        & &  & & 10\\
        & &  & & 6 \\
        & &  & & 2 \\
        & & & & 3 \\
        & & & & 3 \\
        & & 6 \ar[uuuurr,mapsto]& & 2 \\
        6 \ar[rru,-]& &  3 \ar[uuurr,mapsto] & & 5 \\
        6  \ar[rr,-] \ar[urr,-] & &  2 \ar[uurr,mapsto] & & 2 \\
        & &  & \mathclap{g '\circ f'} & 
        \end{tikzcd} \]
        \item Compute the associated layout 
        \[
        L_{g' \circ f'} = ((2,3),6):((10,60),360).
        \]
        \item The layout $L_{g' \circ f'}$ is coalesced over $(6,6)$, so 
        \[
        B \circ A = ((2,3),6):((10,60),360).
        \]
    \end{enumerate} 
\end{example}

\subsection{More general compositions}

The graphical calculus we have developed naturally extends to compute the composition $B \circ A$ of a tractable layout $A$ with an arbitrary CuTe layout $B$. Informally, we do this by allowing our tuples to have entries in $\mathbb{Q}_{>0} \supset \mathbb{Z}_{>0}$. We illustrate this extension with an example computation.

Consider the layouts $A = (4,4):(4,1)$ and $B = (8,8):(3,7)$. The layout $A$ is tractable, and its standard representation is the tuple morphism $f$ shown below. 
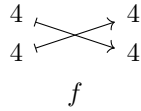
\[ \begin{tikzcd} [row sep = 1, column sep = 8]
4 \ar[drr,mapsto] & & 4\\
4 \ar[urr,mapsto] & & 4\\
 & f & \\
\end{tikzcd} \]
The layout $B$ is not tractable, but we may still depict $B$ using the diagram 
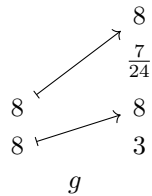
\[ \begin{tikzcd} [row sep = 1, column sep = 8]
 & & 8\\
 & & \frac{7}{24}\\
8 \ar[uurr,mapsto] & & 8 \\
8 \ar[urr,mapsto] & & 3\\
  &g& \\
\end{tikzcd} \]
This diagram does not correspond to an honest tuple morphism since the ``codomain tuple" $(3,8,\frac{7}{24},8) $ has non-integer entries. However, it still encodes the layout $B$ via the usual prefix product formula, and is still admissible as an input to our composition algorithm: We can apply Algorithm \ref{mutualrefinementalgorithm} to obtain the mutual refinement
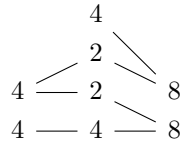
\[ \begin{tikzcd} [row sep = 1, column sep = 8]
 & & 4 & & \\
 & & 2 & & \\
4 \ar[rr,-] \ar[urr,-] & & 2 & & \ar[llu,-] \ar[lluu,-] 8\\
4 \ar[rr,-] & & 4 & & \ar[ll,-] \ar[llu,-] 8\\
\end{tikzcd} \]
form the diagram 
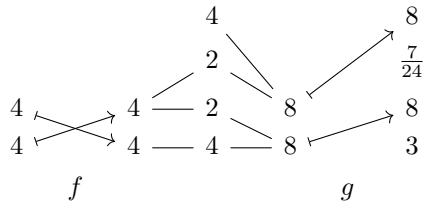
\[ \begin{tikzcd} [row sep = 1, column sep = 8]
& & & & 4 & & & &  8\\
& & & & 2 & & & & \frac{7}{24}\\
4 \ar[drr,mapsto] & & 4 \ar[rr,-] \ar[urr,-] & & 2 & & \ar[llu,-] \ar[lluu,-] 8 \ar[uurr,mapsto] & & 8 \\
4 \ar[urr,mapsto] & & 4 \ar[rr,-] & & 4 & & \ar[ll,-] \ar[llu,-] 8 \ar[urr,mapsto] & & 3\\
& f & & & & & & g & \\
\end{tikzcd} \]
resolve this diagram to obtain 
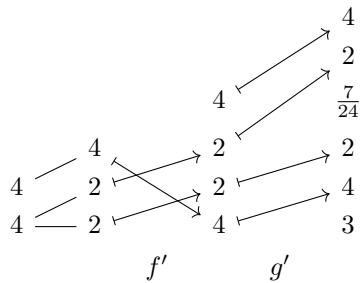
\[ \begin{tikzcd} [row sep = 1, column sep = 8]
 & & & & & & 4 \\
 & & & & & & 2 \\
 & & & & 4 \ar[uurr,mapsto] & & \frac{7}{24} \\
 & & 4 \ar[ddrr,mapsto] & & 2 \ar[uurr,mapsto] & & 2 \\
4 \ar[urr,-] & & 2 \ar[urr,mapsto] & & 2 \ar[urr,mapsto] & & 4\\
4 \ar[rr,-] \ar[urr,-] & & 2 \ar[urr,mapsto] & & 4 \ar[urr,mapsto] & & 3 \\
 & & & f' & & g' & \\
\end{tikzcd} \]
and compose $f'$ and $g'$ to obtain the diagram
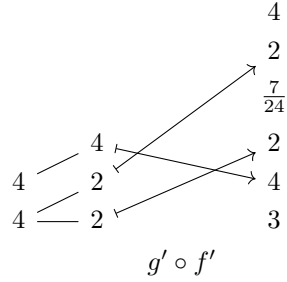
\[ \begin{tikzcd} [row sep = 1, column sep = 8]
 & & & & 4 \\
 & & & & 2 \\
 & & & & \frac{7}{24} \\
 & & 4 \ar[drr,mapsto] & & 2 \\
4 \ar[urr,-] & & 2 \ar[uuurr,mapsto] & & 4\\
4 \ar[rr,-] \ar[urr,-] & & 2 \ar[uurr,mapsto] & & 3\\
 & & & g' \circ f' & \\
\end{tikzcd} \]
The encoded layout is $((2,2),4):((12,7),3)$, which is coalesced over $(4,4)$, so we conclude that 
\[
B \circ A = ((2,2),4):((12,7),3).
\]

\subsection{Admissibility for composition}

In \cite{shah2024layout}, the author introduces the notion of {\it admissibility for composition}, which is a sufficient condition for the composition $B \circ A$ of layouts $A$ and $B$ to exist. Let's recall the definition of admissibility for composition. As in \cite{shah2024layout}, we restrict our attention to flat layouts with no shape entries equal to $1$, and we assume that the first layout in our composition has no strides equal to $0$. 

\begin{definition}\label{admissibilityforcomposition}
    Suppose
    \[
    A = (s_1,\dots,s_m):(d_1,\dots,d_m)
    \]
    is a flat layout with no $s_i = 1$ and no $d_i = 0$. Suppose $B$ is a flat layout with 
    \[
    \shape(B) = (u_1,\dots,u_p).
    \]
    We say $A$ and $B$ are {\it admissible for composition} if the following conditions hold.
    \begin{enumerate}
        \item For each $ 1 \leq i \leq m$, there exists $1 \leq k \leq \ell \leq p$ such that 
        \begin{enumerate}
            \item $u_1 \cdots u_{k-1}$ divides $d_i$, 
            \item $d_i$ divides $u_1 \cdots u_k$ (properly if $k < p$),
            \item $u_1 \cdots u_{\ell-1}$ divides $s_id_i$, 
            \item $s_i d_i$ divides $u_1 \cdots u_\ell$ (properly if $\ell < p$).
        \end{enumerate}
        \item The intervals 
        \[
        [d_i, d_i(s_i-1)] \cap [1,s_1 \cdots s_{m-1})
        \]
        are pairwise disjoint.
    \end{enumerate}
\end{definition}

\begin{remark}
    The indices $k,\ell$ in the definition above are referred to as ``division indices" in \cite{shah2024layout}.
\end{remark}

\begin{remark}
    In \cite{shah2024layout}, the author works with the ``extended layout function" of a layout, which may be considered as the layout function of the layout obtained by replacing the final shape entry $s_m$ with $\infty$. The definition we give here is the appropriate analogue for working with ordinary layout functions.
\end{remark}

\begin{lemma}\label{prefixproductdivisibilitylemma}
    Suppose $T = (t_1,\dots,t_n)$ and $U = (u_1, \dots , u_p)$ are tuples of positive integers, and suppose $(T,U)$ admits a mutual refinement. Then for any prefix products $t_1 \cdots t_j$ and $u_1 \cdots u_k$ of $T$ and $U$, respectively, either 
    \begin{enumerate}
        \item $t_1 \cdots t_{j}$ is greater than $u_1 \cdots u_{k}$, or 
        \item $t_1 \cdots t_{j}$ divides $u_1 \cdots u_{k}$.
    \end{enumerate}
\end{lemma}

\begin{proof}
    Let's choose some mutual refinement $(T',U')$ of $(T,U)$, and write $(u_1',\dots,u_{p'}')$ for the flattening of $U'$. Any prefix product of $T$ or $U$ is also a prefix product of $U'$, and since prefix products of a fixed tuple of positive integers satisfy $x \leq y \Rightarrow x \mid y$, the result follows.
\end{proof}

\begin{theorem}
    Suppose $A$ is a flat tractable layout with no shape entries equal to $1$ and no stride entries equal to $0$. Suppose $B$ is a flat tractable layout. Let $f:S \to T$ and $g:U \to V$ denote the standard representation of $A$ and $\coalesce(B)$, respectively. If $T$ and $U$ admit a mutual refinement, then $A$ and $B$ are admissible for composition.
\end{theorem}

\begin{proof}
    Let's write $S = (s_1,\dots,s_m)$, $T = (t_1,\dots,t_n)$, $U = (u_1,\dots,u_p)$, and let's write $\alpha: \langle m \rangle_* \to \langle n \rangle_*$ for the map over which $f$ lies. We need to check that the conditions from Definition \ref{admissibilityforcomposition} hold.
    \begin{enumerate}
        \item Suppose $1 \leq i \leq m$. Then $d_i = t_1 \cdots t_{j-1}$ for some $j$, namely $j = \alpha(i)$. Suppose we have a mutual refinement $(T', U')$ of $(T,U)$, and write $(T')^\flat = (t_1',\dots,t_{n'}')$ and $(U')^\flat = (u_1',\dots,u_{p'}')$. 
        \begin{itemize}
            \item (a) and (b): Since $T'$ refines $T$, there there exists some $1 \leq a \leq n'$ such that 
            \[d_i = t_1 \cdots t_{j-1} = t_1' \cdots t_{a}' = u_1' \cdots u_a'.\]
            Take the maximal $k \in \langle p \rangle$ such that $u_1 \cdots u_{k-1} \leq u_1' \cdots u_a'$. 
            \begin{itemize} 
                \item Suppose $k < p$. We observe that 
                \[
                u_1 \cdots u_{k-1} \leq d_i < u_1 \cdots u_{k}.
                \]
                where the second inequality holds by maximality of $k \in \langle p \rangle$. Lemma \ref{prefixproductdivisibilitylemma} implies that $u_1 \cdots u_{k-1}$ divides $ d_i$ and $d_i$ divides $u_1 \cdots u_{k}$ properly. 
                \item Suppose $k = p$. We observe that 
                \[
                u_1 \cdots u_{k-1} \leq d_i = t_1 \cdots t_{j-1} < t_1 \cdots t_n \leq u_1 \cdots u_{p} = u_1 \cdots u_k.
                \]
                Lemma \ref{prefixproductdivisibilitylemma} implies that $u_1 \cdots u_{k-1}$ divides $ d_i$ and $d_i$ divides $u_1 \cdots u_{k}$ (properly, though we don't require this).
            \end{itemize}

            \item (c) and (d): Again, since $T'$ refines $T$, there exists some $1 \leq b \leq n'$ such that 
            \[
            s_id_i = t_1 \cdots t_j = t_1' \cdots t_b' = u_1' \cdots u_b'.
            \]
            Take the maximal $\ell \in \langle p \rangle $ such that $u_1 \cdots u_{\ell - 1} \leq u_1' \cdots u_b'$.
            \begin{itemize} 
                \item Suppose $\ell < p$. We observe that 
                \[
                u_1 \cdots u_{\ell-1} \leq s_i d_i < u_1 \cdots u_{\ell}.
                \]
                where the second inequality holds by maximality of $\ell \in \langle p \rangle$. Lemma \ref{prefixproductdivisibilitylemma} implies that $u_1 \cdots u{\ell-1}$ divides $ s_id_i$ and $s_id_i$ divides $u_1 \cdots u_{\ell}$ properly. 
                \item Suppose $\ell = p$. We observe that 
                \[
                u_1 \cdots u_{\ell-1} \leq s_id_i = t_1 \cdots t_{j} \leq t_1 \cdots t_n \leq u_1 \cdots u_{p} = u_1 \cdots u_k.
                \]
                Lemma \ref{prefixproductdivisibilitylemma} implies that $u_1 \cdots u_{k-1}$ divides $ d_i$ and $d_i$ divides $u_1 \cdots u_{k}$.
            \end{itemize}
        \end{itemize}
        \item For any $i \neq i'$ in $\langle m \rangle$, we have  $d_i = t_1 \cdots t_{j-1}$, $s_i = t_j$, $d_{i'} = t_1 \cdots t_{j'-1}$, and $s_{i'} = t_{j'}$, where $j = \alpha(i)$ and $j' = \alpha(i')$. We then have 
        \[
        [d_i, d_i(s_i-1)] = [t_1 \cdots t_{j-1} , t_1 \cdots t_{j-1}(t_j - 1)]
        \]
        and 
        \[
        [d_{i'}, d_{i'}(s_{i'}-1)] = [t_1 \cdots t_{{j'}-1} , t_1 \cdots t_{{j'}-1}(t_{j'} - 1)]
        \]
        If $j' > j$, then 
        \[
        t_1 \cdots t_{j-1}(t_j - 1) < t_1 \cdots t_{{j'}-1}
        \]
        so the intervals do not overlap, and similarly if $j < j'$.
    \end{enumerate}
\end{proof}

\section{Logical division and logical product}
In this section we illustrate how the composition algorithm \ref{tractablelayoutcompositionalgorithm} can be used to compute logical division and logical product. 
\subsection{Logical division examples}
Recall that if $A$ and $B$ are layouts, the logical division $A \oslash B$ is defined as 
\[
A \oslash B = A \circ (B,B^c)
\]
where
\[
B^c = \comp(B, \size(A)).
\]

\begin{example}
    Suppose we want to compute the logical division $A \oslash B$ where $A = (8,8):(8,1)$ and $B = (2,2):(1,4)$. Then we can write $A = L_g$, $B = L_h$ and $B^c = L_{h^c}$ where $f$ and $f^c$ are the tuple morphisms shown below.
    
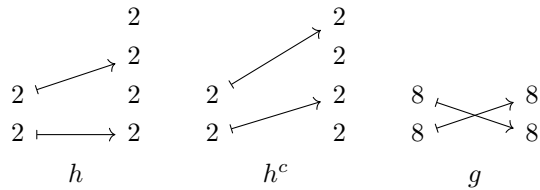
\[ \begin{tikzcd} [row sep = 1, column sep = 8]
     & & 2 & & & & 2 & & & & \\
     & & 2 & & & & 2 & & & & \\
    2 \ar[urr,mapsto] & & 2 & & 2 \ar[uurr,mapsto] & & 2 & & 8 \ar[drr,mapsto] & & 8\\
    2 \ar[rr,mapsto] & & 2 & & 2 \ar[urr,mapsto] & & 2 & & 8 \ar[urr,mapsto] & & 8\\
     & h & & & & h^c & & & & g &   \\
    \end{tikzcd} \]
    It follows that $(B,B^c)$ is encoded by the nested tuple morphism $f = (h,h^c)$ shown below. 
    \[ \begin{tikzcd} [row sep = 1, column sep = 8]
      & & 2 \ar[rr,mapsto] & & 2 \\
      & & 2 \ar[drr,mapsto] & & 2\\
    4 \ar[urr,-] \ar[uurr,-] & & 2 \ar[urr,mapsto] & & 2 \\
    4 \ar[rr,-] \ar[urr,-]  & & 2 \ar[rr,mapsto] & & 2\\
    & & & f & 
    \end{tikzcd} \]
    We then proceed with our composition algorithm as before. We use algorithm \ref{mutualrefinementalgorithm} to find the mutual refinement 
    \[ \begin{tikzcd} [row sep = 1, column sep = 8]
       & & 4 & & \\
       2 \ar[rr,-] & & 2 & & \\
       2 \ar[rr,-]& & 2 & & \\
     2 \ar[rr,-] & & 2 & & \ar[uull,-] \ar[uuull,-] 8\\
     2 \ar[rr,-] & & 2 & & \ar[ll,-] \ar[ull,-] \ar[uull,-] 8\\
    \end{tikzcd} \]
    form the diagram 
    \[ \begin{tikzcd} [row sep = 1, column sep = 8]
      & & & & & & 4 & & & & \\
      & & 2 \ar[rr,mapsto] & & 2 \ar[rr,-] & & 2 & & & & \\
      & & 2 \ar[drr,mapsto] & & 2 \ar[rr,-]& & 2 & & & & \\
    4 \ar[urr,-] \ar[uurr,-] & & 2 \ar[urr,mapsto] & & 2 \ar[rr,-] & & 2 & & \ar[uull,-] \ar[uuull,-] 8 \ar[drr,mapsto] & & 8\\
    4 \ar[rr,-] \ar[urr,-]  & & 2 \ar[rr,mapsto] & & 2 \ar[rr,-] & & 2 & & \ar[ll,-] \ar[ull,-] \ar[uull,-] 8 \ar[urr,mapsto] & & 8\\
    & & & f & & & & & &  g & 
    \end{tikzcd} \]
    resolve the diagram to obtain
    \[ \begin{tikzcd} [row sep = 1, column sep = 8]
    & & & & 4 \ar[dddrr,mapsto]& & 2 \\
      & & 2 \ar[rr,mapsto] & & 2  \ar[dddrr,mapsto] & & 2\\
      & & 2 \ar[drr,mapsto] & & 2  \ar[uurr, mapsto] & & 2\\
    4 \ar[urr,-] \ar[uurr,-] & & 2 \ar[urr,mapsto] & & 2 \ar[uurr,mapsto] & & 4\\
    4 \ar[rr,-] \ar[urr,-]  & & 2 \ar[rr,mapsto] & & 2  \ar[uurr,mapsto] & & 2\\
    & & & f & & g' & 
    \end{tikzcd} \]
    and compose $f$ and $g'$ to obtain
    \[
    \begin{tikzcd} [row sep = 1, column sep = 8]
      & &   & & 2\\
      & & 2 \ar[dddrr,mapsto] & & 2\\
      & & 2 \ar[urr,mapsto] & & 2\\
    4 \ar[urr,-] \ar[uurr,-] & & 2 \ar[uuurr,mapsto] & & 4\\
    4 \ar[rr,-] \ar[urr,-]  & & 2 \ar[uurr,mapsto] & & 2\\
     & & & g' \circ f & \\
    \end{tikzcd} 
    \]
    The layout encoded by this nested tuple morphism is 
    \[
    L_{g' \circ f} = ((2,2),(2,2)):((8,32),(16,1))
    \]
    which is coalesced over $((2,2),(2,2))$, so we conclude that 
    \[
    A \oslash B = ((2,2),(2,2)):((8,32),(16,1)).
    \]
\end{example}

\subsection{Logical product examples}

Recall that if $A$ and $B$ are layouts, the logical product $A \otimes B$ is defined as 
\[
A \otimes B = (A, A^c \circ B)
\]
where 
\[
A^c= \comp(A, \size(A)\cdot \cosize(B)).
\]
In particular, if we want to compute $A \otimes B$ by hand, it suffices to compute $A^c \circ B$, and then concatenate the result with $A$. 

\begin{example}
    Suppose we want to compute the logical product $A \otimes B$ where $A = (2,2):(1,2)$ and $B = (5,5):(5,1)$. Then 
    \begin{align*}
    A^c & = \comp(A,\size(A)\cdot \cosize(B))\\
     & = \comp(A, 100)\\
     & = (25):(4). 
    \end{align*}
    We proceed as in the previous section.
    \begin{enumerate}
        \item Take the standard representations of $B$ and $\coalesce(A^c) = A^c$.
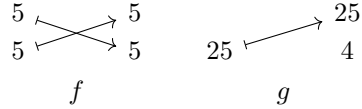
        \[ \begin{tikzcd} [row sep = 1, column sep = 8]
        5 \ar[drr,mapsto] & & 5 & & & & 25\\
        5 \ar[urr,mapsto] & & 5 & & 25 \ar[urr,mapsto] & & 4\\
        & f & & & & g & \\
        \end{tikzcd} \]
        \item Apply Algorithm \ref{mutualrefinementalgorithm} to obtain the mutual refinement 
        
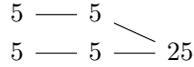
\[ \begin{tikzcd} [row sep = 1, column sep = 8]
        5 \ar[rr,-] & & 5 & & \\
        5 \ar[rr,-] & & 5 & & \ar[ll,-] \ar[ull,-] 25 \\
        \end{tikzcd} \]
        \item Form the diagram 
        
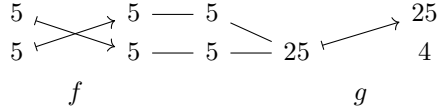
\[ \begin{tikzcd} [row sep = 1, column sep = 8]
        5 \ar[drr,mapsto] & & 5 \ar[rr,-] & & 5 & & & & 25 \\
        5 \ar[urr,mapsto] & & 5 \ar[rr,-] & & 5 & & \ar[ll,-] \ar[ull,-] 25 \ar[urr,mapsto] & & 4\\
         & f & & & & & & g & \\
        \end{tikzcd} \]
        \item Resolve the diagram to obtain
        
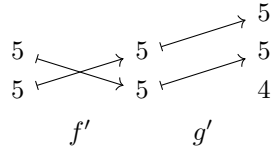
\[ \begin{tikzcd}[row sep = 1, column sep = 8]
         & & & & 5\\
        5 \ar[drr,mapsto] & & 5 \ar[urr,mapsto] & & 5 \\
        5 \ar[urr,mapsto] & & 5 \ar[urr,mapsto] & & 4\\
        & f' & & g' & \\
        \end{tikzcd} \]
        \item Compose $f'$ and $g'$ to obtain
        
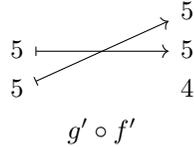
\[ \begin{tikzcd}[row sep = 1, column sep = 8]
         & & 5\\
        5 \ar[rr,mapsto] & & 5\\
        5 \ar[uurr,mapsto] & & 4\\
        & g' \circ f' & \\
        \end{tikzcd} \]
        \item Compute the encoded layout
        \[
        L_{g' \circ f'} = (5,5):(20,4).
        \]
        \item The layout $L_{g' \circ f'}$ is coalesced over $(5,5)$, so 
        \[
        A^c \circ B = (5,5):(20,4).
        \]
    \end{enumerate}
We conclude that 
        \[
        A \otimes B = ((2,2),(5,5)):((1,2),(20,4)).
        \]

\end{example}

\appendix

\chapter{An introduction to categories}\label{categorytheoryappendix}
Throughout this work, we freely use the language of \emph{categories} which are mathematical objects which abstract the notion of {\it morphisms} and their {\it composition}. The purpose of this appendix is to provide a concise and user-friendly introduction to the basics of categories. In particular, we aim to the answer the following questions:
\begin{enumerate}
    \item What is a category?
    \item What is a functor?
\end{enumerate}
Those capable of answering these questions with confidence, and with examples in mind, will be able to understand the most important conepts and constructions in the current work. For those interested in learning the more advanced concepts from category theory, such as {\it natural transformations}, {\it adjunctions}, and {\it (co)limits}, we recommend \cite{categoriesfortheworkingmathematician}.

\section{What is a category?}
We begin by addressing the first question. Before giving a definition, let's consider a motivating example. Suppose $X$ and $Y$ are sets. A {\it function} $f:X \to Y$ assigns to each element $x \in X$ some element $f(x) \in Y$. We refer to $X$ as the {\it domain} of $f$ and to $g$ as the {\it codomain} of $Y$.

\begin{example}\label{functionexample1} There is a function $f:\mathbb{Z} \to \mathbb{Z}$ given by 
\[
f(x) = 2x.
\]
\end{example}
\begin{example}\label{functionexample2} There is a function $g:\mathbb{Z} \to \mathterm{Bool}$, where $\mathterm{Bool} = \{\mathbf{True},\mathbf{False}\}$, given by 
\[
g(x) = \begin{cases}\mathbf{True} & x \text{ is even,}\\
\mathbf{False} & x \text{ is odd.}
\end{cases}.
\]
\end{example}
If $f:X \to Y$ and $g:Y \to Z$ are functions, then we can {\it compose} $f$ and $g$: The composite of $f$ and $g$ is the function $g \circ f:X \to Z$ given by 
\[
(g \circ f)(x) = g(f(x)).
\]
\begin{example}If $f$ and $g$ are the functions of Examples \ref{functionexample1} and \ref{functionexample2}, then the composite $g \circ f:\mathbb{Z} \to \mathterm{Bool}$ is given by 
\[
(g \circ f)(x) = \mathbf{True}.
\]
\end{example}
Composition of functions satisfies two essential properties. First, composition is {\it associative}: if $f$ and $g$ are composable, and $g$ and $h$ are composable, then 
    \[
    h \circ (g \circ f) = (h \circ g) \circ f.
    \]
Second, every set $X$ has an {\it identity} function $\id_X:X \to X$ given by 
    \[\id_X(x) = x.\]
    If $f:X \to Y$ is any function, then precomposing with $\id_X$ or post-composing with $\id_Y$ leaves the function $f$ unchanged:
    \[
    f \circ \id_X = f = \id_Y \circ f.
    \]

In pure and applied mathematics, there are many instances where we have some collection of {\bf objects}, and {\bf morphisms} between those objects, which have the same formal behavior of sets and functions: morphisms can be composed in an associative fashion, and objects admit identity morphisms. While functions between sets are the prototypical example, the objects in a category need not be sets, and the morphisms in a category need not be functions. We will see many such examples later on. To capture this recurring structure, we define the notion of a {\bf category}.

\begin{definition}
A {\it category} $\catstyle{C}$ consists of
\begin{enumerate}[left = 0pt]
    \item a collection of {\it objects}:
    \[\ob(\catstyle{C})=\{X,Y,Z,\dots \}.\]
    These objects may be sets, tuples, numbers, vector spaces, matrices, or some other mathematical structure, depending on the category $\catstyle{C}$.
    \item a collection of {\it morphisms} between those objects:
    \[
    \mor(\catstyle{C}) = \{f,g,h,\dots\}.
    \]
    Each morphism $f:X \to Y$ in $\catstyle{C}$ has a {\it domain} $X$
    and a {\it codomain} $Y$, which are objects in $\catstyle{C}$.
    \item a {\it composition rule}: If $f:X \to Y$ and $g:Y \to Z$ are morphisms in $\catstyle{C}$, then there is a morphism 
    \[g \circ f:X \to Z\]
    called the {\it composite} of $f$ and $g$. Composition of morphisms in $\catstyle{C}$ is {\it associative}, in that
    \[h \circ (g \circ f) = (h \circ g) \circ f,\]
    when defined.
    \item {\it identity morphisms}: If $X$ is an object in $\catstyle{C}$, then there is a morphism \[\id_X:X \to X\]
    called the {\it identity morphism} on $X$. If $f:X \to Y$ is any morphism in $\catstyle{C}$, then 
    \begin{align*}
    f \circ \id_X & = f = \id_Y \circ f.
    \end{align*}
    \end{enumerate}
\end{definition}

Let's take a look at some important examples of categories. We begin with the motivating example.
\begin{example}
        There is a category $\catstyle{Set}$ whose objects are sets, and whose morphisms are functions. The composition of morphisms is given by functional composition:
        \[
        (g \circ f)(x) = g(f(x))
        \]
        and the identity morphism on a set $X$ is the identity function 
        \[
        \id_X(x) = x.
        \]
\end{example}

\begin{example}
        There is a category $\catstyle{Vect}$ whose objects are the vector spaces $\mathbb{R}^n$ for $n \geq 0$, and whose morphisms are {\it matrices}. Specifically, a morphism
        \[A: \mathbb{R}^n \to \mathbb{R}^m\]
        in $\catstyle{Vect}$ is a $m \times n$ matrix $A$. Composition in $\catstyle{Vect}$ is given by taking matrix products:
        \[
        B \circ A = BA,
        \]
        and the identity morphism on $\mathbb{R}^n$ is the $n \times n$ matrix 
        \[
        \id_{\mathbb{R}^n} = I_n =  \begin{bmatrix}
            1 & 0 & \cdots & &  0\\
            0 & 1 &  & &  \\
            \vdots  & & \ddots& &  \vdots \\
             & & & 1 & 0 \\
             0&  & \cdots  & 0 & 1
        \end{bmatrix}.
        \]
\end{example}

\begin{example}
    There is a category $\catstyle{Div}$ whose objects are integers $a \geq 1$, and in which there is a unique morphism 
    \[
    \mathterm{div}_{a}^b:a \to b
    \]
    if $a$ divides $b$. If $a$ divides $b$ and $b$ divides $c$, then $a$ divides $c$, which means that we have a well defined composition rule 
    \[
    \mathterm{div}_b^c \circ \mathterm{div}_a^b = \mathterm{div}_a^c,
    \]
    and the identity morphism 
    \[
    \id_a = \mathterm{div}_a^a
    \]
    exists since every positive integer $a$ divides itself. 
\end{example}

In addition to the definition of a category, there are a few important categorical concepts that we need to understand. For instance, it is important to understand the notion of an {\it isomorphism}, which generalizes the notion of a {\it bijection} of sets.
\begin{definition}
Suppose $\catstyle{C}$ is a category, and suppose $f:X \to Y$ is a morphism in $\catstyle{C}$. We say $f$ is an {\it isomorphism} if there exists a morphism $f^{-1}: Y \to X$ in $\catstyle{C}$ such that 
\begin{enumerate}
\item $f^{-1} \circ f = \id_X$, and 
\item $f \circ f^{-1} = \id_Y$.
\end{enumerate}
\end{definition}

\begin{example}
    In the category $\catstyle{Set}$, an isomorphism is a {\it bijection}: a function $f:X \to Y$ such that for each $y \in Y$, there exists a unique $x \in X$ with $f(x) = y$. For example, the function $f:\mathbb{Z} \to \mathbb{Z}$ given by 
    \[f(x) = x+10\]
    is a bijection, with inverse $f^{-1}:\mathbb{Z} \to \mathbb{Z}$ given by 
    \[
    f^{-1}(x) = x-10.
    \]
\end{example}

\begin{example}
    In the category $\catstyle{Vect}$, an isomorphism is an {\it invertible matrix}. For example, the matrix 
    \[
    A = \begin{bmatrix}
        3 & 2\\
        1 & 1\\
    \end{bmatrix}
    \]
    is invertible with inverse 
    \[
    A^{-1} = \begin{bmatrix}
        1 & -2\\
        -1 & 3\\
    \end{bmatrix}
    \]
    since 
    \[
    A^{-1}A= \begin{bmatrix}
        1 & 0\\
        0 & 1\\
    \end{bmatrix} = AA^{-1}
    \]
\end{example}

\begin{example}
    In the category $\catstyle{Div}$, the only isomorphisms are the identity morphisms 
    \[\id_a = \mathterm{div}_a^a.\]
    This is because if $a$ divides $b$ and $b$ divides $a$, then $a = b$. 
\end{example}

\section{What is a functor?}
Next, we turn our attention to the second question.
\begin{definition} 
Suppose $\catstyle{C}$ and $\catstyle{D}$ are categories. A {\it functor} $F:\catstyle{C} \to \catstyle{D}$ consists of  
\begin{enumerate}
    \item for each object $X$ in $\catstyle{C}$, an object $FX$ in $\catstyle{D}$, and
    \item for each morphism $f:X \to Y$ in $\catstyle{C}$, a morphism 
    \[Ff:FX \to FY\]
    in $\catstyle{D}$,
\end{enumerate}
satisfying the following properties:
\begin{enumerate}
    \item $F$ is compatible with composition: If $f$ and $g$ are composable morphisms in $\catstyle{C}$, then 
    \[
    F(g \circ f) = Fg \circ  Ff.
    \]
    \item $F$ is compatible with identities: If $X$ is an object in $\catstyle{C}$, then 
    \[
    F\id_X = \id_{FX}.
    \]
\end{enumerate}
\end{definition}

\begin{example}
    There is a functor $F:\catstyle{Div} \to \catstyle{Set}$ defined as follows. On objects, $F$ is given by
    \begin{align*}
    Fa & = [0,a] = \{x \in \mathbb{R} \mid 0 \leq x \leq a\}.
    \end{align*}
    and on morphisms, $F$ is given by 
    \[
    \begin{tikzcd}
        F\mathterm{div}_a^b(x) = \frac{b}{a}\cdot x.
    \end{tikzcd}
    \]
    Let's verify that $F$ is a functor.
    \begin{enumerate}
        \item $F$ is compatible with composition: If $a$ divides $b$ and $b$ divides $c$, then
        \begin{align*}
        (F \mathterm{div}_b^c \circ F \mathterm{div}_a^b)(x) = F \mathterm{div}_b^c ( F \mathterm{div}_a^b(x)) & = \frac{c}{b}\cdot(\frac{b}{a}\cdot x)  = \frac{c}{a}\cdot x  = F\mathterm{div}_a^c(x).
        \end{align*}
        \item $F$ is compatible with identities: If $a \geq 1$, then 
        \begin{align*}
            F\id_a(x) = F \mathterm{div}_a^a(x) = \frac{a}{a} \cdot x = \id_{Fa}(x).
        \end{align*}
    \end{enumerate}
\end{example}

\printbibliography

\end{document}